%% file: paper.tex
\documentclass[11pt]{article}
\usepackage[T1]{fontenc}
\usepackage[utf8]{inputenc}
\usepackage{lmodern,microtype}
\usepackage[margin=1.05in]{geometry}
\usepackage{amsmath,amssymb,amsthm,mathtools}
\usepackage{booktabs,longtable,array,graphicx,xcolor}
\usepackage[most]{tcolorbox}
\usepackage{tikz}
\usetikzlibrary{arrows.meta,positioning}
\usepackage[round,authoryear]{natbib}
\usepackage[colorlinks=true,linkcolor=blue!45!black,citecolor=blue!45!black,urlcolor=blue!55!black]{hyperref}
\usepackage{cleveref,aliascnt}
\usepackage{xurl}
\usepackage{pdflscape}
\usepackage{placeins,needspace,capt-of}
\definecolor{revisionred}{HTML}{B00020}
\definecolor{ink}{HTML}{16384D}
\definecolor{teal}{HTML}{087F8C}
\definecolor{ochre}{HTML}{B36A19}
\newtheorem{theorem}{Theorem}[section]
\newaliascnt{lemma}{theorem}
\newtheorem{lemma}[lemma]{Lemma}
\aliascntresetthe{lemma}
\crefname{lemma}{Lemma}{Lemmas}
\Crefname{lemma}{Lemma}{Lemmas}
\newaliascnt{proposition}{theorem}
\newtheorem{proposition}[proposition]{Proposition}
\aliascntresetthe{proposition}
\crefname{proposition}{Proposition}{Propositions}
\Crefname{proposition}{Proposition}{Propositions}
\newaliascnt{corollary}{theorem}
\newtheorem{corollary}[corollary]{Corollary}
\aliascntresetthe{corollary}
\crefname{corollary}{Corollary}{Corollaries}
\Crefname{corollary}{Corollary}{Corollaries}
\newaliascnt{definition}{theorem}

\aliascntresetthe{definition}
\crefname{definition}{Definition}{Definitions}
\Crefname{definition}{Definition}{Definitions}
\theoremstyle{remark}\newaliascnt{remark}{theorem}
\newtheorem{remark}[remark]{Remark}
\aliascntresetthe{remark}
\crefname{remark}{Remark}{Remarks}
\Crefname{remark}{Remark}{Remarks}
\crefname{section}{Section}{Sections}
\crefname{subsection}{Section}{Sections}
\crefname{figure}{Figure}{Figures}
\crefname{table}{Table}{Tables}
\crefname{theorem}{Theorem}{Theorems}
\crefname{appendix}{Appendix}{Appendices}
\Crefname{appendix}{Appendix}{Appendices}
\newcommand{\E}{\mathbb E}

\newcommand{\eps}{\varepsilon}
\newcommand{\hb}{h_2}

\newcommand{\one}{\mathbf 1}
\newcommand{\bits}{\{0,1\}}
\newcommand{\osc}{\operatorname{osc}}

\newcommand{\certpath}[1]{\textup{\path{#1}}}
\newtcolorbox{certbox}[1][]{enhanced,breakable,colback=ink!3,colframe=ink!65,boxrule=.5pt,arc=1pt,left=8pt,right=8pt,top=7pt,bottom=7pt,#1}
\newtcolorbox{examplebox}[1][]{enhanced,breakable,colback=teal!3,colframe=teal!60,boxrule=.5pt,arc=2pt,left=8pt,right=8pt,top=7pt,bottom=7pt,#1}

\hypersetup{pdftitle={Binary Deletion Channel Capacity to Within One Hundredth of a Bit},pdfauthor={Dimitris Papailiopoulos},pdfsubject={Uniform capacity approximation with certified numerical bounds}}
\title{\textbf{Binary Deletion Channel Capacity\\to Within One Hundredth of a Bit}}
\author{Dimitris Papailiopoulos\thanks{This work was carried out with interacting language-model agents, including GPT-5.6 Sol, GPT-6 Astra, and Fable 5.1, in Codex and Claude Code. The agents developed the mathematical arguments, designed and ran the numerical checks, reviewed the proofs, and drafted and revised the manuscript. They exchanged intermediate results and transferred work across research, verification, and writing phases. The author posed the original question, steered the research directions, and directed the presentation. The recorded project began on 25 August 2026; the uniform numerical certificate was completed on 10 September, followed by further review and numerical replay.}\\[.4em]{\normalsize Microsoft Research \& University of Wisconsin}}
\date{}
\begin{document}
\maketitle
\begin{abstract}
The exact capacity of the binary deletion channel remains unknown despite decades of work on achievable rates and converse bounds. We establish a computer-assisted approximation whose error is below $0.0095$ bits per transmitted bit, uniformly over all deletion probabilities. The mean certified error bound, with uniform weighting of deletion probability, is below $0.006522$. The estimate is the midpoint of explicit lower and upper bounds. For the converse, a stationary-source reduction is combined with finite inequalities covering every allowed input configuration. Two constructions control the unobserved input beyond a finite window: one uses a common outside survivor sequence and bounds omitted deletion patterns, while the other cancels an entropy term to make outside probabilities enter linearly. For the lower bound, finite-state inputs combine output-entropy estimates with selected disjoint counts of compatible deletion masks; independent-run inputs retain additional uncertainty about output-run boundaries. Directed numerical checks establish the finite inequalities. An analytic comparison between deletion probabilities then extends the pointwise bounds over the entire parameter range. The lower endpoint supplies rates within $0.019$ bits of capacity in the asymptotic coding sense. We give the derivations, recorded computational costs, and complete numerical inputs and programs needed to verify the result.
\end{abstract}
\clearpage
{\setlength{\parskip}{0pt}\small\tableofcontents}
\clearpage
\input{sections/introduction}
\par\bigskip\Needspace{8\baselineskip}
\part{Foundations and the approximation theorem}\label{part:foundations}
\input{sections/setting}
\input{sections/proof_overview}
\input{sections/chapter_upper}
\input{sections/survivor}
\input{sections/shared}
\input{sections/posterior}
\input{sections/small_deletion}
\input{sections/chapter_lower}
\input{sections/source_lower}
\input{sections/renewal}
\input{sections/chapter_assembly}

\input{sections/transport}
\input{sections/arithmetic}
\input{sections/reproduction}
\input{sections/discussion}
\appendix
\crefalias{section}{appendix}
\clearpage
\part*{Supplementary derivations}
\addcontentsline{toc}{part}{Supplementary derivations}
\input{sections/standard}
\input{sections/fair_input}

\input{sections/small_assembly}
\input{sections/equivalences}
\setlength{\bibsep}{3pt}
\bibliographystyle{plainnat}
\bibliography{references}
\end{document}

%% file: sections/introduction.tex
\section{Introduction}\label{sec:introduction}
The binary deletion channel independently removes each transmitted bit with probability $d$. The receiver sees the surviving bits in their original order, but does not know their former positions. Unlike an erasure, a deletion leaves no marker. Thus the channel removes both a bit and the information needed to locate the bits that follow it; \cref{fig:intro-channel} illustrates the distinction. Its capacity $C(d)$ is the largest number of information bits per transmitted bit that can be communicated with vanishing error probability as the transmitted length grows.

\input{figures/intro_channel}

The problem belongs to the early study of synchronization errors. Shannon used recovery of text with missing letters to illustrate source redundancy \citep[Section~7]{Shannon48}; this was not a formula for deletion-channel capacity. \Citet{Gallager61} studied decoding with synchronization errors, and \citet{Dob67} established a coding theorem for channels with insertions and deletions. That theorem expresses deletion-channel capacity as a limit of finite-block mutual-information maxima. Numerical work began soon afterward: \citet{VD68} already studied computer evaluation for a symbol-deletion channel. An exact expression for this capacity has remained unknown for nearly six decades. The difficulty is not independence of the deletions themselves, but the missing correspondence between input and output positions.

For example, output $0$ can result from input $00$ by deleting either bit. From input $01$, it can result only by deleting the second bit. The probability distribution of the transmitted sequence therefore affects how much uncertainty remains about its deleted positions. Long blocks are used in coding for both the binary symmetric channel and the deletion channel. The difference is that a memoryless channel with identified output positions admits a one-bit capacity optimization, whereas hidden deletion positions prevent that usual simplification. \Cref{sec:setting} defines the finite-block optimization precisely, and \cref{sec:standard} supplies the information inequalities used below.

Independent losses without position information arise in packet communication, and synchronization errors also occur in DNA synthesis and sequencing \citep{DG01,DG06,Maarouf23}. These applications often include substitutions, insertions, or multiple reads as well. The binary deletion channel isolates the cost of missing positions in a setting simple enough to state exactly, yet difficult enough that even its one-parameter capacity function remains unknown. An accurate capacity approximation gives a quantitative target against which source constructions and codes can be assessed.

\subsection{Capacity approximation}
We construct lower and upper functions $L(d)$ and $U(d)$ and use the midpoint of their interval to estimate capacity:
\begin{equation}\label{eq:intro-main}
 L(d)\le C(d)\le U(d),\qquad
 \widehat C(d)=\frac{L(d)+U(d)}2,\qquad
 |C(d)-\widehat C(d)|\le\eps_\star<\frac1{100}
 \quad(0\le d\le1).
\end{equation}
The first two inequalities put capacity in the interval. Every point in an interval is at most half its width from its midpoint; hence a width bound $U(d)-L(d)\le2\eps_\star$ proves the last inequality. The computer-assisted bound gives $\eps_\star<0.0095$ bit per input bit; \cref{thm:main} states its exact value. Its guarantee holds at every real deletion probability, including probabilities between the finitely many numerical evaluation points. \Cref{fig:intro-capacity} compares the enclosure with selected prior bounds.

\begin{figure}[!htbp]\centering
\includegraphics[width=.92\linewidth]{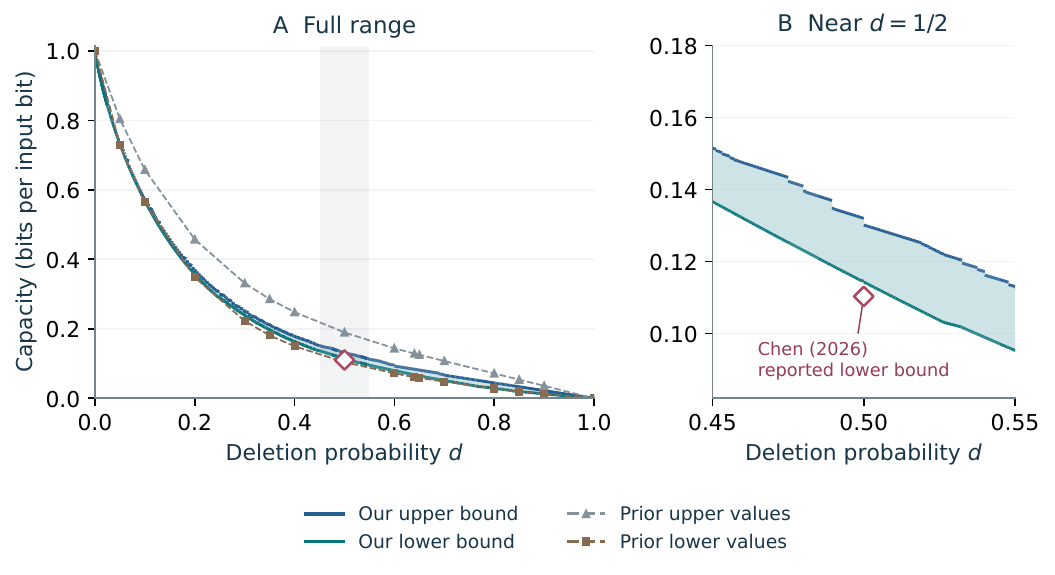}
\caption{Capacity bounds in bits per input bit. Left: our lower and upper bounds enclose capacity over the full parameter range; the gray vertical band marks the range enlarged on the right. Historical markers show the selected values of \citet{DM07,VTR13,RC23,PR26}; dashed lines between them are visual guides. The separate marker at $d=1/2$ shows the reported lower bound of \citet{Chen26}, whose computation was not independently replayed here. The right panel makes the new enclosure visible without changing the vertical units.}\label{fig:intro-capacity}
\end{figure}

The certified error varies with $d$. The mean of $(U(d)-L(d))/2$, with equal weight on all $d\in[0,1]$, is at most $0.006522$ bit per input bit. This is a bound on mean estimation error, not a measurement of the unknown error. The lower endpoint has a separate coding meaning: every nonnegative rate strictly below $L(d)$ is achievable in the Shannon sense: a sequence of increasing-length codes has error probability tending to zero. The midpoint need not be achievable. Accordingly, the midpoint guarantee is half the one-sided tolerance for the achievable lower bound. The one-bit result of \citet{etkin2008} concerns an achievable approximation to a capacity region; our statement concerns a numerical estimate of the scalar capacity function.

\subsection{Prior bounds and the contribution of this work}\label{sec:intro-related}
Lower bounds have improved by choosing input sequences with useful dependence and recovering information discarded by simpler decoding analyses. Markov inputs, in which a bit depends on a finite recent history, were studied by \citet{DG01,DG06}. Trellis methods provide another finite-state route \citep{CK15}. Run-based constructions describe the lengths of maximal equal-bit segments, such as the runs $000$ and $11$ in $00011$ \citep{DM06,DM07}. \Citet{KD10,VTR13} improved the accounting of the information such sources retain. \Citet{RC23} optimized run-length distributions and obtained stronger pointwise values as well as a lower bound proportional to $1-d$ throughout the range. \Cref{tab:intro-history} records selected numerical milestones, with the source of each value identified.

\Citet[Theorem~27]{Chen26} reports $C(1/2)>0.11032415$ for an explicit renewal source. We include this stronger reported comparison in \cref{tab:intro-history}. We have not independently replayed that computation; the public record inspected on September 12, 2026 listed a preprint but no machine-readable certificate.

\input{data/intro_history}

Upper bounds have followed several complementary routes. Revealing some deleted positions or block boundaries gives the receiver an easier channel whose capacity upper-bounds the original one \citep{DMP07,FD10}. Output-distribution bounds and convex duality offer another route \citep{DMP07,Cher19,CR19}; for suitable finite-state channels, graph-based output distributions can be combined with dynamic programming \citep{HSPKS21}. The recent calculations of \citet{PR26} improve finite auxiliary-channel bounds through a parallel Blahut--Arimoto implementation. Parameter comparisons allow computed bounds to extend beyond their evaluation points \citep{Dalai11,RD15}. These contributions address different parts of the problem; selected pointwise comparisons do not by themselves give a prior uniform approximation theorem.

We prove a uniform numerical enclosure of capacity from explicit finite inequalities. For upper bounds, the inequalities cover every input distribution; for lower bounds, finite calculations evaluate specified input processes. The principal analytic task is to replace dependence on arbitrarily distant input positions by finite expressions without restricting the allowed inputs. The principal numerical task is to verify all required expressions with controlled rounding and combine their values over the full parameter range.

The converse begins with the classical upper bound obtained by comparing the actual output distribution with a chosen auxiliary distribution \citep{DMP07,HSPKS21}. A short group of received bits can originate beyond any fixed input window. We derive two tests that cover that unknown input: the first checks every common outside survivor string and bounds the omitted deletion cases; the second cancels a conditional entropy so that the remaining expression is an average over outside strings, bounded by its largest value. Their derivation overviews define these constructions and give the full finite tests in \eqref{eq:shared-row} and~\eqref{eq:posterior-row}. A separate run-length converse follows the modified-deletion approach of \citet{KM10}. These are specific sufficient conditions for the deletion channel, not a new general duality framework.

For finite-state inputs, we lower-bound output entropy and residual deletion-mask entropy. For independent run lengths, we evaluate the classical run-group rate and restore conditional uncertainty about where output runs begin. Each complete rate uses one fixed source. The counting and conditioning principles have direct predecessors: the action-entropy identity in \citet{HOS16}, the monotone synchronization-state refinements in \citet{ISW16}, and the run-based analyses of \citet{DM07,KD10}. For renewal sources, the endpoint-window method, its monotone conditional-mutual-information increments, integer segmentation posteriors, and directed entropy evaluation were developed in \citet{Chen26}. We use these constructions in \cref{sec:renewal}. The raw-bit endpoints in \cref{sec:source} differ from renewal-run endpoints, but the underlying conditioning and counting principles are shared. Our contribution is the specified finite tests, source laws and selected event counts, together with their verified combination into a capacity approximation over the full deletion-probability interval.

Small-deletion expansions describe the endpoint near zero \citep{KMS10,KM13}; capacity-approaching coding theorems address asymptotic encoding and decoding \citep{TPFV22,Rub22,pernice2022}. Neither task by itself evaluates the unrestricted capacity function at interior parameters. Likewise, results for uniform input \citep{DSV12,PIW24,JP26}, finite block length \citep{MD26}, or a fixed number of deletions \citep{SSBY25} concern different quantities. The present work evaluates information rates; it does not supply an efficient finite encoder and decoder attaining the reported rates. Surveys provide broader context \citep{Mit09,MBT10,CR21}.

\subsection{From the methods to a uniform theorem}
The finite tests establish pointwise bounds. To obtain a uniform approximation, these bounds must also cover the probabilities between evaluation points.

For $d<1$, the normalized capacity $C(d)/(1-d)$ is nonincreasing \citep{RD15}. If $C(a)\le u$ and $C(b)\ge\ell$, then for $a\le d\le b<1$,
\[
 (1-d)\frac{\ell}{1-b}\le C(d)\le(1-d)\frac{u}{1-a}.
\]
Normalized capacity at $d$ lies between its values at $a$ and $b$; multiplying by $1-d$ proves the display. Thus one can bound the vertical difference between these capacity estimates throughout a parameter interval. Exact arithmetic checks these differences and verifies coverage. A lower bound proportional to $1-d$, valid throughout the range, handles the final interval approaching one; $C(0)=1$ and $C(1)=0$ are exact.

\Cref{tab:reading-route} gives a reading route by purpose. The short roadmap in \cref{sec:proof-overview} gives the whole argument. Each method then begins with its own derivation overview and continues directly to the complete proof and finite computation.

\begin{table}[!htbp]\centering\small
\begin{tabular}{@{}>{\raggedright\arraybackslash}p{5.5cm}>{\raggedright\arraybackslash}p{8.1cm}@{}}\toprule
Purpose & Where to read\\\midrule
Read the result at a high level & The Introduction and \cref{sec:proof-overview}; \cref{thm:main} states the guarantee.\\
Understand one upper-bound method & \Cref{sec:roadmap-shared,sec:roadmap-posterior,sec:roadmap-runs}; full proofs in \cref{sec:shared,sec:posterior,sec:small-deletion}.\\
Understand how a lower bound is computed & Finite-state inputs: \cref{sec:source-computation}. Renewal inputs: \cref{sec:renewal-baseline-computation,cor:renewal-universal-assembly}.\\
Separate the analytic proof from its numerical premises & Upper calculations: \cref{sec:shared-computation,sec:posterior-computation}. All method checks, experiments and reproduction: \cref{sec:reproduction}.\\
Understand prior work and what remains open & The contribution discussion above and the bottlenecks in \cref{sec:limits}.\\\bottomrule
\end{tabular}
\caption{Reading route through one paper. The global roadmap gives the whole argument; each method overview is located beside its complete derivation. Numerical verification and limitations have separate destinations.}\label{tab:reading-route}
\end{table}

%% file: figures/intro_channel.tex
\begin{figure}[!htbp]\centering
\begin{tikzpicture}[x=.72cm,y=.7cm,font=\small,text=ink]
\node[anchor=east] at(-.6,2.3){Input position};
\node[anchor=east] at(-.6,1.5){Transmitted bits};
\foreach \i/\b in {1/1,2/0,3/1,4/1,5/0,6/0}{
 \node at(\i,2.3){\i};
 \node[draw=black!35,rounded corners=1pt,minimum size=.5cm,fill=white] at(\i,1.5){$\b$};
}
\foreach \i in {2,5}{\draw[ink,line width=.8pt] (\i-.27,1.18)--(\i+.27,1.82);}
\node[anchor=west,align=left] at(7.0,1.5){Positions 2 and 5\\are removed in this realization.};
\draw[-{Latex},black!50] (3.5,.9)--(3.5,.2);
\node[anchor=east] at(-.6,-.35){Deletion output};
\foreach \i/\b in {1/1,2/1,3/1,4/0}{\node[draw=teal!65,fill=teal!4,rounded corners=1pt,minimum size=.5cm] at(\i,-.35){$\b$};}
\node[anchor=west,align=left] at(7.0,-.35){The receiver gets $1110$.\\The missing positions are not marked.};
\node[anchor=east] at(-.6,-1.55){Erasure output};
\foreach \i/\b in {1/1,2/{?},3/1,4/1,5/{?},6/0}{\node[draw=black!35,rounded corners=1pt,minimum size=.5cm] at(\i,-1.55){$\b$};}
\node[anchor=west,align=left] at(7.0,-1.55){An erasure marker preserves\\each surviving bit's input position.};
\end{tikzpicture}
\caption{A deletion removes a bit and closes the space it occupied. In this example, deleting positions 2 and 5 from $101100$ produces $1110$. The deletion receiver observes neither the erased-position markers nor the input-position indices shown for explanation. An erasure channel would instead return $1?11?0$, retaining the six position slots.}\label{fig:intro-channel}
\end{figure}

%% file: data/intro_history.tex
\begin{table}[!htbp]
\centering\small
\setlength{\tabcolsep}{4pt}
\begin{tabular}{@{}>{\raggedright\arraybackslash}p{4.5cm}rrrrr@{}}
\toprule
Source & $d=0.1$ & $d=0.3$ & $d=0.5$ & $d=0.7$ & $d=0.9$\\
\midrule
\multicolumn{6}{@{}l}{\textit{Capacity lower bounds: larger is stronger}}\\[3pt]
\citet{DM06,MD06} & .537944 & .161947 & .059250 & .035550 & .011850\\
\citet{DM07} & .561960 & .222430 & .101860 & .045324 & .012378\\
\citet{VTR13} & .563800 & .222500 & --- & --- & ---\\
\citet{RC23} & --- & --- & .104075 & .047726 & .012379\\
\citet{Chen26}, reported & --- & --- & .11032415 & --- & ---\\
\textbf{This work} & \textbf{.566606} & \textbf{.239035} & \textbf{.114540} & \textbf{.051276} & \textbf{.012928}\\[4pt]
\midrule
\multicolumn{6}{@{}l}{\textit{Capacity upper bounds: smaller is stronger}}\\[3pt]
\citet{DMP07} & .704279 & .446639 & .311087 & .208112 & .100000\\
\citet{FD10} & .689000 & .362000 & .212000 & .126000 & .049000\\
\citet{RD15} & --- & --- & --- & .124290 & .041430\\
\citet{RC23} & .676200 & .351300 & .201500 & .112350 & .037450\\
\citet{PR26} & .657700 & .331400 & .189600 & .107340 & .035780\\
\textbf{This work} & \textbf{.572443} & \textbf{.250840} & \textbf{.130162} & \textbf{.066530} & \textbf{.022177}\\
\bottomrule
\end{tabular}
\caption{Selected historical capacity bounds, in bits per input bit. A dash means that no value from that source is included in this comparison. The first lower row takes the stronger of \citet{DM06} and the universal $0.1185(1-d)$ bound of \citet{MD06}. The \citet{RD15} row uses $0.4143(1-d)$ for $d\ge0.65$. At $d=0.7,0.9$, the 2023 and 2026 upper rows use $0.3745(1-d)$ and $0.3578(1-d)$, respectively. The \citet{Chen26} entry reproduces its strict reported lower bound; we did not replay its numerical computation, and its public record inspected on September 12, 2026 did not supply a machine-readable certificate. Our displayed lower values are rounded down and upper values up. The \citet{DMP07} row uses Table~I of the authors' prepublication version. Other historical entries reproduce reported decimals, with trailing zeros for alignment. The table compares specified points, not prior uniform approximation guarantees.}
\label{tab:intro-history}
\end{table}

%% file: sections/setting.tex
\section{Channel model and capacity approximation}\label{sec:setting}
\subsection{Codes, input distributions, and capacity}
Fix a deletion probability $d$, known to the encoder and decoder. We write $X_1^n=(X_1,\ldots,X_n)$ for a binary input sequence of length $n$; this notation lists the first and last input indices. Independently of that sequence, let $D_i=1$ if position $i$ is deleted and $D_i=0$ otherwise. The $D_i$ are independent, with $\Pr(D_i=1)=d$. The output $Y_n$ consists of the undeleted bits in order. Its observed length is $N_n\sim\operatorname{Bin}(n,1-d)$. The receiver knows $Y_n$ and hence $N_n$, but not the deleted positions.

A length-$n$ codebook is a set $\mathcal C_n\subseteq\bits^n$ of $M_n$ distinct sequences. An encoder assigns one sequence to each of $M_n$ equally likely messages, and a decoder estimates the message from $Y_n$. Its rate is $\log_2M_n/n$ bits per transmitted bit. A rate is achievable if some sequence of codes has error probability tending to zero and rates with lower limit at least that rate. Capacity is the supremum of achievable rates. This definition is called \emph{operational} because it concerns encoders and decoders. All rates below use transmitted bits as the denominator, and logarithms have base two unless $\ln$ is written.

Dobrushin's theorem identifies this capacity with the information limit
\begin{equation}\label{eq:capacity-limit}
 C_n(d)=\frac1n\max_{P_{X_1^n}}I(X_1^n;Y_n),\qquad
 C(d)=\lim_{n\to\infty}C_n(d).
\end{equation}
The maximum is over probability distributions on all $2^n$ binary sequences, not over codebooks or decoders. For example, probabilities $3/4$ and $1/4$ on $00$ and $11$ are allowed; equally likely messages using the codebook $\{00,11\}$ instead give probabilities $1/2$ and $1/2$. The deletion channel meets Dobrushin's hypotheses: each input symbol independently produces itself or an empty string, its output length is bounded by one, and its mean is $1-d>0$ for $d<1$. The endpoint $d=1$ has constant empty output and capacity zero directly \citep{Dob67}; see also Theorem~2.7 of the extended version of \citet{pernice2022}.

The existence of this limit and its coding interpretation require a theorem. They do not follow for every channel merely by defining mutual information. For memoryless channels, the usual coding theorem reduces the optimization to a single input symbol. For the deletion channel, $C_n(d)$ is an upper bound on $C(d)$ and need not be achievable. For instance, $C_1(d)=1-d$: if the only bit survives, it is known, and otherwise the output is empty. A fair input bit attains this value. In a long transmission, however, an empty position is not identified. The finite-block comparison in \cref{sec:proof-overview} quantifies this distinction.

\subsection{The exact functions being reported}
The file \path{data/cells.csv} defines a capacity enclosure on 4,931 closed intervals covering $[0,1]$. A row contains interval endpoints $x_i,y_i$, normalized upper and lower coefficients $A_i,B_i$, and the identifiers of the bounds that justify them. All these numbers are stored as exact fractions. On the interior of that interval, define
\begin{equation}\label{eq:cell-functions}
 L_i(d)=(1-d)B_i,\qquad U_i(d)=(1-d)A_i.
\end{equation}
Thus the table specifies functions at every real parameter, rather than only values at a finite mesh.

The factor $1-d$ comes from normalized-capacity monotonicity \citep{RD15}. Specifically, an upper bound $u$ at $a\le x_i$ gives $A_i=u/(1-a)$, and a lower bound $\ell$ at $b\ge y_i$ gives $B_i=\ell/(1-b)$ when $b<1$. Multiplication by $1-d$ converts these rates per expected surviving bit back to rates per input bit. A bound $C(d)\ge c(1-d)$ valid for every $d$ supplies $B_i=c$ directly. The underlying point bounds, called \emph{anchors}, are listed in \path{data/anchors.csv}. Both tables, the manuscript source, and the verification programs are available at \url{https://github.com/anadim/binary-deletion-channel-capacity}. Section~\ref{sec:transport} proves every extension direction.

Set $L=L_i$ and $U=U_i$ on interval interiors. At a boundary belonging to two adjacent intervals, both rows give valid enclosures. Use their intersection: the larger lower value and the smaller upper value. Finally set $L(0)=U(0)=1$ and $L(1)=U(1)=0$, the noiseless and complete-deletion capacities. These rules define $L$, $U$, and $\widehat C=(L+U)/2$ everywhere.

\begin{theorem}[Capacity to within one hundredth of a bit]\label{thm:main}
Let $L$ and $U$ be defined from the rational table \path{data/cells.csv}. Under the arithmetic model in \cref{sec:arithmetic}, the analytic bounds and completed numerical checks establish
\[
 L(d)\le C(d)\le U(d),\qquad U(d)-L(d)\le2\eps_\star
 \quad(d\in[0,1]).
\]
Consequently their midpoint satisfies
\begin{equation}\label{eq:main}
\boxed{\quad |\widehat C(d)-C(d)|\le\eps_\star
=\frac{2908466681147}{306240000000000}<0.0095<\frac1{100}
\quad(d\in[0,1]).\quad}
\end{equation}
The largest computed cell width occurs on $[521/1000,1043/2000]$.
\end{theorem}
\begin{proof}[Proof structure]
The methods in \cref{sec:survivor,sec:shared,sec:posterior,sec:small-deletion,sec:source,sec:renewal} prove finite sufficient inequalities for upper and lower bounds. The completed numerical checks identified in \cref{sec:reproduction} establish those inequalities for the reported parameters. The comparison in \cref{sec:transport} then gives the cell enclosures above.

On cell $[x_i,y_i]$, subtraction in \eqref{eq:cell-functions} gives width $(1-d)(A_i-B_i)$. Since $A_i\ge B_i$ and $d\ge x_i$, this is at most $(1-x_i)(A_i-B_i)$. Exact rational arithmetic checks all 4,931 rows and finds maximum $2\eps_\star$. Intersecting enclosures at a shared boundary can only decrease their width. Capacity lies between the two endpoints, so its distance from their midpoint is at most half the width. The complete roadmap in \cref{sec:proof-overview} explains the finite inequalities that supply the point bounds.
\end{proof}

\begin{corollary}[Achievable rates within $0.019$ bits]\label{cor:achievable}
Every rate strictly below $L(d)$ is achievable, and
\[
 0\le C(d)-L(d)\le\Delta_\star:=2\eps_\star
 =\frac{2908466681147}{153120000000000}<0.019
 \quad(d\in[0,1]).
\]
\end{corollary}
\begin{proof}
The lower bound gives $L(d)\le C(d)$, so every strictly smaller rate is achievable by the definition of capacity. Subtract $L(d)$ from $C(d)\le U(d)$ and apply the width bound in \cref{thm:main}.
\end{proof}
The midpoint is an estimate, and need not be achievable. The achievable-rate tolerance is twice its error guarantee. This distinction also separates our scalar estimation statement from the achievable-region convention of \citet{etkin2008}. Neither statement is a uniform relative-error guarantee.

\subsection{Accuracy over the parameter range}
\Cref{tab:points} gives representative values, and \cref{fig:approximation} enlarges the enclosure. For $r(d)=(U(d)-L(d))/2$, integration of the pointwise error inequality gives
\[
 \int_0^1|\widehat C(d)-C(d)|\,\mathrm dd
 \le\int_0^1r(d)\,\mathrm dd\le0.006521572012.
\]
The last bound is the exact integral of the piecewise-linear radius, rounded upward. It uses uniform weighting of deletion probability. Section~\ref{sec:discussion} derives the integral and the distribution of tighter tolerances; the actual error relative to unknown capacity is not observed.

\input{data/representative_points.tex}
\begin{center}\begin{minipage}{\linewidth}\centering
\includegraphics[width=\linewidth]{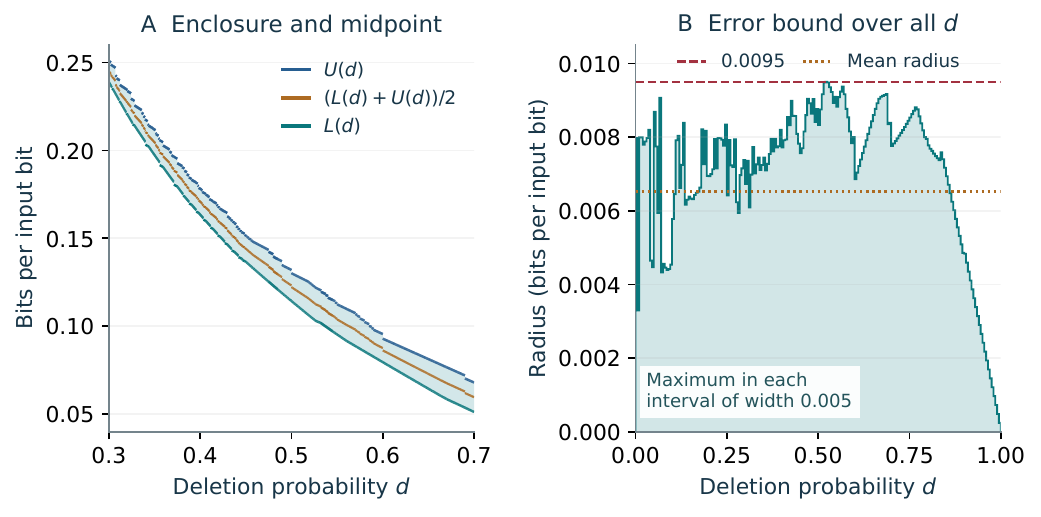}
\captionof{figure}{Capacity enclosure and approximation error. Left: the certified lower and upper bounds, with their midpoint, over $0.3\le d\le0.7$. Right: for each interval of width $0.005$, the plotted height is the largest certified radius $(U-L)/2$ in that interval. This conservative display bounds the radius at every $d$ without smoothing the certificate. The dashed line is $0.0095$ bit per input bit; the dotted line is the mean radius, computed from the exact integral over $[0,1]$.}\label{fig:approximation}
\end{minipage}\end{center}

\subsection{Analytic and numerical verification}
An analytic result states a finite set of inequalities sufficient for a capacity bound. A numerical checker then evaluates every required expression for fixed input probabilities, auxiliary output probabilities, and potentials. These arrays are exact numerical parameters of the proof, not communication codebooks. Interval arithmetic supplies enclosing endpoints for logarithms and entropy expressions: an enclosure $[0.127,0.129]$ proves an upper bound of $0.13$, whereas $[0.127,0.131]$ does not. Search chooses useful parameters; its convergence or optimality is unnecessary once the complete inequalities pass.

The proof assumes correct integer operations and the stated enclosing arithmetic, including the compiler and hardware operations used to implement directed rounding. Section~\ref{sec:arithmetic} gives their error bounds. Section~\ref{sec:reproduction} distinguishes evaluation of the 894,758-byte cell table, exact recomposition of recorded bounds, and replay of the numerical inequalities. It also identifies the large external arrays and the resources required for a full replay. No proof-assistant formalization of the analytic arguments or checker correctness is claimed.
\FloatBarrier

%% file: data/representative_points.tex
\begin{table}[!htbp]\centering\small
\begin{tabular}{@{}rrrrr@{}}\toprule
$d$ & Lower (down) & Upper (up) & Displayed estimate & Error at most\\\midrule
0 & 1.000000000 & 1.000000000 & 1.000000000 & 0.000000000\\
0.01 & 0.922119599 & 0.928537991 & 0.925328794 & 0.003209197\\
0.05 & 0.729835514 & 0.736458607 & 0.733147060 & 0.003311547\\
0.1 & 0.566606555 & 0.572442887 & 0.569524720 & 0.002918167\\
0.2 & 0.356748472 & 0.367278532 & 0.362013501 & 0.005265031\\
0.3 & 0.239035421 & 0.250839634 & 0.244937527 & 0.005902107\\
0.4 & 0.163509876 & 0.178217231 & 0.170863553 & 0.007353678\\
0.5 & 0.114540813 & 0.130161383 & 0.122351098 & 0.007810285\\
0.6 & 0.079600140 & 0.092932209 & 0.086266174 & 0.006666035\\
0.64 & 0.067490721 & 0.083638988 & 0.075564854 & 0.008074134\\
0.7 & 0.051276954 & 0.066529623 & 0.058903288 & 0.007626335\\
0.8 & 0.028844591 & 0.044353082 & 0.036598836 & 0.007754246\\
0.9 & 0.012928926 & 0.022176541 & 0.017552733 & 0.004623808\\
1 & 0.000000000 & 0.000000000 & 0.000000000 & 0.000000000\\
\bottomrule\end{tabular}
\caption{The current certificate at representative probabilities. The error column includes rounding of the displayed midpoint. Exact values and the selected anchors are retained in the data register.}\label{tab:points}
\end{table}

%% file: sections/proof_overview.tex
\section{Proof roadmap}\label{sec:proof-overview}
An upper bound must hold for every input distribution; a lower bound needs one input process with the asserted information rate. Figure~\ref{fig:proofmap} shows the overall argument. Each method begins with a derivation overview beside its full proof, so the reader can follow one construction from its entropy identity to the finite test without reading the other methods first.

\subsection{Why direct finite-block optimization is too large}
A direct first approach is to maximize the mutual information of $n$ transmitted bits. The comparison of \citet{FD10}, proved here in \cref{lem:block}, gives
\[
 C_n(d)-\frac{\log_2(n+1)}n\le C(d)\le C_n(d).
\]
For the upper direction, supply the receiver with boundaries between successive length-$n$ block outputs. This makes the channel easier. Removing those boundaries costs at most $\log_2(n+1)$ bits per block, because its output length has $n+1$ possible values; independent input blocks give the lower direction. At $n=512$, half this tolerance is below $10/1024<1/100$, since $513<2^{10}$. The optimization, however, has $2^{512}$ input sequences. The following methods replace that optimization by smaller finite sufficient tests while leaving transmitted length and input dependence unrestricted.

\input{figures/proof_map_inline}

\subsection{Three upper-bound constructions}
Start with $I(X;Y)=H(Y)-H(Y\mid X)$. For a converse, bound output entropy above and the residual channel randomness below. An auxiliary probability assignment $Q$ gives $H(Y)\le\mathbb E[-\log_2Q(Y)]$ by nonnegativity of relative entropy. The expectation uses the true output law. Section~\ref{sec:roadmap-upper} proves this inequality, removes the unknown input distribution by a maximum, and reduces the entropy difference to a fixed number of surviving output symbols.

Those symbols can originate arbitrarily far into the input. The first construction checks every inspected input word and every allowed outside survivor string, using the same outside string for all retained deletion cases. It pays explicitly for the cases not retained. The second rewrites the local cost, using an entropy cancellation, as an average over outside strings. That average is at most its largest term, so checking every fixed outside string covers every outside distribution. Their complete tests are \eqref{eq:shared-row} and \eqref{eq:posterior-row}; derivations and proofs are together in \cref{sec:shared,sec:posterior}.

A separate converse is useful near zero deletion. It reverses selected deletions, including those that can merge neighboring runs, and bounds the entropy cost of that change. An auxiliary law describes the resulting output runs. A scalar inequality is checked for every permitted run-length case; averaging then preserves the same bound for any input-run distribution. Section~\ref{sec:small-deletion} derives the run-description bound, its pair and triple corrections, and the finite row and tail tests. A scalar specialization at the end illustrates the calculation.

\subsection{Two lower-bound constructions}
For a chosen source, the deletion mask $D_1^n$ and input $X_1^n$ determine the trace $Y_n$. Expanding their conditional joint entropy in the two possible orders gives
\[
 I(X_1^n;Y_n)=H(Y_n)-nh_2(d)+H(D_1^n\mid X_1^n,Y_n).
\]
Thus output entropy and residual mask entropy both enter with a positive sign; the independent deletion-mask entropy is the subtracted term. Every estimate in this identity must refer to the same source.

For a finite-state input, the source transition matrix determines all input-word probabilities and the state transitions between survivors. Revealing a boundary state supplies a computable lower bound on entropy per survivor. Counting deletion masks compatible with an input and output supplies the second positive term. A nonnegative decomposition allows selected disjoint events to be counted with their original probabilities; omitted positive terms weaken the result safely. Section~\ref{sec:source} starts with the rate formula and its dependencies, then specifies the source and derives the two entropy estimates beside their finite calculations.

For independent input-run lengths, describe the input runs that merge into one output run. The classical run-group rate is supplemented by conditional uncertainty about where output runs begin. Finite run blocks compute that correction, while explicit bounds control every subtracted tail. A Poisson repeat calculation, with its own normalization and channel-transfer proof, supplies a lower bound proportional to $1-d$ valid throughout the parameter range. These calculations and their dependencies appear together in \cref{sec:renewal}.

\subsection{Where computation enters and how the bounds combine}
The mathematical results prove that their finite tests imply capacity bounds. Numerical search proposes source probabilities, auxiliary distributions, potentials, and truncation lengths. A selected candidate is frozen as an exact numerical object before verification. Interval arithmetic checks all expressions required by the selected theorem, including remainders. Upper tests use upper interval endpoints; lower tests use lower endpoints for positive terms and upper endpoints for subtracted terms. Optimizer scores alone are not proof premises.

After completing each bound, take the largest lower value and the smallest upper value that are valid at a given parameter. The monotonicity of $C(d)/(1-d)$ extends checked point bounds in the directions derived in \cref{sec:setting}. Exact rational arithmetic verifies interval coverage and width, then the midpoint gives the approximation. Section~\ref{sec:roadmap-numerics} identifies which families are selected and where; \cref{sec:arithmetic,sec:reproduction} give arithmetic and replay details. Section~\ref{sec:limits} states why the computation stopped, which interval currently controls the error, and what further improvements would require.
\FloatBarrier

%% file: figures/proof_map_inline.tex
\begin{figure}[!htbp]\centering
\begin{tikzpicture}[>=Stealth,text=ink,
 box/.style={draw=ink!40,rounded corners=2pt,align=left,text width=5.85cm,inner sep=7pt,font=\small\hyphenpenalty10000},
 wide/.style={box,text width=12.65cm,align=center},arr/.style={->,thick,draw=ink!60}]
\node[wide,fill=ink!3] (goal) {\textbf{One target, two different proof obligations}\\
$L(d)\le C(d)\le U(d)$; midpoint error at most $(U(d)-L(d))/2$.};
\node[box,minimum height=2.35cm,fill=ink!3,below=.45cm of goal.south west,anchor=north west] (u) {\textbf{Upper: cover every input distribution}\\
Start from $I=H(Y)-H(Y\mid X)$.\\
An auxiliary output law bounds $H(Y)$ above; \cref{lem:increments} bounds channel entropy below.};
\node[box,minimum height=2.35cm,fill=teal!4,right=.45cm of u] (l) {\textbf{Lower: choose one input process}\\
Start from $I=H(Y)-nh_2(d)+H(D\mid X,Y)$; \cref{lem:mask-identity}.\\
Both unknown entropy terms need lower bounds.};
\node[box,minimum height=2.55cm,fill=ink!3,below=.4cm of u] (uf) {\textbf{Make the converse finite}\\
\Cref{sec:shared}: one common outside sequence and a remainder.\\
\Cref{sec:posterior}: entropy cancellation.\\
Run alternative near zero: \cref{sec:small-deletion}.};
\node[box,minimum height=2.55cm,fill=teal!4,below=.4cm of l] (lf) {\textbf{Make the source rate finite}\\
\Cref{lem:birch,lem:trace-rate}: output entropy.\\
\Cref{lem:alignment-genie,lem:positive-cell}: residual mask entropy.\\
Alternative run-source calculation: \cref{sec:renewal}.};
\node[box,minimum height=2.65cm,fill=ochre!12,below=.4cm of uf] (uc) {\textbf{Compute an upper bound}\\
Choose probability tables and a potential.\\
Bound \emph{every} required row and remainder above: \eqref{eq:shared-row} or \eqref{eq:posterior-row}.\\
Run alternative: \eqref{eq:small-certificate-bound}.};
\node[box,minimum height=2.65cm,fill=ochre!12,below=.4cm of lf] (lc) {\textbf{Compute a lower bound}\\
Choose transition or run probabilities.\\
Evaluate weighted finite sums below:\\
$(1-d)b_- -h_+ +\mathcal A_-$; \cref{thm:lower}.\\
Run alternative: \eqref{eq:renewal-block-credit}.};
\node[wide,fill=ink!3,below=.5cm of uc.south east,anchor=north] (cover) {\textbf{Extend proved point bounds to every real $d$: \cref{lem:transport}}\\
An upper point at $a$ and a lower point at $b$ give, for $a\le d\le b<1$,\\[2pt]
$(1-d)\ell_b/(1-b)\le C(d)\le(1-d)u_a/(1-a)$.\\
Check coverage and interval widths by exact rational arithmetic; \eqref{eq:continuum-row}.};
\draw[arr](goal.south)--++(0,-.18)-|(u.north);\draw[arr](goal.south)--++(0,-.18)-|(l.north);
\draw[arr](u)--(uf);\draw[arr](l)--(lf);\draw[arr](uf)--(uc);\draw[arr](lf)--(lc);
\draw[arr](uc.south)--(uc.south|-cover.north);\draw[arr](lc.south)--(lc.south|-cover.north);
\end{tikzpicture}
\caption{How the lemmas lead to a numerical capacity approximation. Read either column downward. Pale boxes describe analytic reductions; shaded computation boxes identify the finite computation and its required inequality direction. Search selects the parameters in those boxes; controlled arithmetic checks them. The final comparison extends completed bounds between evaluated probabilities. Here $D$ is the deletion mask, $b_-$ lower-bounds entropy per survivor, $h_+$ upper-bounds $h_2(d)$, and $\mathcal A_-$ lower-bounds residual mask entropy per input bit. Method alternatives are combined only after each has proved a complete bound.}\label{fig:proofmap}
\end{figure}
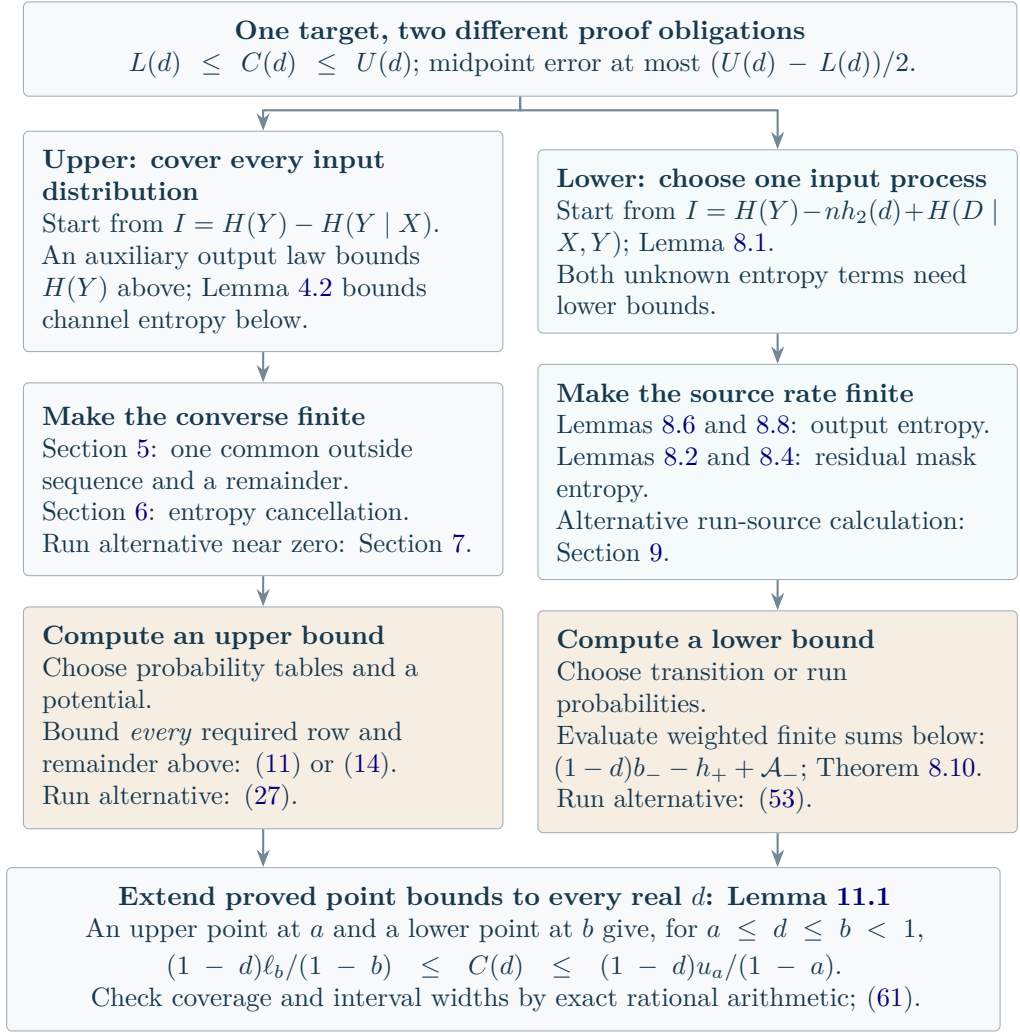

%% file: sections/chapter_upper.tex
\part{Upper bounds}\label{part:upper}
An achievable rate needs one input process; an upper bound must cover them all. \Cref{sec:survivor} first proves a common entropy inequality. \Cref{sec:shared} bounds its uninspected input positions by a common outside survivor sequence and an omitted-probability term. \Cref{sec:posterior} instead cancels its nonlinear entropy term before taking a finite maximum. These are alternative sufficient tests, not successive improvements that must both be applied. \Cref{sec:small-deletion} develops a separate run-length converse near zero deletion probability.

Each construction ends with the finite inequalities used by its checker. A proposed probability table or potential enters the proof only through those inequalities; finding an optimal proposal is unnecessary. The final enclosure takes the smallest available valid upper bound at each deletion probability, as described in \cref{sec:roadmap-numerics}.

%% file: sections/survivor.tex
\section{A source-independent converse from survivor entropies}\label{sec:survivor}
Fix $0\le d<1$. For an input $X_1^n$ and its received sequence $Y_n$, we upper-bound $I(X_1^n;Y_n)=H(Y_n)-H(Y_n\mid X_1^n)$ using $m$ consecutive surviving bits, with $m$ fixed as $n\to\infty$. The output-entropy bound comes from an auxiliary distribution; the conditional-entropy bound retains deletion randomness after the input is specified. We first justify stationarity as a proof device, then derive both estimates and their difference.

\input{sections/roadmap_upper}

\subsection{Reduction from arbitrary inputs to stationary inputs}
A stationary process has the same finite-word probabilities at every starting position. This makes a sum of local expected costs equal to its length times one expectation. The following block construction shows that proving the same upper bound for every stationary input loses no capacity.

\begin{lemma}[Stationary bounds imply an unrestricted converse]\label{lem:stationary}
Fix $d$. Suppose every stationary binary source satisfies
\[
\limsup_{n\to\infty}\frac1n I(X_1^n;Y_n)\le u.
\]
Then $C(d)\le u$.
\end{lemma}
\begin{proof}
Fix a block length $q$ and an arbitrary law $P$ on binary words of length $q$. Form a doubly infinite input by concatenating independent words drawn from $P$. Randomize the position of the block boundaries by an independent offset uniform on $\{0,\ldots,q-1\}$. Shifting the resulting sequence by one position cycles the offset and, when necessary, relabels independent blocks. Its distribution is therefore unchanged by a shift, which proves stationarity.

A length-$n$ segment intersects at most $n/q+2$ blocks and contains at least $n/q-2$ complete blocks. The two incomplete end blocks explain both constants. Supply the receiver with the offset and the individual trace length of each intersected block, calling this extra information $G$. The offset has $q$ possibilities, and each trace length has at most $q+1$ possibilities. Entropy is at most log alphabet size, and joint entropy is at most the sum of individual entropies by \cref{lem:entropy-chain,lem:side-information}. Consequently
\[
H(G)\le\log_2q+(n/q+2)\log_2(q+1).
\]
Given $G$, the receiver can split the concatenated trace at the block boundaries. Conditional on the offset, the complete input blocks are independent, their deletion processes are independent, and their separate input--trace pairs are therefore independent. The chain rule then adds their mutual informations. Each complete block contributes $I_P(X_1^q;Y_q)$, while discarding the end blocks cannot increase information by data processing. It follows that
\[
I(X_1^n;Y_n,G)\ge(n/q-2)I_P(X_1^q;Y_q).
\]
For this comparison, expand the revealed mutual information first with respect to the offset. Its nonnegative offset term can be discarded, and the conditional term contains the independent complete-block pairs just described. To remove $G$, we do not assume that it is independent of the input. We use the general side-information inequality, which gives $I(X_1^n;Y_n)\ge I(X_1^n;Y_n,G)-H(G)$. Combining the preceding bounds, dividing by $n$, and taking $n\to\infty$ gives
\[
\liminf_{n\to\infty}\frac1n I(X_1^n;Y_n)
\ge\frac1q I_P(X_1^q;Y_q)-\frac{\log_2(q+1)}q.
\]
In this limit $q$ was fixed: only the number of concatenated blocks grew. The assumed upper bound for stationary processes is $u$, so the right side is at most $u$. The distribution $P$ was arbitrary; maximizing $I_P/q$ gives $C_q(d)$. We now let $q\to\infty$. By \eqref{eq:capacity-limit}, $C_q(d)\to C(d)$, while $\log_2(q+1)/q\to0$. Thus $C(d)\le u$.
\end{proof}
It therefore suffices to prove one upper bound for every stationary input process. In the rest of this section we fix such a process, with no further restriction on its distribution or dependence between bits.

\subsection{The surviving sequence and its conditional entropy}
Place the raw origin immediately before input position one, and write $F=(X_1,X_2,\ldots)$ for the entire input after that origin. We condition on $F$ to measure the randomness introduced by the channel, not to give the receiver the input. Let $G_1$ be the position of the first retained bit. For $j\ge2$, let $G_j$ be the distance from retained position $j-1$ to retained position $j$. A distance of $r$ means $r-1$ consecutive deletions followed by one retained bit. These events are independent of the input and give
\[
\Pr(G_j=r)=d^{r-1}(1-d),\qquad r=1,2,\ldots.
\]
The sum of these probabilities is one because $d<1$. The spacings are independent, and the $j$th survivor is the input bit at the cumulative position,
\[
Z_j=X_{G_1+\cdots+G_j}.
\]
This equation defines the infinite survivor sequence $Z$. It uses the same deletion rule as the finite channel, but continues beyond the end of a transmitted block.

Fixing $F$ fixes every input bit but leaves the deletion positions random. If the first two input bits are $01$, retaining position one produces a first output zero, whereas deleting position one and retaining position two produces a first output one. The conditional entropy $A_k=H(Z_1^k\mid F)$ averages this channel randomness over input realizations for the first $k$ survivors. Set $A_0=0$. By the chain rule, the increase when we include one more survivor is
\[
a_k=A_k-A_{k-1}=H(Z_k\mid Z_1^{k-1},F).
\]
The quantity $a_k$ is therefore the entropy of survivor $k$ conditional on the entire raw input and survivors $1,\ldots,k-1$. Although $F$ appears in the conditioning, conditional entropy averages the entropy for each input realization with its source probability. The deletion process supplies the remaining randomness.

The next lemma turns one local quantity $a_m$ into a bound for a long received string. It first proves that later conditional entropies are at least $a_m$, then accounts for the random number of survivors.

\begin{lemma}[Increasing conditional survivor increments]\label{lem:increments}
The numbers $a_k$ just defined are nondecreasing and lie in $[0,1]$. For every fixed output-word length $m\ge1$,
\[
\liminf_{n\to\infty}\frac1nH(Y_n\mid X_1^n)\ge(1-d)a_m.
\]
\end{lemma}
\begin{proof}
\emph{Comparison of successive increments.}
The chain-rule expression for $a_{k+1}$ is $H(Z_{k+1}\mid Z_1^k,F)$. Revealing the first spacing $G_1$ can only reduce this conditional entropy, by \cref{lem:side-information}. We will show that the reduced entropy is exactly $a_k$; together these statements give
\[
a_{k+1}\ge H(Z_{k+1}\mid Z_1^k,F,G_1)=a_k.
\]
To verify the equality, first fix $G_1=r$. Once $F$ and $r$ are known, $Z_1=X_r$ is known, so conditioning on $Z_1$ adds nothing. The remaining survivors use the independent spacings $G_2,G_3,\ldots$ and the shifted input $(X_{r+1},X_{r+2},\ldots)$. The discarded prefix $(X_1,\ldots,X_r)$ does not affect their deletion randomness once that shifted input is fixed. Hence the conditional uncertainty in $Z_{k+1}$ after $Z_2,\ldots,Z_k$ is the same functional of the shifted raw future as $a_k$ is of the original raw future. Stationarity makes the distributions of these two futures identical. Averaging over the raw input therefore gives $a_k$ for each $r$, and averaging over $G_1$ keeps that value. Finally, the conditioned symbol is binary, so its entropy is between zero and one. This proves monotonicity and boundedness.

\emph{Comparison of finite and infinite outputs.}
Use the same input and the same deletion indicators to construct both the finite trace $Y_n$ and the infinite survivors $Z$. If $N_n$ bits survive among the first $n$ input positions, then $Y_n=(Z_1,\ldots,Z_{N_n})$. For any deterministic $k$, \cref{lem:random-prefix} therefore applies and yields
\[
A_k=H(Z_1^k\mid F)\le H(Y_n\mid F)+k\Pr(N_n<k).
\]
If $Y_n$ has at least $k$ symbols, it determines $Z_1^k$. Otherwise the missing suffix has entropy at most $k$ bits. Multiplying this upper bound by the probability of the second event gives $k\Pr(N_n<k)$. The length of $Y_n$ determines which event occurred, so no extra event indicator is needed.

The full raw future $F$ contains the finite input $X_1^n$. Revealing more input cannot increase conditional entropy, again by \cref{lem:side-information}. Combining that fact with the preceding inequality gives
\[
H(Y_n\mid X_1^n)\ge H(Y_n\mid F)
\ge A_k-k\Pr(N_n<k).
\]
This is the finite inequality from which the asymptotic lower bound follows. It is valid for every $k$, with no replacement of random length by its mean.

\emph{Choice of the prefix length.}
Choose $k=k_n=\lfloor n(1-d)-n^{2/3}\rfloor$, which is positive for all sufficiently large $n$ because $d<1$ is fixed. The event $N_n<k_n$ implies that $N_n$ is at least $n^{2/3}$ below its mean. The binomial variance bound following \cref{lem:random-prefix} gives
\[
\Pr(N_n<k_n)\le d(1-d)n^{-1/3}\longrightarrow0,
\qquad \frac{k_n}{n}\longrightarrow1-d.
\]
To finish, monotonicity gives $a_j\ge a_m$ whenever $j\ge m$. Since the earlier increments are nonnegative, summing them gives $A_{k_n}=\sum_{j=1}^{k_n}a_j\ge(k_n-m+1)a_m$ once $k_n\ge m$. Substitute this into the finite inequality and divide by $n$:
\[
\frac1nH(Y_n\mid X_1^n)
\ge\frac{k_n-m+1}{n}a_m
-\frac{k_n}{n}\Pr(N_n<k_n).
\]
The first coefficient tends to $1-d$, and the second term tends to zero by the preceding probability bound. Taking the lower limit proves the claim.
\end{proof}
For any fixed $m$, the proof gives a lower bound $(1-d)a_m$ on the conditional output entropy per input bit as $n\to\infty$. The factor $1-d$ is the asymptotic number of survivors per input bit. This estimate will be subtracted from an upper bound on output entropy. Its validity does not require $m$ to increase with the transmitted length.

\subsection{An output predictor bounds output entropy}
We now upper-bound $H(Y_n)$ by a quantity that can be specified with a finite table. For each output history $v$ of length $m-1$, choose two positive numbers $Q(0\mid v)$ and $Q(1\mid v)$ that sum to one. The table can be chosen freely; it need not give the true next-bit probabilities of the source. It defines the logarithmic loss
\[
\ell_Q(z_1^m)=-\log_2Q(z_m\mid z_1^{m-1}).
\]
For example, assigning probability $1/4$ to the observed next bit gives loss two bits. We use the expected loss as an upper bound on entropy, by \cref{lem:cross-entropy} and the standard comparison-distribution argument \citep[Section~2.6]{CT06}; this is not a decoder or an estimate of decoding complexity. The same table is used for every input, and its arguments contain received symbols only. Positivity and finiteness of the table make all its logarithmic losses bounded.

\begin{lemma}[Fixed output-predictor upper bound]\label{lem:predictor}
Every stationary binary input satisfies
\[
\limsup_{n\to\infty}\frac1nH(Y_n)
\le(1-d)\E\ell_Q(Z_1^m),
\]
where $Z_1^m$ consists of $m$ consecutive survivors of that input.
\end{lemma}
\begin{proof}
\emph{Encode the length and then the observed symbols.}
Let $y=y_1\cdots y_k$ be any binary word of length $0\le k\le n$. Assign its length probability $1/(n+1)$, give its first $\min\{k,m-1\}$ symbols probability $1/2$ each, and use $Q$ for the remaining symbols. This defines the comparison distribution
\[
q_n(y)=\frac{2^{-\min\{k,m-1\}}}{n+1}
\prod_{j=m}^{k}Q(y_j\mid y_{j-m+1}^{j-1}).
\]
A product with no factors equals one. For each fixed length, summing over the last symbol removes a factor whose two probabilities sum to one. Repeating this sum for each earlier symbol gives total mass $1/(n+1)$ at that length. Summing over the $n+1$ lengths proves normalization.

The actual output need not have the uniform length distribution or the conditional probabilities $Q$. Nevertheless, \cref{lem:cross-entropy} gives $H(Y_n)\le\E[-\log_2q_n(Y_n)]$, with expectation under the actual output distribution. Taking the negative logarithm of the displayed product produces $\log_2(n+1)$ for the length, at most $m-1$ for the initial symbols, and the sum of $\ell_Q$ over the remaining output symbols. These are entropy-bound terms measured in bits.

\emph{Assign output losses to raw input positions.}
It remains to bound that sum without assuming that trace length is deterministic. Extend the stationary input and the independent deletions to both sides of the origin. At each retained raw position, assign the loss of predicting its symbol from the preceding $m-1$ survivors. At deleted positions assign zero. For every trace of the first $n$ positions, its internal predictor losses are among these $n$ assigned costs. The extra costs concern only the earliest survivors whose predecessors lie before the block; they are nonnegative. Thus the expected sum of logarithmic losses within the trace is at most the expected sum of these $n$ assigned terms.

For example, retained positions $1,3,6$ contribute the internal two-symbol losses for $(X_1,X_3)$ and $(X_3,X_6)$ at positions $3$ and $6$. Position $1$ may carry an additional nonnegative loss using a survivor before the block. Dropping it can only reduce the sum.

\emph{Identify the law seen at a retained position.}
The expected cost at a raw position is its retention probability $1-d$ times the expected loss conditional on retention. We must check that this conditional loss has the same distribution as the forward survivor loss appearing in the statement. Fix the retained raw position at zero and first fix the $m-1$ spacings to its predecessors. In their left-to-right order, call these positive spacings $g_1,\ldots,g_{m-1}$. The sampled raw positions are $-(g_1+\cdots+g_{m-1}),\ldots,-g_{m-1},0$. Translation by $g_1+\cdots+g_{m-1}$ turns them into $0,g_1,g_1+g_2,\ldots,g_1+\cdots+g_{m-1}$. Stationarity gives the same ordered input-symbol law at these translated positions. 

Average over the spacings, whose independent geometric law is unaffected by their order and is independent of the input. In the forward construction of $Z_1^m$, the spacings $G_2,\ldots,G_m$ have exactly this law; its additional first spacing $G_1$ translates all $m$ sampled positions together. Conditioning on $G_1$ and using stationarity once more removes that common translation. Thus the two ordered survivor blocks have the same law. This uses translations and independent gap averaging, not a reversal of the input process.

\emph{Sum over positions and remove the finite boundary terms.}
Each raw position consequently has expected cost $(1-d)\E\ell_Q(Z_1^m)$. Adding the $n$ equal expectations gives the finite bound
\[
H(Y_n)\le\log_2(n+1)+(m-1)
+n(1-d)\E\ell_Q(Z_1^m).
\]
Divide by $n$. The first two terms vanish because $m$ is fixed, giving the claimed upper limit.
\end{proof}
The table $Q$ must be positive and normalized and use only the stated output history; no accuracy assumption is required.

\subsection{Combining the two entropy bounds}
Write $F=(X_1,X_2,\ldots)$ for the entire raw input after the chosen origin. Conditional on this input, deletion randomness produces a distribution $P_F$ on the first $m$ survivors. Here $m$ counts retained bits, not input positions. The comparison table $Q(z_m\mid z_1^{m-1})$ assigns probabilities to the last survivor given the preceding $m-1$ survivors.

For a nonnegative array $P$ on $m$-bit words, let $P_{m-1}(v)=P(v0)+P(v1)$ be its prefix marginal. Define the conditional relative-entropy expression

\begin{equation}\label{eq:psi}
\Psi_Q(P)=\sum_{z_1^m\in\bits^m}P(z_1^m)\log_2
\frac{P(z_1^m)}{P_{m-1}(z_1^{m-1})Q(z_m\mid z_1^{m-1})}.
\end{equation}
To interpret this formula for a probability distribution, split its logarithm into $-\log_2Q(z_m\mid z_1^{m-1})$ and $\log_2[P(z_1^m)/P_{m-1}(z_1^{m-1})]$. The first part averages the predictor loss. The second part is the negative conditional entropy of the last survivor given its prefix. Thus
\[
\Psi_Q(P)=\E_P\ell_Q(Z_1^m)-H_P(Z_m\mid Z_1^{m-1}).
\]
Averaging this identity for $P=P_F$ gives $\E\Psi_Q(P_F)=\E\ell_Q(Z_1^m)-a_m$, because $a_m$ is precisely that last-symbol entropy conditional also on $F$.

\begin{corollary}[The common survivor rate bound]\label{cor:survivor-rate}
For every stationary input and every positive normalized output table $Q$,
\[
\limsup_{n\to\infty}n^{-1}I(X_1^n;Y_n)
\le (1-d)\mathbb E\Psi_Q(P_F).
\]
Here $P_F$ is the first-$m$-survivor distribution conditional on the entire input $F$, and the expectation averages over that input.
\end{corollary}
\begin{proof}
Apply \cref{lem:predictor} to the positive entropy term in $I(X_1^n;Y_n)$ and \cref{lem:increments} to the term being subtracted. For two bounded sequences, the upper limit of their difference is at most the upper limit of the first minus the lower limit of the second; this follows by bounding each sequence beyond a sufficiently large index. We obtain
\begin{equation}\label{eq:reward-bridge}
\limsup_{n\to\infty}\frac1n I(X_1^n;Y_n)
\le(1-d)\bigl(\E\ell_Q(Z_1^m)-a_m\bigr)
=(1-d)\E\Psi_Q(P_F).
\end{equation}
\end{proof}
The separate inequality $I(X_1^n;Y_n)\le H(X_1^n)$ follows from nonnegativity of conditional entropy. For a stationary input, division by $n$ and \cref{lem:entropy-rate} turn it into the alternative upper bound $h(X)$, its input entropy rate. Later certificates can combine these two valid bounds.

\begin{lemma}[Mixtures and omitted probability masses]\label{lem:survivor-cost}
For nonnegative arrays $P,P'$ on $m$-bit words, the function $\Psi_Q$ in \eqref{eq:psi} is nonnegative, convex, and homogeneous. In particular,
\[
\Psi_Q(P+P')\le\Psi_Q(P)+\Psi_Q(P').
\]
\end{lemma}
\begin{proof}
In \eqref{eq:psi}, the denominator $P_{m-1}Q$ is linear in $P$ and has the same total mass as $P$. 

\emph{Mixtures cannot exceed their average cost.}
Fix $0<\theta<1$ and put $P_\theta=\theta P+(1-\theta)P'$. For each output prefix $v$ and last bit $a$, apply log-sum (\cref{lem:logsum}) to the two numerator masses $\theta P(va),(1-\theta)P'(va)$ and the corresponding denominator masses $\theta P(v)Q(a\mid v),(1-\theta)P'(v)Q(a\mid v)$. Their sums are $P_\theta(va)$ and $P_\theta(v)Q(a\mid v)$, because prefix marginalization is linear. Summing the resulting inequalities over $v,a$ gives
\[
 \Psi_Q(P_\theta)\le\theta\Psi_Q(P)+(1-\theta)\Psi_Q(P').
\]
This is convexity; weights zero and one give equality.

\emph{The cost is nonnegative and scales with mass.}
To prove nonnegativity, apply log-sum once to the entire numerator and denominator arrays. Their totals agree, so its right side is total mass times $\log_21=0$. Multiplying $P$ by a constant cancels that constant inside the ratio, so $\Psi_Q(cP)=c\Psi_Q(P)$ for $c\ge0$. 

\emph{Separate retained and omitted masses.}
Convexity at the midpoint then gives
\[
\Psi_Q(P+P')=2\Psi_Q\bigl((P+P')/2\bigr)
\le\Psi_Q(P)+\Psi_Q(P').
\]
This is the subadditivity that permits separate charges for retained and omitted probability masses. Zero entries are interpreted by continuity; the positivity of $Q$ prevents a finite positive numerator from meeting a zero code probability.

\end{proof}

The converse now reduces to bounding $(1-d)\mathbb E\Psi_Q(P_F)$ for every input process. The transmitted length has disappeared, but $P_F$ still depends on an infinite input. The next two sections replace that dependence by finite tests covering every possible uninspected sequence.

%% file: sections/roadmap_upper.tex
\subsection{Proof plan: bound the two entropies in opposite directions}\label{sec:roadmap-upper}
An upper bound must hold for every input distribution. We begin with
$I(X_1^n;Y_n)=H(Y_n)-H(Y_n\mid X_1^n)$: upper-bound the first entropy and lower-bound the second. Both are needed. For the fixed input $01$ at $d=1/2$, the outputs $\varnothing,0,1,01$ are equally likely, so both entropies equal two and their difference is zero. Random output need not convey information.

The standard comparison-distribution argument supplies the first direction \citep[Section~2.6]{CT06}. If $p$ is the true distribution of $Y$ and $q$ is any positive probability distribution, then
\begin{equation}\label{eq:roadmap-kl}
\mathbb E_p[-\log_2q(Y)]-H(Y)
=\sum_y p(y)\log_2\frac{p(y)}{q(y)}=D(p\Vert q)\ge0.
\end{equation}
Substituting the entropy sum and combining logarithms gives the equality; \cref{lem:logsum} gives the inequality. Only the complete weighted sum is nonnegative. The comparison need not hold separately for each output. The expectation still uses the actual channel probabilities, not the chosen $q$.

The finite-channel dual bound follows by subtracting conditional output entropy and bounding an average over inputs by its largest value:
\[
 I_P(X_1^n;Y_n)
 \le\sum_xP(x)D(W_n(\cdot\mid x)\Vert q_n)
 \le\max_xD(W_n(\cdot\mid x)\Vert q_n).
\]
Here $W_n(y\mid x)$ is the known deletion-channel probability. The unknown input law $P$ disappears from the last expression. Directly evaluating it at large $n$ is expensive because there are $2^n$ inputs. Output-distribution and dynamic-programming converses use this same general principle \citep{DMP07,HSPKS21}.

We replace the full-block calculation by one involving $m$ consecutive received bits. The proof has three dependencies:
\begin{enumerate}
\item \Cref{lem:stationary} reduces the problem to stationary inputs, so each input position has the same expected local cost.
\item \Cref{lem:predictor} bounds output entropy using a finite table $Q$ of next-output-bit probabilities. \Cref{lem:increments} lower-bounds conditional output entropy using the uncertainty of the $m$th survivor when the full input is known.
\item \Cref{cor:survivor-rate} subtracts the two estimates. \Cref{sec:shared,sec:posterior} then turn that rate inequality into finite tests covering every input.
\end{enumerate}
The present section is analytic. It does not yet provide a finite computation: the conditional survivor distribution still depends on input bits arbitrarily far away. The next two sections resolve precisely that dependence.
\Cref{fig:upper-proof-routes} shows the common entropy argument and the two finite-window reductions that follow.

%% file: sections/shared.tex
\section{A converse from a shared survivor tail}\label{sec:shared}
For the entire input $F$ after an origin, let $P_F$ be the distribution of its first $m$ survivors. The preceding section bounds its information rate by $(1-d)\mathbb E\Psi_Q(P_F)$, where $\Psi_Q$ is output logarithmic loss minus conditional survivor entropy. We inspect only $R$ input positions. The first survivor can lie beyond them, so we condition on a fixed number of outside survivors and bound the probability that they are insufficient.

The outside survivor sequence must be common to every deletion pattern inside the window. Indeed, fixing the input and all deletion indicators after position $R$ fixes that outside sequence; varying the first $R$ indicators does not change it. The conditioning is in the proof, not information supplied to the receiver.

\input{sections/roadmap_shared}

\subsection{A finite window and the probability it omits}
Choose integer parameters: a raw-window length $R\ge1$, a number $m\ge1$ of desired output symbols, and a cutoff $0\le b\le m$. This section uses $0\le d<1$; at $d=1$ capacity is exactly zero. The state $(w,u)$ contains the raw word $w\in\bits^R$ and the first $m-b$ surviving bits $u$ strictly beyond that window. The word $u$ records survivors, not the next $m-b$ raw bits.

Include precisely the deletion patterns with at least $b$ survivors in the window. If a pattern retains $j<m$ symbols there, append the first $m-j$ symbols of $u$; the inequality $m-j\le m-b$ ensures that $u$ contains enough symbols. If $j\ge m$, take the first $m$ survivors already in the window. For each resulting word $z$, let $M_{w,u}(z)$ be the sum of the original probabilities of the included patterns producing $z$. Thus $M_{w,u}$ is a nonnegative mass array on $m$-bit output words. Its total mass is the probability that the window contains at least $b$ survivors, which can be less than one.

The omitted probability does not depend on the contents of $w$. There are $\binom Rj$ masks with $j$ survivors. Each has probability $(1-d)^jd^{R-j}$, because its $j$ survival events and $R-j$ deletion events are independent. Therefore
\begin{equation}\label{eq:beta}
\beta=\sum_{j=0}^{\min\{b-1,R\}}\binom Rj(1-d)^jd^{R-j}.
\end{equation}
An empty sum is zero, and endpoint powers use the binomial-probability convention $0^0=1$. The upper summation limit also covers a cutoff larger than the raw window. The retained masks have total probability $1-\beta$, so $M_{w,u}$ has this same mass. We keep that mass unchanged: dividing by $1-\beta$ would change the probability weights in the converse.

The example in \cref{fig:intro-window} has $R=3,m=2,b=1,w=010,u=1$ and $d=1/2$. Its retained masses for $00,01,10,11$ are $1/8,4/8,1/8,1/8$; the only omitted pattern deletes all three bits. Their sum is $7/8$, not one. Keeping these original weights is necessary for the entropy subtraction.

\subsection{Bounding the omitted deletion patterns}
The function $\Psi_Q$ extends from probability distributions to nonnegative mass arrays such as $M_{w,u}$. For such an array $P$, the prefix mass is $P(v)=P(v0)+P(v1)$, where $v$ has length $m-1$. The chosen predictor assigns probabilities $Q(0\mid v),Q(1\mid v)$ summing to one, and its logarithmic loss is $\ell_Q(vz)=-\log_2Q(z\mid v)$. Expanding \eqref{eq:psi} gives
\begin{equation}\label{eq:roadmap-shared-cost}
\Psi_Q(P)=\sum_zP(z)\ell_Q(z)
-\sum_{v\in\bits^{m-1}}P(v)\hb\!\left(\frac{P(v1)}{P(v)}\right).
\end{equation}
A prefix of mass zero contributes zero. To obtain this formula, separate $\log(P(vz)/(P(v)Q(z\mid v)))$ into $\log(P(vz)/P(v))$ and $-\log Q(z\mid v)$, then sum over the two last bits. The first sum is logarithmic loss weighted by the original masses. In the second sum, each prefix contributes its mass times the binary entropy of the last bit conditional on that prefix. This interpretation applies even when the total mass is less than one.

By \cref{lem:survivor-cost}, $\Psi_Q(P+A)\le\Psi_Q(P)+\Psi_Q(A)$. We use this inequality to separate retained masks from omitted masks, while convexity lets us disclose the outside survivors for the bound.

\begin{lemma}[Shared-tail domination]\label{lem:shared-tail}
For each raw future $F$ with window $w$,
\[
\Psi_Q(P_F)\le\E[\Psi_Q(M_{w,U})\mid F]+\beta\ell_{\max},
\qquad \ell_{\max}=\max_{z\in\bits^m}\ell_Q(z).
\]
Here $U$ is the random survivor prefix beyond the window, and $\beta$ is the omitted-mask probability in \eqref{eq:beta}.
\end{lemma}
\begin{proof}
\emph{Separate the two classes of masks.}
First fix the entire survivor sequence beyond the window. Only the window mask remains random. Its full output-prefix law is $M_{w,u}+A$, where $A$ collects the omitted masks and has mass $\beta$. By the subadditivity just proved,
\[
\Psi_Q(M_{w,u}+A)\le\Psi_Q(M_{w,u})+\Psi_Q(A).
\]
\emph{Bound its logarithmic term and discard its subtracted entropy.}
The expression $\Psi_Q(A)$ is at most $\sum_z A(z)\ell_Q(z)$ because its entropy term is subtracted and is nonnegative. Each logarithmic loss is at most $\ell_{\max}$ and the masses $A(z)$ sum to $\beta$. Therefore
\[
\Psi_Q(A)\le\sum_zA(z)\ell_Q(z)
\le\ell_{\max}\sum_zA(z)=\beta\ell_{\max}.
\]
\emph{Remove the information that fixed the tail.}
For the fixed raw input $F$, the deletion indicators beyond the window make the outside survivor sequence random. Averaging the conditional distributions $M_{w,u}+A$ over that randomness gives $P_F$. Convexity bounds $\Psi_Q$ of this average by the average of $\Psi_Q(M_{w,u}+A)$. Substituting the preceding inequalities gives the claimed upper bound for this $F$.
\end{proof}
Multiplication by $1-d$ gives the omitted-pattern contribution $(1-d)\beta\ell_{\max}$ in bits per input bit. It depends only on the probability of omission and the largest logarithmic loss of $Q$, not on the unknown words produced by omitted patterns.

For the example in \cref{fig:intro-window}, $\beta=1/8$. If $Q$ assigns $1/2$ to both bits, then $\ell_{\max}=1$ and the omitted contribution is $(1/2)(1/8)=1/16$ per input bit. The retained contribution is still included; the final test checks all windows and outside words.

\Cref{fig:alignment-and-window} illustrates the deletion ambiguity and the backward state update used next.

\begin{figure}[tb]\centering
\includegraphics[width=\linewidth]{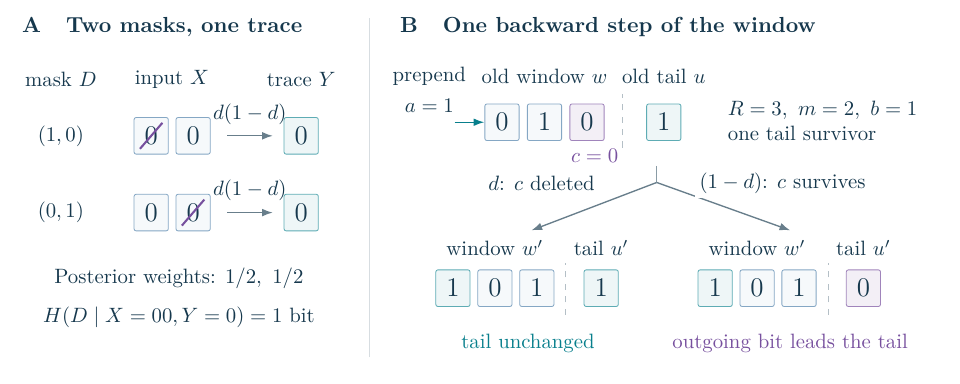}
\caption{Unobserved deletion positions and a shared-tail transition. A mask entry $D_i=1$ denotes deletion, with probability $d$, and a zero denotes survival, with probability $1-d$.
\textbf{A.} Both displayed masks take $X=00$ to $Y=0$ with probability $d(1-d)$. For $0<d<1$, conditioning on this input and trace leaves two equally likely masks and hence one bit of uncertainty.
\textbf{B.} The state contains a raw window $w$ of length $R=3$ and one survivor $u$ strictly beyond it. We seek $m=2$ survivors and retain window masks with at least $b=1$ survivor. Prepending $a=1$ changes $w=010$ to $w'=101$ and sends the outgoing bit $c=0$ into the tail. Deleting $c$ leaves $u'=1$; surviving $c$ gives $u'=0$. The old tail uses only deletion indicators after $c$, so its value gives no information about the deletion indicator of $c$. Bits run from earliest to latest, and dashed lines separate a raw window from its survivor tail.}
\label{fig:alignment-and-window}
\end{figure}

\subsection{A backward shift and its two possible next states}
To average finite-state inequalities, we need the actual transition of their states. Recall that $w=w_1\cdots w_R$ is a raw window and $u$ contains $m-b$ survivors beyond it. Moving the window one raw position backward prepends a bit $a$ and removes its former last bit $c$. In symbols,
\[
w'=aw_1\cdots w_{R-1},\qquad c=w_R.
\]
The new raw window is deterministic once $w,a$ are given. The outgoing bit $c$ now lies immediately before the old survivor tail. If it is deleted, that tail stays $u$. If it survives, prepend $c$ to $u$ and discard its last symbol to keep the length $m-b$. Call the resulting word $\psi_cu$. When $m=b$, both tails are the empty word, denoted $\varnothing$.

A potential $V$ is any real-valued function on this finite state set. Its conditional expected value after the move is therefore
\begin{equation}\label{eq:physical-transition}
(P_aV)(w,u)=dV(w',u)+(1-d)V(w',\psi_cu).
\end{equation}
The two coefficients are the probabilities that $c$ is deleted or retained. Its deletion indicator is independent of the raw bits $w,a$ and of the indicators strictly after $c$ that determine $u$. This independence justifies using the probabilities $d$ and $1-d$ conditional on the old state and prepended bit. It verifies the fixed-operator hypothesis of \cref{lem:potential} for every input law, regardless of its memory.

\paragraph{Worked example: balancing costs on a trajectory.}
Consider an abstract process that alternates between states $A$ and $B$, with local costs $2$ and $0$. A separate maximum over the two costs gives $2$, whereas each pair of steps costs $2$ and the average is $1$. Assign potentials $V(A)=1$ and $V(B)=0$. Adding the new potential minus the old potential changes the two costs to
\[
\underbrace{2+V(B)-V(A)}_{\text{step from }A}=1,
\qquad
\underbrace{0+V(A)-V(B)}_{\text{step from }B}=1.
\]
Each complete cycle still costs $2$: the amount subtracted on the first step is restored on the second. Over an incomplete trajectory, only the final potential minus the initial potential remains. These are abstract cost units, chosen to illustrate the averaging argument, rather than numerical bounds on capacity.

In the deletion-channel state process, the next state is random rather than forced to alternate. The expression $(P_aV)(w,u)$ is its expected next potential. At an absorbing state, where the next state equals the current state with probability one, this expected difference is zero. A positive probability of a self-loop alone does not force zero correction: transitions to other states also enter the expectation.

The potential enters the capacity bound only through $(P_aV)(w,u)-V(w,u)$. Under a stationary source, its expected value is zero because the old and new states have the same distribution. For a finite sequence of moves, the expected differences sum to the expected final potential minus the expected initial potential. The absolute value of this difference is bounded by the range of $V$, independently of the number of moves; see \cref{lem:potential}.

\subsection{Combining output information with input entropy}
All 177 shared-tail instances in the published certificate use the mixture weight $t=0$. Their test therefore uses only the output-information bound, the omitted-pattern term, and the potential difference; no input probability table $K$ enters its value. We give the more general mixture below to state the full valid test.
Let $W$ be the $R$-bit input window and $U$ the first $m-b$ survivors after it. The retained weights $M_{W,U}$ have mass $1-\beta$; $\ell_{\max}$ is the largest loss assigned by $Q$.

With this notation, the survivor argument and the missing-mask lemma give $(1-d)\E\Psi_Q(M_{W,U})+(1-d)\beta\ell_{\max}$. Independently, information cannot exceed the input entropy rate $h(X)$. A number bounded by both $A$ and $B$ is also bounded by $(1-t)A+tB$ for every $0\le t\le1$. We use this elementary observation to combine them.
The input entropy rate can also be bounded by a finite probability table. Let $X_i$ be the bit to be prepended, $W$ the next $R$ raw bits, and $U$ the first $m-b$ survivors after those bits. For each $(w,u)$, choose positive probabilities $K(0\mid w,u),K(1\mid w,u)$ summing to one. Conditional cross entropy bounds $H(X_i\mid W,U)$ by the expected negative logarithm of this table. The full comparison is
\[
h(X)=H(X_i\mid X_{i+1}^{\infty})
\le H(X_i\mid W,U)
\le\E[-\log_2K(X_i\mid W,U)].
\]
The first equality is the stationary entropy-rate formula read backward; finite-block entropy is unchanged by reversing the order of the chain rule. For the first inequality, conditioning also on the entire raw future can only reduce entropy. Once that future is known, the extra random coins generating $U$ are independent of $X_i$, so this additional conditioning gives exactly the left-hand entropy. The last inequality is conditional cross entropy, \cref{lem:cross-entropy}.

Thus $K$ bounds input entropy and $Q$ bounds output entropy. Both are auxiliary distributions; neither prescribes a receiver algorithm.

\begin{certbox}[title={The shared-tail certificate}]
Fix $d,t,Q,K,V,R,m,b$ as above. For every raw window $w\in\bits^R$, tail $u\in\bits^{m-b}$, and prepended bit $a\in\bits$, require
\begin{align}\label{eq:shared-row}
&(1-t)(1-d)\Psi_Q(M_{w,u})+(P_aV)(w,u)-V(w,u)\notag\\
&\quad-t\log_2K(a\mid w,u)
 +(1-t)(1-d)\beta\ell_{\max}\le u_\star.
\end{align}
The first and last terms bound output information, with weight $1-t$. The term involving $K$ bounds input entropy, with weight $t$. The potential change has zero stationary average. The retained measure $M_{w,u}$ keeps its original mass $1-\beta$.
\end{certbox}

\begin{theorem}[The finite shared-tail test implies a capacity bound]\label{thm:shared}
If \eqref{eq:shared-row} holds for all its rows, then $C(d)\le u_\star$.
\end{theorem}
\begin{proof}
\emph{Bound one arbitrary stationary source.}
Fix a stationary input and denote its information-rate limsup by $\overline I$. Combining the two bounds just derived gives
\[
\overline I\le\E\!\left[(1-t)(1-d)\Psi_Q(M_{W,U})
-t\log_2K(X_i\mid W,U)
+(1-t)(1-d)\beta\ell_{\max}\right].
\]
\emph{Average the finite inequalities.}
Adding the expected potential change leaves this right side unchanged, by stationarity. The resulting expression is the average of the left side of \eqref{eq:shared-row}. Every row is at most $u_\star$, so its average is at most $u_\star$.

\emph{Remove the stationarity reduction.}
This proves the bound for every stationary source. The block-concatenation argument in \cref{lem:stationary} extends it to unrestricted channel capacity.
\end{proof}
Thus the computation only has to verify a finite set of explicitly defined inequalities. The probability tables and potential need not minimize this maximum: every proposal passing the complete test gives the claimed bound.

\subsection{Where computation enters: evaluate every finite row}\label{sec:shared-computation}
The analytic task is complete: \cref{thm:shared} has reduced capacity to a stated finite test. The numerical input consists of exact parameters $d,R,m,b,t$, normalized tables $Q,K$, and potential values $V$. Search changes these parameters to lower the largest row value; verification fixes them and checks every row.

For each fixed proposal, each entropy term in the row can be evaluated by pairs of masses. Fix an output prefix $v$ of length $m-1$ and let $p_0=M_{w,u}(v0)$, $p_1=M_{w,u}(v1)$. Substituting these masses in the definition of $\Psi_Q$ gives its prefix contribution
\[
p_0\log_2\frac{p_0}{(p_0+p_1)Q(0\mid v)}
+p_1\log_2\frac{p_1}{(p_0+p_1)Q(1\mid v)}.
\]
There are only two possible last bits, which explains both summands and the shared prefix mass $p_0+p_1$. Zero masses contribute zero. Summing this expression over all prefixes gives $\Psi_Q(M_{w,u})$.

The window-mask probabilities are computed by processing one raw bit at a time. For each partial trace, deletion leaves the trace unchanged and multiplies its mass by $d$; survival appends the bit and multiplies the mass by $1-d$. Every mask has exactly one sequence of these choices, so after $R$ steps the recursion counts it exactly once. Keep the masks meeting the cutoff and complete their traces with $u$ as defined above.

For each row, the checker computes an interval containing its exact value and keeps the upper endpoint, including the potential difference and the complete omitted-pattern term. The arithmetic rules are given in \cref{sec:arithmetic}. There are $2^R$ possible raw windows, $2^{m-b}$ outside survivor words, and two prepended bits, hence $2^{R+m-b+1}$ rows. If every upper endpoint is at most $u_\star$, the hypothesis of \cref{thm:shared} holds.

Every stationary input generates only states among these checked rows. The remainder covers every omitted mask and the potential has zero mean change. Hence the finite computation proves an unrestricted capacity bound by \cref{lem:stationary}. Its cost is exponential in $R$ and $m-b$, but independent of transmitted length. Recorded costs and the division between search and verification are reported in \cref{sec:reproduction}; the theorem does not assume that the search found a best table.

%% file: sections/roadmap_shared.tex
\subsection{Proof plan: count a window, bound what it misses}\label{sec:roadmap-shared}
We inspect $R$ input bits but need $m$ received bits. Some of those received bits can come from outside the window. The construction records $m-b$ outside survivors and counts masks leaving at least $b$ survivors inside. All counted masks use the \emph{same} outside sequence: changing deletions inside the window cannot change deletions after it.

\begin{enumerate}
\item Form the unnormalized distribution $M_{w,u}$ of the counted cases. Its missing mass is the explicit binomial probability $\beta$ in \eqref{eq:beta}.
\item \Cref{lem:shared-tail} bounds their information cost by $\Psi_Q(M_{w,u})$ and charges at most $\beta\ell_{\max}$ for the omitted cases. This covers every possible distant input.
\item Add a bounded function's change between overlapping windows. Its stationary average is zero, but it can reduce the largest local cost. Combine with input entropy when useful.
\item Check the single finite inequality \eqref{eq:shared-row} for every window, outside word, and prepended bit. \Cref{thm:shared} proves that its largest checked value upper-bounds capacity.
\end{enumerate}
The tables $Q,K$ and the bounded function $V$ are chosen to make that largest value small. They are numerical parameters of a sufficient test; they do not restrict the unknown input distribution. The use of a bounded potential follows the average-reward method \citep{Puterman94,HSPKS21}. The deletion-specific work is the common outside sequence and the explicit omitted probability.

\input{figures/intro_window}

%% file: figures/intro_window.tex
\begin{figure}[!htbp]\centering
\begin{tcolorbox}[width=.98\linewidth,colback=teal!3,colframe=teal!65!black,coltext=ink,boxrule=.7pt,arc=2pt,left=9pt,right=9pt,top=8pt,bottom=8pt]
\textbf{Upper-bound example: three inspected input bits, two required output bits}
\par\smallskip
The local entropy calculation uses the \emph{first two bits that survive from the start of the inspected window}. A three-bit input window need not supply both. We must bound this calculation even when one or both output bits come from later input positions.
\par\smallskip
\centering
\fbox{\strut inspected input $010$}\quad\big|\quad
\fbox{\strut later surviving sequence begins with $\mathbf 1$}
\par\smallskip\raggedright
Fix the later input and its deletion pattern, with first outside survivor $1$. Vary only deletions inside $010$. At deletion probability $d=1/2$, its eight deletion patterns still each have probability $1/8$. A cross marks a deleted input bit; a bold $\mathbf1$ in the output comes from outside the window.
\par\smallskip\centering
\begin{tabular}{@{}ccc@{}}\toprule
Inspected input after deletions & First two output bits & Probability\\\midrule
$010$ & $01$ & $1/8$\\
$\times10$ & $10$ & $1/8$\\
$0\times0$ & $00$ & $1/8$\\
$01\times$ & $01$ & $1/8$\\
$0\times\times$ & $0\mathbf1$ & $1/8$\\
$\times1\times$ & $1\mathbf1$ & $1/8$\\
$\times\times0$ & $0\mathbf1$ & $1/8$\\\midrule
$\times\times\times$ & Both bits come from outside & $1/8$\\\bottomrule
\end{tabular}
\par\smallskip\raggedright
\textbf{Seven cases are determined.} The first seven rows use at most the same one outside survivor. Their output-pair masses are $1/8,4/8,1/8,1/8$ for $00,01,10,11$. These masses sum to $7/8$; they are not renormalized. Among retained cases whose first output bit is zero, four of five have second bit one. This probability $4/5$ is used to compute the channel entropy subtracted in the converse; it need not equal the chosen comparison probability $1/2$.
\par\smallskip
\textbf{One case is bounded.} The last row has probability $1/8$. If the comparison table assigns $1/2$ to each next bit, its negative logarithm is always one, so this row contributes at most $1/8$ to the local expected negative logarithm.
\end{tcolorbox}
\caption{A retained-mask calculation for the upper bound. The same outside survivor is used in all seven retained cases, with their original probabilities. The omitted case has mass $1/8$ and is bounded separately. Repeating for outside survivor $0$ and every three-bit input covers all window states. The final converse also subtracts conditional channel entropy and checks the overlapping-window inequalities; this example illustrates its finite calculation.}\label{fig:intro-window}
\end{figure}

%% file: sections/posterior.tex
\section{A converse from posterior coding and entropy cancellation}\label{sec:posterior}
Fix a window of $R$ input bits. Fix the entire input sequence and let $\mu$ be the conditional distribution, over deletions alone, of its first $m$ survivors strictly after that window. This is a distribution on $2^m$ words. It varies with the specified input; the test will cover every such distribution. The information-rate expression $(1-d)\mathbb E\Psi_Q(P_F)$ is nonlinear in these probabilities because it subtracts a conditional entropy. We add a bounded potential whose entropy change cancels that term. The remaining upper expression is linear in $\mu$, so checking it for every fixed outside word covers every distribution. This section derives the cancellation and the resulting finite test.

\input{sections/roadmap_posterior}
\subsection{The distribution of the next surviving word}
Fix $0\le d<1$, a raw-window length $R\ge1$, and an output-word length $m\ge1$. A tail word $z=z_1\cdots z_m$ contains the first $m$ surviving bits strictly after the raw window. If a raw bit $a$ is inserted before this survivor sequence and is retained, the new first $m$ survivors are $az_1\cdots z_{m-1}$. Define $B_az$ to be this word. For a probability distribution $\mu$, the distribution $B_a\mu$ assigns to each word $y$ the sum of $\mu(z)$ over words satisfying $B_az=y$. Thus $B_a$ acts on words by substitution and on distributions by summing the probabilities of their preimages.

The prepended bit is deleted with probability $d$ and retained with probability $1-d$, independently of the old survivor word. On deletion the word remains unchanged; on retention it becomes $B_az$. The new word distribution is therefore
\[
T_a\mu=d\mu+(1-d)B_a\mu.
\]
The two disjoint events have total probability one, so this is a normalized distribution.

Now let $w=w_1\cdots w_R$ be the known raw window. Starting from the distribution $\mu$ of the next $m$ survivors after it, first insert $w_R$, then $w_{R-1}$, and continue to $w_1$. The resulting distribution of the first $m$ survivors from the beginning of the window is
\[
P_w\mu=T_{w_1}\cdots T_{w_R}\mu.
\]
The rightmost operator acts first. For example, $P_{01}\mu=T_0(T_1\mu)$: incorporate the nearer raw bit $1$ into the tail description before incorporating the earlier $0$. In general the operators do not commute.

Now prepend a raw bit $a$ and move the window backward. Write $w'=aw_1\cdots w_{R-1}$ and $c=w_R$ for its new window and outgoing bit. The outgoing bit becomes part of the tail, whose law is now $T_c\mu$. Expanding both operator products gives
\begin{equation}\label{eq:intertwining}
P_{w'}T_c=T_aP_w.
\end{equation}
Indeed, the left side is $T_aT_{w_1}\cdots T_{w_{R-1}}T_{w_R}$, exactly the right side. Both describe the same next-$m$-survivor law after the move. This equality concerns composition order, not commutation.

\subsection{The entropy change of one deletion mixture}
We next compare the old and new word entropies. Let $Z$ have any distribution $\nu$ on length-$m$ words. The variable $J$ records whether the prepended raw bit $a$ is deleted or retained; it is independent of $Z$, with probabilities $d$ and $1-d$. Define the new word by $W=Z$ on deletion and $W=B_aZ$ on retention. Its distribution is $T_a\nu$. This calculation applies to every $\nu$, including the distribution obtained after processing a particular raw window.

\paragraph{Worked example: the output does not determine the deletion indicator.}
Take $m=2$, $d=1/2$, and let the old word $Z$ be uniform on $00,01,10,11$. Prepend the raw bit $a=0$. On deletion the new word is still $Z$; on survival it is $0Z_1$, so the old last bit $Z_2$ is discarded. There are eight equally likely pairs consisting of the deletion indicator and the old word. Their new words have the following probabilities.
\begin{center}
{\small Uniform old word $Z\in\{00,01,10,11\}$; deletion probability $d=1/2$; prepended bit $a=0$.}\par\smallskip
\begin{tabular}{@{}ccc@{}}\toprule
New word $W$ & Probability & $\Pr(J=\mathrm{delete}\mid W)$\\\midrule
$00$ & $3/8$ & $1/3$\\
$01$ & $3/8$ & $1/3$\\
$10$ & $1/8$ & $1$\\
$11$ & $1/8$ & $1$\\\bottomrule
\end{tabular}
\end{center}
For example, $00$ arises from deleting the prepended zero before old word $00$, or retaining it before old word $00$ or $01$. Each explanation has probability $1/8$. One of the three deletes the prepended bit, explaining the conditional probability $1/3$. If $W$ begins with $1$, survival is impossible because it would have put a zero first.

The old entropy is $H(Z)=2$. Revealing whether the prepended bit survived leaves entropy two on deletion and one on survival, so $H(W\mid J)=(2+1)/2=3/2$. In the table, the first new bit is zero with probability $3/4$, and the second is a fair bit independent of the first. The chain rule therefore gives $H(W)=\hb(1/4)+1$. Subtracting the entropy with the selector revealed gives
\[
I(J;W)=H(W)-H(W\mid J)=\hb(1/4)-\frac12.
\]
Thus the change in word entropy is
\[
H(W)-H(Z)=\hb(1/4)-1=-\frac12+I(J;W).
\]
The negative term is the survival probability $1/2$ times the one bit of uncertainty discarded from $Z_2$. The selector information accounts for not being told which branch occurred. The following identity holds for every old-word distribution $\nu$, not just the uniform example.

\begin{lemma}[Selector entropy identity]\label{lem:selector}
For the law and selector just defined,
\[
H(T_a\nu)-H(\nu)
=-(1-d)H_\nu(Z_m\mid Z_1^{m-1})+I(J;W).
\]
\end{lemma}
\begin{proof}
\emph{Separate deletion from survival.}
Conditioning on the selector gives a weighted average:
\[
H(W\mid J)=dH(\nu)+(1-d)H(B_a\nu).
\]
The two weights are the probabilities of its two values. On survival, the first symbol $a$ is fixed and the remaining symbols are $Z_1^{m-1}$. Consequently $H(B_a\nu)=H_\nu(Z_1^{m-1})$.

\emph{Identify the discarded uncertainty.}
The chain rule, \cref{lem:entropy-chain}, also gives
$H(\nu)=H_\nu(Z_1^{m-1})+H_\nu(Z_m\mid Z_1^{m-1})$.
Finally $H(W)=H(W\mid J)+I(J;W)$ by \cref{lem:mixture}. Substitute the preceding expressions, subtract $H(\nu)$, and use $d+(1-d)=1$. The coefficient of the last-symbol conditional entropy is $-(1-d)$, proving the identity.
\end{proof}

\subsection{Bounding the conditional entropy of the deletion indicator}
After the entropy cancellation below, the rate bound will contain $-I(J;W)$. Since $I(J;W)=H(J)-H(J\mid W)$ and $H(J)=\hb(d)$, an upper bound on $H(J\mid W)$ gives an upper bound on $-I(J;W)$. For each raw window $w$, prepended bit $a$, and observed word $y$, choose probabilities $K(\mathrm{delete}\mid w,a,y)$ and $K(\mathrm{survive}\mid w,a,y)$ summing to one. We use their expected negative logarithm to upper-bound this conditional entropy.

Here $K$ is an auxiliary distribution for the deletion indicator; in the preceding section it described an input bit. The known prepended bit $a$ is now a conditioning argument.

The output word restricts which indicator values are possible. A word not beginning with $a$ cannot have arisen by survival, since survival prepends $a$. For such a word, set its deletion probability in $K$ to one. For a word beginning with $a$, choose both entries of $K$ positive. An impossible survival event is never evaluated inside a logarithm. In the worked example, the ideal selector code uses deletion probabilities $1/3,1/3,1,1$ in the four rows of the table. Its entries are the conditional probabilities of deletion and retention in that example. In the certificate, any code satisfying the stated support and normalization conditions gives a valid cross-entropy inequality.

Conditional cross entropy, \cref{lem:cross-entropy} \citep[Section~2.6]{CT06}, says $H(J\mid W)\le\E[-\log_2K(J\mid w,a,W)]$. Since $-I(J;W)=-H(J)+H(J\mid W)$, separating the two selector outcomes gives
\begin{align*}
-I(J;W)\le{}&-\hb(d)
+d\E_\nu[-\log_2K(\mathrm{delete}\mid w,a,Z)]\\
&+(1-d)\E_\nu[-\log_2K(\mathrm{survive}\mid w,a,B_aZ)].
\end{align*}
The different word arguments are essential: deletion shows $Z$, while survival shows $B_aZ$.

The output comparison distribution $Q$ assigns probabilities to the last of $m$ received symbols from the preceding $m-1$. Its logarithmic loss is $\ell_Q(z)=-\log_2Q(z_m\mid z_1^{m-1})$. To write the coming finite inequality, define the following sum of the output logarithmic loss and the two conditional logarithmic terms from $K$:
\begin{align}\label{eq:posterior-loss}
\mathcal L_{w,a}(z)={}&(1-d)\ell_Q(z)
+d[-\log_2K(\mathrm{delete}\mid w,a,z)]\notag\\
&+(1-d)[-\log_2K(\mathrm{survive}\mid w,a,B_az)].
\end{align}
The coefficient $1-d$ on the output term converts from survivors to input bits. The other two coefficients average the deletion and retention events for $J$. All three logarithmic terms are nonnegative because their arguments are probabilities at most one. This nonnegativity will permit the finite sums to be bounded term by term.

\subsection{Choosing a potential that removes the nonlinear term}
For this step, $\mu$ is the distribution of the first $m$ survivors after the raw window, and $\nu=P_w\mu$ is the distribution of the first $m$ survivors when the window is included. The rate expression uses the latter distribution. Choose one finite real coefficient $v_w(z)$ for each raw window $w$ and each outside survivor word $z$. Define the bounded potential\[
\mathcal V(w,\mu)=-H(P_w\mu)+\sum_{z\in\bits^m}v_w(z)\mu(z).
\]
The negative entropy is prescribed by the identity in \cref{lem:selector}; its change will cancel the conditional entropy in the rate bound. The coefficients $v_w(z)$ remain free parameters in the finite inequalities. Since $0\le H(P_w\mu)\le m$ and the second term is a probability-weighted average of finitely many coefficients, $\mathcal V$ is bounded uniformly over $w$ and $\mu$, including distributions with zero probabilities.

The potential is chosen before the unknown tail law is maximized over. Its entropy term is fixed by the channel identity; only the finite coefficients $v_w(z)$ are optimized. The next lemma carries out the cancellation and identifies precisely what the checker must evaluate.

\subsection{A finite test covering every tail distribution}
We explain the matrix notation in the finite test. For a tail word $z$, let $(P_w^\top\mathcal L_{w,a})(z)$ be the expected value of $\mathcal L_{w,a}$ when the tail is deterministically $z$ and the window mask is random. Averaging first over that mask and then over $z\sim\mu$ gives
\[
\E_{P_w\mu}\mathcal L_{w,a}
=\sum_z\mu(z)(P_w^\top\mathcal L_{w,a})(z).
\]
This is the finite interchange of the two sums; it is the meaning of the transpose here.

If the outside law assigns probabilities $\theta$ and $1-\theta$ to two words, linearity makes its upper expression the same weighted average of their row values. Thus bounding each row by $u_\star$ bounds every mixture by $\theta u_\star+(1-\theta)u_\star=u_\star$. The argument extends by summation to any distribution on the finite outside alphabet.

\begin{certbox}[title={The posterior-code certificate}]
For every window $w\in\bits^R$, prepended bit $a\in\bits$, and tail word $z\in\bits^m$, check
\begin{align}\label{eq:posterior-row}
G(w,a,z)={}&(P_w^\top\mathcal L_{w,a})(z)-\hb(d)\notag\\
&+d[v_{w'}(z)-v_w(z)]
+(1-d)[v_{w'}(B_cz)-v_w(z)]\le u_\star.
\end{align}
Here $w'=aw_1\cdots w_{R-1}$ and $c=w_R$. The output predictor $Q$ is shared across all raw windows. The selector code in $\mathcal L$ uses the prepended bit $a$; the potential's tail transition uses the outgoing bit $c$.
\end{certbox}

\begin{lemma}[The cancelled expression is an average of finite rows]\label{lem:posterior-linear}
For every $w\in\bits^R$, $a\in\bits$, and probability law $\mu$ on the $m$ outside survivors,
\[
(1-d)\Psi_Q(P_w\mu)+\mathcal V(w',T_c\mu)-\mathcal V(w,\mu)
\le\sum_z\mu(z)G(w,a,z).
\]
\end{lemma}
\begin{proof}
Put $\nu=P_w\mu$ and $h_\nu=H_\nu(Z_m\mid Z_1^{m-1})$. Expanding the rate expression and using \cref{lem:selector} gives
\[
\begin{aligned}
(1-d)\Psi_Q(\nu)&=(1-d)\mathbb E_\nu\ell_Q-(1-d)h_\nu,\\
-H(T_a\nu)+H(\nu)&=(1-d)h_\nu-I(J;W).
\end{aligned}
\]
The second line reverses the signs of the selector identity because the potential contains negative entropy. Equation~\eqref{eq:intertwining} identifies $T_a\nu$ with the new window's survivor distribution. Adding the lines cancels the two conditional-entropy terms. Applying the conditional cross-entropy bound for the deletion-decision table $K$ leaves at most $\mathbb E_{P_w\mu}\mathcal L_{w,a}-h_2(d)$.

It remains to expand the linear part of the potential. On deletion of the outgoing bit $c$, the tail word stays $z$; on survival it becomes $B_cz$. Hence its change is
\[
\sum_z\mu(z)\{d[v_{w'}(z)-v_w(z)]
 +(1-d)[v_{w'}(B_cz)-v_w(z)]\}.
\]
The old potential is subtracted once because $d+(1-d)=1$. Finally, average the window-mask expectation first for a fixed tail $z$, then over $z\sim\mu$. By the definition of $P_w^\top\mathcal L$, this gives the first term of $G$ for each $z$. The constant $-h_2(d)$ can also be averaged because $\sum_z\mu(z)=1$. Together these terms give exactly the claimed weighted sum of rows.
\end{proof}

\begin{theorem}[The finite entropy-cancellation test implies a capacity bound]\label{thm:posterior}
Fix $0\le d<1$ and integers $R,m\ge1$. Let $Q(\cdot\mid v)$ be positive and normalized for every $v\in\bits^{m-1}$. For every $w,a,y$, let $K(\cdot\mid w,a,y)$ be normalized, with both entries positive if $y_1=a$, and deletion probability one if $y_1\ne a$. Let all $v_w(z)$ be finite real numbers. If \eqref{eq:posterior-row} holds for every $(w,a,z)\in\bits^R\times\bits\times\bits^m$, then $C(d)\le u_\star$.
\end{theorem}
\begin{proof}
\emph{Cover every tail law by averaging the checked rows.}
By \cref{lem:posterior-linear}, for every tail law $\mu$,
\[
(1-d)\Psi_Q(P_w\mu)
+\mathcal V(w',T_c\mu)-\mathcal V(w,\mu)
\le\sum_z\mu(z)G(w,a,z)\le u_\star.
\]
The last inequality follows simply by multiplying each coordinate bound by the nonnegative mass $\mu(z)$ and adding; the masses sum to one. This is why deterministic tail words cover all tail distributions.

\emph{Cancel the potential along the physical state process.}
Apply this inequality to any stationary raw input, first conditioning on its entire raw sequence. Given that sequence, the outgoing bit $c$ is fixed and its deletion is independent of all deletions after it. The conditional outside-survivor law therefore changes from $\mu$ to $T_c\mu$ when the window moves backward. This assertion would not follow merely by conditioning on the finite window, because that window can be correlated with the unobserved input. On the state consisting of $w$ and this full-input conditional law $\mu$, the backward update is the deterministic map $(w,\mu)\mapsto(w',T_c\mu)$. This is the fixed transition in \cref{lem:potential}; it does not assert that an observed tail word determines $\mu$. Now average over the stationary raw sequence. The bounded potential has the same expected value before and after that shift, so its expected change is zero by \cref{lem:potential}. We conclude $(1-d)\E\Psi_Q(P_F)\le u_\star$. The survivor converse \eqref{eq:reward-bridge} bounds this source's information rate by that quantity.

\emph{Pass from stationary inputs to capacity.}
The source was arbitrary among stationary inputs. Finally \cref{lem:stationary} extends the bound to channel capacity.
\end{proof}

\subsection{Where computation enters: evaluate the finite maximum}\label{sec:posterior-computation}
The fixed numerical inputs are $d,R,m$, the positive normalized output table $Q$, the supported normalized deletion-decision table $K$, and the coefficients $v_w(z)$. Search seeks a small maximum of \eqref{eq:posterior-row}. Verification evaluates an upper bound on every row for the frozen proposal. No candidate input source is part of this upper-bound calculation.

For each window $w$, prepended bit $a$, and fixed outside survivor word $z$, the first term of \eqref{eq:posterior-row} averages $\mathcal L_{w,a}$ over the window masks. We compute this expectation by grouping masks with the same retained word.

Let $A_k(v)$ be the total probability of masks retaining exactly $k<m$ window symbols with trace $v$. Let $C(v)$ be the probability of masks retaining at least $m$ symbols whose first $m$ survivors are $v$.
These masses can be computed without listing the masks separately. Start with probability one at the empty trace. On reading each raw bit, deletion multiplies a trace's mass by $d$ and keeps its trace unchanged; survival multiplies by $1-d$ and appends the bit. Once a trace reaches length $m$, keep only its first $m$ symbols and leave that state unchanged when later bits are processed. The two branch probabilities then sum to one. Induction on the processed raw length shows that each state contains exactly the mass of its described masks.

A mask retaining $k<m$ window symbols needs exactly $\ell=m-k$ tail symbols. If its trace is $v$, its full word is $vz_1^\ell$. Group masks by this number $\ell$. For a tail prefix $t$ of length $\ell$, define
\[
F_\ell(t)=\sum_{v\in\bits^{m-\ell}}A_{m-\ell}(v)\mathcal L_{w,a}(vt)
\qquad(1\le\ell\le m).
\]
Each term is mask probability times its loss. For $\ell=0$, no tail is needed, so instead set $F_0(\varnothing)=\sum_{v\in\bits^m}C(v)\mathcal L_{w,a}(v)$. Every mask belongs to exactly one of these groups. Therefore
\begin{equation}\label{eq:prefix-contraction}
(P_w^\top\mathcal L_{w,a})(z)=\sum_{\ell=0}^mF_\ell(z_1^\ell).
\end{equation}
At level $\ell$, there are $2^\ell$ prefixes $t$ and $2^{m-\ell}$ words $v$, so computing the whole level uses $2^m$ nonnegative products. The $m+1$ levels require $O(m2^m)$ work. Using upper interval endpoints for the probabilities and logarithmic terms gives an upper bound because all these products are nonnegative. The signed potential differences are evaluated separately. Finally the complete coordinate check covers $2^R$ windows, two actions, and $2^m$ tails, hence $2^{R+m+1}$ rows.

The $O(m2^m)$ count is for constructing the displayed sums for one fixed $(w,a)$ from its mask masses, not for checking the whole certificate. There are $2^R$ windows and two prepended bits, and the potential terms must also be checked for all $2^m$ outside words. The cost therefore grows exponentially with the two window parameters even though it is independent of transmitted length. Recorded verification times appear in \cref{sec:reproduction}.

In the published cover, posterior anchors supply the selected upper bounds on cells from $d=27/50$ to $d=1$; the endpoint $C(1)=0$ is exact. These selections describe the numerical certificate, not the validity domain of \cref{thm:posterior}.

%% file: sections/roadmap_posterior.tex
\subsection{Proof plan: cancel entropy before maximizing}\label{sec:roadmap-posterior}
This construction retains all $m$ outside survivors, so no window masks are omitted. The difficulty is their unknown probability distribution. Taking its average is easy for expected loss, but the conditional entropy subtracted in \eqref{eq:reward-bridge} is nonlinear in those probabilities.

\Cref{lem:selector} supplies the needed cancellation. Prepending a retained input bit inserts one fixed symbol and removes the last symbol from the first-$m$-survivor word. Its entropy change therefore contains exactly the conditional entropy appearing in the rate bound. The proof proceeds as follows:
\begin{enumerate}
\item Derive the entropy change by conditioning on whether the prepended bit was deleted (\cref{lem:selector}).
\item Upper-bound the remaining uncertainty about this deletion decision with an auxiliary conditional distribution $K$, using cross entropy.
\item Include negative word entropy in a bounded potential. \Cref{lem:posterior-linear} shows explicitly that the conditional entropies cancel and the resulting bound is linear in the outside probabilities.
\item Check \eqref{eq:posterior-row} for every fixed outside word. Every outside distribution is an average of these words; \cref{thm:posterior} averages the rows and cancels the potential.
\end{enumerate}
Thus the finite maximum is taken \emph{after} the entropy cancellation. No enumeration of input distributions is needed. The comparison-distribution and potential framework is standard \citep{DMP07,HSPKS21}; the identity below specifies the entropy term used in this deletion-channel test. The complete row calculation follows the proof in the same section.

\input{figures/upper_proof_routes}

%% file: figures/upper_proof_routes.tex
\par\medskip\noindent\begin{minipage}{\linewidth}\centering
\begin{tikzpicture}[>=Stealth,text=ink,
 box/.style={draw=ink!50,rounded corners=2pt,align=center,text width=6.05cm,inner sep=6pt,font=\small\hyphenpenalty10000},
 wide/.style={box,text width=13.1cm},arr/.style={->,thick,draw=ink!60}]
\node[wide,fill=ink!3] (base) {\textbf{Common entropy bound: \cref{sec:survivor}}\\
\Cref{lem:predictor}: upper-bound output entropy.\quad
\Cref{lem:increments}: lower-bound conditional output entropy.\\
Their difference bounds every stationary information rate; \cref{lem:stationary} removes the stationarity restriction.};
\node[box,fill=ink!3,minimum height=2cm,below left=.5cm and .20cm of base.south] (s) {\textbf{Shared outside survivors: \cref{sec:shared}}\\
\Cref{lem:shared-tail}: count window masks using one common outside sequence; bound omitted probability explicitly.};
\node[box,fill=ink!3,minimum height=2cm,below right=.5cm and .20cm of base.south] (p) {\textbf{Entropy cancellation: \cref{sec:posterior}}\\
\Cref{lem:selector}: condition on whether a prepended input bit survives. The entropy change cancels the nonlinear conditional entropy.};
\node[box,fill=ochre!14,minimum height=2.35cm,below=.5cm of s] (st) {\textbf{Check \eqref{eq:shared-row} in every row}\\
Known mask probabilities, auxiliary distributions, a remainder, and a potential change.\\
\Cref{thm:shared}: all rows $\le u$ imply \mbox{$C(d)\le u$}.};
\node[box,fill=ochre!14,minimum height=2.35cm,below=.5cm of p] (pt) {\textbf{Check \eqref{eq:posterior-row} in every row}\\
The remaining bound is linear in the outside probabilities; fixed outside words suffice.\\
\Cref{thm:posterior}: all rows $\le u$ imply \mbox{$C(d)\le u$}.};
\draw[arr](base.south)-|(s.north);\draw[arr](base.south)-|(p.north);
\draw[arr](s)--(st);\draw[arr](p)--(pt);
\end{tikzpicture}

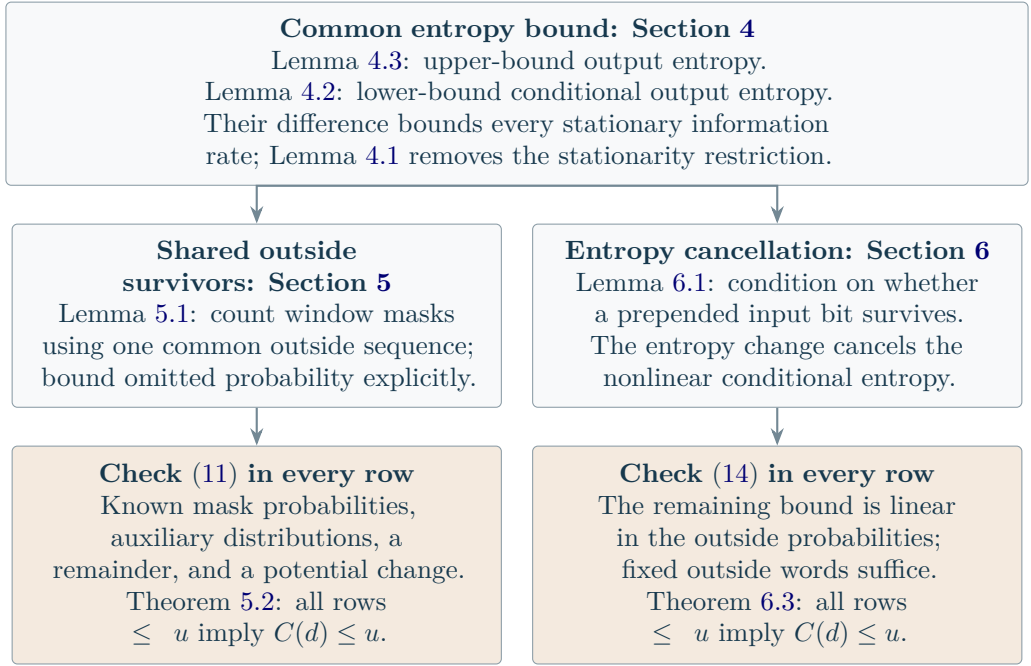
\captionof{figure}{Dependencies of the two upper bounds based on consecutive surviving bits. Both begin with the same entropy subtraction, but turn its dependence on distant input positions into different finite tests. A potential is a bounded function whose expected change cancels under stationarity. The shaded bottom boxes are the expressions evaluated by the numerical checkers; each row is a required inequality, not a sampled input. The threshold $u$ is in bits per input bit.}\label{fig:upper-proof-routes}
\end{minipage}\par\medskip

%% file: sections/small_deletion.tex
\section{A run-based converse for small deletion probabilities}
\label{sec:small-deletion}
A run is a maximal block of equal input bits; $000110$, for example, has runs $000$, $11$, and $0$. Deletions create two kinds of ambiguity. Within a run, several deleted positions can give the same trace: either deletion in $00$ leaves $0$. A deleted singleton also merges its same-bit neighbors: deleting the middle $1$ in $00100$ leaves $0000$. We bound the entropy of these ambiguities and the output entropy to obtain a converse.

The converse must cover every input law, including dependent run lengths and arbitrarily long runs. We first reverse selected deletions to make the remaining run counts recoverable. We bound the reversed-mask entropy, retaining dependence between neighboring masks. A probability-code inequality then removes the unknown run law. Its finite verification includes an analytic bound for every unrepresented run length.

Throughout this section $d$ is the deletion probability, so a bit survives with probability $1-d$. Logarithms have base two unless written $\ln$. We use the entropy of nonnegative masses $w_1,\ldots,w_r$ in the form
\[
\Phi(w_1,\ldots,w_r)=W\log W-\sum_{i=1}^r w_i\log w_i,
\qquad W=\sum_iw_i.
\]
For $W>0$ this is $W H(w_1/W,\ldots,w_r/W)\ge0$; at $W=0$ it is zero. We use $0\log0=0$ throughout. Its derivative in a positive mass $w_i$ is $\log(W/w_i)\ge0$, so it is increasing in each mass, with boundary values defined by continuity.

\input{sections/roadmap_runs}
\subsection{Run frequencies and the class of inputs needed for the converse}
Fix a distribution on length-$q$ binary blocks, concatenate independent blocks with that law, and choose the origin uniformly within a block. The input is stationary: every finite-word distribution is unchanged by a shift. Its run lengths may be dependent. By \cref{lem:stationary}, a common upper bound for these sources bounds capacity. We first fix $q$ and its law and let the transmitted length $n$ tend to infinity; only afterward do we increase $q$.

Unless the source is constant, for each fixed bit value the probability that a complete block consists entirely of that value is strictly below one. A long run requires many consecutive constant blocks of the same value, whose probability decreases geometrically by independence. Thus run lengths have exponential tails. Constant sources have zero information rate and are handled separately.

The two incomplete boundary runs have bounded expected lengths for this fixed source. Describing their lengths and the run count costs $O(\log n)$ bits in total; the first few complete context runs have bounded expected descriptions. These costs vanish after division by $n$. Their bounds may depend on $q$ and its law, because the final information-rate bound will not.

For this fixed nonconstant source, define the information rate $\overline I=\limsup_n n^{-1}I(X_1^n;Y_n)$. Define its output-entropy rate by $\overline H(Y)=\limsup_n n^{-1}H(Y_n)$. For a stationary mask $Z$, write $\overline H(Z\mid X)=\limsup_n n^{-1}H(Z_1^n\mid X_1^n)$. All are per input bit. The input entropy rate is $h(X)=\lim_n n^{-1}H(X_1^n)$.

Let $\rho=\Pr\{X_0\ne X_{-1}\}>0$ be the expected number of run beginnings per input bit. Condition on a run beginning at the origin and denote this law by $\Pr^\circ$; write $L_0,L_1,\ldots$ for its complete run lengths. This law is stationary under shifts by one run. For any bounded function of finitely many successive run lengths, compare its sum over beginnings in $[1,n]$ with the sum after advancing each beginning to its successor. At most the first and last terms differ. Taking expectations, dividing by $n$, and using translation stationarity proves invariance under this run shift; no sample-path frequency limit is assumed.

Define the per-input-bit frequencies
\[
\xi_\ell=\rho\Pr^\circ\{L_0=\ell\},\qquad
z_{\ell r}=\rho\Pr^\circ\{L_0=\ell,L_1=r\}.
\]
Here $\xi_\ell$ counts length-$\ell$ runs and $z_{\ell r}$ counts consecutive lengths $\ell,r$, both per input bit. Runs partition the input, and each has one predecessor and one successor. Hence
\[
\sum_\ell\ell\xi_\ell=1,\qquad
\sum_rz_{\ell r}=\sum_rz_{r\ell}=\xi_\ell.
\]
These identities can be obtained from expected counts in $[1,n]$: length-weighted run counts cover all but the two boundary pieces, and incoming and outgoing pair counts differ at the boundary only. Their expected boundary costs are bounded for the fixed block source. Divide by $n$ and let $n$ increase. Consequently $\rho\E^\circ L_0=1$ and $\sum_\ell\xi_\ell=\rho$: these are frequencies per input bit, not probability distributions. Write $L_{-1}$ for the preceding run length under the same law.

To combine the channel bound with $I(X_1^n;Y_n)\le H(X_1^n)$, we need an input-entropy bound in these frequencies. We obtain one by coding the next length from the preceding length. This does not assume that the actual run process is Markov.

Describe the initial bit, run count, boundary lengths, and first complete run. For each later run, assign conditional probability $z_{\ell r}/\xi_\ell$ after length $\ell$. Only rows with $\xi_\ell>0$ are used; at a zero-frequency length, choose any probability law. For a fixed run count, products of these probabilities form a normalized distribution. Cross entropy, \cref{lem:cross-entropy}, bounds the true entropy by its expected negative logarithm even if earlier runs affect the true conditional law. The descriptions already discussed cost $o(n)$, so
\begin{equation}
\overline I\le h(X)\le\rho H^\circ(L_1\mid L_0)
=-\sum_{\ell,r}z_{\ell r}\log\frac{z_{\ell r}}{\xi_\ell}.
\label{eq:run-input-entropy}
\end{equation}
The passage from finite descriptions to this sum is justified by their nonnegative expected code lengths. The next-run entropy is finite: comparison with the geometric code $2^{-r}$ bounds it by $\E^\circ L_1<\infty$. Sum over bounded run lengths first, then increase the cutoff. Boundary descriptions remain $o(n)$ for the fixed source. The first inequality follows from information being at most input entropy. The equality uses pair probability $z_{\ell r}/\rho$ and conditional probability $z_{\ell r}/\xi_\ell$. Omitting the preceding length gives the weaker bound $h(X)\le-\sum_\ell\xi_\ell\log(\xi_\ell/\rho)$, again without a run-independence assumption.

\subsection{A modified mask whose run counts can be parsed}
Let $D$ be the deletion mask, $X$ the input and $Y$ its trace. Independence of $D$ and $X$, and the fact that $(X,D)$ determines $Y$, give $I(X;Y)=H(Y)-H(D)+H(D\mid X,Y)$; see also \citet[eq.~(2), Example~3]{HOS16}. To obtain the identity, expand $H(Y,D\mid X)$ first as $H(D\mid X)+H(Y\mid X,D)=H(D)$ and then as $H(Y\mid X)+H(D\mid X,Y)$. A converse therefore needs an upper bound on the residual mask entropy. We split $D$ into a parsable mask and a record of reversed deletions.

Following the extended-run modification of \citet[Section IV-B]{KM10}, enlarge each input run by adjoining the input bit immediately before it and the input bit immediately after it. If the original mask has at least two deletions in this extended run, reverse every deletion inside the original run. Denote the remaining mask by $\widehat D$ and the reversed deletions by $Z^{\rm err}=D\mathbin\oplus\widehat D$. Thus $D$ is the disjoint union of its retained and reversed deletions.

Each input run now loses at most one symbol. A retained deletion has no original deletion in either adjacent boundary bit. If a singleton is deleted, neither neighbor has a retained deletion: any deletion there would meet the singleton in its extended run and be reversed.

For example, take consecutive complete runs $00$, $1$, $000$. Deleting only the singleton joins undeleted parents of lengths $2$ and $3$. Deleting both that singleton and the last bit of $00$ instead causes both deletions to be reversed. The error mask $Z^{\rm err}$ records those two positions; its entropy will pay for the modification.

\begin{lemma}[Extended-run parsing bound]
\label{lem:extended-run-parser}
For the fixed source above and $0<d\le1/2$, put
\[
w_\ell=\ell d(1-d)^{\ell+1}\log\ell.
\]
Let $\zeta_\ell$ be the entropy of the following law on binary length-$\ell$ vectors: its zero vector has mass $(1-d)^\ell+\ell d(1-d)^{\ell+1}$; each singleton has mass $d(1-d)^{\ell-1}[1-(1-d)^2]$; each vector of weight $j\ge2$ has mass $d^j(1-d)^{\ell-j}$. Then
\begin{align}
\overline I&\le\overline H(Y)-h_2(d)+\overline H(Z^{\rm err}\mid X)+\sum_\ell\xi_\ell w_\ell,\notag\\
\overline H(Z^{\rm err}\mid X)&\le\sum_\ell\xi_\ell\zeta_\ell.
\label{eq:modified-mask-upper}
\end{align}
Moreover, $\zeta_\ell\le\ell h_2(d)$.
\end{lemma}
\begin{proof}
\textbf{Step 1: recover deletion counts when the mask has the required structure.}
Suppose a mask deletes at most one bit per run, and every deleted singleton has completely undeleted adjacent runs. Compare an input run of length $\ell$ with the same-bit output run of length $k$. With one parent, $k=\ell$ or $\ell-1$, fixing the deletion count. Multiple parents can join only across deleted opposite-bit singletons; all parents are undeleted, so $k>\ell$. Their strictly increasing cumulative lengths identify the final parent. An opposite first output bit signals a deleted initial singleton. Repeating determines every run count. Disclose the two incomplete boundary runs and their deletion masks; their bounded expected lengths and length descriptions cost $o(n)$ for the fixed source.

\textbf{Step 2: relate the parsable trace to the actual observation.}
To transfer this parser to the observed trace $Y$, we prove that $Y$ and the modified trace $\widehat Y$ determine each other given $X,Z^{\rm err}$.

The parser applies to $(X,\widehat Y)$. For the reverse direction, use the known word
\[
W=\operatorname{delete}(X,Z^{\rm err}).
\]
Removing the retained deletions from $W$ gives $Y$. Within an original run, either no deletion was reversed or all were. Thus a run shortened by $Z^{\rm err}$ has no retained deletion.

A merged $W$-run consists of original same-bit runs joined across runs erased by $Z^{\rm err}$. Each contributing run borders an erased original boundary bit. Any deletion in that contributor would therefore have been reversed. Hence merged $W$-runs have no retained deletions; a $W$-run with a retained deletion is an intact original run and loses at most one bit.

If an intact singleton of $W$ has a retained deletion, its neighboring original boundary bits survive. Other deletions in those neighboring runs were reversed, so the neighboring $W$-runs have no retained deletions, even if shortened or merged. The parser therefore applies to $(W,Y)$. Its positive counts belong to intact original runs; all other original counts are zero. These counts construct $\widehat Y$. Conversely, parsing $(X,\widehat Y)$ and adding the known reversed counts constructs $Y$. Neither direction requires the retained deletion positions.

\textbf{Step 3: bound the additional entropy by the error-mask entropy.}

Apply the chain rule, the interconversion just proved, and conditioning-reduces-entropy:
\begin{align*}
H(D\mid X,Y)
&\le H(Z^{\rm err}\mid X,Y)+H(D\mid X,Y,Z^{\rm err})\\
&\le H(Z^{\rm err}\mid X)+H(\widehat D\mid X,\widehat Y,Z^{\rm err})\\
&\le H(Z^{\rm err}\mid X)+H(\widehat D\mid X,\widehat Y).
\end{align*}
Given $Z^{\rm err}$, each of $D$ and $\widehat D$ determines the other. Disclosing incomplete boundary runs in the finite-word argument adds $o(n)$ entropy and does not affect the rate bound.

\textbf{Step 4: bound the entropy of retained deletion positions.}

The parser identifies the runs with retained deletions. A length-$\ell$ run has at most $\ell$ possible positions, giving conditional entropy at most $\log\ell$. A specified deletion remains exactly when the other $\ell-1$ run bits and both outside boundary bits survive, an event of probability $d(1-d)^{\ell+1}$. The $\ell$ positions are disjoint alternatives. Their total probability times $\log\ell$ is $w_\ell$, and multiplication by $\xi_\ell$ converts this to entropy per input bit.

\textbf{Step 5: evaluate the entropy of each run's error vector.}

For a length-$\ell$ run, the error vector is zero if no bit was deleted, or if exactly one was deleted and both outside bits survive. These events have total probability $(1-d)^\ell+\ell d(1-d)^{\ell+1}$. A specified singleton error requires its unique deletion and at least one outside deletion, giving $d(1-d)^{\ell-1}[1-(1-d)^2]$. Every specified weight-$j\ge2$ original mask is wholly reversed, with probability $d^j(1-d)^{\ell-j}$. These cases give $\zeta_\ell$. The chain rule bounds joint error entropy by the sum of the marginal run entropies, yielding $\overline H(Z^{\rm err}\mid X)\le\sum_\ell\xi_\ell\zeta_\ell$ without independence.

This law transfers mass $d(1-d)^{\ell+1}$ from each singleton to zero. When $d\le1/2$, zero initially has at least each singleton's mass. Transferring $x$ from mass $q$ to mass $p\ge q$ changes their entropy with derivative $\log_2((q-x)/(p+x))\le0$. Hence $\zeta_\ell\le\ell h_2(d)$. The infinite run sums are finite: $\zeta_\ell\le\ell h_2(d)$ and $w_\ell\le\ell\log_2\ell$, whose expectations are finite under the exponential run tail. Apply the finite-word bound first, take expectations, and then divide by input length. Nonnegative sums pass to the limit by increasing finite cutoffs; the displayed summable bounds control the omitted lengths. This gives \eqref{eq:modified-mask-upper}.
\end{proof}

For evaluation, put $\eta(t)=-t\log_2t$. Only the zero and singleton probabilities differ from the independent mask law, so
\begin{align}
\zeta_\ell={}&\ell h_2(d)
+\eta((1-d)^\ell+\ell d(1-d)^{\ell+1})-\eta((1-d)^\ell)\notag\\
&+\ell\{\eta(d(1-d)^{\ell-1}[1-(1-d)^2])-\eta(d(1-d)^{\ell-1})\}.
\label{eq:zeta-row}
\end{align}
There is one zero vector and $\ell$ singleton vectors, explaining their respective coefficients. The higher-weight vectors cancel in this subtraction.

\subsection{An output-entropy bound from run counts}
To bound $\overline H(Y)$, let $y_k$ count length-$k$ output runs per input bit and assign them any positive probability law $R_k$. For a fixed total trace length, products of the $R_k$ over its compositions sum to at most one: they are disjoint events in which independent lengths first reach that total. Completing this subprobability law by an unused symbol permits cross entropy. Describing the trace length and first bit costs $O(\log n)$, hence
\[
\overline H(Y)\le\sum_{k\ge1}y_k[-\log_2R_k].
\]
Each length-$k$ output run incurs negative logarithm $-\log_2R_k$, explaining the frequency coefficient $y_k$.

\begin{lemma}[Output-run counts from input-run frequencies]\label{lem:output-run-counts}
For the fixed stationary source, let $y_k$ be the number of length-$k$ output runs per input bit and $\xi_\ell$ the number of length-$\ell$ input runs per input bit. Then
\begin{align}
\sum_kky_k&=1-d,\notag\\
\sum_ky_k&\le\sum_\ell\xi_\ell[1-(1+(1-d)^2)d^\ell],\notag\\
y_k&\ge\sum_{\ell\ge k}\xi_\ell(1-d)^2\binom\ell k(1-d)^kd^{\ell-k}.
\label{eq:run-count-constraints}
\end{align}
\end{lemma}
\begin{proof}
The first line counts surviving bits. For the second, count every nonempty input run, then subtract one for each empty run flanked by nonempty neighbors: it joins two same-bit runs. These joins are distinct, and each neighbor survives with probability at least $1-d$, giving reduction at least $(1-d)^2d^\ell$. For the third, a length-$\ell$ run with exactly $k$ survivors and nonempty neighbors gives a distinct length-$k$ output run. The binomial probability times the neighbor lower bound proves the line. Only deletion coins in distinct runs are independent; their lengths need not be.
\end{proof}

To use these three relations, the output-code coefficients must have matching signs.

Choose the output code so that its length has an affine part and a nonpositive correction:
\[
-\log_2R_k=A_e+\gamma k+\Delta_k,
\qquad A_e\ge0,\quad\Delta_k\le0.
\]
The code constructed below has this form. In $A_e\sum_ky_k+\gamma\sum_kky_k+\sum_k\Delta_ky_k$, use the upper count bound for $A_e\ge0$, the exact bit count for $\gamma$, and the lower count bound for $\Delta_k\le0$. Together with \eqref{eq:modified-mask-upper}, this gives $\overline I\le B_0+\sum_\ell\xi_\ell g_\ell$, where
\begin{align}
B_0&=(1-d)\gamma-h_2(d),\notag\\
g_\ell&=\zeta_\ell+w_\ell+A_e[1-(1+(1-d)^2)d^\ell]
+\sum_{k=1}^\ell\Delta_k(1-d)^2\binom\ell k(1-d)^kd^{\ell-k}.
\label{eq:small-run-reward}
\end{align}
The terms $\zeta_\ell,w_\ell$ bound the error-mask and retained-position entropies. The other terms bound output description length. We will remove the unknown frequencies using $\sum_\ell\ell\xi_\ell=1$.

\subsection{Information shared by adjacent error masks}
The modification couples neighboring error vectors through their boundary deletion coins. Their joint entropy is therefore at most the sum of their marginal entropies minus their mutual information. We retain a computable lower bound on this subtraction.

For example, take adjacent singleton runs at $d=1/10$, with independent outside coins. Each reversed-mask bit has probability $d[1-(1-d)^2]=19/1000$ of being one, whereas both are one with probability $d^2=1/100$. Thus the reversed masks are dependent despite independent original coins. To retain this dependence with a fixed alphabet, classify an error vector relative to its facing boundary: $0$ for zero; $1$ for a boundary singleton; $2$ for another singleton; $3$ for weight at least two with boundary bit zero; and $4$ for weight at least two with boundary bit one.

\textbf{Step 1: keep a small, sufficient description of each error vector.}
These categories concern the reversed mask $Z^{\rm err}$. Inside a run, it takes the original mask if the extended run contains at least two original deletions, and is zero otherwise. The facing boundary is where neighboring modifications share coins.

For a length-$\ell$ run, let $A_\ell(t,b\mid c)$ be the probability of category $t$ and original boundary-mask bit $b$, conditional on the adjacent run's original boundary-mask bit $c$. The opposite outside bit is still random. Its nonzero entries are
\begin{align}
A_\ell(0,0\mid c)&=(1-d)^\ell[1+(\ell-1)d\mathbf1_{c=0}],\notag\\
A_\ell(0,1\mid c)&=d(1-d)^\ell\mathbf1_{c=0},\notag\\
A_\ell(1,1\mid c)&=d(1-d)^{\ell-1}[d+c(1-d)],\notag\\
A_\ell(2,0\mid c)&=(\ell-1)d(1-d)^{\ell-1}[d+c(1-d)],\notag\\
A_\ell(3,0\mid c)&=(1-d)-(1-d)^\ell-(\ell-1)d(1-d)^{\ell-1},\notag\\
A_\ell(4,1\mid c)&=d-d(1-d)^{\ell-1}.
\label{eq:five-category-table}
\end{align}
For category zero with $b=0$, either all run bits survive or one nonboundary deletion remains retained. The latter requires $c=0$ and survival of the other outside bit, giving $(\ell-1)d(1-d)^\ell\mathbf1_{c=0}$. With $b=1$, zero error requires the unique boundary deletion and both outside bits surviving. These are the first two lines.

A singleton error requires one original deletion and an outside deletion. Given $c$, the outside event has probability $d+c(1-d)$; there is one location for category $1$ and $\ell-1$ for category $2$. For category $3$, subtract the zero mask and nonboundary singletons from probability $1-d$ that the facing bit survives. For category $4$, subtract its singleton from probability $d$ that the facing bit is deleted. All masks of weight at least two are reversed.

\textbf{Step 2: join the two conditional tables without recounting a coin.}

For adjacent lengths $\ell,r$, orient both categories toward their common boundary. Sum over their original boundary bits $b,c$ to obtain
\begin{equation}
P_{\ell r}(t,u)=\sum_{b,c\in\{0,1\}}A_\ell(t,b\mid c)A_r(u,c\mid b).
\label{eq:pair-category-table}
\end{equation}
Each factor includes its own boundary coin once and conditions on the other. The remaining coins are independent; summing removes the two boundary bits.

\textbf{Step 3: turn the category dependence into an entropy saving.}

Let $c_{\ell r}=I(P_{\ell r})$ be the mutual information of the two categories under the table. It is nonnegative by \cref{lem:logsum}. For consecutive full error vectors $E_i$, the chain rule gives total entropy $\sum_iH(E_i)-\sum_iI(E_i;E_1,\ldots,E_{i-1})$. Keeping only $E_{i-1}$ decreases the subtracted information. Replacing both vectors by their facing categories decreases it again by \cref{lem:data-processing}. Apply this argument conditional on the input, average over adjacent lengths, and divide by input length:
\begin{equation}
\overline H(Z^{\rm err}\mid X)
\le\sum_\ell\xi_\ell\zeta_\ell-\sum_{\ell,r}z_{\ell r}c_{\ell r}.
\label{eq:pair-credit}
\end{equation}
The frequency $z_{\ell r}$ counts these pairs per input bit. Left- and right-facing categories of a run may differ, but both reductions start from its full error vector and are compatible.

\subsection{A further entropy reduction across a known middle mask}
To use one more neighboring run, retain the chain-rule term involving the two outer error vectors conditional on the full middle vector. Coarsening that conditioning could change the conditional mutual information in either direction.

For three consecutive error vectors, expand
$I(E_i;E_{i-1},E_{i-2})=I(E_i;E_{i-1})+I(E_i;E_{i-2}\mid E_{i-1})$.
The first term supplies the pair correction. Conditional data processing permits replacing only the outer vectors by their inward-facing categories $T,U$. For lengths $p,j,r$, the extra subtraction is $\delta_j(p,r)=I(T;U\mid Z^{\rm err}_{\rm middle})$. With $t_{pjr}=\rho\Pr^\circ\{L_{-1}=p,L_0=j,L_1=r\}$, its rate contribution is $\sum_{p,j,r}t_{pjr}\delta_j(p,r)$. Each term describes a triple with middle length $j$; the frequencies count its middle run once. This subtraction uses the original entropy chain rule, so overlapping triples do not double-count a separate entropy term.

We now derive a finite formula while retaining that full conditioning. For a nonnegative joint table $J$ with total mass $M$, row sums $J_{t+}$ and column sums $J_{+u}$, define
\[
\mathcal I(J)=\sum_{t,u}J_{tu}\log J_{tu}
-\sum_tJ_{t+}\log J_{t+}-\sum_uJ_{+u}\log J_{+u}+M\log M.
\]
Expanding the mutual information of $J/M$ gives $\mathcal I(J)=M I(J/M)\ge0$, with value zero at $M=0$. The four terms respectively account for joint masses, row masses, column masses and total mass. Keeping the factor $M$ is essential: conditional information averages over middle-vector probabilities. If $J$ contains the masses for one middle vector, this is its probability-weighted conditional-information contribution.

\textbf{Step 1: separate the possible complete middle error vectors.}
Form a separate joint table for each complete middle error vector and add its $\mathcal I$ value; averaging the tables first would change the conditioning.

Write $A=A_p$ and $B=A_r$ for the outer-run tables just derived: $A(t,a\mid b)$ is the probability that the left outer run has category $t$ and original boundary-mask bit $a$, conditional on the adjacent middle bit $b$; $B$ has the corresponding right-outer meaning. For $j\ge2$, let $J_{\rm zero}$ be the joint category table when the entire middle error vector is zero. This happens either when all middle bits survive or when one is deleted and both outer boundary bits survive. Splitting that singleton between first, last, and the $j-2$ interior positions gives
\begin{align*}
J_{\rm zero}(t,u)={}&(1-d)^j\Big(\sum_aA(t,a\mid0)\Big)\Big(\sum_cB(u,c\mid0)\Big)\\
&+d(1-d)^{j-1}\{A(t,0\mid1)B(u,0\mid0)
+A(t,0\mid0)B(u,0\mid1)\\
&\hspace{36mm}+(j-2)A(t,0\mid0)B(u,0\mid0)\}.
\end{align*}
When all middle bits survive, sum the outer facing bits $a,c$ freely. A retained middle singleton requires $a=c=0$; its first, last, or interior location gives conditional arguments $(1,0)$, $(0,1)$, or $(0,0)$, respectively. Each singleton mask has probability $d(1-d)^{j-1}$.

\textbf{Step 2: account for singleton and multiple-error middle vectors.}
A singleton error requires that the middle deletion was reversed, so at least one outside boundary bit must be deleted.

For a specified singleton with first and last mask bits $(b,e)$, the category table is
\[
J_{be}(t,u)=d(1-d)^{j-1}
\sum_{(a,c)\ne(0,0)}A(t,a\mid b)B(u,c\mid e),
\quad (b,e)\in\{(1,0),(0,1),(0,0)\}.
\]
The excluded pair $(a,c)=(0,0)$ belongs to $J_{\rm zero}$. A middle error of weight at least two fixes its entire original mask; the remaining left and right coins are independent and contribute zero information. Thus
\begin{equation}
\delta_j(p,r)=\mathcal I(J_{\rm zero})+\mathcal I(J_{10})+\mathcal I(J_{01})+(j-2)\mathcal I(J_{00}).
\label{eq:full-bridge}
\end{equation}
The $j-2$ interior singletons are distinct conditioning events with equal contributions; their tables are not merged.

\textbf{Step 3: handle a one-bit middle run separately.}
For $j=1$, the first and last positions are the same coin, so the preceding two-boundary formula cannot be substituted directly. A zero middle error means either this bit survives or it is deleted while both outside bits survive; a nonzero error means it is deleted with at least one outside deletion. Accordingly
\begin{align*}
J_{\rm zero}(t,u)&=(1-d)\Big(\sum_aA(t,a\mid0)\Big)\Big(\sum_cB(u,c\mid0)\Big)
+dA(t,0\mid1)B(u,0\mid1),\\
J_1(t,u)&=d\sum_{(a,c)\ne(0,0)}A(t,a\mid1)B(u,c\mid1).
\end{align*}
The two possible error vectors give $\delta_1(p,r)=\mathcal I(J_{\rm zero})+\mathcal I(J_1)\ge0$.

\subsection{A neighboring-run correction to output entropy}
The output bound used $1-d$ as a lower bound on a neighboring run's survival probability. For length $p$, the exact value is $1-d^p$. Retaining part of this improvement gives another pair correction, separate from the mask-entropy corrections.

Recall that $A_e\ge0$ and $\Delta_k\le0$ are the output-code coefficients. For a central length-$\ell$ run, collect its expected correction into
\[
F_\ell=\sum_{k=1}^\ell\Delta_k\binom\ell k(1-d)^kd^{\ell-k}\le0,
\qquad G_\ell=A_ed^\ell-F_\ell\ge0,
\qquad a_\ell=d-d^\ell\ge0.
\]
In the output bound, the coefficient of the probability that both neighbors survive is $-A_ed^\ell+F_\ell=-G_\ell\le0$. Increasing that probability therefore lowers the upper bound.

A neighbor of length $p$ survives with probability $1-d^p=(1-d)+a_p$. For neighbors of lengths $p,r$, independence of their deletion coins gives
\[
(1-d^p)(1-d^r)=(1-d)^2+(1-d)(a_p+a_r)+a_pa_r
\ge(1-d)^2+(1-d)(a_p+a_r).
\]
The inequality drops $a_pa_r\ge0$. The additional reduction for central length $\ell$ is at least $(1-d)(a_p+a_r)G_\ell$. Average with the triple frequencies and sum out each unused neighbor to obtain the pair coefficient
\begin{equation}
e_{\ell r}=(1-d)(a_\ell G_r+a_rG_\ell)\ge0.
\label{eq:neighbor-edge}
\end{equation}
The total subtraction is $\sum_{\ell,r}z_{\ell r}e_{\ell r}$: its two summands account for the left and right neighbors, without another factor of two. It can be added to the mask-entropy corrections because it bounds output entropy.

\subsection{Eliminating the unknown run law by a positive inequality}
We now remove the unknown run law. Both input entropy, bounded in \eqref{eq:run-input-entropy}, and the channel expression $B_0+\sum_{\ell,r}z_{\ell r}(g_r-c_{\ell r}-e_{\ell r})$ upper-bound the information rate. Their convex combination with weights $\omega$ and $1-\omega$, $0<\omega<1$, does too. We construct a subprobability code whose logarithm contains this same combination, allowing cross entropy to bound it for every run law.

Here $g_r$ is the uncorrected run coefficient, $c_{\ell r}$ the mask correction, and $e_{\ell r}$ the output correction; $B_0=(1-d)\gamma-h_2(d)$ is independent of run frequencies. We use their balance identities $\sum_rz_{\ell r}=\sum_rz_{r\ell}=\xi_\ell$ and length normalization $\sum_\ell\ell\xi_\ell=1$.

The target is $(1-\omega)B_0+\beta$. A term $-\beta r$ in the code logarithm becomes $-\beta$ after length normalization. A difference of bounded potentials cancels by frequency balance. This is the standard cross-entropy and average-reward potential argument, applied here to run lengths; see \cref{lem:cross-entropy} and \citet{Puterman94,HSPKS21}.

First omit the triple correction. Choose a positive vector $v=(v_\ell)$ with both $v$ and $1/v$ bounded. The scalar $\beta$ is the proposed rate contribution multiplying run length, not a deletion probability. Require
\begin{equation}
\sum_{r\ge1}2^{[(1-\omega)(g_r-c_{\ell r}-e_{\ell r})-\beta r]/\omega}v_r
\le v_\ell\qquad(\ell\ge1).
\label{eq:run-positive-rows}
\end{equation}
\begin{theorem}[Positive run rows give a capacity upper bound]\label{thm:run-positive}
Fix $0<d\le1/2$, the output-description coefficients in \eqref{eq:small-run-reward}, the valid nonnegative pair corrections $c_{\ell r},e_{\ell r}$, and $0<\omega<1$. If a positive vector $v$, with both $v$ and $1/v$ bounded, satisfies \eqref{eq:run-positive-rows} for every positive integer $\ell$, then
\[
 C(d)\le (1-\omega)B_0+\beta.
\]
\end{theorem}

For example, $M=\bigl(\begin{smallmatrix}1/4&1\\1/8&1/4\end{smallmatrix}\bigr)$ and $v=(2,1)$ give $Mv=(3/2,1/2)\le v$. Although the first unweighted row sums to $5/4$, dividing the weighted rows by their current potentials gives masses $3/4$ and $1/2$. The potential accounts for each transition's destination.

\begin{proof}
\textbf{Step 1: interpret a checked row as a probability code.}
For $\xi_\ell>0$, the next-length probabilities under the run-origin law are $z_{\ell r}/\xi_\ell$. Entropies and expectations in this step use that law. Zero-frequency rows contribute nothing when averaged.
Divide the summand by $v_\ell$ to define $q_\ell(r)$, with $\sum_rq_\ell(r)\le1$. Normalize by this sum and apply cross entropy. Since its logarithm is nonpositive, the true next-run law satisfies $H(L_1\mid L_0=\ell)+\E[\log_2q_\ell(L_1)\mid L_0=\ell]\le0$.

\textbf{Step 2: recover the entropy and the known run coefficients.}

Expand the logarithm of the defined $q_\ell(r)$ and multiply by $\omega$. Writing $R$ here for the next run length only, we obtain
\begin{align*}
\omega H(R\mid\ell)&+(1-\omega)\E[g_R-c_{\ell R}-e_{\ell R}\mid\ell]
-\beta\E[R\mid\ell]\\
&\le\omega\log v_\ell-\omega\E[\log v_R\mid\ell].
\end{align*}
The exponent supplies the run coefficient and length term; $v_R/v_\ell$ supplies the potential difference. Exponential run tails and bounded $\log v$ make these expectations finite.

\textbf{Step 3: remove the potential and normalize by input length.}

Multiply by $\xi_\ell$ and sum. Incoming and outgoing frequencies of each length agree, cancelling the two potential sums. The comparison probabilities are not the unknown input-run transition law: this cancellation uses frequency balance, not a Markov assumption on that law. The length term is $\beta\sum_r r\xi_r=\beta$. Adding the constant $(1-\omega)B_0$ from the convex combination yields
\begin{equation}
\boxed{C(d)\le(1-\omega)B_0+\beta.}
\label{eq:small-certificate-bound}
\end{equation}
This bound holds for every fixed stationary concatenation of length-$q$ blocks. Its parameters depend on $d$ and the certificate, not the source or $q$. Boundary costs vanish with $q$ fixed; then \cref{lem:stationary} allows $q\to\infty$ and gives capacity. Constant sources have zero information rate.
\end{proof}

\subsection{Keeping a finite amount of neighboring-run memory}
To retain $\delta_j(p,r)$, store preceding length $p$ whenever current length $j\le J_0$, for a chosen cutoff $J_0$; otherwise store only $j$. Call the state $S$. After next length $r$, the state $\operatorname{next}(S,r)$ stores $r$ and, when $r\le J_0$, also $j$.

Define the retained correction by
\[
\Gamma(S,r)=c_{jr}+e_{jr}
+\begin{cases}
\delta_j(p,r),&\text{if }S\text{ stores the preceding length }p,\\
0,&\text{otherwise}.
\end{cases}
\]
The pair terms concern $(j,r)$; the triple term also needs $p$. Omitting any nonnegative correction remains valid. Thus the bridge checker can retain pair and triple mask corrections, and the neighboring-run checker pair mask and output corrections, within the same notation $\Gamma$.

Choose positive potential values $v(S)$, with both $v$ and $1/v$ bounded. The state version of \eqref{eq:run-positive-rows} is the explicit requirement
\begin{equation}\label{eq:memory-run-row}
\sum_{r\ge1}
2^{[(1-\omega)(g_r-\Gamma(S,r))-\beta r]/\omega}
 v\bigl(\operatorname{next}(S,r)\bigr)
\le v(S)\qquad\text{for every }S.
\end{equation}
The destination potential uses the full next state, including any stored predecessor.

For averaging, recall the run-origin law $\Pr^\circ$ and run density $\rho$. Define transition frequencies per input bit by
\[
\mu(S,r)=\rho\Pr^\circ\{S,L_1=r\},
\qquad \mu_S=\sum_r\mu(S,r).
\]
Run stationarity balances state visits and arrivals; next-run lengths count all input bits:
\begin{equation}\label{eq:memory-run-balance}
\mu_{S'}=\sum_{S,r:\,\operatorname{next}(S,r)=S'}\mu(S,r),
\qquad
\sum_{S,r}r\mu(S,r)=1.
\end{equation}
Also $\sum_S\mu_S=\rho$. These are per-input-bit frequencies, not a probability distribution over states.

Coding the next length given $S$, with initial context disclosed, gives $h(X)\le\rho H^\circ(L_1\mid S)$. Each next length updates the state deterministically. Products of the assigned conditional probabilities therefore form a normalized law for each run count, even if the true source uses longer memory. Divide the row summand by $v(S)$ to define the subprobability code $q_S(r)$. Cross entropy gives $H^\circ(L_1\mid S=s)+\E^\circ[\log_2q_s(L_1)\mid S=s]\le0$ at each state of positive frequency. Expand, multiply by $\omega\mu_s$, sum, and use the length identity:
\begin{align*}
&\omega\rho H^\circ(L_1\mid S)
 +(1-\omega)\sum_{S,r}\mu(S,r)[g_r-\Gamma(S,r)]-\beta\\
&\qquad\le\omega\sum_{S,r}\mu(S,r)
\left[\log_2v(S)-\log_2v\bigl(\operatorname{next}(S,r)\bigr)\right]=0.
\end{align*}
Frequency balance in \eqref{eq:memory-run-balance} cancels the potential sums. Bounded $\log v$ justifies rearranging them.

Combine $h(X)$ with $B_0+\sum_{S,r}\mu(S,r)[g_r-\Gamma(S,r)]$ as before. The display bounds their weighted average by $(1-\omega)B_0+\beta$. The same block-concatenation argument proves \eqref{eq:small-certificate-bound}.

\textbf{Finite reduction of the state space.}
The test still ranges over all current and next lengths. A finite calculation is sufficient only with analytic bounds for unrepresented lengths; omitting unavailable nonnegative corrections weakens the bound safely.

Choose $K\ge J_0$ and represent lengths greater than $K$ by $*$. A long current run has state $B_*$ and potential one. A current $j\le J_0$ stores a predecessor in $\{1,\ldots,K,*\}$; a current $J_0<j\le K$ stores none. This gives $J_0(K+1)+(K-J_0)+1=(J_0+1)K+1$ states. Omit corrections requiring unknown lengths. Every long next run leads to $B_*$.

For the pair-only construction, define the uncorrected positive weight and a bound on its tail by
\[
u_r=2^{[(1-\omega)g_r-\beta r]/\omega},\qquad
T\ge\sum_{r>K}u_r.
\]
Omitting corrections bounds unknown-length weights by $u_r$. Every short row uses its known short corrections and the complete tail bound $T$, since a long destination has potential one. For any current length greater than $K$, dropping corrections bounds the complete weighted row by $\sum_{r\le K}u_rv_r+T$, which is checked against its potential one. This single row covers every long current row, so $K+1$ checks prove the infinite pair inequality. For states with triple memory, the corresponding long-current row is
\[
\sum_{r\le K}u_r v(\operatorname{next}(B_*,r))+T\le1.
\]
It uses the same tail bound and omission rule. When a short next run stores a long predecessor, its destination is the state with predecessor $*$, whose own potential is retained. Only a long next run uses destination potential one.

\subsection{An explicit bound for every long-run contribution}
To supply $T\ge\sum_{r>K}u_r$, we bound $g_r$ by an affine function, exponentiate, and sum a geometric series. We choose the output-run law $R_k$ to support both this bound and the earlier coefficient signs. Here $K$ is the input-run cutoff. Recall that $0<\omega<1$ weights the input-entropy bound, $\beta$ multiplies the next input-run length, and $g_r$ collects its channel-entropy terms. These are fixed test parameters and coefficients, not parameters of the unknown input law. Choose rational $1/2\le\theta<1$, factors $b_k\ge1$ for $1\le k\le H$, and $\epsilon_0=d^2$, with $0<d\le1/2$. Set
\begin{align*}
Z_R&=\sum_{k=1}^Hb_k\theta^k+\frac{\theta^{H+1}}{1-\theta},\\
R_k^g&=\frac{\theta^k}{Z_R}
\begin{cases}b_k,&k\le H,\\1,&k>H,\end{cases}
\qquad
R_k=(1-\epsilon_0)R_k^g+\frac{\epsilon_0}{k(k+1)}.
\end{align*}
The normalizer $Z_R$ sums the finite prefix and geometric tail of $R^g$. Also $1/[k(k+1)]=1/k-1/(k+1)$ sums to one. Thus their mixture $R$ is a positive probability law.

\textbf{Step 1: obtain two bounds from the two code components.}
The geometric component bounds description length by a line; the $1/[k(k+1)]$ component gives a logarithmic bound. Each follows by lower-bounding the mixture probability by one component.

Define $\gamma=-\log_2\theta$ and $A_e=\log_2Z_R-\log_2(1-\epsilon_0)$. Since $b_k\ge1$ and $\theta\ge1/2$, we have $Z_R\ge\theta/(1-\theta)\ge1$, hence $A_e\ge0$. Also $R_k\ge(1-\epsilon_0)\theta^k/Z_R$, so $-\log_2R_k\le A_e+\gamma k$. Thus $\Delta_k=-\log_2R_k-A_e-\gamma k\le0$, as required by the output-count argument.

The heavier mixture component gives a second inequality, $R_k\ge\epsilon_0/[k(k+1)]$. Taking negative logarithms and subtracting the affine part yields
\[
\Delta_k\le-\gamma k-\log_2\epsilon_0-A_e+\log_2(k(k+1)).
\]
\textbf{Step 2: average the output-run bound for one input run.}
The coefficient $g_\ell$ concerns an input run, whereas $R_k$ describes an output run. Their connection is the binomial survivor count.

Let $N_\ell\sim\operatorname{Bin}(\ell,1-d)$ and set $\Delta_0=0$ for the following expectation. Rewriting \eqref{eq:small-run-reward} gives
\[
g_\ell=\zeta_\ell+w_\ell+A_e[1-(1+(1-d)^2)d^\ell]
+(1-d)^2\E[\Delta_{N_\ell}\mathbf1_{N_\ell\ge1}].
\]
The expectation replaces the positive-survivor sum. Here $\zeta_\ell$ is the reversed-mask entropy and $w_\ell$ the remaining deletion-position entropy bound for a length-$\ell$ input run. Since $\E N_\ell=\ell(1-d)$, the term $-\gamma k$ contributes $-\gamma\ell(1-d)$. Bound its constant by $C_0=\max(0,-\log_2\epsilon_0-A_e)$ and use $\log_2(k(k+1))\le2\log_2(\ell+1)$ for $1\le k\le\ell$. These nonnegative bounds need only probability at most one. Finally use $\zeta_\ell\le\ell h_2(d)$ and bound the $A_e$ term by $A_e$. Multiplication by the two-neighbor factor $(1-d)^2$ yields
\[
g_\ell\le\ell\bigl(h_2(d)-(1-d)^3\gamma\bigr)
+A_e+(1-d)^2C_0+2(1-d)^2\log_2(\ell+1)+w_\ell.
\]
The third power arises from two neighbor-survival factors and the mean survivor count inside the run.

\textbf{Step 3: replace the remaining nonlinear terms by a line.}

Choose $L\ge K$, $L\ge1$. Concavity of $\log_2(x+1)$ bounds it by its tangent at $L+1$, giving
\[
2(1-d)^2\log_2(\ell+1)
\le2(1-d)^2\log_2(L+2)-t(L+1)+t\ell,
\qquad t=\frac{2(1-d)^2}{(L+2)\ln2}.
\]
We also recall the remaining ambiguity term $w_\ell=\ell d(1-d)^{\ell+1}\log_2\ell$. For real $x>1$, the logarithmic derivative of $d(1-d)^{x+1}\log_2x$ is $\ln(1-d)+1/(x\ln x)$. Its second term decreases with $x$. The checker verifies it is already negative at $x=L+1$. Thus, for every $\ell\ge L+1$,
\[
 w_\ell/\ell\le\eta=d(1-d)^{L+2}\log_2(L+1).
\]

Collecting the constant and linear terms now gives $g_\ell\le I_0+S_0\ell$, where
\begin{align*}
I_0&=A_e+(1-d)^2C_0+2(1-d)^2\log_2(L+2)-t(L+1),\\
S_0&=h_2(d)-(1-d)^3\gamma+t+\eta.
\end{align*}
Substitute this affine bound into $u_\ell=2^{[(1-\omega)g_\ell-\beta\ell]/\omega}$. Because $0<\omega<1$, the direction is preserved: $u_\ell\le2^{\alpha I_0+\sigma\ell}$, with
\[
\alpha=\frac{1-\omega}{\omega},\qquad
\sigma=\frac{(1-\omega)S_0-\beta}{\omega}.
\]
\textbf{Step 4: sum the bound over every remaining run length.}
The slope must be negative to sum the infinite bound; a finite prefix alone cannot establish this.

The checker requires an entirely negative interval enclosure for $\sigma$. Then $0<2^\sigma<1$, so the geometric series sums exactly to
\begin{equation}
\sum_{\ell>L}u_\ell\le
2^{\alpha I_0}\sum_{\ell=L+1}^{\infty}(2^\sigma)^\ell
=\frac{2^{\alpha I_0+\sigma(L+1)}}{1-2^\sigma}.
\label{eq:complete-heavy-tail}
\end{equation}

Evaluate $K<r\le L$ upward and add \eqref{eq:complete-heavy-tail}. The resulting $T$ covers every omitted next length, without restricting input-run lengths.

\begin{corollary}[A finite check covers every run length]\label{cor:finite-run-check}
Fix the finite state representation and potential described above, including potential one for long current runs. Suppose the checked short rows include an upper bound $T$ on every contribution from next lengths $r>K$, and the checked long-current rows omit only nonnegative corrections. If all these upper row values satisfy \eqref{eq:memory-run-row}, then the bound \eqref{eq:small-certificate-bound} holds. The tail bound $T$ may be obtained by summing $K<r\le L$ and adding \eqref{eq:complete-heavy-tail}, provided its slope $\sigma$ is strictly negative.
\end{corollary}
\begin{proof}
For represented lengths, the checked terms bound the corresponding terms of the infinite row individually. For $r>K$, omitting a nonnegative correction increases the exponent, and the next state's potential is one. Thus the added $T$ bounds their complete sum. A long current length has no retained length-specific correction; its checked row consequently bounds every such current length. This proves every infinite row. The state-frequency and cross-entropy argument above then gives \eqref{eq:small-certificate-bound}. Finally, $\sigma<0$ makes the ratio $2^\sigma$ smaller than one, so \eqref{eq:complete-heavy-tail} bounds all lengths beyond $L$, rather than only a finite sample.
\end{proof}

\paragraph{Where computation enters.}
The input to this checker is a proposed output-run distribution, $\omega,\beta$, the finite cutoffs, and the positive potential values. These are comparison parameters, not an estimated run distribution of an input source. The checker evaluates the finite entropy tables defining $g_r$ and the retained corrections, bounds the finite sums upward, verifies a negative upper endpoint for $\sigma$, and adds the complete tail bound. It then checks every represented row against its potential value. \Cref{cor:finite-run-check} turns these numerical inequalities into $C(d)\le(1-\omega)B_0+\beta$. Search may change any permitted comparison parameters; a proof requires the complete check only for the final proposal. The finite cutoffs control cost and approximation error, and do not cap actual input-run lengths.

\subsection{The scalar construction used on the first positive interval}
The first positive interval uses potential one and no pair or triple corrections. All rows then reduce to one scalar sum, still requiring a complete tail bound. This construction is checked only on $17/10000\le d\le1/100$. Set
\begin{align*}
\theta&=\tfrac12+\tfrac{13}{8}d,\qquad
R_1=\tfrac{1-d}{2},\qquad R_2=\tfrac14-\tfrac d8,
\qquad R_3=\tfrac18+\tfrac d{16},\\
T_0&=\tfrac18+\tfrac{9d}{16},\qquad
R_k=T_0(1-\theta)\theta^{k-4}\quad(k\ge4).
\end{align*}
On this interval every displayed probability is positive and $0<\theta<1$. The tail sums to $T_0$ by the geometric-series formula. Direct addition gives $R_1+R_2+R_3=7/8-9d/16=1-T_0$, proving normalization.

Set $A=\log_2[\theta^4/(T_0(1-\theta))]$ and $\gamma=-\log_2\theta$. In this scalar construction, $A$ plays the role of $A_e$ above. For $k\ge4$, expand the logarithm of the geometric formula to obtain $-\log_2R_k=A+\gamma k$. For $k=1,2,3$ define $\Delta_k=-\log_2R_k-A-\gamma k$. The required conditions $A\ge0$ and $\Delta_k\le0$ reduce respectively to $\theta^4\ge T_0(1-\theta)$ and $R_k\theta^4\ge T_0(1-\theta)\theta^k$. They involve only rational products at each exact checked parameter.

With potential one and omitted nonnegative entropy corrections, all run rows are bounded by the same scalar condition
\[
\sum_{\ell\ge1}2^{[(1-\omega)g_\ell-\beta\ell]/\omega}\le1.
\]
This is \eqref{eq:run-positive-rows} with constant potential and zero corrections. For its tail, $\Delta_k\le0$ and $\zeta_\ell\le\ell h_2(d)$ imply $g_\ell\le\ell h_2(d)+A+w_\ell$. The same derivative test gives $w_\ell\le\eta\ell$ beyond $L$. With $\alpha=(1-\omega)/\omega$ and $\sigma=((1-\omega)h_2(d)-\beta)/\omega+\alpha\eta<0$, the remaining terms are bounded by $2^{\alpha A+\sigma\ell}$, whose complete sum is $2^{\alpha A+\sigma(L+1)}/(1-2^\sigma)$. This completes the scalar check.

\paragraph{Contribution to the final bound.}
The published enclosure selects the run-based family as a whole, including its non-scalar tests, between $d=0.0017$ and $d=3/128$. Its bound is valid wherever the stated row and remainder conditions hold; selection does not restrict validity. The fair-input lower refinements and interval assembly appear in \cref{sec:fair-input,sec:small-assembly}.

%% file: sections/roadmap_runs.tex
\subsection{Proof plan: describe ambiguities between runs}\label{sec:roadmap-runs}
This is a separate converse, used near zero deletion probability. It starts from
\[
 I(X_1^n;Y_n)=H(Y_n)-nh_2(d)+H(D_1^n\mid X_1^n,Y_n).
\]
Here an \emph{upper} bound on deletion-pattern ambiguity is required; the lower-bound constructions later in the paper need the opposite direction for the same identity.

\begin{enumerate}
\item Following \citet{KM10}, reverse selected deletions until the remaining deletion counts can be recovered run by run. \Cref{lem:extended-run-parser} pays for the reversed mask and bounds the remaining position ambiguity.
\item \Cref{lem:output-run-counts} bounds output entropy through run counts. Pair and triple corrections then subtract explicitly identified redundancies from the mask and output descriptions.
\item \Cref{thm:run-positive} removes the unknown run frequencies: a positive row with sum at most one defines a valid comparison distribution for the next run length.
\item \Cref{cor:finite-run-check} reduces the rows to finitely many lengths and an explicit bound on all longer lengths. These are the quantities checked numerically.
\end{enumerate}
Run lengths remain unrestricted and may be dependent. Short-run state variables belong to the comparison distribution, not to an assumed model of the input. Omitting a nonnegative correction gives a weaker valid bound; omitting an infinite tail does not.

\Cref{fig:run-proof-route} summarizes how these steps produce the final finite test.

\input{figures/run_proof_route}

%% file: figures/run_proof_route.tex
\begin{figure}[!htbp]\centering
\begin{tikzpicture}[>=Stealth,text=ink,
 box/.style={draw=ink!50,rounded corners=2pt,align=left,text width=12.8cm,inner sep=6pt,font=\small},arr/.style={->,thick,draw=ink!60}]
\node[box,fill=ink!3] (m) {\textbf{Simplify the deletion pattern and account for the change.}\\
\Cref{lem:extended-run-parser} bounds the remaining position uncertainty. Equations~\eqref{eq:modified-mask-upper} and~\eqref{eq:small-run-reward} include the reversed deletions and output-description cost.};
\node[box,fill=ink!3,below=.4cm of m] (r) {\textbf{Retain information about neighboring runs.}\\
Pair and triple event tables give nonnegative corrections. Insert the selected corrections in the positive row condition \eqref{eq:memory-run-row}; averaging removes the unknown run distribution.};
\node[box,fill=ochre!14,below=.4cm of r] (c) {\textbf{Check finitely many lengths and bound all longer ones.}\\
The explicit tail bound \eqref{eq:complete-heavy-tail} covers every omitted length. Passing all row and tail inequalities proves the upper bound \eqref{eq:small-certificate-bound}.};
\draw[arr](m)--(r);\draw[arr](r)--(c);
\end{tikzpicture}
\caption{The small-deletion converse in \cref{sec:small-deletion} starts with the modified-deletion idea of \citet{KM10}. The first stage accounts for what was changed, the second retains local information, and the third makes the inequalities finite without restricting possible input-run lengths. The final numerical bound depends on both checked finite rows and proved infinite-tail bounds.}\label{fig:run-proof-route}
\end{figure}
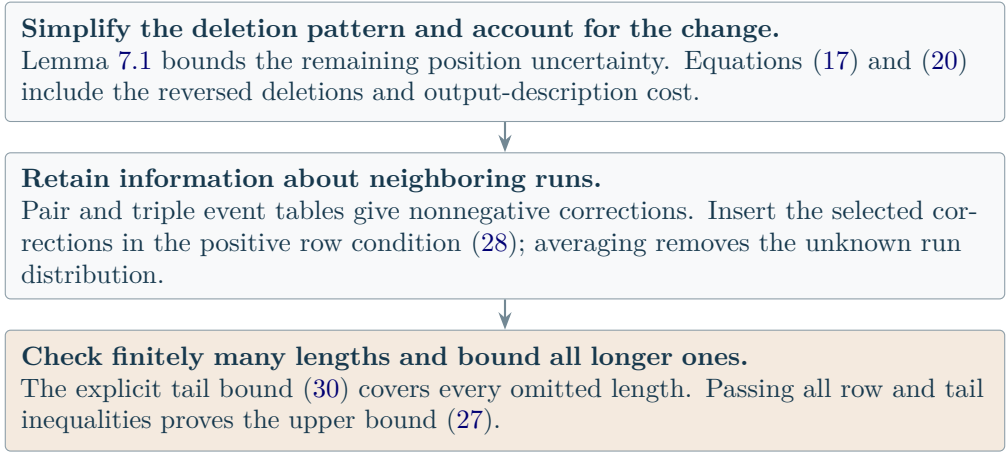

%% file: sections/chapter_lower.tex
\part{Lower bounds}\label{part:lower}
A capacity lower bound needs one specified input source with a sufficiently large information rate. For every block length $n$, its mutual information is at most $nC_n(d)$; division by $n$ and the capacity limit therefore turn a source-rate lower estimate into a capacity lower bound. This part develops two finite calculations, following the source-based approach of \citet{DG06,VTR13,RC23}.

\textbf{Finite-state inputs, \cref{sec:source}.} Choose an exact transition matrix, stationary state probabilities, and a binary symbol for each state. Compute lower estimates of output entropy and of deletion-mask entropy conditional on input and trace. The chain-rule identity adds these estimates and subtracts the known mask entropy. The state need not be determined by a finite history of observed input bits.

\textbf{Renewal inputs, \cref{sec:renewal}.} Choose independent lengths for alternating zero and one runs. The jigsaw construction of \citet{DM07} bounds the rate by describing groups of input runs that produce one output run. Retaining uncertainty about their starts, as in \citet{KD10}, improves the rate. Poisson simulation \citep{MD06} converts one repeat-channel calculation into a lower bound for every deletion probability.

Each complete calculation uses one source throughout: increasing mask ambiguity is useful only after accounting for the same source's output entropy. We maximize the resulting whole-source bounds. Supplementary fair-input refinements in \cref{sec:fair-input} improve some pointwise values and the mean approximation error; \cref{sec:family-simplification} checks their role in the worst-case tolerance.

%% file: sections/source_lower.tex
\section{Achievable rates from exact stationary inputs}\label{sec:source}
We choose an exact stationary binary source and compute a lower bound on its information rate, following the source-based approach of \citet{DG06,VTR13,RC23}. Stationarity means that every finite word has the same probability at every input position. Deletions are independent of the source. Throughout, $d<1$; the endpoint $C(1)=0$ is exact.

\input{sections/roadmap_lower}

\subsection{Specifying the input process}\label{sec:source-law}
Choose a finite state space, rational transition matrix $T$, stationary masses $\pi$, and one emitted bit per state. The entry $T_{ij}$ is the probability of the next state being $j$ given current state $i$. With $\one$ denoting the all-ones vector, check

\[
T_{ij}\ge0,\quad T\one=\one,\quad \pi^\top T=\pi^\top,
\quad\pi_i>0,\quad\pi^\top\one=1.
\]
These conditions normalize each transition row and $\pi$, and ensure that one transition preserves $\pi$. Starting with these masses makes the source stationary. Several states may emit the same bit. The certificate contains several exact finite-state sources; the 120-state source developed below is one worked instance. Each instance uses its own transition table and stationary masses. No finite-order Markov property of the emitted bits is assumed.

Let $E_b$ be diagonal, with entry one on states emitting $b$ and zero elsewhere. The probability of a raw word is
\[
 \Pr(X_1^h=x_1^h)=\pi^\top E_{x_1}T E_{x_2}\cdots T E_{x_h}\one.
\]
Each $E_{x_i}$ removes paths emitting the wrong bit at position $i$; each $T$ advances one input position. The final multiplication sums the ending state. Thus every word weight is a finite sum of rational products.

If the stored table labels transitions with bits, verify that all transitions entering a state carry that state's emitted bit. Complement symmetry may reduce enumeration only when a checked state permutation exchanges emitted bits while preserving $T$ and $\pi$.

The recorded sources include 32-state capped-run inputs, 64-state input-context chains, and 48-state run-context chains, as well as the 120-state refinement below. The parent source plan contains one 32-state source, three 64-state sources, and seven 48-state sources; later source refinements retain the same finite-state verification contract. These are different input laws evaluated by the same bound, not successive terms to be added together.

Embedding counts below depend only on the raw word and may be reused across sources. Their probability weights must be recomputed from that source's $T$ and $\pi$.

\paragraph{Why allow dependence between input bits?}
A run is a maximal segment of equal bits. A run of length $\ell$ disappears with probability $d^\ell$, since all of its bits must be deleted. Increasing repetition reduces this probability but also makes the input more predictable. Furthermore, deleting the middle run of $000\mid1\mid000$ merges its neighbors. These effects motivate adjusting run lengths and their dependence. Markov and run-length inputs have a long history in deletion-channel lower bounds \citep{DG06,DM07,KD10,RC23}; our finite-state representation makes selected features of that dependence adjustable.

The 120-state source used below remembers the current bit, the previous completed run length in classes $1,2,3+$, the run before that in classes $1,2+$, and the current run age in classes $1,2,\ldots,9,10+$. A plus sign means ``at least.'' Thus the state count is $2\cdot3\cdot2\cdot10=120$. Each state has two outgoing transitions: continue the current run, or switch bits and start a new run. Continuing increments the capped age; switching resets the age to one and updates the two completed-run classes. Complement symmetry ties corresponding zero and one states, leaving 60 free switching probabilities.

To select its 60 switching probabilities, search evaluates both finite entropy estimates in \eqref{eq:lower-plan}, adjusts those probabilities to increase their combined rate, and then freezes a candidate for exact verification. The detailed calculation is in \cref{sec:source-computation}.

This generator remembers run features, not a fixed number of preceding bits. Even a long current run need not erase its memory of the previous run. The chosen classes restrict the input family; no optimality theorem singles out these classes or the number 120. The numerical comparison in \cref{sec:source-computation} tests the benefit of refining the current-run age.

\subsection{The information identity and a finite alignment window}
Write $Y_n$ for the output of $X_1^n$ and $D_1^n$ for its deletion mask. The following identity is the starting point of the lower bound; it is the deletion specialization of the action-entropy identity in \citet{HOS16}.
\begin{lemma}[Separating input information from deletion randomness]\label{lem:mask-identity}
For every distribution of $X_1^n$ independent of the deletion mask,
\begin{equation}\label{eq:alignment-identity}
I(X_1^n;Y_n)=H(Y_n)-n\hb(d)+H(D_1^n\mid X_1^n,Y_n).
\end{equation}
\end{lemma}
\begin{proof}
Expand $H(D_1^n,Y_n\mid X_1^n)$ with the mask first. Its value is $n\hb(d)$: the mask is independent of the input, and input plus mask determines the output. Expanding with the output first gives $H(Y_n\mid X_1^n)+H(D_1^n\mid X_1^n,Y_n)$. Equate the two expansions and substitute for $H(Y_n\mid X_1^n)$ in $I(X_1^n;Y_n)=H(Y_n)-H(Y_n\mid X_1^n)$.
\end{proof}

Different masks can produce the same trace. Their remaining uncertainty is the positive correction in \eqref{eq:alignment-identity}. We lower-bound it by revealing additional mask entries: conditioning decreases entropy (\cref{lem:side-information}). These revelations are proof devices, not receiver observations.

For a window of $H$ input bits, $Y_H$ denotes its trace. The next lemma turns the uncertainty about its first deletion indicator into a lower bound per input bit.
\begin{lemma}[Finite-window alignment genie]\label{lem:alignment-genie}
For a stationary input and an integer $H\ge1$, put
\[
g_H(d)=H(D_1\mid X_1^H,Y_H).
\]
Then
\[
\liminf_{n\to\infty}\frac1nH(D_1^n\mid X_1^n,Y_n)\ge g_H(d).
\]
\end{lemma}
\begin{proof}
The chain rule gives
\[
H(D_1^n\mid X_1^n,Y_n)
=\sum_{i=1}^n H(D_i\mid D_1^{i-1},X_1^n,Y_n).
\]
For $i\le n-H+1$, also reveal the mask entries after position $i+H-1$. Conditioning decreases entropy, so this summand is at least the entropy of $D_i$ with every deletion decision outside its length-$H$ window known. Only the first decision in the window is measured; the remaining $H-1$ decisions are left unknown.

The revealed mask and input determine the output before and after the window. Removing those two strings from $Y_n$ identifies the window output. For a fixed full input, the conditional probability of an inside mask is proportional to its deletion probability times the indicator that it produces this window output. Both factors depend only on the inside mask and input. Thus the unrevealed ambiguity is exactly the one defining $g_H(d)$. Stationarity makes its average independent of $i$.

The final $H-1$ chain-rule summands are nonnegative. Discard them to get
\[
 H(D_1^n\mid X_1^n,Y_n)\ge(n-H+1)g_H(d).
\]
Divide by $n$ and let $n\to\infty$ with $H$ fixed. Overlapping windows cause no double counting: each summand measures a different first decision $D_i$.
\end{proof}

\subsection{Increasing the window without recounting old uncertainty}
A larger window can leave more uncertainty about its first deletion decision. To see why, reveal only its final deletion decision and compare with the shorter window.
\begin{lemma}[Nonnegative increments of mask uncertainty]\label{lem:mask-increments}
For a stationary input, $g_1=0$, and for $k\ge2$,
\[
\Delta_k:=g_k-g_{k-1}
=I(D_1;D_k\mid X_1^k,Y_k)\ge0.
\]
Consequently $g_H=\sum_{k=2}^H\Delta_k$.
\end{lemma}
\begin{proof}
Knowing $D_k$ allows us to recover $Y_{k-1}$ from $Y_k$: remove its last symbol if position $k$ survived, and remove nothing if it was deleted. Conversely, $Y_{k-1},X_k,D_k$ determine $Y_k$. Conditional on $X_1^{k-1},Y_{k-1}$, the posterior of $D_1$ does not depend on $X_k$ or the independent decision $D_k$. This follows by the same mask-probability calculation used in \cref{lem:alignment-genie}. Therefore
\[
H(D_1\mid X_1^k,Y_k,D_k)=g_{k-1}.
\]
Subtract this equality from the definition of $g_k$. The difference is the displayed conditional mutual information, which is nonnegative by \cref{lem:side-information}. For one input bit the output length determines deletion, giving $g_1=0$. Summing the differences telescopes.
\end{proof}
Each $\Delta_k$ is an average of nonnegative conditional mutual informations over complete input-output observations. Keeping any subset of those observations gives a lower estimate. Individual logarithmic summands inside a mutual information need not be nonnegative; they cannot be discarded separately. The next calculation groups counts so that every retained contribution has the required sign.

\subsection{Counting compatible deletion patterns}
For input $x=x_1\cdots x_h$ and trace $y$ of length $j$, an embedding lists the retained positions that spell $y$. Each has mask probability $(1-d)^jd^{h-j}$. Let $a_x(y)$ count embeddings omitting position one and $b_x(y)$ count those using it.

These integers have a finite recursion. Let $N_t(v)$ count embeddings of $v=v_1\cdots v_q$ in $x_1^t$. For nonempty $v$,
\[
 N_t(v)=N_{t-1}(v)+\mathbf1_{\{x_t=v_q\}}N_{t-1}(v_1^{q-1}),
 \qquad N_t(\varnothing)=1\ (t\ge0),\quad N_0(v)=0\ (v\ne\varnothing).
\]
The first term omits position $t$; the second uses it and removes the matched last symbol. Apply this recursion to the suffix $x_2^h$. Its counts for $y$ give $a_x(y)$, and its counts for $y_2^j$, multiplied by $\mathbf1_{\{x_1=y_1\}}$, give $b_x(y)$. This avoids listing every mask separately, although enumerating all raw words is still exponential. Subsequence recursions are used in \citet{ISW16}; run-segmentation posteriors based on integer counts are given by \citet[Lemma~13]{Chen26}. The run-based endpoint hierarchy in that work \citep[Theorem~11]{Chen26} is analogous to, but uses different variables from, the raw-bit identity in \cref{lem:mask-increments}.

For a fixed input $x$ of length $h$ and output $y$ of length $j$, all compatible masks have probability $(1-d)^jd^{h-j}$. Thus, if the two counts are $a$ and $b$, the output has conditional probability $(a+b)(1-d)^jd^{h-j}$, and the posterior survival probability of position one is $b/(a+b)$. Their product with binary entropy is the contribution we need. Impossible outputs have both counts zero and contribute zero.

For nonnegative masses, $\Phi(u,v)=(u+v)\log_2(u+v)-u\log_2u-v\log_2v$, with $0\log_20=0$, equals total mass times binary conditional entropy. Multiplying $\Phi(a_x(y),b_x(y))$ by the common mask probability gives this trace's contribution.

Sum over traces, then average with the exact source-word probabilities:
\[
F_j(x)=\sum_{y\in\bits^j}\Phi(a_x(y),b_x(y)),\qquad
f_j(h)=\E F_j(X_1^h).
\]
Set impossible lengths, $F_0$, and $F_j(\varnothing)$ to zero. Empty traces determine the entire mask; for $h=1$, trace length also determines the first mask entry. These cases therefore have zero entropy.

Group masks by their survivor count and multiply each expected count entropy by the corresponding mask probability:
\begin{equation}\label{eq:gH}
g_H(d)=\sum_{j=1}^H (1-d)^jd^{H-j}f_j(H).
\end{equation}
For $00\to0$, the two classes have counts $(1,1)$, so $\Phi(1,1)=2$: observation probability $2d(1-d)$ times one bit of uncertainty. Under fair independent input, $00$ and $11$ each have probability $1/4$; $01$ and $10$ contribute zero. Hence $g_2(d)=d(1-d)$. By \cref{lem:alignment-genie}, this is a lower bound on residual mask entropy per input bit; the complete rate also includes the other two terms of \eqref{eq:alignment-identity}.

\subsection{Selecting contributions and reusing a rectangular set of counts}
To implement \cref{lem:mask-increments} using the count table, append position $h$. Revealing its mask entry leaves $F_j(x_1^{h-1})$ if deleted, or $F_{j-1}(x_1^{h-1})$ if retained. The additional count entropy is
\begin{equation}\label{eq:positive-cell}
E_j(x_1^h)=F_j(x_1^h)-F_j(x_1^{h-1})-F_{j-1}(x_1^{h-1}),
\qquad e_j(h)=\E E_j(X_1^h).
\end{equation}

\begin{examplebox}[title={A selected contribution: appending a zero to $01$}]
For $01$, each one-symbol trace determines whether the first bit survived, so $F_1(01)=0$. Appending zero gives $010$. Trace $0$ now has two embeddings with different first mask entries, while trace $1$ remains unambiguous. Therefore
\[
E_1(010)=F_1(010)-F_1(01)-F_0(01)=2-0-0=2.
\]
The subtractions remove the earlier uncertainty; the empty-trace term is zero. At $d=1/2$, the mask factor is $1/8$. Fair independent input assigns probability $1/8$ to $010$, so retaining this word contributes $2(1/8)(1/8)=1/32$ bit per input bit.

\end{examplebox}

\begin{lemma}[Every incremental alignment cell is nonnegative]\label{lem:positive-cell}
For all raw words and valid indices, $E_j(x_1^h)\ge0$.
\end{lemma}
\begin{proof}
For each trace of length $j$, split its two-class count vector into embeddings omitting and using position $h$. Revealing that choice decreases uncertainty about the first mask entry. For count vectors $(u_0,u_1)$ and $(v_0,v_1)$, this is the homogeneous-entropy inequality $\Phi(u_0+v_0,u_1+v_1)\ge\Phi(u_0,u_1)+\Phi(v_0,v_1)$ (\cref{lem:logsum}).

Summed over traces, the omitted-position terms give $F_j(x_1^{h-1})$. In the retained-position terms, remove position $h$ and the matching final trace symbol; this gives exactly $F_{j-1}(x_1^{h-1})$. Hence $F_j(x_1^h)\ge F_j(x_1^{h-1})+F_{j-1}(x_1^{h-1})$, which is $E_j(x_1^h)\ge0$.
\end{proof}
The count increments are exactly the entropy increments of \cref{lem:mask-increments} after their deletion probabilities are restored:
\begin{equation}\label{eq:count-increment-bridge}
\Delta_h(d)=\sum_{j=1}^h(1-d)^jd^{h-j}e_j(h).
\end{equation}
Indeed, substitute $e_j(h)=f_j(h)-f_j(h-1)-f_{j-1}(h-1)$. The first weighted sum is $g_h$ by \eqref{eq:gH}; the second is $d g_{h-1}$; the third is $(1-d)g_{h-1}$ after shifting $j$ by one. Their difference is $g_h-g_{h-1}$. Every sum uses the same input law. Thus the finite coefficients below compute the previously defined nonnegative increments.

A count cell is indexed by raw depth $h$ and survivor count $j$. Since its increment is nonnegative, selected distinct cells give a lower sum. The following identity evaluates all cells with $j\le k$ and $h\le H$ from boundary counts.

\begin{lemma}[Finite rectangle identity]\label{lem:rectangle}
For integers $1\le k\le H$ and one fixed stationary source,
\begin{align}\label{eq:rectangle}
\sum_{j=1}^k\sum_{h=j}^H (1-d)^jd^{h-j}e_j(h)
={}&\sum_{c=k}^{H-1}(1-d)^{k+1}d^{c-k}f_k(c)
+\sum_{j=1}^k (1-d)^jd^{H-j}f_j(H).
\end{align}
For $k=H$, this is $g_H(d)$.
\end{lemma}
\begin{proof}
Substitute the definition of $e_j(h)$. For $h<H$, the coefficient of $f_j(h)$ is its own weight $(1-d)^jd^{h-j}$ minus $(1-d)^jd^{h+1-j}$ from the next depth. When $j<k$, also subtract $(1-d)^{j+1}d^{h-j}$ from the next survivor count. Interior coefficients vanish because $d+(1-d)=1$.

At $j=k$, the last subtraction is absent, leaving $(1-d)^{k+1}d^{h-k}$. At $h=H$, neither later-depth term occurs, leaving $(1-d)^jd^{H-j}$. These are the two boundary sums. If $k=H$, the first sum is empty and the second is \eqref{eq:gH}.
\end{proof}
Any finite selection of cells lies in a larger complete triangle. Nonnegativity bounds the selected sum by that triangle's $g_H(d)$, and \cref{lem:alignment-genie} bounds $g_H(d)$ by the full mask-entropy rate. No infinite-sum interchange is needed.

The cached base uses $k=12$, $H=28$. Substituting these values into the rectangle identity gives
\begin{equation}\label{eq:base-rectangle}
\sum_{c=12}^{27}(1-d)^{13}d^{c-12}f_{12}(c)
+\sum_{j=1}^{12}(1-d)^jd^{28-j}f_j(28).
\end{equation}
Additional contributions must have $h>28$ or $j\ge13$ to avoid the base rectangle. Within one cell $(j,h)$, disjoint raw-word or prefix events may be summed with their original source probabilities. The certificate checks that these events are disjoint. Lower estimates covering the same event are alternatives: take their maximum, not their sum.

A prefix event means all length-$h$ words extending a specified prefix. A complete length-$h$ word is the special case in which that event contains one word. For example, prefixes $00$ and $01$ describe disjoint events, whereas $0$ and $00$ overlap. In general, two prefixes at the same count cell are disjoint exactly when neither extends the other. A bound for the whole cell overlaps every subevent and cannot be added to it. Omitted words and mass lost to downward rounding receive zero contribution; retained probabilities are never renormalized. Their combined role in the lower bound is
\[
\mathcal A_-\le
\underbrace{\text{base rectangle}}_{\eqref{eq:base-rectangle}}
+\underbrace{\text{selected disjoint extensions}}_{\text{nonnegative count contributions}}
\le\liminf_{n\to\infty}\frac{H(D_1^n\mid X_1^n,Y_n)}n.
\]
The current extension checker represents selected events by finite lists of complete binary words. In the 48-bit extension it checks 16,384 distinct stored words and excludes previously credited prefixes at each relevant shorter depth. This is a concrete implementation of the disjoint-event argument, not a checker for arbitrary symbolic prefix descriptions. The source-specific weighting and directed count evaluation are given in \cref{sec:alignment-arithmetic}.

\subsection{A lower bound on output entropy}
For the same finite-state input, we next lower-bound output entropy. Let $T$ and $\pi$ be its raw-state transition matrix and stationary masses. Between retained positions, a gap of $r$ contains $r-1$ deletions followed by a survivor, with probability $(1-d)d^{r-1}$. The state advances by $T^r$, so

\begin{equation}\label{eq:survivor-matrix}
M=\sum_{r\ge1}(1-d)d^{r-1}T^r=(I-dT)^{-1}(1-d)T.
\end{equation}
To verify the inverse, multiply $S_N=\sum_{r=0}^N(dT)^r$ by $I-dT$. Intermediate terms cancel, giving $I-(dT)^{N+1}$. Each row of the remaining power has mass $d^{N+1}\to0$, proving the geometric-series identity. The mixture weights sum to one and each $T^r$ preserves $\pi$, so $M$ is stochastic and $\pi^\top M=\pi^\top$.

Start $S_0\sim\pi$, advance states with $M$, and call their emitted bits $Z_r$. The index $r$ now counts survivors. The output entropy rate is $h(Z)$ per survivor.

A finite output history alone gives an upper entropy-rate estimate. For a lower estimate, \citet{Birch62} additionally reveals the boundary state. First consider an example where observed bits already identify the state.
\begin{examplebox}[title={A source whose survivor entropy can be calculated by hand}]
Take a stationary two-state source with equal state masses and flip probability $1/4$:
\[
T=\begin{pmatrix}3/4&1/4\\1/4&3/4\end{pmatrix}.
\]
At $d=1/2$, survivor gaps have probabilities $2^{-r}$. If $p_r$ is the flip probability after $r$ raw steps, then $p_{r+1}=\tfrac14(1-p_r)+\tfrac34p_r$ and $p_0=0$. Thus $1-2p_r=2^{-r}$.

Average over the survivor gap:
\[
\sum_{r\ge1}2^{-r}\frac{1-2^{-r}}2
=\frac12\left(1-\sum_{r\ge1}4^{-r}\right)
=\frac12(1-1/3)=\frac13.
\]
The geometric sums are $\sum_{r\ge1}2^{-r}=1$ and $\sum_{r\ge1}4^{-r}=1/3$. The survivor-time matrix is therefore
\[
M=\begin{pmatrix}2/3&1/3\\1/3&2/3\end{pmatrix}.
\]
Each observed bit identifies the state, so the output entropy is $h_2(1/3)$ per survivor, or $\tfrac12h_2(1/3)$ per input bit. The mask terms are still required for the complete information rate.

With several states emitting each bit, the state is not observed. Revealing it gives the lower bound below.
\end{examplebox}

\begin{lemma}[Birch lower bound]\label{lem:birch}
For $m\ge1$,
\[
B_m(d)=H(Z_m\mid Z_1^{m-1},S_0)\le h(Z).
\]
\end{lemma}
\begin{proof}
For the stationary two-sided survivor process, \cref{lem:entropy-rate} gives
\[
h(Z)=H(Z_m\mid Z_{-\infty}^{m-1})
\ge H(Z_m\mid Z_{-\infty}^{m-1},S_0)
=H(Z_m\mid Z_1^{m-1},S_0).
\]
The inequality reveals the boundary state, which can only reduce entropy. Given $S_0$, the Markov property makes outputs after time zero independent of outputs through time zero. This proves the last equality and the lower-bound direction.
\end{proof}

\begin{corollary}[Moving the disclosed state farther into the past]\label{cor:birch-monotone}
For the same source and deletion probability, $B_{m+1}(d)\ge B_m(d)$.
\end{corollary}
\begin{proof}
Reveal also the next surviving state $S_1$. Conditioning decreases entropy, so
\[
H(Z_{m+1}\mid S_0,Z_1^m)
\ge H(Z_{m+1}\mid S_0,S_1,Z_1^m)
=H(Z_{m+1}\mid S_1,Z_2^m).
\]
The equality holds because $S_1$ determines $Z_1$ and separates the future from $S_0$. Stationarity identifies the last entropy with $B_m(d)$. Increasing $m$ moves the extra state disclosure farther from the predicted bit, reducing its advantage.
\end{proof}

Form $M_b$ by keeping columns of $M$ whose destination state emits $b$ and zeroing the others. Conditional on $S_0=i$, the joint probability of history $v=v_1\cdots v_{m-1}$ and next bit $b$ is $p_{i,v,b}$. With $e_i$ the coordinate vector for state $i$, it equals
\[
p_{i,v,b}=e_i^\top M_{v_1}\cdots M_{v_{m-1}}M_b\one.
\]
Each matrix restricts one emitted bit; multiplication by $\one$ sums the ending state. The history has mass $p=p_{i,v,0}+p_{i,v,1}$ and conditional next-bit probability $p_{i,v,1}/p$. Its entropy contribution is $p\hb(p_{i,v,1}/p)=\Phi(p_{i,v,0},p_{i,v,1})$, or zero if $p=0$. Sum over histories and initial states:
\begin{equation}\label{eq:birch-row}
B_m(d)=\sum_i\pi_i\sum_{v\in\bits^{m-1}}
\Phi(p_{i,v,0},p_{i,v,1}).
\end{equation}
Arb interval arithmetic encloses the matrix inverse and every history probability, with checked normalization and nonnegativity. The lower endpoint of the resulting entropy sum is the required $b_-$.

\subsection{Conversion to input-bit units and valid entropy reuse}
For the fixed input, $Z$ denotes the infinite survivor sequence and $Y_n$ the trace of $n$ raw bits. Its random length $N_n$ is binomial with mean $n(1-d)$. The following conversion accounts for that fluctuation.
\begin{lemma}[Trace entropy per input bit]\label{lem:trace-rate}
For a stationary binary input and $0\le d<1$,
\[
\lim_{n\to\infty}\frac{H(Y_n)}n= (1-d)h(Z).
\]
\end{lemma}
\begin{proof}
Couple $Y_n=Z_1^{N_n}$ with the fixed prefix $Z_1^k$, where $k=\lfloor n(1-d)\rfloor$. From either prefix, disclose the length and the missing suffix to recover the other. Each extra binary symbol has entropy at most one, giving
\begin{align*}
H(Y_n)&\le H(Z_1^k)+H(N_n)+\E(N_n-k)_+,\\
H(Z_1^k)&\le H(Y_n)+\E(k-N_n)_+.
\end{align*}
In the second inequality $Y_n$ already reveals $N_n$. Combining the two directions yields
\[
|H(Y_n)-H(Z_1^k)|\le H(N_n)+\E|N_n-k|.
\]
The count has at most $n+1$ values, so $H(N_n)\le\log_2(n+1)$. Its variance is $nd(1-d)$; Cauchy--Schwarz and $|k-\E N_n|\le1$ give $\E|N_n-k|\le\sqrt{nd(1-d)}+1$. The error thus vanishes after division by $n$. Finally $k/n\to1-d$ and $H(Z_1^k)/k\to h(Z)$, proving the result.
\end{proof}

An output-entropy estimate may be reused at a larger deletion probability only for the identical source. The needed comparison concerns entropy per survivor.
\begin{lemma}[Independent thinning increases entropy per survivor]\label{lem:thinning-entropy}
For a fixed stationary finite-alphabet process $Z$, let $V$ retain its symbols independently with a positive probability. Then $h(V)\ge h(Z)$. Consequently, for the \emph{same} stationary raw input law and $a\le d<1$, $h(Z^{(d)})\ge h(Z^{(a)})$.
\end{lemma}
\begin{proof}
Condition on retaining position zero, so $V_0=Z_0$. Let $\tau_{-k}^{-1}$ denote the original positions of the $k$ preceding retained symbols. Reveal these positions and the entire original past $Z_{-\infty}^{-1}$. Together they determine $V_{-k}^{-1}$, so conditioning reduces entropy:
\[
 H(V_0\mid V_{-k}^{-1})
 \ge H(Z_0\mid Z_{-\infty}^{-1},\tau_{-k}^{-1})
 =h(Z).
\]
The equality holds because the thinning positions are independent of the source, including after conditioning on retention of the origin. The entropy-rate identity in \cref{lem:entropy-rate} identifies the remaining conditional entropy with $h(Z)$. Letting $k\to\infty$ on the left gives $h(V)\ge h(Z)$.

For $a\le d<1$, retain each $a$-deletion survivor with probability $(1-d)/(1-a)$. The combined retention probability is $1-d$, producing the $d$-deletion output of the same source. Apply the first conclusion.
\end{proof}
The checker binds the source and verifies $a\le d$. Only entropy per survivor is reused; alignment contributions still use the target $d$.

\subsection{Assembly of the stationary-source certificate}
Evaluate both positive terms of \eqref{eq:alignment-identity} downward and the subtracted mask entropy upward. Output entropy alone needs the factor $1-d$ to convert survivor units to input-bit units.

\begin{certbox}[breakable=false,title={A stationary-source lower certificate}]
For one exact stationary source, let $b_-$ lower-bound \eqref{eq:birch-row}, using \cref{lem:thinning-entropy} if reused. Let $\mathcal A_-(d)$ be the base rectangle plus distinct additional mask contributions, per input bit, and let $h_+\ge\hb(d)$. Accept any value satisfying
\begin{equation}\label{eq:lower-row}
\ell_\star\le (1-d)b_- -h_+ +\mathcal A_-(d).
\end{equation}
All source masses refer to the same input law. All counted cells and prefixes are distinct in the stated sense.
\end{certbox}

\begin{theorem}[Validity of the stationary-source lower bound]\label{thm:lower}
A certificate satisfying the preceding requirements proves $C(d)\ge\ell_\star$.
\end{theorem}
\begin{proof}
The finite-window lemma and nonnegative increments lower-bound mask entropy per input bit by $\mathcal A_-(d)$. The Birch bound and trace-rate lemma lower-bound output entropy by $(1-d)b_-$. Substitute into \eqref{eq:alignment-identity} and subtract $h_+\ge\hb(d)$ to obtain the right side of \eqref{eq:lower-row}.

For every $n$, the source mutual information is at most $nC_n(d)$. Divide by $n$ and use $C_n(d)\to C(d)$ from \cref{sec:setting}. Thus the accepted value lower-bounds capacity.
\end{proof}
These are information-theoretic achievable rates. The entropy calculation does not itself construct an efficient encoder or decoder.

\subsection{Source selection and the actual finite computation}\label{sec:source-computation}
Theorem~\ref{thm:lower} is valid for any source satisfying its hypotheses. Computation has two roles: search for a useful source, then establish its lower bound. These operations use the same analytic expression but have different guarantees.

\paragraph{The search objective.}
For fixed $d$, output length $m$, and a fixed selection $\mathcal S$ of nonnegative mask contributions, parameterize the input transition probabilities by $\theta$. Let $\underline g_{\mathcal S}(\theta,d)$ be their exact selected sum and $B_m(\theta,d)$ the Birch expression. Search tries to increase
\begin{equation}\label{eq:source-search-objective}
L_{m,\mathcal S}(\theta,d)
=(1-d)B_m(\theta,d)-h_2(d)+\underline g_{\mathcal S}(\theta,d).
\end{equation}
Every attempt recomputes the stationary distribution, the survivor matrix $M$, and the source-word weights. Only the embedding counts are unchanged. For each input depth $h$, aggregate the retained counts and deletion factors into a source-independent coefficient $A_{h,d}(x)\ge0$. Then
\[
\underline g_{\mathcal S}(\theta,d)=\sum_h\sum_{x\in\bits^h}
P_\theta(X_1^h=x)A_{h,d}(x).
\]
This is why a count cache can support many source searches. It does not make the complete word enumeration small: the base includes $2^{27}$ depth-28 words after complement symmetry, plus $2^{27}-2^{11}$ shorter-depth entries. Selected extensions avoid requiring all words at larger depths.

The recorded search uses floating-point matrix operations, differentiation of \eqref{eq:source-search-objective}, and L-BFGS-B updates of the transition parameters \citep{BLNZ95}. It optimizes a lower estimate of the source rate, not the unknown source rate itself. Earlier sampled searches supplied candidates; sampled objectives are not numerical proof premises.

\paragraph{A controlled enlargement of the source family.}
The recorded refinement at $d=3/5$ enlarged the current-age classes from $1,2,3,4+$ to $1,\ldots,4,5+$, then to $1,\ldots,7,8+$ and $1,\ldots,9,10+$. This gives 48, 60, 96 and 120 states while preserving the completed-run memory. At initialization, each newly split state inherits its old transition probabilities. Summing the new stationary masses over each old state reproduces its old mass, and the projected transitions reproduce the old chain. Thus the initial expanded generator gives exactly the same input law. Optimization then allows the formerly tied probabilities to differ. The records compare the resulting sources with the same $m=13$ and retained count selection; see \cref{tab:source-search}.

\paragraph{The certificate calculation.}
Freeze the selected probabilities as exact rationals. For the 120-state source they are multiples of $2^{-48}$, with an exactly solved stationary law. For this recorded $m=13$ example, evaluate the $120\cdot2^{12}=491{,}520$ state/history conditional-entropy terms in \eqref{eq:birch-row}, the base count weights in \eqref{eq:base-rectangle}, and the selected disjoint extensions. Interval and integer arithmetic give $b_-$, $\mathcal A_-$, and $h_+$ in \eqref{eq:lower-row}. The checker does not differentiate or rerun the optimizer.

For one recorded $d=3/5$ calculation, these numbers are approximately
\[
b_-=0.725031271517,\quad
\mathcal A_-=0.760404915678,\quad h_+=0.970950594455.
\]
These displayed decimals are explanatory approximations; the stored rational endpoints give the certified inequality
\[
C(3/5)\ge0.079466829830.
\]
Equivalently, the output entropy is bounded below by about $0.290012509$ bit per input bit, and conditional output entropy is bounded above by about $0.210545679$. Their difference is the proved rate. Later additional counts improved this point; \cref{sec:computation-experiments} separates those improvements from source search and reports the measured costs.

\begin{remark}[What a better source score establishes]\label{rem:source-score}
Let $a_\theta$ be the residual mask-entropy rate and $R_\theta$ the source information rate. Where these limits exist, \eqref{eq:alignment-identity} gives
\[
R_\theta-L_{m,\mathcal S}(\theta,d)
=(1-d)\bigl[h_\theta(Z)-B_m(\theta,d)\bigr]
+\bigl[a_\theta-\underline g_{\mathcal S}(\theta,d)\bigr].
\]
Each difference is nonnegative by the preceding lemmas. A better finite score can reflect a better source, tighter estimates for that source, or both. Further source optimization and more accurate evaluation of a fixed source address different uncertainties. Neither the chosen source family nor a terminated numerical search is claimed optimal.
\end{remark}

%% file: sections/roadmap_lower.tex
\subsection{The rate calculation and its dependencies}\label{sec:roadmap-lower}
Fix a stationary input process. There are two sources of uncertainty in its output: information about the transmitted bits and randomness introduced by deletions. The identity
\[
 I(X_1^n;Y_n)=H(Y_n)-nh_2(d)+H(D_1^n\mid X_1^n,Y_n)
\]
separates them. Here $D_i=1$ means deletion. The last term restores uncertainty between deletion patterns that give the same output; subtracting the full mask entropy would otherwise subtract too much.

We will prove a lower bound of the form
\begin{equation}\label{eq:lower-plan}
 C(d)\ge (1-d)B_m-h_2(d)+\underline g.
\end{equation}
The number $B_m$ lower-bounds entropy per received bit; $\underline g$ lower-bounds residual mask entropy per transmitted bit. The factor $1-d$ puts both terms in transmitted-bit units. Their proofs have separate tasks:
\begin{enumerate}
\item \Cref{lem:alignment-genie} turns uncertainty in a finite input window into a rate. \Cref{lem:mask-increments,lem:positive-cell,lem:rectangle} then express a lower estimate $\underline g$ as selected nonnegative counts.
\item \Cref{lem:birch,lem:trace-rate} give $B_m$ from the transition matrix between surviving positions and convert it to output entropy per input bit.
\item \Cref{thm:lower} combines those estimates for the \emph{same} input process. \Cref{sec:source-computation} evaluates this formula and explains how the process was selected.
\end{enumerate}
\begin{center}\small
\begin{tabular}{@{}p{4.2cm}p{4.2cm}p{4.2cm}@{}}\toprule
Target & Finite expression & Computed endpoint\\\midrule
Output entropy rate $h(Z)$ & Birch entropy $B_m$ & $b_-\le B_m$\\
Residual mask entropy per input bit & Selected sum $\underline g\le g_H$ & $\mathcal A_-\le\underline g$\\
Deletion entropy $h_2(d)$ & Explicit binary entropy & Upper endpoint $h_+\ge h_2(d)$\\\bottomrule
\end{tabular}
\end{center}
The next subsection specifies the generator of input bits. No auxiliary output distribution from the converse is used in this calculation.

%% file: sections/renewal.tex
\section{Achievable rates from renewal inputs}
\label{sec:renewal}
A run is a maximal constant segment: $111100$ has run lengths four and two. Choose independent complete run lengths with finite-support distribution $P=(P_\ell)$; successive binary runs alternate. Write $\overline L=\sum_\ell\ell P_\ell$ and $J=\max\{\ell:P_\ell>0\}$. Channel randomness is independent of these lengths.

We evaluate the jigsaw rate of \citet{DM07} and retain conditional run-start entropy, following \citet{KD10}. For the run-segmentation refinement, \citet[Theorem~11 and Lemma~13]{Chen26} gives an increasing endpoint-window entropy hierarchy and integer posterior counts; that work also develops directed entropy evaluation. These run endpoints differ from the individual raw-bit deletion indicators in \cref{sec:source}. The identity below separates the jigsaw term from the remaining start uncertainty; finite lower estimates of both give an achievable rate for the same $P$.

\input{sections/roadmap_renewal}
\subsection{A stationary input with prescribed run lengths}
Starting at a fresh run boundary is generally not stationary. A uniformly located input position is more likely to lie in a long run, in proportion to its length. We construct this stationary initial position explicitly.

State $(b,r)$ emits $b$ and records $r$ positions remaining, including the current one. If $r>1$, move to $(b,r-1)$; otherwise draw $\ell\sim P$ and move to $(1-b,\ell)$. Start with masses
\begin{equation}
 \pi(b,r)=\frac{\Pr_P\{\ell\ge r\}}{2\overline L}.
 \label{eq:renewal-stationary}
\end{equation}
For $r<J$, incoming mass to $(b,r)$ is $\pi(b,r+1)+P_r\pi(1-b,1)=[\Pr_P\{\ell\ge r+1\}+P_r]/(2\overline L)=\pi(b,r)$. At $r=J$ the first term is absent.

Summing over bits and remaining lengths gives $\sum_r\Pr_P\{\ell\ge r\}/\overline L=1$: each length $\ell$ contributes to exactly $\ell$ tail probabilities. The initial masses are therefore normalized and invariant. Fresh lengths at completed boundaries still have law $P$.

\emph{Example.} If lengths one and two each have probability $1/2$, then $\overline L=3/2$, $\pi(0,1)=\pi(1,1)=1/3$, and $\pi(0,2)=\pi(1,2)=1/6$. Every run visits remaining length one; only length-two runs visit remaining length two.

\subsection{Grouping input runs into output runs}
Deleting the middle run of $00\mid1\mid000$ can merge the two zero runs into one output run. The trace no longer marks where one zero run's contribution ends and the other's begins. The jigsaw calculation groups input runs by the output run they produce.

A \emph{type} lists the complete input runs from one positive-output run up to, but excluding, the next opposite-bit run with positive output. Independent fresh lengths and channel counts make successive complete types independent. We compute entropy per type, then divide by its mean raw length.

Supplying these group boundaries can reveal more than the input requires given its trace. The conditional entropy of the boundaries measures that excess description. Subtracting it from the conditional input description increases the achievable rate.

For $0\mid1\mid0\to0$, keeping only the first or last zero each has mass $d^2(1-d)$. With previous output bit one, they give equally likely different global run starts. With previous bit zero, both merge into the preceding output run and give no new start. A local first positive run is therefore not necessarily a global output-run start.

At $d=1/2$, the first case contributes $2(1/8)\cdot1=1/4$ bit of expected block entropy: posterior uncertainty must be multiplied by observation probability.

\paragraph{Step 1: describe the output of one input run.}
Let $\beta_\ell(k)$ be the probability that an input run of length $\ell$ produces $k$ output bits, and put $z_\ell=\beta_\ell(0)$. Independent survival gives the binomial count distribution
\[
 \beta_\ell(k)=\binom\ell k(1-d)^kd^{\ell-k},\qquad z_\ell=d^\ell.
\]
We record the Poisson version alongside the deletion law because the same run-group calculation will later yield the all-parameter bound in \cref{sec:renewal-poisson}. For Poisson repeats of intensity $\lambda$, each input bit produces an independent Poisson$(\lambda)$ count. Its generating function is $e^{\lambda(t-1)}$ by the exponential series. Multiplying $\ell$ such functions gives $e^{\lambda\ell(t-1)}$, whose coefficients are
\[
 \beta_\ell(k)=e^{-\lambda\ell}\frac{(\lambda\ell)^k}{k!},\qquad
 z_\ell=e^{-\lambda\ell}.
\]
Averaging disappearance over the source gives $D_0$. Bayes' rule gives length distributions $A$ conditional on positive output and $B$ conditional on no output:
\[
 D_0=\sum_\ell P_\ell z_\ell,\qquad
 A_\ell=\frac{P_\ell(1-z_\ell)}{1-D_0},\qquad
 B_\ell=\frac{P_\ell z_\ell}{D_0}.
\]
The evaluated parameters satisfy $0<D_0<1$. These conditional distributions are determined by $P$ and the channel, not free comparison distributions. Here $B_\ell$ is distinct from the start-set variables below.

\paragraph{Step 2: find the group that produces one output run.}
Starting at a positive run, let $I$ count vanished opposite-bit runs before the next positive opposite-bit run. Let $Z$ be the initial length, $S_i$ the vanished opposite-bit lengths, and $R_i$ the intervening same-bit lengths. The type is $T=(Z,S_1,R_1,\ldots,S_I,R_I)$. Here $T$ denotes a run list, not the transition matrix of \cref{sec:source}.

Each opposite-bit run independently vanishes with probability $D_0$. Initial and vanished lengths have the conditional laws above; intervening same-bit lengths are unrestricted. Hence
\begin{equation}
 \Pr\{I=i\}=(1-D_0)D_0^i,\qquad Z\sim A,\quad S_i\sim B,\quad R_i\sim P,
 \label{eq:renewal-type-law}
\end{equation}
Draw $I$, then draw the listed lengths independently from their indicated laws. Given them, draw channel counts with the corresponding conditions. Let $K$ be the total output-run length. The next positive opposite-bit run restarts this construction, making complete pairs $(T,K)$ independent and identically distributed.

\paragraph{Step 3: convert one group into a rate per input bit.}
Since $\Pr(I\ge j)=D_0^j$, summing tail probabilities gives $\E I=D_0/(1-D_0)$. Put $L_0=\sum_\ell \ell P_\ell z_\ell$. Averaging lengths under $A$, $B$, and $P$ gives
\[
 \mathbb EZ=\frac{\overline L-L_0}{1-D_0},\qquad
 \mathbb ES_1=\frac{L_0}{D_0},\qquad \mathbb ER_1=\overline L.
\]
A type contains its initial run and $I$ independent pairs, so
\[
 \mathbb E|T|_{\rm raw}
 =\mathbb EZ+\mathbb EI(\mathbb ES_1+\mathbb ER_1)
 =\frac{\overline L-L_0+L_0+D_0\overline L}{1-D_0}.
\]
The reciprocal mean length is the type density $\nu$ per input bit; the stopping argument below justifies this conversion.

Describe $I$, the initial length, and the $I$ pairs. Independence gives $H(T)=H(I)+H(A)+\E I[H(B)+H(P)]$. The geometric self-information is $-\log_2(1-D_0)-I\log_2D_0$, with mean $H(I)=h_2(D_0)/(1-D_0)$. Substitution yields
\begin{equation}
 \begin{aligned}
 \mathbb E|T|_{\rm raw}&=\frac{(1+D_0)\overline L}{1-D_0},
 &\nu&=\frac{1-D_0}{(1+D_0)\overline L},\\
 H(T)&=H(A)+\frac{h_2(D_0)+D_0(H(B)+H(P))}{1-D_0}.
 \end{aligned}
 \label{eq:renewal-type-entropy}
\end{equation}

The same elementary count bound will control boundary-output lengths in the proof and omitted Poisson counts in the computation.
\begin{lemma}[An entropy bound from a count's mean]\label{lem:renewal-count-entropy}
A nonnegative integer variable $N$ of mean $\mu$ satisfies
\[
 H(N)\le g(\mu):=(\mu+1)\log_2(\mu+1)-\mu\log_2\mu
 \le\log_2(\mu+1)+1/\ln2.
\]
Moreover, $g$ is concave on $[0,\infty)$, with $g(0)=0$.
\end{lemma}
\begin{proof}
For $\mu>0$, compare with the geometric probabilities $Q(j)=(\mu+1)^{-1}[\mu/(\mu+1)]^j$. Averaging their negative logarithms gives $\log_2(\mu+1)+\mu\log_2(1+1/\mu)=g(\mu)$. Cross entropy bounds $H(N)$ by this value, and $\ln(1+x)\le x$ gives the last inequality. Differentiation gives $g''(\mu)=-1/[\ln(2)\mu(1+\mu)]<0$. At mean zero, $N=0$ almost surely; continuity gives both the entropy bound and concavity at the endpoint.
\end{proof}

Let $B_n^{\rm start}$ mark input-run indices beginning actual global output runs in $X_1^n$. Positive contributions that merge have no separate marks. Write $I_P$ for this source's information rate. The next identity expresses it as a jigsaw term $J_P$ plus conditional start entropy $a_P$.
\begin{proposition}[The jigsaw rate and conditional run-start entropy]
\label{prop:renewal-identity}
The stationary source information rate satisfies
\begin{equation}
 I_P=J_P+a_P,\qquad
 J_P=\frac{H(P)}{\overline L}-\nu H(T\mid K),\qquad
 a_P=\lim_{n\to\infty}\frac1nH(B_n^{\rm start}\mid X_1^n,Y_n)\ge0.
 \label{eq:renewal-identity}
\end{equation}
In particular $J_P+a$ is a capacity lower bound whenever $a\le a_P$ is established for this same source.
\end{proposition}
\begin{proof}
We first calculate entropy with the run starts supplied, then remove their conditional entropy by the chain rule. Finite input windows cut through at most two types; we control these boundary descriptions before taking rates.

\textbf{Step 1: encode the marked word and control its boundaries.}
Complete $(T,K)$ pairs and the two boundary fragments determine $(X_1^n,B_n^{\rm start},Y_n)$: concatenate the raw lengths with alternating bits, mark each type's start, and emit its $K$ output bits. Conversely, marked starts divide the raw-run list into types, and output-run lengths give $K$. An initial bit and boundary positions complete this two-way description.

Since $|T|_{\rm raw}\le J(2I+1)$, a raw type longer than $J(2q+1)$ requires $I>q$, of probability $D_0^{q+1}$. Thus raw type length has an exponential tail.

A stationary endpoint sees length $\ell$ with mass proportional to $\ell\Pr(|T|_{\rm raw}=\ell)$, because that type contains $\ell$ possible endpoint positions. This extra factor preserves all finite moments. For a fresh renewal sequence, summing over possible starting positions gives the same length-weighted upper bound. Both boundary raw descriptions consequently have uniformly bounded mean length.

Deletion output length is at most raw length. For Poisson repeats, conditional on a type of raw length $L$, its unconditioned same-bit count has mean at most $\lambda L$. Conditioning the first run to be positive divides probabilities by at least $1-e^{-\lambda}$, since its length is at least one. For $t>0$,
\[
 \mathbb E[e^{tK}\mid T]
 \le \frac{\exp\{\lambda L(e^t-1)\}}{1-e^{-\lambda}}.
\]
The right side has finite expectation for sufficiently small $t>0$: its exponent coefficient $\lambda(e^t-1)$ then lies below the raw-length exponential decay rate. Endpoint length bias preserves this integrability. Encode boundary raw lengths and output counts in unary, with bits literally; these descriptions have uniformly finite mean length and therefore bounded entropy.

\textbf{Step 2: compute the rate of the marked description.}
For a fresh type sequence, put $L_j=|T_j|_{\rm raw}$, $S_j=L_1+\cdots+L_j$, and $\tau=\min\{j:S_j\ge n\}$. The stopped list has product probability, so its self-information adds across types. The event $\tau\ge j$ depends only on earlier lengths, independently of type $j$. Its expected entropy is therefore $\E\tau\,H(T,K)$.

To bound the crossing type, sum over its possible starting position $r<n$. At most one renewal occurs at each integer $r$, and its next length is independent of arrival there. Thus
\[
 \mathbb E L_\tau
 \le\sum_{r=0}^{n-1}\mathbb E[L_1\mathbf1_{\{L_1\ge n-r\}}]
 \le\sum_{u\ge1}\mathbb E[L_1\mathbf1_{\{L_1\ge u\}}]
 =\mathbb E L_1^2.
\]
For a realized $L_1$, the last sum has $L_1$ identical terms, giving $\E L_1^2$. Since $n\le S_\tau<n+L_\tau$, the expected overshoot is bounded. Expanding $S_\tau=\sum_jL_j\mathbf1_{\{\tau\ge j\}}$ and using the same independence gives $\E S_\tau=\E\tau\,\E L_1=n+O(1)$. Hence $\E\tau/n\to1/\E L_1=\nu$.

Let $U$ be the stopped type list and $V$ the finite marked input--output object. The chain rule gives
\[
 |H(U)-H(V)|\le H(U\mid V)+H(V\mid U).
\]
Each conditional entropy is bounded by its two-way boundary description above. Removing the crossing type and adding the stationary initial residual therefore changes entropy by $o(n)$. Dividing by $n$ gives
\[
 \frac1nH(X_1^n,B_n^{\rm start},Y_n)\longrightarrow\nu H(T,K).
\]
\textbf{Step 3: compute the input and output rates in the same units.}
For input alone, apply the same stopped-sum argument to independent raw lengths of mean $\overline L$ and entropy $H(P)$. The entropy rate is $H(P)/\overline L$.

For output alone, complete run lengths have law $K$ and alternating bits. The same argument gives entropy $H(K)/\E K$ per output symbol; the initial bit and incomplete run have vanishing rate.

To convert to input units, let $N_n=|Y_n|$ and $q_n=\lfloor\mathbb EN_n\rfloor$. The prefix comparison in the proof of \cref{lem:trace-rate} gives
\[
 |H(Y_n)-H(V_1^{q_n})|\le H(N_n)+\mathbb E|N_n-q_n|,
\]
where $V$ is the same infinite output. The reason applies to either channel: disclose the length and missing binary suffix to recover either prefix from the other. The count mean and variance are $n(1-d),nd(1-d)$ for deletion and $\lambda n,\lambda n$ for repeats. Thus Cauchy--Schwarz gives $\mathbb E|N_n-q_n|\le\sqrt{\operatorname{Var}(N_n)}+1$, while \cref{lem:renewal-count-entropy} gives $H(N_n)=O(\log n)$. The total error is $O(\log n+\sqrt n)$ and vanishes after division by $n$.

Each type produces mean $\E K$ output bits and there are $\nu n+o(n)$ types in expectation. Boundary outputs have bounded mean, so $\E N_n/n\to\nu\E K$. Multiply this density by entropy $H(K)/\E K$ per output symbol to obtain $H(Y_n)/n\to\nu H(K)$; the random-length error above vanishes after division by $n$.

\textbf{Step 4: remove the marks and recover the omitted information.}
We now have all three entropy rates. The only entropy lost on removing the marks is the conditional uncertainty of their positions. The chain rule (\cref{lem:entropy-chain}) writes
$H(X_1^n,B_n^{\rm start},Y_n)=H(X_1^n,Y_n)+H(B_n^{\rm start}\mid X_1^n,Y_n)$.
Substitute this into $I(X_1^n;Y_n)=H(X_1^n)+H(Y_n)-H(X_1^n,Y_n)$ to obtain
\[
 I(X_1^n;Y_n)=H(X_1^n)+H(Y_n)-H(X_1^n,B_n^{\rm start},Y_n)
 +H(B_n^{\rm start}\mid X_1^n,Y_n).
\]
This is an exact finite-word identity. Before taking its limit, we must also show that the information rate exists for the correlated input. Let $X_A,X_B$ be two adjacent raw blocks, and let $Y_A,Y_B$ be their separately observed traces. The actual output is their concatenation. Revealing the division between traces can only increase mutual information by data processing (\cref{lem:data-processing}). Given the input blocks, the two channel outputs are independent, even though the input blocks themselves may be correlated. Expanding mutual information and then using $H(Y_A,Y_B)\le H(Y_A)+H(Y_B)$ gives
\begin{align*}
 I(X_A,X_B;\operatorname{concat}(Y_A,Y_B))
 &\le I(X_A,X_B;Y_A,Y_B)\\
 &=H(Y_A,Y_B)-H(Y_A\mid X_A)-H(Y_B\mid X_B)\\
 &\le I(X_A;Y_A)+I(X_B;Y_B).
\end{align*}
Stationarity identifies both block marginals, giving $I_{m+n}\le I_m+I_n$. For $n=kq+r$, $0\le r<q$, repeated subadditivity and $I_r\le r$ give $I_n\le kI_q+r$. Thus $\limsup_n I_n/n\le I_q/q$ for every $q$; the lower limit is at least $\inf_q I_q/q$ by definition. The rate exists. Substituting the established entropy limits into the chain-rule identity proves \eqref{eq:renewal-identity}; conditional entropy is nonnegative, so $a_P\ge0$.

The information rate of this source cannot exceed capacity (\cref{sec:setting}). Therefore lower estimates of $J_P$ and $a_P$ for the same $P$ give a capacity lower bound.
\end{proof}

\subsection{Where computation enters: the run-group baseline}\label{sec:renewal-baseline-computation}
Fix the run-length law $P$, with mean $\overline L$ and maximum $J$. A type $T$ groups input runs producing one output run of length $K$, and $\nu$ counts types per input bit. To evaluate $J_P=H(P)/\overline L-\nu H(T\mid K)$, use

$H(T\mid K)=H(T)+H(K\mid T)-H(K)$, by the chain rule. The finite formula \eqref{eq:renewal-type-entropy} gives $H(T)$. We need an upper bound for the subtracted conditional count entropy and a lower bound for $H(K)$. Merges can involve arbitrarily many runs even though each raw run has length at most $J$. The computation fixes $P$ and the channel parameter, then chooses a merge-count cutoff $q$ and an output-count cutoff $K_0$. These are numerical cutoffs, not restrictions on the channel. We first derive the finite probabilities and then the omitted-entropy bounds.
\paragraph{Step 1: compute the count law for a known type.}

A type contains an initial positive run of length $z$, then $i$ same-bit runs separated by vanished opposite-bit runs. Let $Q_i=P^{*i}$ be the distribution of their total same-bit length $r$, with $Q_0(0)=1$. This convolution is determined by $P$; it is unrelated to an upper-bound comparison distribution.

The unrestricted output count has law $\beta_{z+r}$. Subtract cases where the initial run vanishes, of mass $z_z\beta_r(k)$, then divide by its positive-output probability $1-z_z$:
\begin{equation}
 W_{z,r}(k)=\frac{\beta_{z+r}(k)-z_z\beta_r(k)}{1-z_z},\qquad k\ge1.
 \label{eq:renewal-W}
\end{equation}
Here $z_z$ is the disappearance probability at length $z$. Multiply the conditional count entropy by the type mass $(1-D_0)D_0^i$, initial-length mass $A_z=P_z(1-z_z)/(1-D_0)$, and added-length mass $Q_i(r)$. Cancel $1-D_0$ to obtain
\begin{equation}
 \mathbb EH(K\mid T)=\sum_{i\ge0}D_0^i\sum_zP_z(1-z_z)
 \sum_r Q_i(r)H(W_{z,r}).
 \label{eq:renewal-conditional-sum}
\end{equation}

\paragraph{Step 2: compute the output-run law without knowing its type.}
For the positive term $H(K)$, let $F(t)=\sum_\ell P_\ell\sum_{k\ge0}\beta_\ell(k)t^k$ and $f_k=[t^k]F(t)$, so $F(0)=D_0$. A positive initial run has generating function $(F(t)-D_0)/(1-D_0)$. Averaging the number of additional same-bit runs gives $\sum_{i\ge0}(1-D_0)D_0^iF(t)^i=(1-D_0)/(1-D_0F(t))$. Multiply these factors:
\begin{equation}
 \mathbb Et^K=\frac{F(t)-D_0}{1-D_0F(t)}.
 \label{eq:renewal-pgf}
\end{equation}
Let $\kappa_k=\Pr\{K=k\}$; an output run is positive, so $\kappa_0=0$. Multiply \eqref{eq:renewal-pgf} by its denominator and compare coefficients. Since $f_0=D_0$, move the term $D_0^2\kappa_k$ to the left and divide:
\begin{equation}
 \kappa_k=\frac{f_k+D_0\sum_{j=1}^k f_j\kappa_{k-j}}{1-D_0^2}.
 \label{eq:kappa-recurrence}
\end{equation}
For deletion, $f_j=0$ for $j>J$, so the recurrence uses a finite polynomial. For either channel, $\sum_{k=1}^{K_0}-\kappa_k\log_2\kappa_k\le H(K)$: omitted entropy terms are nonnegative. This is the required lower direction because $H(K)$ enters $J_P$ positively.

\paragraph{Step 3: upper-bound the cost of omitted merge counts.}
For the subtracted $H(K\mid T)$, omitted terms require an upper bound. In deletion, a type with $I=i$ produces at most $J(i+1)$ positive output lengths, hence entropy at most $\log_2[J(i+1)]$. After retaining $I\le q$, the omitted mass is $D_0^{q+1}$. Conditional on omission, $I-(q+1)$ has the original geometric distribution, so $\E[I+1\mid I>q]=q+2+D_0/(1-D_0)$. Concavity of the logarithm bounds its mean by the logarithm at that mean:
\begin{equation}
 \sum_{i>q}\Pr(I=i)\mathbb E[H(K\mid T)\mid I=i]
 \le D_0^{q+1}\log_2\!\left[J\left(q+2+\frac{D_0}{1-D_0}\right)\right].
 \label{eq:renewal-deletion-tail}
\end{equation}
This tail bound completes the conditional-entropy estimate for deletion. The numerical inputs and rounding directions can now be stated independently of the recurrence implementation. The published deletion-channel replay uses merge cutoff $q=7$, output-count cutoff $K_0=192$, and 256-bit Arb arithmetic. Increasing $q$ enlarges the collection of retained run types; increasing $K_0$ retains more count probabilities. In each case the formulas above still bound the complete omitted tail.
\begin{lemma}[Finite entropy bounds give a run-group rate]\label{lem:renewal-finite-baseline}
For one chosen source $P$, let $r_-\le H(P)/\overline L$, $t_+\ge H(T)$, $c_+\ge H(K\mid T)$, and $k_-\le H(K)$. With the type density $\nu$ from \eqref{eq:renewal-type-entropy},
\[
 J_P\ge r_- -\nu(t_++c_+-k_-).
\]
A downward evaluation of the whole right side, including any interval uncertainty in $\nu$, is a valid lower bound $J_P^-$.
\end{lemma}
\begin{proof}
By the chain rule, $H(T\mid K)=H(T)+H(K\mid T)-H(K)\le t_++c_+-k_-$. Multiply by the positive density $\nu$ and subtract from the lower estimate of $H(P)/\overline L$ in \eqref{eq:renewal-identity}. The inequality reverses on subtracting, giving the claim. The entropy quantities in parentheses are per group; multiplication by groups per input bit puts them in the same units as $r_-$.
\end{proof}

\emph{Example.} With $J=2$, $D_0=1/2$, and $q=2$, the omitted mass is $1/8$ and the conditional mean of $I+1$ is five. The tail is at most $\log_2(10)/8$ bits per type; multiplying by $\nu$ converts it to bits per input bit.

For Poisson repeats, output counts are unbounded. Put $M=\lambda\overline L$ and use the concave mean-entropy bound $g$ from \cref{lem:renewal-count-entropy}. This replaces the deletion channel's finite-alphabet bound $\log_2[J(i+1)]$.

Conditioning an initial length-$z$ run to be positive changes its mean count to $\lambda z/(1-z_z)$. Averaging against $A_z=P_z(1-z_z)/(1-D_0)$ cancels $1-z_z$ and gives $M/(1-D_0)$; every added same-bit run contributes mean $M$. Conditional on $I>q$, their mean number is $q+1+D_0/(1-D_0)$. The omitted conditional-entropy sum is at most its mass $D_0^{q+1}$ times $g$ of this mean, by concavity. The resulting upper bound is
\begin{equation}
 D_0^{q+1}g\!\left(M\left[q+1+\frac{1+D_0}{1-D_0}\right]\right).
 \label{eq:prc-merge-tail}
\end{equation}

\paragraph{Step 4: upper-bound the remaining Poisson count tail.}
Even for retained merge counts, each Poisson law $W_{z,r}$ has infinitely many output counts. For $\theta>1$, its generating function is at most $G=\exp(\lambda(z+r)(\theta-1))/(1-e^{-\lambda z})$: drop the negative conditioning correction from its numerator.

Nonnegative generating-function terms give $W_{z,r}(k)\le G\theta^{-k}$. Beyond cutoff $K_0$, put $B=G\theta^{-(K_0+1)}$ and $\rho=\theta^{-1}$. The omitted masses are bounded by $B\rho^j$. Verify $B<1/e$; on this range $-x\log_2x$ is increasing because its derivative is $-\log_2x-1/\ln2\ge0$. Replacing each probability by its upper bound and summing gives
\begin{equation}
 \sum_{k>K_0}-W_{z,r}(k)\log_2W_{z,r}(k)
 \le\frac{B}{1-\rho}\left[-\log_2B+\frac{\rho}{1-\rho}\log_2\theta\right].
 \label{eq:prc-output-tail}
\end{equation}
The summed expression is $\sum_{j\ge0}B\rho^j[-\log_2B+j\log_2\theta]$. The geometric sums are $\sum_j\rho^j=1/(1-\rho)$ and $\sum_jj\rho^j=\rho/(1-\rho)^2$; the latter follows by writing $j=\sum_{u=1}^j1$ and summing the nonnegative tails. Substitution proves the displayed bound. For each retained type, weight its Poisson entropy-tail bound by its original type probability. Add these weighted bounds and the omitted-merge bound \eqref{eq:prc-merge-tail} to the retained conditional-entropy sum. This gives $c_+$, and \cref{lem:renewal-finite-baseline} gives $J_P^-$. The count recurrence may stop at $K_0$ for the positive $H(K)$ term, but a corresponding cutoff in the subtracted $H(K\mid T)$ term requires precisely these tail charges.

\subsection{Alignment entropy from finite run blocks}
Recall that $a_P$ is conditional global-start entropy per input bit, $D_0=\sum_\ell P_\ell\beta_\ell(0)$ is a run's disappearance probability, and $\overline L=\sum_\ell\ell P_\ell$ is mean input-run length. To improve $J_P^-$, lower-bound $a_P$. Reveal the trace length of each finite input-run block. These lengths separate block traces but need not identify the input runs that started them.

Fix $k$ complete lengths $\ell=(\ell_1,\ldots,\ell_k)$ and set their first bit to zero by complement symmetry. Let $c$ be the preceding output bit and $B^c$ the input-run indices beginning global output runs. The first positive run is marked only if its bit differs from $c$. For block trace $Y_{\rm block}$, define
\[
 h_k(c)=\mathbb E_\ell H(B^c\mid\ell,Y_{\rm block},c).
\]
Average over independent lengths with law $P$. Looking backwards from a run boundary, the number of wholly vanished runs is geometric with parameter $1-D_0$. Even counts give preceding bit opposite with probability $(1-D_0)\sum_{j\ge0}D_0^{2j}=1/(1+D_0)$; odd counts give equal bit with probability $D_0/(1+D_0)$. This past is independent of the fresh block.

\begin{lemma}[A finite block reveals positive alignment entropy]
\label{lem:renewal-fixed-block}
For either channel,
\begin{equation}
 a_P\ge\widetilde A_k(P)=\frac{h_k(1)+D_0h_k(0)}{(1+D_0)k\overline L}.
 \label{eq:renewal-block-credit}
\end{equation}
No first-positive-run index is revealed in this bound.
\end{lemma}
\begin{proof}
\textbf{Step 1: retain the correct start variable.}
Take $r$ blocks of $k$ complete runs, with full length list $\ell$. Let $B^+$ mark the first positive run and subsequent output-bit changes. Draw an independent preceding bit $C_0$ with the stationary relative probabilities above. Remove the first mark if its bit equals $C_0$, obtaining the actual starts $B^{C_0}$.

The removal is determined by $(B^+,Y,C_0)$, so it decreases conditional entropy. Revealing local output lengths $M_1,\ldots,M_r$ decreases it again:
\[
 H(B^+\mid\ell,Y)\ge H(B^{C_0}\mid\ell,Y,C_0)
 \ge H(B^{C_0}\mid\ell,Y,C_0,M_1,\ldots,M_r),
\]
\textbf{Step 2: use the disclosed boundaries to separate blocks.}
The lengths divide $Y$ into block traces, including empty ones. Together with $C_0$, they determine each block's preceding bit. Given the raw lengths, channel counts are independent across blocks; conditioning on separate block outputs preserves this product distribution. The last conditional entropy therefore adds over blocks.

At every boundary, the backwards geometric calculation gives the same preceding-bit mixture, independent of fresh lengths. Each block's mean contribution is $[h_k(1)+D_0h_k(0)]/(1+D_0)$.

\textbf{Step 3: remove boundary conventions and divide by raw length.}
To compare with the stationary global starts $B^{\rm glob}$, disclose the first positive raw-run index, with one extra value if all runs vanish. This recovers $B^+$ from $B^{\rm glob}$ and has at most $rk+1$ possibilities. Hence
\[
 H(B^+\mid\ell,Y)
 \le H(B^{\rm glob}\mid\ell,Y)+\log_2(rk+1).
\]
The discrepancy is logarithmic. The first-positive index only bounds this boundary error; it is not given in the local quantities $h_k(c)$.

The raw length $L^{(r)}$ of $rk$ independent runs has mean $rk\overline L$ and variance $rk\operatorname{Var}_P(\ell)$. Cauchy--Schwarz gives $\E|L^{(r)}-\lfloor rk\overline L\rfloor|\le\sqrt{rk\operatorname{Var}_P(\ell)}+1$. Couple the random-length word with this fixed-length window. To recover either marked input--output object, describe the differing raw suffix, output suffix, start indicators, and cropping lengths.

The expected raw suffix length is $O(\sqrt r)$. Deletion produces no more output than raw bits. For Poisson repeats, cropping depends only on input lengths, so the expected output suffix length is $\lambda$ times the raw suffix length. Its count has entropy $O(\log r)$ by the geometric bound. Binary suffixes and start indicators need at most one bit per position; cropping lengths cost $O(\log r)$. Both recovery descriptions thus have entropy $O(\sqrt r+\log r)$. Apply the chain-rule entropy comparison to the marked and unmarked objects, then subtract, to bound the change in conditional start entropy by the same order.

The stationary initial partial run contributes only an integrable boundary description. Divide the finite inequality by $rk\overline L$ and let $r$ grow. All boundary costs vanish, leaving \eqref{eq:renewal-block-credit}.
\end{proof}

\begin{lemma}[Count posterior alternatives before taking entropy]\label{lem:renewal-posterior-masses}
Fix input-run lengths, a received word, and the preceding output bit. Group channel-count paths by their \emph{actual global start set}. If these groups have original masses $p_1,\ldots,p_q$, their contribution to conditional start entropy is
\[
 \Phi(p_1,\ldots,p_q)=p\log_2p-\sum_i p_i\log_2p_i,
 \qquad p=\sum_i p_i.
\]
Dropping paths after this grouping, or dropping complete observations, gives a lower bound when all remaining masses are kept unchanged.
\end{lemma}
\begin{proof}
For this observation, $p$ is its probability and $p_i/p$ is the posterior probability of start set $i$. Multiplying $-\sum_i(p_i/p)\log_2(p_i/p)$ by $p$ gives the displayed expression. For $p_i>0$, its derivative in that mass is $\log_2(p/p_i)\ge0$; continuity covers zero masses. Removing paths decreases group masses, so it cannot increase the entropy contribution. Removing a complete observation discards a nonnegative contribution. Neither operation justifies renormalizing the retained probabilities.
\end{proof}
Two paths of mass $1/4$ each contribute $1/2$ bit if their start sets differ, but zero if both have the same start set and hence form one alternative. Counting paths instead of distinct start sets would therefore overstate the lower bound. The finite calculation must first identify the alternatives whose uncertainty is being measured.

Two estimates of the same conditional entropy may be compared by taking their maximum. Their sum is justified only when they bound disjoint terms of the stated entropy decomposition. We use this distinction in the short- and long-block constructions below.

\subsection{Three-run and four-run deletion formulas}
We now evaluate $h_k(c)$ for three and four runs. Keep the same $P$, $D_0=\sum_\ell P_\ell d^\ell$, and binomial masses $\beta_n(k)=\binom nk(1-d)^kd^{n-k}$, zero outside their support.

\paragraph{Three runs: uncertainty between two possible starts.}
For two same-bit runs of lengths $z,r$, with the intervening opposite run absent, output count $u$ starts in the first run with mass $\beta_{z+r}(u)-d^z\beta_r(u)$, or in the later run with mass $d^z\beta_r(u)$. The subtraction removes cases where the first run vanishes. The second alternative is impossible for $u>r$. Denote their total entropy contribution by
\[
 Q(z,r)=\sum_{u=1}^{r}\Phi\!\left(\beta_{z+r}(u)-d^z\beta_r(u),\ d^z\beta_r(u)\right).
\]
The middle run vanishes with average probability $D_0$. A new global start also requires preceding bit opposite, of probability $1/(1+D_0)$. Average outer lengths with $P_zP_r$ and divide by mean block length $3\overline L$:
\begin{equation}
 \widetilde A_3=\frac{D_0}{3\overline L(1+D_0)}\sum_{z,r}P_zP_rQ(z,r).
 \label{eq:renewal-A3}
\end{equation}

\paragraph{Four runs: merge indistinguishable start patterns first.}
For lengths $(a,b,c,e)$ with bits $0,1,0,1$ and output $0^u1^v$, possible start patterns are $(1,2)$, $(1,4)$, and $(3,4)$. Independent run counts give their masses $x,y,z$:
\begin{align*}
 x&=\beta_a(u)d^c[\beta_{b+e}(v)-d^b\beta_e(v)],\\
 y&=d^b[\beta_{a+c}(u)-d^a\beta_c(u)]\beta_e(v),\\
 z&=d^{a+b}\beta_c(u)\beta_e(v).
\end{align*}
For $x$, run three vanishes and run two starts the ones; its bracket removes cases where run two vanishes. For $y$, run two vanishes and run one starts the zeros. For $z$, runs one and two vanish. These events are disjoint.

With preceding bit one, the three global start sets differ, contributing $\Phi(x,y,z)$. With preceding bit zero, the first zero output merges into the past. Patterns $(1,4)$ and $(3,4)$ then both have start set $\{4\}$, so their masses must be merged, giving $\Phi(x,y+z)$.

Zero-only output requires runs two and four to vanish, giving $d^{b+e}Q(a,c)$ when the preceding bit is one. The corresponding one-only contribution is $d^{a+c}Q(b,e)$ for preceding bit zero. Adding these terms gives
\begin{align*}
 G_1(a,b,c,e)&=d^{b+e}Q(a,c)+\sum_{u=1}^{\max(a,c)}\sum_{v=1}^{e}\Phi(x,y,z),\\
 G_0(a,b,c,e)&=d^{a+c}Q(b,e)+\sum_{u=1}^{\max(a,c)}\sum_{v=1}^{e}\Phi(x,y+z).
\end{align*}
Outside these ranges, or with more output-run changes, at most one start pattern has positive mass and entropy is zero. Average the four source lengths, mix preceding-bit cases with relative weights $1,D_0$, and divide by $4\overline L$:
\begin{equation}
 \widetilde A_4=\frac{\sum_{a,b,c,e}P_aP_bP_cP_e[G_1+D_0G_0]}{4\overline L(1+D_0)}.
 \label{eq:renewal-A4}
\end{equation}
\paragraph{Complete the ordinary deletion-channel bound.}
Evaluate either block sum downward, allowing selected length words with their original masses. Both estimates lower-bound the same $a_P$, so take their maximum. Adding the jigsaw lower endpoint and using nonnegativity of capacity gives
\begin{equation}
 \boxed{C(d)\ge\max\{0,J_P^-+\max(\widetilde A_3^-,\widetilde A_4^-)\}.}
 \label{eq:renewal-ordinary-final}
\end{equation}
All terms must use the same source $P$; separate whole-source rates may then be maximized.

\subsection{Poisson simulation and an all-parameter lower bound}\label{sec:renewal-poisson}
To obtain a lower bound for every $d$ from one calculation, use the Poisson simulation of \citet{MD06}. Let $C_{\rm PRC}(\lambda)$ be capacity per input bit when each bit independently produces a Poisson$(\lambda)$ number of copies.

Repeat each input bit randomly before deletion, choosing the copying mean to leave Poisson$(\lambda)$ surviving copies. The proof verifies this count distribution and controls the random transmitted length.

\emph{Example.} At $d=3/4$, Poisson copying of mean two leaves mean $1/2$ after deletion. It simulates intensity $\lambda=1/2$ using about two transmitted bits per original bit, hence half the original rate.
\begin{lemma}[Poisson simulation]
\label{lem:poisson-transfer}
For every $\lambda>0$ and $0\le d\le1$,
\begin{equation}
 C(d)\ge\frac{1-d}{\lambda}C_{\rm PRC}(\lambda).
 \label{eq:poisson-transfer}
\end{equation}
\end{lemma}
\begin{proof}
\textbf{Step 1: match the channel law exactly.}
Fix $0<d<1$ and put $r=\lambda/(1-d)$. Independently replace each original bit by Poisson$(r)$ copies. Given copy count $L$, the survivor generating function is $[d+(1-d)t]^L$. Averaging the Poisson count gives
\[
 \mathbb E[d+(1-d)t]^L
 =\exp\{r[d+(1-d)t-1]\}
 =\exp\{\lambda(t-1)\}.
\]
This is Poisson$(\lambda)$. Independence of copying and deletion across original bits gives the complete repeat-channel distribution.

\textbf{Step 2: impose a fixed transmission budget.}
For slack $\delta>0$, impose transmission length $N=\lceil(r+\delta)n\rceil$. The preprocessing length $L$ is Poisson$(rn)$ and independent of the input. Markov's inequality applied to $e^{tL}$ gives, for $t>0$,
\[
 \Pr(L>N)
 \le \exp\{n[r(e^t-1)-t(r+\delta)]\}.
\]
The exponent's bracket has value zero and derivative $-\delta$ at $t=0$. It is negative at some small fixed $t>0$, so overflow is exponentially unlikely. On overflow transmit a fixed length-$N$ word; otherwise append padding.

Off overflow $E=\{L>N\}$, the ideal output $Y_{\rm PRC}$ is a prefix of the padded output $Y_{\rm padded}$. Let $K_0$ be its length there, zero on overflow. Its mean is at most $\lambda n$, so the geometric entropy bound gives $H(K_0)=O(\log n)$. Disclosing $(E,K_0)$ recovers the ideal output off overflow; on overflow it carries at most $n$ additional bits about the binary input $U_1^n$. The chain rule gives
\[
 I(U_1^n;Y_{\rm PRC})
 \le I(U_1^n;Y_{\rm padded})
       +H(E,K_0)+n\Pr(E).
\]
Conditional on $E=0$, no additional mutual information remains; conditional on $E=1$, it is at most $H(U_1^n\mid E=1)\le n$. Also $H(E,K_0)\le1+H(K_0)$ and $n\Pr(E)=o(n)$. Thus
\[
 I(U_1^n;Y_{\rm PRC})\le I(U_1^n;Y_{\rm padded})+O(\log n)+o(n)
 \le N C_N(d)+o(n).
\]
\textbf{Step 3: compare rates and remove the budget slack.}
The second inequality uses data processing through the randomized length-$N$ input, whose mutual information is at most $NC_N(d)$. Maximize over original input laws and divide by $n$. Both channels satisfy the capacity-limit hypotheses of \cref{sec:setting}: independent finite output words per input symbol and finite mean output length. Since $N/n\to r+\delta$, the result is $C_{\rm PRC}(\lambda)\le(r+\delta)C(d)$. Let $\delta\downarrow0$.

At $d=1$ both sides are zero. At $d=0$, use $C_{\rm PRC}(\lambda)\le\lambda$: binary output entropy is at most mean length $\lambda n$ plus $O(\log n)$ length entropy. Divide by $n$.
\end{proof}

For repeat intensity $\lambda$ and source $P$, put $M=\lambda\overline L$ and $D_0=\sum_\ell P_\ell e^{-\lambda\ell}$. Divide the jigsaw identity by $\lambda$, insert $a\le a_P$, and expand $H(T\mid K)$:
\[
 \frac{I_P}{\lambda}\ge
 \frac{H(P)}{M}-\frac{1-D_0}{(1+D_0)M}
 \{H(T)+\mathbb EH(K\mid T)-H(K)\}+\frac{a}{\lambda},
 \qquad M=\lambda\overline L.
\]
Evaluate the type and output-count entropies with the directed sums and tails above. A lower endpoint for $I_P/\lambda$, multiplied by $1-d$, then lower-bounds $C(d)$ for every $d$.

\subsection{Input-defined blocks and disjoint alignment layers}
To strengthen this repeat-channel rate, we sum start entropies from disjoint classes of input blocks. Their addition requires one common partition, defined before observing the output.

\paragraph{Choose a partition determined by the input.}
Choose anchor lengths $\mathcal A\subseteq\operatorname{supp}P$, with $p_A=P(\mathcal A)>0$, and split after every run in $\mathcal A$. These are run boundaries, unrelated to deletion-parameter point bounds.

A block has lengths $\ell=(\ell_1,\ldots,\ell_k)$, the first $k-1$ outside $\mathcal A$ and the last inside, with probability $\prod_iP_{\ell_i}$. Its run count is geometric of mean $1/p_A$. The stopped-sum identity gives mean raw length $\overline L/p_A$, since whether run $i$ is needed is determined before drawing its length.

\emph{Example.} With equiprobable lengths one and two, split after each two. For example, $(2)$, $(1,2)$, and $(1,1,2)$ have probabilities $1/2$, $1/4$, and $1/8$. The mean raw block length is three. Selecting the first two classes keeps masses $1/2,1/4$, not their renormalizations by $3/4$.

\paragraph{Adjust the previous-bit law to this partition.}
The preceding run now has an anchor length. Let $z_\ell=e^{-\lambda\ell}$, $D_0=\sum_\ell P_\ell z_\ell$, and $z_A$ be its average disappearance probability. The previous output bit equals or opposes the new block's first bit with probabilities
\[
 z_A=\mathbb E[z_\ell\mid\ell\in\mathcal A],\qquad
 \pi_{\rm same}=\frac{z_A}{1+D_0},\qquad \pi_{\rm opp}=1-\pi_{\rm same}.
\]
For equal bits, the preceding anchor must vanish, of probability $z_A$. Earlier lengths still have law $P$; summing the required parity of vanished runs gives $z_A(1-D_0)\sum_{j\ge0}D_0^{2j}=z_A/(1+D_0)$. The opposite case is the complement. The past is independent of the fresh block; a uniform binary phase removes parity bias.

\begin{lemma}[Start entropy from one input-defined partition]\label{lem:renewal-partition-credit}
Fix independent finite-support run lengths with law $P$ and mean $\overline L$, a Poisson-repeat intensity $\lambda>0$, and a set $\mathcal A$ of lengths with $p_A=P(\mathcal A)>0$. Split the input after each $\mathcal A$-run. For a block with lengths $\ell=(\ell_1,\ldots,\ell_k)$, let $h_\ell(c)$ be its conditional global-start entropy, with first input bit zero and preceding output bit $c$. With the partition-specific probabilities $\pi_{\rm opp},\pi_{\rm same}$ derived above,
\begin{equation}
 a_P\ge\frac{p_A}{\overline L}\sum_{\ell\ \rm anchor\ block}
 \left(\prod_iP_{\ell_i}\right)
 [\pi_{\rm opp}h_\ell(1)+\pi_{\rm same}h_\ell(0)].
 \label{eq:anchor-credit}
\end{equation}
\end{lemma}
\begin{proof}
The input already determines the partition; it is chosen without observing deletion or repeat counts. Reveal its block-output lengths, as in \cref{lem:renewal-fixed-block}. This decreases start entropy and separates the block observations. The preceding output bit is independent of the fresh block and has the partition-specific mixture computed above. A block has its original probability $\prod_iP_{\ell_i}$ and mean raw length $\overline L/p_A$, so averaging its entropy and dividing by that mean gives the display. Geometric block counts and bounded input-run lengths give the same vanishing boundary error as in the finite-block proof. Each summand is nonnegative. Consequently, disjoint selected classes of this one partition may be added; discarded classes receive zero contribution.
\end{proof}
Division by $\lambda$ replaces the block-density factor $p_A/\overline L$ by $p_A/M$, where $M=\lambda\overline L$. We use this normalized convention in the repeat-channel computations below.

\paragraph{Select the source and two disjoint classes of blocks.}
The intensity $\lambda=19/100$ has a direct precedent: \citet[Section~4.2]{RC23} use it in their Poisson-repeat lower-bound search. The hash-bound source payload identified by \path{evidence/parent/plans/prc.json} records its origin as probability optimization of the input-defined-partition bound. The relevant objective is a finite lower estimate of $I_P^{\rm PRC}(\lambda)/\lambda$: the normalized run-group rate plus selected start-entropy contributions for the same source. The payload preserves the resulting exact probabilities and labels its floating score \texttt{discovery\_coefficient\_only}; that score is not a proof premise.

The preserved payload does not record an optimization trajectory or establish why the seven support points were selected. We therefore treat the support and intensity as an inherited fixed proposal and certify its exact probabilities. We do not claim a global optimization over intensities or supports, or a theorem selecting these particular lengths. The mathematical guarantee comes from evaluating this specified law with complete entropy-tail bounds.

The certificate uses $\lambda=19/100$ and the exact support and anchor set
\[
 \{12,36,71,120,187,277,397\},\qquad
 \mathcal A=\{36,71,120,187,277,397\}.
\]
Only length 12 is nonanchor. Write $P_\ell=n_\ell/2^{28}$ with numerators
\begin{center}\small
\begin{tabular}{@{}rrrrrrr@{}}\toprule
$\ell=12$ & $36$ & $71$ & $120$ & $187$ & $277$ & $397$\\\midrule
137637854 & 68537725 & 38773509 & 16606658 & 5499720 & 1215838 & 164152\\\bottomrule
\end{tabular}
\end{center}
The numerators sum to $2^{28}$. The two entropy calculations retain $3\le k\le8$ and $9\le k\le16$ runs, respectively. These are disjoint terms of \eqref{eq:anchor-credit} and may be added. The earlier fixed four-run bound overlaps them and is not added.

\subsection{Counting long-block alignments by histograms}
For $\lambda=19/100$, consider the second block class: $9\le k\le16$ input runs with lengths $(a,\ldots,a,L)$, $a=12$, and terminal anchor $L\in\mathcal A$. We sum possible actual start sets for each received word, combining output-count orderings only when their entropy contributions agree.

Set the first input bit to zero and retain preceding output bit one. Consider outputs with $1\le m\le k-2$ runs, counts $u_1,\ldots,u_{m-1}\in\{1,\ldots,8\}$ and $u_m\in\{1,\ldots,64\}$. Here $m$ counts output runs and $k$ input runs. Omitted cases have nonnegative contributions; retained masses stay unchanged.

\paragraph{Step 1: assign a mass to each actual start path.}
Use zero-based input-run indices. A zero-first output starts at an even index $i_0$; alternating output bits require odd gaps $1+2r_j$, $r_j\ge0$, between starts. Put $t=\lambda a$. A nonterminal gap $1+2r$ contains $r$ vanished opposite-bit runs, of mass $e^{-tr}$, and $r+1$ same-bit runs with total Poisson mean $(r+1)t$.

The proposed first run must contribute positively. Given total count $u$, its zero-count probability is the ratio $e^{-(r+1)t}(rt)^u/u!$ to $e^{-(r+1)t}((r+1)t)^u/u!$, namely $(r/(r+1))^u$. Exclude this event and multiply the factors:
\begin{equation}
 f_r(u)=e^{-tr}\Pr\{\operatorname{Pois}((r+1)t)=u\}
 \left[1-\left(\frac r{r+1}\right)^u\right].
 \label{eq:histogram-edge}
\end{equation}
For the terminal output run starting at $i$, let $S$ and $O$ be the total same-bit and opposite-bit raw lengths from $i$ to the block end. Opposite bits must vanish, of mass $e^{-\lambda O}$; the same-bit count has mean $\lambda S$. Given count $u_m$, the proposed first run vanishes with probability $((S-\ell_i)/S)^{u_m}$, by the same Poisson ratio. The terminal factor is
\[
 e^{-\lambda O}\Pr\{\operatorname{Pois}(\lambda S)=u_m\}
 \left[1-\left(\frac{S-\ell_i}{S}\right)^{u_m}\right].
\]
Runs before $i_0$ must vanish, giving $e^{-ti_0}$. Prefix, edge, and terminal factors describe disjoint events indexed by actual start sets, after summing indistinguishable count histories.

\paragraph{Step 2: evaluate equal output-count orderings together.}
All nonterminal lengths equal $a$. Permuting the first $m-1$ counts together with their gaps preserves the edge-product mass, total gap, endpoints, and boundary factors. It bijects valid start paths: positive gaps keep every reordered partial path between those endpoints.

Therefore entropy depends only on the histogram $(n_1,\ldots,n_8)$ of nonterminal output counts. Its number of distinct orderings is $(m-1)!/\prod_un_u!$: label all positions, then divide by permutations of equal counts. Evaluate once and multiply by this integer.

\emph{Example.} Counts $112,121,211$ have histogram $n_1=2,n_2=1$ and multiplier $3!/2!=3$. The terminal count is separate. Distinct actual start sets remain separate posterior alternatives.

Define $G_u(z)=\sum_{r\ge0}f_r(u)z^r$, with formal variable $z$. In a product of $m-1$ edge factors, exponent $j=\sum r_i$ gives final start $i_0+(m-1)+2j$. It cannot exceed $k-1$, so coefficients above $\lfloor(k-m)/2\rfloor$ cannot contribute. For each coefficient, sum the terminal factor over allowed even $i_0$.

Histograms satisfy $n_1+\cdots+n_8=m-1$. Arranging $m-1$ stars and seven separators gives $\binom{m+6}{7}$ possibilities, including one when $m=1$. The checker enumerates that number and applies each multinomial weight.

\paragraph{Step 3: propagate mass and entropy without cancellation.}
Store each path collection as $(p,E)$, where $p=\sum_i p_i$ and $E=p\log_2p-\sum_i p_i\log_2p_i$. For the second collection, write $q=\sum_jq_j$ and $F=q\log_2q-\sum_jq_j\log_2q_j$. Independent products have masses $p_iq_j$. Expanding their logarithms gives
$\sum_{i,j}p_iq_j\log_2(p_iq_j)=q\sum_i p_i\log_2p_i+p\sum_jq_j\log_2q_j$.
Subtract from $pq\log_2(pq)$ to obtain
\[
 (p,E)\otimes(q,F)=(pq,qE+pF).
\]

For disjoint collections, entropy within the collections contributes $E+F$. Uncertainty about which collection occurred contributes $\Phi(p,q)=(p+q)\log_2(p+q)-p\log_2p-q\log_2q$, with $0\log_20=0$. Thus
\[
 (p,E)\oplus(q,F)=(p+q,E+F+\Phi(p,q)).
\]
All pair entries are nonnegative. The product formula is increasing in each entry, as is the union formula: $\partial\Phi(p,q)/\partial p=\log_2(1+q/p)\ge0$, with the same argument for $q$ and continuity at zero. Lower endpoints therefore propagate downward through both operations.

Polynomial convolution consists of disjoint sums of independent products, so these operations lower-bound the selected path entropy. Restore source mass $P(a)^{k-1}P(L)$ and multiply by $(p_A/M)\pi_{\rm opp}$ from \eqref{eq:anchor-credit}. Here $p_A=P(\mathcal A)$, $M=\lambda\overline L$, and $\pi_{\rm opp}$ is the preceding-opposite-bit probability: these factors supply block density and repeat normalization. Omitted terms receive no renormalization.

The first evaluator also allows the first positive bit to equal the preceding bit. It then merges paths differing only in that unmarked first start before evaluating entropy, as required by the definition of $B^c$.

\paragraph{The finite calculation and its normalization.}
For the repeat-channel baseline, the replay uses merge cutoff $q=5$, output-count cutoff $K_0=192$, and 224-bit Arb arithmetic. The subtracted conditional count entropies use separate adaptive cutoffs: each is enlarged until its geometric entropy-tail test passes. Thus $192$ is the cutoff for the positive output entropy $H(K)$, not a universal cap on every Poisson count. The short-block calculation retains at most six output runs, nonterminal output counts through eight, and a final count through 64; the long-block calculation uses the count ranges stated above and 256-bit Arb initialization. These choices limit computation, while upper tail charges or discarded nonnegative contributions preserve validity.

The repeat-channel calculation freezes $\lambda$, the rational law $P$, the input-defined block partition, the output-count cutoffs, and the two disjoint ranges of block lengths. It then performs three computations: the finite baseline sums and their upper tails; the short-block start-set sums; and the long-block histogram calculation just derived. Start-set masses and their entropies are bounded downward. The preceding-bit probabilities and the block density must use the same partition; replacing them by the fixed-block mixture would change the quantity being bounded.

Keep all three computational outputs in the normalized units used above. Let $\alpha_{\rm base}=J_P^-/\lambda$, and let $\alpha_{\rm short},\alpha_{\rm long}$ be the downward start-entropy sums using block-density factor $p_A/M$, for $3\le k\le8$ and $9\le k\le16$, respectively. These factors already include division by $\lambda$. An upper tail is needed for every truncated subtracted baseline entropy; omitted start-set contributions may be discarded because they are positive. Recorded component costs appear in \cref{sec:reproduction}.

\begin{corollary}[Assembly of the all-parameter renewal lower bound]\label{cor:renewal-universal-assembly}
For the same exact source, intensity, and input-defined partition, let $\alpha_{\rm base}$ be the certified normalized baseline, and let $\alpha_{\rm short},\alpha_{\rm long}$ be certified normalized contributions from the two disjoint block classes. If
\[
 \alpha_\star\le\alpha_{\rm base}+\alpha_{\rm short}+\alpha_{\rm long},
\]
then $C(d)\ge(1-d)\alpha_\star$ for every $0\le d\le1$. The stated source and finite checks certify
\[
 \alpha_\star=\frac{125150570101}{1000000000000}=0.125150570101.
\]
\end{corollary}
\begin{proof}
\Cref{lem:renewal-finite-baseline} bounds the run-group term. \Cref{lem:renewal-partition-credit} permits the two start-entropy contributions to be added because their block-length ranges are disjoint terms of one input-defined partition. Dividing \cref{prop:renewal-identity} by $\lambda$ therefore bounds $I_P^{\rm PRC}(\lambda)/\lambda$ below by the displayed sum. The source rate is at most repeat-channel capacity; \cref{lem:poisson-transfer} converts its normalized lower bound to $(1-d)\alpha_\star$.
\end{proof}
At each deletion probability, the final lower bound takes the maximum of this complete bound and the separate finite-state and ordinary-renewal rates.

%% file: sections/roadmap_renewal.tex
\subsection{Proof plan: one run-length source, two entropy terms}\label{sec:roadmap-renewal}
This is an alternative achievable-rate calculation for an input whose successive complete run lengths are independent. It does not add a correction to a different finite-state source. All terms below use the \emph{same} chosen run-length law $P$.

If the receiver were told which input runs begin each output run, the remaining ambiguity would separate into independent groups. The classical jigsaw rate charges for this extra description \citep{DM07}. Some of the supplied group boundaries remain uncertain even after the input and output are known; recovering that excess charge improves the lower bound \citep{KD10,Chen26}.

\begin{enumerate}
\item \Cref{prop:renewal-identity} gives $I_P=J_P+a_P$: the run-group rate $J_P$ plus conditional uncertainty $a_P$ about actual output-run starts. Both are measured per input bit.
\item \Cref{lem:renewal-finite-baseline} gives a finite lower estimate $J_P^-$. Its positive entropy term may omit nonnegative contributions; every omitted contribution to a subtracted entropy needs an upper tail bound.
\item \Cref{lem:renewal-fixed-block} lower-bounds $a_P$ by a finite block calculation. \Cref{lem:renewal-posterior-masses} specifies what is counted and which events must be merged before entropy is evaluated.
\end{enumerate}
For ordinary deletion, the resulting rate is \eqref{eq:renewal-ordinary-final}. A second calculation uses Poisson repeats: \cref{lem:poisson-transfer} converts its rate to a deletion-channel bound valid for every $d$. The repeat calculation uses one input-defined partition (\cref{lem:renewal-partition-credit}), so its additional block contributions are disjoint. \Cref{cor:renewal-universal-assembly} states their complete normalized sum and the transfer to all $d$. These two complete rates and the rates of \cref{sec:source} are alternatives; the final lower bound takes their maximum.

\Cref{fig:lower-proof-routes} summarizes how these steps produce the final finite test.

\input{figures/lower_proof_routes}

%% file: figures/lower_proof_routes.tex
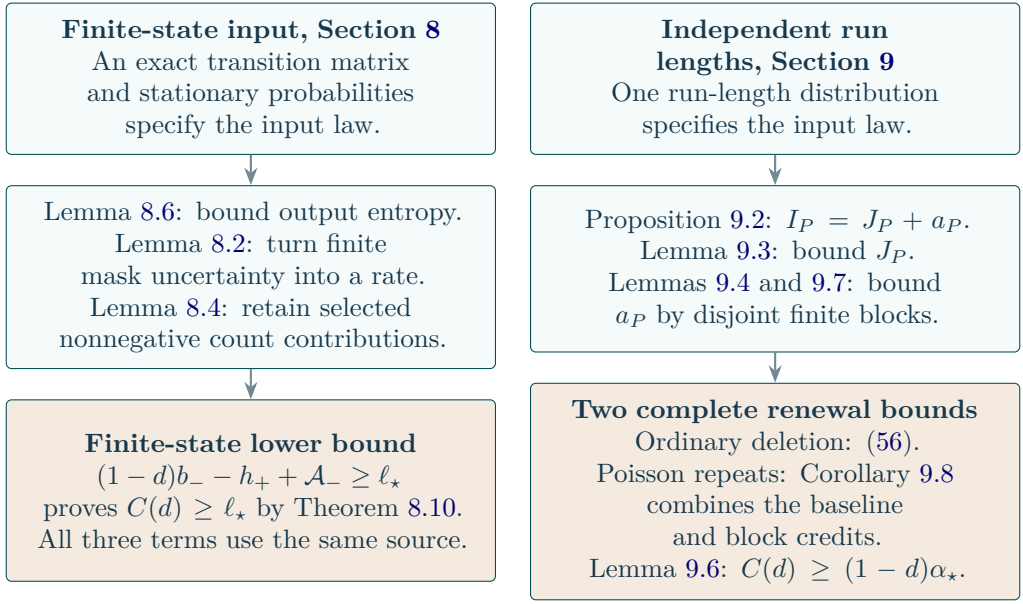
\begin{figure}[tb]\centering
\begin{tikzpicture}[>=Stealth,text=ink,
 box/.style={draw=teal!60!black,rounded corners=2pt,align=center,text width=6.05cm,inner sep=6pt,font=\small\hyphenpenalty10000},arr/.style={->,thick,draw=ink!60}]
\node[box,fill=teal!4,minimum height=1.55cm] (s) {\textbf{Finite-state input, \cref{sec:source}}\\
An exact transition matrix and stationary probabilities specify the input law.};
\node[box,fill=teal!4,minimum height=1.55cm,right=.45cm of s] (r) {\textbf{Independent run lengths, \cref{sec:renewal}}\\
One run-length distribution specifies the input law.};
\node[box,fill=teal!4,minimum height=2.2cm,below=.4cm of s] (sc) {\Cref{lem:birch}: bound output entropy.\\
\Cref{lem:alignment-genie}: turn finite mask uncertainty into a rate.\\
\Cref{lem:positive-cell}: retain selected nonnegative count contributions.};
\node[box,fill=teal!4,minimum height=2.2cm,below=.4cm of r] (rc) {\Cref{prop:renewal-identity}: $I_P=J_P+a_P$.\\
\Cref{lem:renewal-finite-baseline}: bound $J_P$.\\
\Cref{lem:renewal-fixed-block,lem:renewal-partition-credit}: bound $a_P$ by disjoint finite blocks.};
\node[box,fill=ochre!14,minimum height=2.4cm,below=.4cm of sc] (st) {\textbf{Finite-state lower bound}\\
\mbox{$(1-d)b_- -h_+ +\mathcal A_-\ge\ell_\star$}\\
proves $C(d)\ge\ell_\star$ by \cref{thm:lower}.\\
All three terms use the same source.};
\node[box,fill=ochre!14,minimum height=2.4cm,below=.4cm of rc] (rt) {\textbf{Two complete renewal bounds}\\
Ordinary deletion: \eqref{eq:renewal-ordinary-final}.\\
Poisson repeats: \cref{cor:renewal-universal-assembly} combines the baseline and block credits.\\
\Cref{lem:poisson-transfer}: $C(d)\ge(1-d)\alpha_\star$.};
\draw[arr](s)--(sc);\draw[arr](r)--(rc);\draw[arr](sc)--(st);\draw[arr](rc)--(rt);
\end{tikzpicture}
\caption{How the lower bounds are computed. Each column evaluates one specified input process. For the finite-state source, $b_-$ bounds output entropy per survivor from below, $h_+$ bounds deletion entropy $h_2(d)$ from above, and $\mathcal A_-$ bounds residual mask entropy per input bit from below. For independent run lengths, $J_P$ is the run-group rate and $a_P$ is the conditional uncertainty about output-run starts. Ordinary deletion and Poisson repeats give separate completed bounds; the repeat calculation is normalized by intensity once before transfer. The final lower bound is the maximum of completed source rates.}\label{fig:lower-proof-routes}
\end{figure}

%% file: sections/chapter_assembly.tex
\part{Uniform approximation and numerical verification}\label{part:assembly}
The preceding parts give two kinds of sufficient conditions. A converse condition bounds the information rate of every input distribution from above. A source calculation bounds the information rate of one specified input process from below. We now evaluate these conditions and combine their conclusions into an approximation valid for every deletion probability $d\in[0,1]$. This requires both rigorous numerical endpoints and a proof that no parameter values between the evaluated points have been missed.

\Cref{sec:transport} proves the comparison between deletion probabilities and explains the exact interval check. \Cref{sec:small-assembly} records the supplementary combination of specialized bounds near zero, after both its upper and lower ingredients are available. \Cref{sec:arithmetic} derives the rounding and truncation rules used in the finite evaluations, and \cref{sec:reproduction} records the completed checks. The logical order is important. First, an analytic result states that specified inequalities imply a capacity bound. Next, a checker establishes every required inequality for a fixed numerical object, such as a predictor or an input law. Finally, the parameter comparison extends those capacity bounds to intervals. The numerical search that selected the objects is not a premise of these implications.

After these steps, the only remaining calculation is the error bound for a midpoint. Once $L(d)\le C(d)\le U(d)$ holds throughout the parameter interval, the midpoint has error at most $(U(d)-L(d))/2$. The computation verifies this upper bound on error, without needing to know the location of the true capacity inside the enclosure.

\input{sections/roadmap_numerics}

%% file: sections/roadmap_numerics.tex
\section{Combining the bounds and verifying their numerical premises}\label{sec:roadmap-numerics}
The preceding constructions end at explicit inequalities. Their proofs apply to any parameters satisfying those inequalities. Computation begins by choosing candidate input distributions, auxiliary distributions, truncation lengths, and potential values. It then bounds each required expression by an interval containing its exact value. For an upper test, the upper interval endpoint must pass; for a lower test, the lower endpoint must pass. For example, an expression enclosed in $[0.127,0.129]$ is proved at most $0.13$; an enclosure $[0.127,0.131]$ does not establish that assertion. \Cref{sec:arithmetic,sec:reproduction} specify the arithmetic and recorded checks.

These finite tests establish bounds at chosen deletion probabilities. To cover intermediate probabilities, use the theorem of \citet{RD15}, proved here in \cref{lem:transport}: $C(d)/(1-d)$ is nonincreasing for $d<1$. If $C(a)\le u$ and $C(b)\ge\ell$, then for $a\le d\le b<1$,
\[
 (1-d)\frac{\ell}{1-b}\le C(d)\le(1-d)\frac{u}{1-a}.
\]
Indeed, the normalized capacity at $d$ lies between its values at $a$ and $b$; multiplication by $1-d$ gives the display. The normalized endpoint bounds are constants on this interval. Their difference times $1-d$ is largest at its left endpoint, so one exact endpoint comparison bounds the width throughout the interval. Check that the intervals cover $[0,1]$, include the exact capacities at zero and one, and take the midpoint to obtain \cref{thm:main}.

Different methods are effective in different parameter ranges. Their values are combined by taking the largest completed lower bound and smallest completed upper bound after extension. One must not take output entropy from one source and mask entropy from another to form a lower bound. \Cref{tab:chapter-contributions} and \cref{fig:chapter-roles} identify the actual selections in this certificate; they are not restrictions on where the methods are valid.

\input{data/chapter_contributions}
\begin{figure}[!htbp]\centering
\input{figures/chapter_roles_inline}
\caption{Which construction supplies each bound in the published approximation. At every deletion probability, the upper and lower bounds are selected separately from complete valid candidates. The enlarged near-zero view shows additional lower estimates for independent fair input. Each legend entry points to its derivation; ``run-length upper bound'' is the modified-deletion run-length converse of \cref{sec:small-deletion}. These selected regions depend on the available certificate; they do not delimit the mathematical validity of a method.}\label{fig:chapter-roles}
\end{figure}
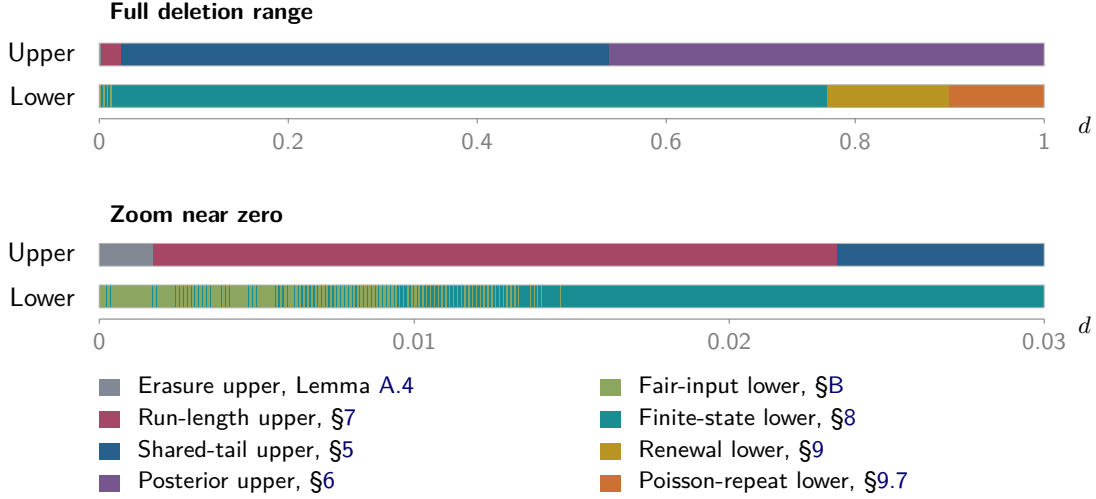

The achieved tolerance is not an optimality claim. Larger windows can improve entropy estimates, but the number of states and rows grows exponentially. More counting terms can improve a lower bound, but only when their probability weights and disjointness are controlled. A different input source may improve one entropy term while worsening the other. Numerical search and stronger analytic inequalities can both help; the stopping point of the present certified calculation does not prove that either has been exhausted. \Cref{sec:limits} separates these remaining obstacles from the numerical result established here.

%% file: data/chapter_contributions.tex
\begin{table}[!htbp]\centering\small
\begin{tabular}{@{}>{\raggedright\arraybackslash}p{3.0cm}>{\raggedright\arraybackslash}p{4.8cm}>{\raggedright\arraybackslash}p{4.0cm}>{\raggedright\arraybackslash}p{2.2cm}@{}}
\toprule
Construction & Contribution to the final enclosure & Where the published cells select it & Derivation\\\midrule
Erasure upper bound & Revealing positions gives $U(d)=1-d$ & $0<d<0.0017$ & \Cref{lem:erasure}\\[4pt]
Run-length upper bound & Handles rare deletions and unbounded input runs & $0.0017<d<3/128$ & \Cref{sec:small-deletion}\\[4pt]
Shared-tail upper bound & Finite window with a common outside survivor word and an omitted-case charge & $3/128<d<0.54$ & \Cref{sec:shared}\\[4pt]
Posterior upper bound & Entropy cancellation covers every outside-survivor distribution & $0.54<d<1$ & \Cref{sec:posterior}\\\midrule
Fair-input lower bounds & Additional lower-bound contributions near zero; not required for the worst-case tolerance & Selected cells below $0.016$, interleaved with finite-state sources & \Cref{sec:fair-input}\\[4pt]
Finite-state sources & A complete information-rate lower bound for each chosen source & Selected cells below $0.77$; all cells from $0.016$ to $0.77$ & \Cref{sec:source}\\[4pt]
Renewal lower bound & Achievable rates from prescribed run lengths & $0.77<d<0.9$ & \Cref{sec:renewal}\\[4pt]
Poisson-repeat lower bound & One achievable-rate coefficient times $1-d$ & $0.9<d<1$ & \Cref{sec:renewal}\\
\bottomrule
\end{tabular}
\caption{Where each construction enters the published enclosure. These are selections made by the certificate, not the domains on which the underlying theorems are valid. At a shared boundary, the larger lower endpoint and smaller upper endpoint are used; $d=0,1$ are exact. The two record families called Markov and stationary-source are grouped here under finite-state sources, and the iid and small-run lower records under fair-input refinements. The exact transition $0.0234375$ is $3/128$. The fair-input refinements still support the reported pointwise values and mean; \cref{sec:family-simplification} explains why the worst-case tolerance can be proved without them.}
\label{tab:chapter-contributions}
\end{table}

%% file: figures/chapter_roles_inline.tex
\definecolor{erasure}{HTML}{808994}\definecolor{runs}{HTML}{A64766}\definecolor{shared}{HTML}{28608C}\definecolor{posterior}{HTML}{77568F}\definecolor{fair}{HTML}{8CA65E}\definecolor{finite}{HTML}{198E92}\definecolor{renewal}{HTML}{B89422}\definecolor{poisson}{HTML}{CC7133}
\begin{tikzpicture}[x=12.5cm,y=1cm,font=\sffamily\small]
\node[anchor=west,font=\sffamily\bfseries\footnotesize] at(0,0.4){Full deletion range};
\node[anchor=east] at(-.015,-0.15){Upper};
\fill[erasure] (0.000000000000,-0.3) rectangle (0.001700000000,0);
\fill[runs] (0.001700000000,-0.3) rectangle (0.023437500000,0);
\fill[shared] (0.023437500000,-0.3) rectangle (0.540000000000,0);
\fill[posterior] (0.540000000000,-0.3) rectangle (1.000000000000,0);
\draw[black!30] (0,-0.3) rectangle (1,0);
\node[anchor=east] at(-.015,-0.7000000000000001){Lower};
\fill[fair] (0.000000000000,-0.8500000000000001) rectangle (0.000225000000,-0.55);
\fill[finite] (0.000225000000,-0.8500000000000001) rectangle (0.000244140625,-0.55);
\fill[fair] (0.000244140625,-0.8500000000000001) rectangle (0.000350000000,-0.55);
\fill[finite] (0.000350000000,-0.8500000000000001) rectangle (0.000366210938,-0.55);
\fill[fair] (0.000366210938,-0.8500000000000001) rectangle (0.000475000000,-0.55);
\fill[finite] (0.000475000000,-0.8500000000000001) rectangle (0.000488281250,-0.55);
\fill[fair] (0.000488281250,-0.8500000000000001) rectangle (0.000600000000,-0.55);
\fill[finite] (0.000600000000,-0.8500000000000001) rectangle (0.000610351562,-0.55);
\fill[fair] (0.000610351562,-0.8500000000000001) rectangle (0.000725000000,-0.55);
\fill[finite] (0.000725000000,-0.8500000000000001) rectangle (0.000732421875,-0.55);
\fill[fair] (0.000732421875,-0.8500000000000001) rectangle (0.000850000000,-0.55);
\fill[finite] (0.000850000000,-0.8500000000000001) rectangle (0.000854492188,-0.55);
\fill[fair] (0.000854492188,-0.8500000000000001) rectangle (0.000975000000,-0.55);
\fill[finite] (0.000975000000,-0.8500000000000001) rectangle (0.000976562500,-0.55);
\fill[fair] (0.000976562500,-0.8500000000000001) rectangle (0.001200000000,-0.55);
\fill[finite] (0.001200000000,-0.8500000000000001) rectangle (0.001220703125,-0.55);
\fill[fair] (0.001220703125,-0.8500000000000001) rectangle (0.001325000000,-0.55);
\fill[finite] (0.001325000000,-0.8500000000000001) rectangle (0.001342773438,-0.55);
\fill[fair] (0.001342773438,-0.8500000000000001) rectangle (0.001450000000,-0.55);
\fill[finite] (0.001450000000,-0.8500000000000001) rectangle (0.001464843750,-0.55);
\fill[fair] (0.001464843750,-0.8500000000000001) rectangle (0.001575000000,-0.55);
\fill[finite] (0.001575000000,-0.8500000000000001) rectangle (0.001586914062,-0.55);
\fill[fair] (0.001586914062,-0.8500000000000001) rectangle (0.001700000000,-0.55);
\fill[finite] (0.001700000000,-0.8500000000000001) rectangle (0.001708984375,-0.55);
\fill[fair] (0.001708984375,-0.8500000000000001) rectangle (0.001825000000,-0.55);
\fill[finite] (0.001825000000,-0.8500000000000001) rectangle (0.001831054688,-0.55);
\fill[fair] (0.001831054688,-0.8500000000000001) rectangle (0.001950000000,-0.55);
\fill[finite] (0.001950000000,-0.8500000000000001) rectangle (0.001953125000,-0.55);
\fill[fair] (0.001953125000,-0.8500000000000001) rectangle (0.002050000000,-0.55);
\fill[finite] (0.002050000000,-0.8500000000000001) rectangle (0.002075195312,-0.55);
\fill[fair] (0.002075195312,-0.8500000000000001) rectangle (0.002175000000,-0.55);
\fill[finite] (0.002175000000,-0.8500000000000001) rectangle (0.002197265625,-0.55);
\fill[fair] (0.002197265625,-0.8500000000000001) rectangle (0.002300000000,-0.55);
\fill[finite] (0.002300000000,-0.8500000000000001) rectangle (0.002319335938,-0.55);
\fill[fair] (0.002319335938,-0.8500000000000001) rectangle (0.002425000000,-0.55);
\fill[finite] (0.002425000000,-0.8500000000000001) rectangle (0.002441406250,-0.55);
\fill[fair] (0.002441406250,-0.8500000000000001) rectangle (0.002550000000,-0.55);
\fill[finite] (0.002550000000,-0.8500000000000001) rectangle (0.002563476562,-0.55);
\fill[fair] (0.002563476562,-0.8500000000000001) rectangle (0.002675000000,-0.55);
\fill[finite] (0.002675000000,-0.8500000000000001) rectangle (0.002685546875,-0.55);
\fill[fair] (0.002685546875,-0.8500000000000001) rectangle (0.002800000000,-0.55);
\fill[finite] (0.002800000000,-0.8500000000000001) rectangle (0.002807617188,-0.55);
\fill[fair] (0.002807617188,-0.8500000000000001) rectangle (0.002925000000,-0.55);
\fill[finite] (0.002925000000,-0.8500000000000001) rectangle (0.002929687500,-0.55);
\fill[fair] (0.002929687500,-0.8500000000000001) rectangle (0.003025000000,-0.55);
\fill[finite] (0.003025000000,-0.8500000000000001) rectangle (0.003051757812,-0.55);
\fill[fair] (0.003051757812,-0.8500000000000001) rectangle (0.003150000000,-0.55);
\fill[finite] (0.003150000000,-0.8500000000000001) rectangle (0.003173828125,-0.55);
\fill[fair] (0.003173828125,-0.8500000000000001) rectangle (0.003275000000,-0.55);
\fill[finite] (0.003275000000,-0.8500000000000001) rectangle (0.003295898438,-0.55);
\fill[fair] (0.003295898438,-0.8500000000000001) rectangle (0.003400000000,-0.55);
\fill[finite] (0.003400000000,-0.8500000000000001) rectangle (0.003417968750,-0.55);
\fill[fair] (0.003417968750,-0.8500000000000001) rectangle (0.003525000000,-0.55);
\fill[finite] (0.003525000000,-0.8500000000000001) rectangle (0.003540039062,-0.55);
\fill[fair] (0.003540039062,-0.8500000000000001) rectangle (0.003650000000,-0.55);
\fill[finite] (0.003650000000,-0.8500000000000001) rectangle (0.003662109375,-0.55);
\fill[fair] (0.003662109375,-0.8500000000000001) rectangle (0.003775000000,-0.55);
\fill[finite] (0.003775000000,-0.8500000000000001) rectangle (0.003784179688,-0.55);
\fill[fair] (0.003784179688,-0.8500000000000001) rectangle (0.003875000000,-0.55);
\fill[finite] (0.003875000000,-0.8500000000000001) rectangle (0.003906250000,-0.55);
\fill[fair] (0.003906250000,-0.8500000000000001) rectangle (0.004000000000,-0.55);
\fill[finite] (0.004000000000,-0.8500000000000001) rectangle (0.004028320312,-0.55);
\fill[fair] (0.004028320312,-0.8500000000000001) rectangle (0.004125000000,-0.55);
\fill[finite] (0.004125000000,-0.8500000000000001) rectangle (0.004150390625,-0.55);
\fill[fair] (0.004150390625,-0.8500000000000001) rectangle (0.004250000000,-0.55);
\fill[finite] (0.004250000000,-0.8500000000000001) rectangle (0.004272460938,-0.55);
\fill[fair] (0.004272460938,-0.8500000000000001) rectangle (0.004375000000,-0.55);
\fill[finite] (0.004375000000,-0.8500000000000001) rectangle (0.004394531250,-0.55);
\fill[fair] (0.004394531250,-0.8500000000000001) rectangle (0.004500000000,-0.55);
\fill[finite] (0.004500000000,-0.8500000000000001) rectangle (0.004516601562,-0.55);
\fill[fair] (0.004516601562,-0.8500000000000001) rectangle (0.004625000000,-0.55);
\fill[finite] (0.004625000000,-0.8500000000000001) rectangle (0.004638671875,-0.55);
\fill[fair] (0.004638671875,-0.8500000000000001) rectangle (0.004725000000,-0.55);
\fill[finite] (0.004725000000,-0.8500000000000001) rectangle (0.004760742188,-0.55);
\fill[fair] (0.004760742188,-0.8500000000000001) rectangle (0.004850000000,-0.55);
\fill[finite] (0.004850000000,-0.8500000000000001) rectangle (0.004882812500,-0.55);
\fill[fair] (0.004882812500,-0.8500000000000001) rectangle (0.004975000000,-0.55);
\fill[finite] (0.004975000000,-0.8500000000000001) rectangle (0.005004882812,-0.55);
\fill[fair] (0.005004882812,-0.8500000000000001) rectangle (0.005100000000,-0.55);
\fill[finite] (0.005100000000,-0.8500000000000001) rectangle (0.005126953125,-0.55);
\fill[fair] (0.005126953125,-0.8500000000000001) rectangle (0.005225000000,-0.55);
\fill[finite] (0.005225000000,-0.8500000000000001) rectangle (0.005249023438,-0.55);
\fill[fair] (0.005249023438,-0.8500000000000001) rectangle (0.005350000000,-0.55);
\fill[finite] (0.005350000000,-0.8500000000000001) rectangle (0.005371093750,-0.55);
\fill[fair] (0.005371093750,-0.8500000000000001) rectangle (0.005475000000,-0.55);
\fill[finite] (0.005475000000,-0.8500000000000001) rectangle (0.005493164062,-0.55);
\fill[fair] (0.005493164062,-0.8500000000000001) rectangle (0.005575000000,-0.55);
\fill[finite] (0.005575000000,-0.8500000000000001) rectangle (0.005615234375,-0.55);
\fill[fair] (0.005615234375,-0.8500000000000001) rectangle (0.005700000000,-0.55);
\fill[finite] (0.005700000000,-0.8500000000000001) rectangle (0.005737304688,-0.55);
\fill[fair] (0.005737304688,-0.8500000000000001) rectangle (0.005825000000,-0.55);
\fill[finite] (0.005825000000,-0.8500000000000001) rectangle (0.005859375000,-0.55);
\fill[fair] (0.005859375000,-0.8500000000000001) rectangle (0.005950000000,-0.55);
\fill[finite] (0.005950000000,-0.8500000000000001) rectangle (0.005981445312,-0.55);
\fill[fair] (0.005981445312,-0.8500000000000001) rectangle (0.006075000000,-0.55);
\fill[finite] (0.006075000000,-0.8500000000000001) rectangle (0.006103515625,-0.55);
\fill[fair] (0.006103515625,-0.8500000000000001) rectangle (0.006200000000,-0.55);
\fill[finite] (0.006200000000,-0.8500000000000001) rectangle (0.006225585938,-0.55);
\fill[fair] (0.006225585938,-0.8500000000000001) rectangle (0.006325000000,-0.55);
\fill[finite] (0.006325000000,-0.8500000000000001) rectangle (0.006347656250,-0.55);
\fill[fair] (0.006347656250,-0.8500000000000001) rectangle (0.006425000000,-0.55);
\fill[finite] (0.006425000000,-0.8500000000000001) rectangle (0.006469726562,-0.55);
\fill[fair] (0.006469726562,-0.8500000000000001) rectangle (0.006550000000,-0.55);
\fill[finite] (0.006550000000,-0.8500000000000001) rectangle (0.006591796875,-0.55);
\fill[fair] (0.006591796875,-0.8500000000000001) rectangle (0.006675000000,-0.55);
\fill[finite] (0.006675000000,-0.8500000000000001) rectangle (0.006713867188,-0.55);
\fill[fair] (0.006713867188,-0.8500000000000001) rectangle (0.006800000000,-0.55);
\fill[finite] (0.006800000000,-0.8500000000000001) rectangle (0.006835937500,-0.55);
\fill[fair] (0.006835937500,-0.8500000000000001) rectangle (0.006925000000,-0.55);
\fill[finite] (0.006925000000,-0.8500000000000001) rectangle (0.006958007812,-0.55);
\fill[fair] (0.006958007812,-0.8500000000000001) rectangle (0.007050000000,-0.55);
\fill[finite] (0.007050000000,-0.8500000000000001) rectangle (0.007080078125,-0.55);
\fill[fair] (0.007080078125,-0.8500000000000001) rectangle (0.007175000000,-0.55);
\fill[finite] (0.007175000000,-0.8500000000000001) rectangle (0.007202148438,-0.55);
\fill[fair] (0.007202148438,-0.8500000000000001) rectangle (0.007275000000,-0.55);
\fill[finite] (0.007275000000,-0.8500000000000001) rectangle (0.007324218750,-0.55);
\fill[fair] (0.007324218750,-0.8500000000000001) rectangle (0.007400000000,-0.55);
\fill[finite] (0.007400000000,-0.8500000000000001) rectangle (0.007446289062,-0.55);
\fill[fair] (0.007446289062,-0.8500000000000001) rectangle (0.007525000000,-0.55);
\fill[finite] (0.007525000000,-0.8500000000000001) rectangle (0.007568359375,-0.55);
\fill[fair] (0.007568359375,-0.8500000000000001) rectangle (0.007650000000,-0.55);
\fill[finite] (0.007650000000,-0.8500000000000001) rectangle (0.007690429688,-0.55);
\fill[fair] (0.007690429688,-0.8500000000000001) rectangle (0.007775000000,-0.55);
\fill[finite] (0.007775000000,-0.8500000000000001) rectangle (0.007812500000,-0.55);
\fill[fair] (0.007812500000,-0.8500000000000001) rectangle (0.007900000000,-0.55);
\fill[finite] (0.007900000000,-0.8500000000000001) rectangle (0.007934570312,-0.55);
\fill[fair] (0.007934570312,-0.8500000000000001) rectangle (0.008025000000,-0.55);
\fill[finite] (0.008025000000,-0.8500000000000001) rectangle (0.008056640625,-0.55);
\fill[fair] (0.008056640625,-0.8500000000000001) rectangle (0.008125000000,-0.55);
\fill[finite] (0.008125000000,-0.8500000000000001) rectangle (0.008178710938,-0.55);
\fill[fair] (0.008178710938,-0.8500000000000001) rectangle (0.008250000000,-0.55);
\fill[finite] (0.008250000000,-0.8500000000000001) rectangle (0.008300781250,-0.55);
\fill[fair] (0.008300781250,-0.8500000000000001) rectangle (0.008375000000,-0.55);
\fill[finite] (0.008375000000,-0.8500000000000001) rectangle (0.008422851562,-0.55);
\fill[fair] (0.008422851562,-0.8500000000000001) rectangle (0.008500000000,-0.55);
\fill[finite] (0.008500000000,-0.8500000000000001) rectangle (0.008544921875,-0.55);
\fill[fair] (0.008544921875,-0.8500000000000001) rectangle (0.008625000000,-0.55);
\fill[finite] (0.008625000000,-0.8500000000000001) rectangle (0.008666992188,-0.55);
\fill[fair] (0.008666992188,-0.8500000000000001) rectangle (0.008750000000,-0.55);
\fill[finite] (0.008750000000,-0.8500000000000001) rectangle (0.008789062500,-0.55);
\fill[fair] (0.008789062500,-0.8500000000000001) rectangle (0.008850000000,-0.55);
\fill[finite] (0.008850000000,-0.8500000000000001) rectangle (0.008911132812,-0.55);
\fill[fair] (0.008911132812,-0.8500000000000001) rectangle (0.008975000000,-0.55);
\fill[finite] (0.008975000000,-0.8500000000000001) rectangle (0.009033203125,-0.55);
\fill[fair] (0.009033203125,-0.8500000000000001) rectangle (0.009100000000,-0.55);
\fill[finite] (0.009100000000,-0.8500000000000001) rectangle (0.009155273438,-0.55);
\fill[fair] (0.009155273438,-0.8500000000000001) rectangle (0.009225000000,-0.55);
\fill[finite] (0.009225000000,-0.8500000000000001) rectangle (0.009277343750,-0.55);
\fill[fair] (0.009277343750,-0.8500000000000001) rectangle (0.009350000000,-0.55);
\fill[finite] (0.009350000000,-0.8500000000000001) rectangle (0.009399414062,-0.55);
\fill[fair] (0.009399414062,-0.8500000000000001) rectangle (0.009475000000,-0.55);
\fill[finite] (0.009475000000,-0.8500000000000001) rectangle (0.009521484375,-0.55);
\fill[fair] (0.009521484375,-0.8500000000000001) rectangle (0.009575000000,-0.55);
\fill[finite] (0.009575000000,-0.8500000000000001) rectangle (0.009643554688,-0.55);
\fill[fair] (0.009643554688,-0.8500000000000001) rectangle (0.009700000000,-0.55);
\fill[finite] (0.009700000000,-0.8500000000000001) rectangle (0.009765625000,-0.55);
\fill[fair] (0.009765625000,-0.8500000000000001) rectangle (0.009825000000,-0.55);
\fill[finite] (0.009825000000,-0.8500000000000001) rectangle (0.009887695312,-0.55);
\fill[fair] (0.009887695312,-0.8500000000000001) rectangle (0.009950000000,-0.55);
\fill[finite] (0.009950000000,-0.8500000000000001) rectangle (0.010009765625,-0.55);
\fill[fair] (0.010009765625,-0.8500000000000001) rectangle (0.010075000000,-0.55);
\fill[finite] (0.010075000000,-0.8500000000000001) rectangle (0.010131835938,-0.55);
\fill[fair] (0.010131835938,-0.8500000000000001) rectangle (0.010200000000,-0.55);
\fill[finite] (0.010200000000,-0.8500000000000001) rectangle (0.010253906250,-0.55);
\fill[fair] (0.010253906250,-0.8500000000000001) rectangle (0.010325000000,-0.55);
\fill[finite] (0.010325000000,-0.8500000000000001) rectangle (0.010375976562,-0.55);
\fill[fair] (0.010375976562,-0.8500000000000001) rectangle (0.010425000000,-0.55);
\fill[finite] (0.010425000000,-0.8500000000000001) rectangle (0.010498046875,-0.55);
\fill[fair] (0.010498046875,-0.8500000000000001) rectangle (0.010550000000,-0.55);
\fill[finite] (0.010550000000,-0.8500000000000001) rectangle (0.010620117188,-0.55);
\fill[fair] (0.010620117188,-0.8500000000000001) rectangle (0.010675000000,-0.55);
\fill[finite] (0.010675000000,-0.8500000000000001) rectangle (0.010742187500,-0.55);
\fill[fair] (0.010742187500,-0.8500000000000001) rectangle (0.010800000000,-0.55);
\fill[finite] (0.010800000000,-0.8500000000000001) rectangle (0.010864257812,-0.55);
\fill[fair] (0.010864257812,-0.8500000000000001) rectangle (0.010925000000,-0.55);
\fill[finite] (0.010925000000,-0.8500000000000001) rectangle (0.010986328125,-0.55);
\fill[fair] (0.010986328125,-0.8500000000000001) rectangle (0.011050000000,-0.55);
\fill[finite] (0.011050000000,-0.8500000000000001) rectangle (0.011108398438,-0.55);
\fill[fair] (0.011108398438,-0.8500000000000001) rectangle (0.011150000000,-0.55);
\fill[finite] (0.011150000000,-0.8500000000000001) rectangle (0.011230468750,-0.55);
\fill[fair] (0.011230468750,-0.8500000000000001) rectangle (0.011275000000,-0.55);
\fill[finite] (0.011275000000,-0.8500000000000001) rectangle (0.011352539062,-0.55);
\fill[fair] (0.011352539062,-0.8500000000000001) rectangle (0.011400000000,-0.55);
\fill[finite] (0.011400000000,-0.8500000000000001) rectangle (0.011474609375,-0.55);
\fill[fair] (0.011474609375,-0.8500000000000001) rectangle (0.011525000000,-0.55);
\fill[finite] (0.011525000000,-0.8500000000000001) rectangle (0.011596679688,-0.55);
\fill[fair] (0.011596679688,-0.8500000000000001) rectangle (0.011650000000,-0.55);
\fill[finite] (0.011650000000,-0.8500000000000001) rectangle (0.011718750000,-0.55);
\fill[fair] (0.011718750000,-0.8500000000000001) rectangle (0.011775000000,-0.55);
\fill[finite] (0.011775000000,-0.8500000000000001) rectangle (0.011840820312,-0.55);
\fill[fair] (0.011840820312,-0.8500000000000001) rectangle (0.011875000000,-0.55);
\fill[finite] (0.011875000000,-0.8500000000000001) rectangle (0.011962890625,-0.55);
\fill[fair] (0.011962890625,-0.8500000000000001) rectangle (0.012000000000,-0.55);
\fill[finite] (0.012000000000,-0.8500000000000001) rectangle (0.012084960938,-0.55);
\fill[fair] (0.012084960938,-0.8500000000000001) rectangle (0.012125000000,-0.55);
\fill[finite] (0.012125000000,-0.8500000000000001) rectangle (0.012207031250,-0.55);
\fill[fair] (0.012207031250,-0.8500000000000001) rectangle (0.012250000000,-0.55);
\fill[finite] (0.012250000000,-0.8500000000000001) rectangle (0.012329101562,-0.55);
\fill[fair] (0.012329101562,-0.8500000000000001) rectangle (0.012375000000,-0.55);
\fill[finite] (0.012375000000,-0.8500000000000001) rectangle (0.012451171875,-0.55);
\fill[fair] (0.012451171875,-0.8500000000000001) rectangle (0.012500000000,-0.55);
\fill[finite] (0.012500000000,-0.8500000000000001) rectangle (0.012573242188,-0.55);
\fill[fair] (0.012573242188,-0.8500000000000001) rectangle (0.012600000000,-0.55);
\fill[finite] (0.012600000000,-0.8500000000000001) rectangle (0.012695312500,-0.55);
\fill[fair] (0.012695312500,-0.8500000000000001) rectangle (0.012725000000,-0.55);
\fill[finite] (0.012725000000,-0.8500000000000001) rectangle (0.012817382812,-0.55);
\fill[fair] (0.012817382812,-0.8500000000000001) rectangle (0.012850000000,-0.55);
\fill[finite] (0.012850000000,-0.8500000000000001) rectangle (0.012939453125,-0.55);
\fill[fair] (0.012939453125,-0.8500000000000001) rectangle (0.012975000000,-0.55);
\fill[finite] (0.012975000000,-0.8500000000000001) rectangle (0.013061523438,-0.55);
\fill[fair] (0.013061523438,-0.8500000000000001) rectangle (0.013100000000,-0.55);
\fill[finite] (0.013100000000,-0.8500000000000001) rectangle (0.013183593750,-0.55);
\fill[fair] (0.013183593750,-0.8500000000000001) rectangle (0.013225000000,-0.55);
\fill[finite] (0.013225000000,-0.8500000000000001) rectangle (0.013305664062,-0.55);
\fill[fair] (0.013305664062,-0.8500000000000001) rectangle (0.013325000000,-0.55);
\fill[finite] (0.013325000000,-0.8500000000000001) rectangle (0.013427734375,-0.55);
\fill[fair] (0.013427734375,-0.8500000000000001) rectangle (0.013450000000,-0.55);
\fill[finite] (0.013450000000,-0.8500000000000001) rectangle (0.013549804688,-0.55);
\fill[fair] (0.013549804688,-0.8500000000000001) rectangle (0.013575000000,-0.55);
\fill[finite] (0.013575000000,-0.8500000000000001) rectangle (0.013671875000,-0.55);
\fill[fair] (0.013671875000,-0.8500000000000001) rectangle (0.013700000000,-0.55);
\fill[finite] (0.013700000000,-0.8500000000000001) rectangle (0.013793945312,-0.55);
\fill[fair] (0.013793945312,-0.8500000000000001) rectangle (0.013825000000,-0.55);
\fill[finite] (0.013825000000,-0.8500000000000001) rectangle (0.013916015625,-0.55);
\fill[fair] (0.013916015625,-0.8500000000000001) rectangle (0.013925000000,-0.55);
\fill[finite] (0.013925000000,-0.8500000000000001) rectangle (0.014038085938,-0.55);
\fill[fair] (0.014038085938,-0.8500000000000001) rectangle (0.014050000000,-0.55);
\fill[finite] (0.014050000000,-0.8500000000000001) rectangle (0.014160156250,-0.55);
\fill[fair] (0.014160156250,-0.8500000000000001) rectangle (0.014175000000,-0.55);
\fill[finite] (0.014175000000,-0.8500000000000001) rectangle (0.014282226562,-0.55);
\fill[fair] (0.014282226562,-0.8500000000000001) rectangle (0.014300000000,-0.55);
\fill[finite] (0.014300000000,-0.8500000000000001) rectangle (0.014404296875,-0.55);
\fill[fair] (0.014404296875,-0.8500000000000001) rectangle (0.014425000000,-0.55);
\fill[finite] (0.014425000000,-0.8500000000000001) rectangle (0.014526367188,-0.55);
\fill[fair] (0.014526367188,-0.8500000000000001) rectangle (0.014550000000,-0.55);
\fill[finite] (0.014550000000,-0.8500000000000001) rectangle (0.014648437500,-0.55);
\fill[fair] (0.014648437500,-0.8500000000000001) rectangle (0.014650000000,-0.55);
\fill[finite] (0.014650000000,-0.8500000000000001) rectangle (0.014770507812,-0.55);
\fill[fair] (0.014770507812,-0.8500000000000001) rectangle (0.014775000000,-0.55);
\fill[finite] (0.014775000000,-0.8500000000000001) rectangle (0.014892578125,-0.55);
\fill[fair] (0.014892578125,-0.8500000000000001) rectangle (0.014900000000,-0.55);
\fill[finite] (0.014900000000,-0.8500000000000001) rectangle (0.015014648438,-0.55);
\fill[fair] (0.015014648438,-0.8500000000000001) rectangle (0.015025000000,-0.55);
\fill[finite] (0.015025000000,-0.8500000000000001) rectangle (0.015136718750,-0.55);
\fill[fair] (0.015136718750,-0.8500000000000001) rectangle (0.015150000000,-0.55);
\fill[finite] (0.015150000000,-0.8500000000000001) rectangle (0.015258789062,-0.55);
\fill[fair] (0.015258789062,-0.8500000000000001) rectangle (0.015275000000,-0.55);
\fill[finite] (0.015275000000,-0.8500000000000001) rectangle (0.015747070312,-0.55);
\fill[fair] (0.015747070312,-0.8500000000000001) rectangle (0.015750000000,-0.55);
\fill[finite] (0.015750000000,-0.8500000000000001) rectangle (0.015869140625,-0.55);
\fill[fair] (0.015869140625,-0.8500000000000001) rectangle (0.015875000000,-0.55);
\fill[finite] (0.015875000000,-0.8500000000000001) rectangle (0.015991210938,-0.55);
\fill[fair] (0.015991210938,-0.8500000000000001) rectangle (0.016000000000,-0.55);
\fill[finite] (0.016000000000,-0.8500000000000001) rectangle (0.770000000000,-0.55);
\fill[renewal] (0.770000000000,-0.8500000000000001) rectangle (0.900000000000,-0.55);
\fill[poisson] (0.900000000000,-0.8500000000000001) rectangle (1.000000000000,-0.55);
\draw[black!30] (0,-0.8500000000000001) rectangle (1,-0.55);
\draw[black!50] (0,-1.0)--(1,-1.0);
\draw[black!50] (0.0,-1.0)--(0.0,-1.06) node[below,font=\sffamily\footnotesize]{0};
\draw[black!50] (0.2,-1.0)--(0.2,-1.06) node[below,font=\sffamily\footnotesize]{0.2};
\draw[black!50] (0.4,-1.0)--(0.4,-1.06) node[below,font=\sffamily\footnotesize]{0.4};
\draw[black!50] (0.6,-1.0)--(0.6,-1.06) node[below,font=\sffamily\footnotesize]{0.6};
\draw[black!50] (0.8,-1.0)--(0.8,-1.06) node[below,font=\sffamily\footnotesize]{0.8};
\draw[black!50] (1.0,-1.0)--(1.0,-1.06) node[below,font=\sffamily\footnotesize]{1};
\node[anchor=west,font=\sffamily\footnotesize] at(1.025,-1.08){$d$};
\node[anchor=west,font=\sffamily\bfseries\footnotesize] at(0,-2.25){Zoom near zero};
\node[anchor=east] at(-.015,-2.8){Upper};
\fill[erasure] (0.000000000000,-2.9499999999999997) rectangle (0.056666666667,-2.65);
\fill[runs] (0.056666666667,-2.9499999999999997) rectangle (0.781250000000,-2.65);
\fill[shared] (0.781250000000,-2.9499999999999997) rectangle (1.000000000000,-2.65);
\draw[black!30] (0,-2.9499999999999997) rectangle (1,-2.65);
\node[anchor=east] at(-.015,-3.35){Lower};
\fill[fair] (0.000000000000,-3.5) rectangle (0.007500000000,-3.2);
\fill[finite] (0.007500000000,-3.5) rectangle (0.008138020833,-3.2);
\fill[fair] (0.008138020833,-3.5) rectangle (0.011666666667,-3.2);
\fill[finite] (0.011666666667,-3.5) rectangle (0.012207031250,-3.2);
\fill[fair] (0.012207031250,-3.5) rectangle (0.015833333333,-3.2);
\fill[finite] (0.015833333333,-3.5) rectangle (0.016276041667,-3.2);
\fill[fair] (0.016276041667,-3.5) rectangle (0.020000000000,-3.2);
\fill[finite] (0.020000000000,-3.5) rectangle (0.020345052083,-3.2);
\fill[fair] (0.020345052083,-3.5) rectangle (0.024166666667,-3.2);
\fill[finite] (0.024166666667,-3.5) rectangle (0.024414062500,-3.2);
\fill[fair] (0.024414062500,-3.5) rectangle (0.028333333333,-3.2);
\fill[finite] (0.028333333333,-3.5) rectangle (0.028483072917,-3.2);
\fill[fair] (0.028483072917,-3.5) rectangle (0.032500000000,-3.2);
\fill[finite] (0.032500000000,-3.5) rectangle (0.032552083333,-3.2);
\fill[fair] (0.032552083333,-3.5) rectangle (0.040000000000,-3.2);
\fill[finite] (0.040000000000,-3.5) rectangle (0.040690104167,-3.2);
\fill[fair] (0.040690104167,-3.5) rectangle (0.044166666667,-3.2);
\fill[finite] (0.044166666667,-3.5) rectangle (0.044759114583,-3.2);
\fill[fair] (0.044759114583,-3.5) rectangle (0.048333333333,-3.2);
\fill[finite] (0.048333333333,-3.5) rectangle (0.048828125000,-3.2);
\fill[fair] (0.048828125000,-3.5) rectangle (0.052500000000,-3.2);
\fill[finite] (0.052500000000,-3.5) rectangle (0.052897135417,-3.2);
\fill[fair] (0.052897135417,-3.5) rectangle (0.056666666667,-3.2);
\fill[finite] (0.056666666667,-3.5) rectangle (0.056966145833,-3.2);
\fill[fair] (0.056966145833,-3.5) rectangle (0.060833333333,-3.2);
\fill[finite] (0.060833333333,-3.5) rectangle (0.061035156250,-3.2);
\fill[fair] (0.061035156250,-3.5) rectangle (0.065000000000,-3.2);
\fill[finite] (0.065000000000,-3.5) rectangle (0.065104166667,-3.2);
\fill[fair] (0.065104166667,-3.5) rectangle (0.068333333333,-3.2);
\fill[finite] (0.068333333333,-3.5) rectangle (0.069173177083,-3.2);
\fill[fair] (0.069173177083,-3.5) rectangle (0.072500000000,-3.2);
\fill[finite] (0.072500000000,-3.5) rectangle (0.073242187500,-3.2);
\fill[fair] (0.073242187500,-3.5) rectangle (0.076666666667,-3.2);
\fill[finite] (0.076666666667,-3.5) rectangle (0.077311197917,-3.2);
\fill[fair] (0.077311197917,-3.5) rectangle (0.080833333333,-3.2);
\fill[finite] (0.080833333333,-3.5) rectangle (0.081380208333,-3.2);
\fill[fair] (0.081380208333,-3.5) rectangle (0.085000000000,-3.2);
\fill[finite] (0.085000000000,-3.5) rectangle (0.085449218750,-3.2);
\fill[fair] (0.085449218750,-3.5) rectangle (0.089166666667,-3.2);
\fill[finite] (0.089166666667,-3.5) rectangle (0.089518229167,-3.2);
\fill[fair] (0.089518229167,-3.5) rectangle (0.093333333333,-3.2);
\fill[finite] (0.093333333333,-3.5) rectangle (0.093587239583,-3.2);
\fill[fair] (0.093587239583,-3.5) rectangle (0.097500000000,-3.2);
\fill[finite] (0.097500000000,-3.5) rectangle (0.097656250000,-3.2);
\fill[fair] (0.097656250000,-3.5) rectangle (0.100833333333,-3.2);
\fill[finite] (0.100833333333,-3.5) rectangle (0.101725260417,-3.2);
\fill[fair] (0.101725260417,-3.5) rectangle (0.105000000000,-3.2);
\fill[finite] (0.105000000000,-3.5) rectangle (0.105794270833,-3.2);
\fill[fair] (0.105794270833,-3.5) rectangle (0.109166666667,-3.2);
\fill[finite] (0.109166666667,-3.5) rectangle (0.109863281250,-3.2);
\fill[fair] (0.109863281250,-3.5) rectangle (0.113333333333,-3.2);
\fill[finite] (0.113333333333,-3.5) rectangle (0.113932291667,-3.2);
\fill[fair] (0.113932291667,-3.5) rectangle (0.117500000000,-3.2);
\fill[finite] (0.117500000000,-3.5) rectangle (0.118001302083,-3.2);
\fill[fair] (0.118001302083,-3.5) rectangle (0.121666666667,-3.2);
\fill[finite] (0.121666666667,-3.5) rectangle (0.122070312500,-3.2);
\fill[fair] (0.122070312500,-3.5) rectangle (0.125833333333,-3.2);
\fill[finite] (0.125833333333,-3.5) rectangle (0.126139322917,-3.2);
\fill[fair] (0.126139322917,-3.5) rectangle (0.129166666667,-3.2);
\fill[finite] (0.129166666667,-3.5) rectangle (0.130208333333,-3.2);
\fill[fair] (0.130208333333,-3.5) rectangle (0.133333333333,-3.2);
\fill[finite] (0.133333333333,-3.5) rectangle (0.134277343750,-3.2);
\fill[fair] (0.134277343750,-3.5) rectangle (0.137500000000,-3.2);
\fill[finite] (0.137500000000,-3.5) rectangle (0.138346354167,-3.2);
\fill[fair] (0.138346354167,-3.5) rectangle (0.141666666667,-3.2);
\fill[finite] (0.141666666667,-3.5) rectangle (0.142415364583,-3.2);
\fill[fair] (0.142415364583,-3.5) rectangle (0.145833333333,-3.2);
\fill[finite] (0.145833333333,-3.5) rectangle (0.146484375000,-3.2);
\fill[fair] (0.146484375000,-3.5) rectangle (0.150000000000,-3.2);
\fill[finite] (0.150000000000,-3.5) rectangle (0.150553385417,-3.2);
\fill[fair] (0.150553385417,-3.5) rectangle (0.154166666667,-3.2);
\fill[finite] (0.154166666667,-3.5) rectangle (0.154622395833,-3.2);
\fill[fair] (0.154622395833,-3.5) rectangle (0.157500000000,-3.2);
\fill[finite] (0.157500000000,-3.5) rectangle (0.158691406250,-3.2);
\fill[fair] (0.158691406250,-3.5) rectangle (0.161666666667,-3.2);
\fill[finite] (0.161666666667,-3.5) rectangle (0.162760416667,-3.2);
\fill[fair] (0.162760416667,-3.5) rectangle (0.165833333333,-3.2);
\fill[finite] (0.165833333333,-3.5) rectangle (0.166829427083,-3.2);
\fill[fair] (0.166829427083,-3.5) rectangle (0.170000000000,-3.2);
\fill[finite] (0.170000000000,-3.5) rectangle (0.170898437500,-3.2);
\fill[fair] (0.170898437500,-3.5) rectangle (0.174166666667,-3.2);
\fill[finite] (0.174166666667,-3.5) rectangle (0.174967447917,-3.2);
\fill[fair] (0.174967447917,-3.5) rectangle (0.178333333333,-3.2);
\fill[finite] (0.178333333333,-3.5) rectangle (0.179036458333,-3.2);
\fill[fair] (0.179036458333,-3.5) rectangle (0.182500000000,-3.2);
\fill[finite] (0.182500000000,-3.5) rectangle (0.183105468750,-3.2);
\fill[fair] (0.183105468750,-3.5) rectangle (0.185833333333,-3.2);
\fill[finite] (0.185833333333,-3.5) rectangle (0.187174479167,-3.2);
\fill[fair] (0.187174479167,-3.5) rectangle (0.190000000000,-3.2);
\fill[finite] (0.190000000000,-3.5) rectangle (0.191243489583,-3.2);
\fill[fair] (0.191243489583,-3.5) rectangle (0.194166666667,-3.2);
\fill[finite] (0.194166666667,-3.5) rectangle (0.195312500000,-3.2);
\fill[fair] (0.195312500000,-3.5) rectangle (0.198333333333,-3.2);
\fill[finite] (0.198333333333,-3.5) rectangle (0.199381510417,-3.2);
\fill[fair] (0.199381510417,-3.5) rectangle (0.202500000000,-3.2);
\fill[finite] (0.202500000000,-3.5) rectangle (0.203450520833,-3.2);
\fill[fair] (0.203450520833,-3.5) rectangle (0.206666666667,-3.2);
\fill[finite] (0.206666666667,-3.5) rectangle (0.207519531250,-3.2);
\fill[fair] (0.207519531250,-3.5) rectangle (0.210833333333,-3.2);
\fill[finite] (0.210833333333,-3.5) rectangle (0.211588541667,-3.2);
\fill[fair] (0.211588541667,-3.5) rectangle (0.214166666667,-3.2);
\fill[finite] (0.214166666667,-3.5) rectangle (0.215657552083,-3.2);
\fill[fair] (0.215657552083,-3.5) rectangle (0.218333333333,-3.2);
\fill[finite] (0.218333333333,-3.5) rectangle (0.219726562500,-3.2);
\fill[fair] (0.219726562500,-3.5) rectangle (0.222500000000,-3.2);
\fill[finite] (0.222500000000,-3.5) rectangle (0.223795572917,-3.2);
\fill[fair] (0.223795572917,-3.5) rectangle (0.226666666667,-3.2);
\fill[finite] (0.226666666667,-3.5) rectangle (0.227864583333,-3.2);
\fill[fair] (0.227864583333,-3.5) rectangle (0.230833333333,-3.2);
\fill[finite] (0.230833333333,-3.5) rectangle (0.231933593750,-3.2);
\fill[fair] (0.231933593750,-3.5) rectangle (0.235000000000,-3.2);
\fill[finite] (0.235000000000,-3.5) rectangle (0.236002604167,-3.2);
\fill[fair] (0.236002604167,-3.5) rectangle (0.239166666667,-3.2);
\fill[finite] (0.239166666667,-3.5) rectangle (0.240071614583,-3.2);
\fill[fair] (0.240071614583,-3.5) rectangle (0.242500000000,-3.2);
\fill[finite] (0.242500000000,-3.5) rectangle (0.244140625000,-3.2);
\fill[fair] (0.244140625000,-3.5) rectangle (0.246666666667,-3.2);
\fill[finite] (0.246666666667,-3.5) rectangle (0.248209635417,-3.2);
\fill[fair] (0.248209635417,-3.5) rectangle (0.250833333333,-3.2);
\fill[finite] (0.250833333333,-3.5) rectangle (0.252278645833,-3.2);
\fill[fair] (0.252278645833,-3.5) rectangle (0.255000000000,-3.2);
\fill[finite] (0.255000000000,-3.5) rectangle (0.256347656250,-3.2);
\fill[fair] (0.256347656250,-3.5) rectangle (0.259166666667,-3.2);
\fill[finite] (0.259166666667,-3.5) rectangle (0.260416666667,-3.2);
\fill[fair] (0.260416666667,-3.5) rectangle (0.263333333333,-3.2);
\fill[finite] (0.263333333333,-3.5) rectangle (0.264485677083,-3.2);
\fill[fair] (0.264485677083,-3.5) rectangle (0.267500000000,-3.2);
\fill[finite] (0.267500000000,-3.5) rectangle (0.268554687500,-3.2);
\fill[fair] (0.268554687500,-3.5) rectangle (0.270833333333,-3.2);
\fill[finite] (0.270833333333,-3.5) rectangle (0.272623697917,-3.2);
\fill[fair] (0.272623697917,-3.5) rectangle (0.275000000000,-3.2);
\fill[finite] (0.275000000000,-3.5) rectangle (0.276692708333,-3.2);
\fill[fair] (0.276692708333,-3.5) rectangle (0.279166666667,-3.2);
\fill[finite] (0.279166666667,-3.5) rectangle (0.280761718750,-3.2);
\fill[fair] (0.280761718750,-3.5) rectangle (0.283333333333,-3.2);
\fill[finite] (0.283333333333,-3.5) rectangle (0.284830729167,-3.2);
\fill[fair] (0.284830729167,-3.5) rectangle (0.287500000000,-3.2);
\fill[finite] (0.287500000000,-3.5) rectangle (0.288899739583,-3.2);
\fill[fair] (0.288899739583,-3.5) rectangle (0.291666666667,-3.2);
\fill[finite] (0.291666666667,-3.5) rectangle (0.292968750000,-3.2);
\fill[fair] (0.292968750000,-3.5) rectangle (0.295000000000,-3.2);
\fill[finite] (0.295000000000,-3.5) rectangle (0.297037760417,-3.2);
\fill[fair] (0.297037760417,-3.5) rectangle (0.299166666667,-3.2);
\fill[finite] (0.299166666667,-3.5) rectangle (0.301106770833,-3.2);
\fill[fair] (0.301106770833,-3.5) rectangle (0.303333333333,-3.2);
\fill[finite] (0.303333333333,-3.5) rectangle (0.305175781250,-3.2);
\fill[fair] (0.305175781250,-3.5) rectangle (0.307500000000,-3.2);
\fill[finite] (0.307500000000,-3.5) rectangle (0.309244791667,-3.2);
\fill[fair] (0.309244791667,-3.5) rectangle (0.311666666667,-3.2);
\fill[finite] (0.311666666667,-3.5) rectangle (0.313313802083,-3.2);
\fill[fair] (0.313313802083,-3.5) rectangle (0.315833333333,-3.2);
\fill[finite] (0.315833333333,-3.5) rectangle (0.317382812500,-3.2);
\fill[fair] (0.317382812500,-3.5) rectangle (0.319166666667,-3.2);
\fill[finite] (0.319166666667,-3.5) rectangle (0.321451822917,-3.2);
\fill[fair] (0.321451822917,-3.5) rectangle (0.323333333333,-3.2);
\fill[finite] (0.323333333333,-3.5) rectangle (0.325520833333,-3.2);
\fill[fair] (0.325520833333,-3.5) rectangle (0.327500000000,-3.2);
\fill[finite] (0.327500000000,-3.5) rectangle (0.329589843750,-3.2);
\fill[fair] (0.329589843750,-3.5) rectangle (0.331666666667,-3.2);
\fill[finite] (0.331666666667,-3.5) rectangle (0.333658854167,-3.2);
\fill[fair] (0.333658854167,-3.5) rectangle (0.335833333333,-3.2);
\fill[finite] (0.335833333333,-3.5) rectangle (0.337727864583,-3.2);
\fill[fair] (0.337727864583,-3.5) rectangle (0.340000000000,-3.2);
\fill[finite] (0.340000000000,-3.5) rectangle (0.341796875000,-3.2);
\fill[fair] (0.341796875000,-3.5) rectangle (0.344166666667,-3.2);
\fill[finite] (0.344166666667,-3.5) rectangle (0.345865885417,-3.2);
\fill[fair] (0.345865885417,-3.5) rectangle (0.347500000000,-3.2);
\fill[finite] (0.347500000000,-3.5) rectangle (0.349934895833,-3.2);
\fill[fair] (0.349934895833,-3.5) rectangle (0.351666666667,-3.2);
\fill[finite] (0.351666666667,-3.5) rectangle (0.354003906250,-3.2);
\fill[fair] (0.354003906250,-3.5) rectangle (0.355833333333,-3.2);
\fill[finite] (0.355833333333,-3.5) rectangle (0.358072916667,-3.2);
\fill[fair] (0.358072916667,-3.5) rectangle (0.360000000000,-3.2);
\fill[finite] (0.360000000000,-3.5) rectangle (0.362141927083,-3.2);
\fill[fair] (0.362141927083,-3.5) rectangle (0.364166666667,-3.2);
\fill[finite] (0.364166666667,-3.5) rectangle (0.366210937500,-3.2);
\fill[fair] (0.366210937500,-3.5) rectangle (0.368333333333,-3.2);
\fill[finite] (0.368333333333,-3.5) rectangle (0.370279947917,-3.2);
\fill[fair] (0.370279947917,-3.5) rectangle (0.371666666667,-3.2);
\fill[finite] (0.371666666667,-3.5) rectangle (0.374348958333,-3.2);
\fill[fair] (0.374348958333,-3.5) rectangle (0.375833333333,-3.2);
\fill[finite] (0.375833333333,-3.5) rectangle (0.378417968750,-3.2);
\fill[fair] (0.378417968750,-3.5) rectangle (0.380000000000,-3.2);
\fill[finite] (0.380000000000,-3.5) rectangle (0.382486979167,-3.2);
\fill[fair] (0.382486979167,-3.5) rectangle (0.384166666667,-3.2);
\fill[finite] (0.384166666667,-3.5) rectangle (0.386555989583,-3.2);
\fill[fair] (0.386555989583,-3.5) rectangle (0.388333333333,-3.2);
\fill[finite] (0.388333333333,-3.5) rectangle (0.390625000000,-3.2);
\fill[fair] (0.390625000000,-3.5) rectangle (0.392500000000,-3.2);
\fill[finite] (0.392500000000,-3.5) rectangle (0.394694010417,-3.2);
\fill[fair] (0.394694010417,-3.5) rectangle (0.395833333333,-3.2);
\fill[finite] (0.395833333333,-3.5) rectangle (0.398763020833,-3.2);
\fill[fair] (0.398763020833,-3.5) rectangle (0.400000000000,-3.2);
\fill[finite] (0.400000000000,-3.5) rectangle (0.402832031250,-3.2);
\fill[fair] (0.402832031250,-3.5) rectangle (0.404166666667,-3.2);
\fill[finite] (0.404166666667,-3.5) rectangle (0.406901041667,-3.2);
\fill[fair] (0.406901041667,-3.5) rectangle (0.408333333333,-3.2);
\fill[finite] (0.408333333333,-3.5) rectangle (0.410970052083,-3.2);
\fill[fair] (0.410970052083,-3.5) rectangle (0.412500000000,-3.2);
\fill[finite] (0.412500000000,-3.5) rectangle (0.415039062500,-3.2);
\fill[fair] (0.415039062500,-3.5) rectangle (0.416666666667,-3.2);
\fill[finite] (0.416666666667,-3.5) rectangle (0.419108072917,-3.2);
\fill[fair] (0.419108072917,-3.5) rectangle (0.420000000000,-3.2);
\fill[finite] (0.420000000000,-3.5) rectangle (0.423177083333,-3.2);
\fill[fair] (0.423177083333,-3.5) rectangle (0.424166666667,-3.2);
\fill[finite] (0.424166666667,-3.5) rectangle (0.427246093750,-3.2);
\fill[fair] (0.427246093750,-3.5) rectangle (0.428333333333,-3.2);
\fill[finite] (0.428333333333,-3.5) rectangle (0.431315104167,-3.2);
\fill[fair] (0.431315104167,-3.5) rectangle (0.432500000000,-3.2);
\fill[finite] (0.432500000000,-3.5) rectangle (0.435384114583,-3.2);
\fill[fair] (0.435384114583,-3.5) rectangle (0.436666666667,-3.2);
\fill[finite] (0.436666666667,-3.5) rectangle (0.439453125000,-3.2);
\fill[fair] (0.439453125000,-3.5) rectangle (0.440833333333,-3.2);
\fill[finite] (0.440833333333,-3.5) rectangle (0.443522135417,-3.2);
\fill[fair] (0.443522135417,-3.5) rectangle (0.444166666667,-3.2);
\fill[finite] (0.444166666667,-3.5) rectangle (0.447591145833,-3.2);
\fill[fair] (0.447591145833,-3.5) rectangle (0.448333333333,-3.2);
\fill[finite] (0.448333333333,-3.5) rectangle (0.451660156250,-3.2);
\fill[fair] (0.451660156250,-3.5) rectangle (0.452500000000,-3.2);
\fill[finite] (0.452500000000,-3.5) rectangle (0.455729166667,-3.2);
\fill[fair] (0.455729166667,-3.5) rectangle (0.456666666667,-3.2);
\fill[finite] (0.456666666667,-3.5) rectangle (0.459798177083,-3.2);
\fill[fair] (0.459798177083,-3.5) rectangle (0.460833333333,-3.2);
\fill[finite] (0.460833333333,-3.5) rectangle (0.463867187500,-3.2);
\fill[fair] (0.463867187500,-3.5) rectangle (0.464166666667,-3.2);
\fill[finite] (0.464166666667,-3.5) rectangle (0.467936197917,-3.2);
\fill[fair] (0.467936197917,-3.5) rectangle (0.468333333333,-3.2);
\fill[finite] (0.468333333333,-3.5) rectangle (0.472005208333,-3.2);
\fill[fair] (0.472005208333,-3.5) rectangle (0.472500000000,-3.2);
\fill[finite] (0.472500000000,-3.5) rectangle (0.476074218750,-3.2);
\fill[fair] (0.476074218750,-3.5) rectangle (0.476666666667,-3.2);
\fill[finite] (0.476666666667,-3.5) rectangle (0.480143229167,-3.2);
\fill[fair] (0.480143229167,-3.5) rectangle (0.480833333333,-3.2);
\fill[finite] (0.480833333333,-3.5) rectangle (0.484212239583,-3.2);
\fill[fair] (0.484212239583,-3.5) rectangle (0.485000000000,-3.2);
\fill[finite] (0.485000000000,-3.5) rectangle (0.488281250000,-3.2);
\fill[fair] (0.488281250000,-3.5) rectangle (0.488333333333,-3.2);
\fill[finite] (0.488333333333,-3.5) rectangle (0.492350260417,-3.2);
\fill[fair] (0.492350260417,-3.5) rectangle (0.492500000000,-3.2);
\fill[finite] (0.492500000000,-3.5) rectangle (0.496419270833,-3.2);
\fill[fair] (0.496419270833,-3.5) rectangle (0.496666666667,-3.2);
\fill[finite] (0.496666666667,-3.5) rectangle (0.500488281250,-3.2);
\fill[fair] (0.500488281250,-3.5) rectangle (0.500833333333,-3.2);
\fill[finite] (0.500833333333,-3.5) rectangle (0.504557291667,-3.2);
\fill[fair] (0.504557291667,-3.5) rectangle (0.505000000000,-3.2);
\fill[finite] (0.505000000000,-3.5) rectangle (0.508626302083,-3.2);
\fill[fair] (0.508626302083,-3.5) rectangle (0.509166666667,-3.2);
\fill[finite] (0.509166666667,-3.5) rectangle (0.524902343750,-3.2);
\fill[fair] (0.524902343750,-3.5) rectangle (0.525000000000,-3.2);
\fill[finite] (0.525000000000,-3.5) rectangle (0.528971354167,-3.2);
\fill[fair] (0.528971354167,-3.5) rectangle (0.529166666667,-3.2);
\fill[finite] (0.529166666667,-3.5) rectangle (0.533040364583,-3.2);
\fill[fair] (0.533040364583,-3.5) rectangle (0.533333333333,-3.2);
\fill[finite] (0.533333333333,-3.5) rectangle (1.000000000000,-3.2);
\draw[black!30] (0,-3.5) rectangle (1,-3.2);
\draw[black!50] (0,-3.65)--(1,-3.65);
\draw[black!50] (0.0,-3.65)--(0.0,-3.71) node[below,font=\sffamily\footnotesize]{0};
\draw[black!50] (0.3333333333333333,-3.65)--(0.3333333333333333,-3.71) node[below,font=\sffamily\footnotesize]{0.01};
\draw[black!50] (0.6666666666666666,-3.65)--(0.6666666666666666,-3.71) node[below,font=\sffamily\footnotesize]{0.02};
\draw[black!50] (1.0,-3.65)--(1.0,-3.71) node[below,font=\sffamily\footnotesize]{0.03};
\node[anchor=west,font=\sffamily\footnotesize] at(1.025,-3.73){$d$};
\fill[erasure] (0.0,-4.67) rectangle (0.022,-4.45);
\node[anchor=west,font=\sffamily\footnotesize] at(0.03,-4.55){Erasure upper, Lemma~\ref{lem:erasure}};
\fill[fair] (0.53,-4.67) rectangle (0.552,-4.45);
\node[anchor=west,font=\sffamily\footnotesize] at(0.56,-4.55){Fair-input lower, \S\ref{sec:fair-input}};
\fill[runs] (0.0,-5.09) rectangle (0.022,-4.87);
\node[anchor=west,font=\sffamily\footnotesize] at(0.03,-4.97){Run-length upper, \S\ref{sec:small-deletion}};
\fill[finite] (0.53,-5.09) rectangle (0.552,-4.87);
\node[anchor=west,font=\sffamily\footnotesize] at(0.56,-4.97){Finite-state lower, \S\ref{sec:source}};
\fill[shared] (0.0,-5.51) rectangle (0.022,-5.29);
\node[anchor=west,font=\sffamily\footnotesize] at(0.03,-5.39){Shared-tail upper, \S\ref{sec:shared}};
\fill[renewal] (0.53,-5.51) rectangle (0.552,-5.29);
\node[anchor=west,font=\sffamily\footnotesize] at(0.56,-5.39){Renewal lower, \S\ref{sec:renewal}};
\fill[posterior] (0.0,-5.93) rectangle (0.022,-5.71);
\node[anchor=west,font=\sffamily\footnotesize] at(0.03,-5.81){Posterior upper, \S\ref{sec:posterior}};
\fill[poisson] (0.53,-5.93) rectangle (0.552,-5.71);
\node[anchor=west,font=\sffamily\footnotesize] at(0.56,-5.81){Poisson-repeat lower, \S\ref{sec:renewal-poisson}};
\end{tikzpicture}

%% file: sections/transport.tex
\section{Extension from finite bounds to all deletion probabilities}\label{sec:transport}
\subsection{A comparison between two deletion probabilities}
Suppose a complete numerical check has proved a capacity bound at a deletion probability $a$. A list of such pointwise bounds leaves a second question: what follows at a probability $d$ where no bound was evaluated? The channel-mixture inequality of \citet{RD15} answers this question. We first prove the comparison, then use it to bound capacity on closed intervals and check that these intervals cover $[0,1]$. Choosing which parameters to evaluate belongs to the numerical search; extending a proved bound between parameters is the analytic step developed here.

Fix $a\le d<1$. Select each input position independently with probability $(1-d)/(1-a)$, and send the selected subsequence through an independent deletion channel with probability $a$. Each original bit reaches the receiver with probability $1-d$, exactly as in the $d$-deletion channel. Now also tell the receiver which positions were selected in the first stage. This supplies additional information and therefore gives an upper bound on the mutual information of the original channel. The number of selected positions supplies the multiplicative factor in the following inequality.
\begin{lemma}[Normalized-capacity transport]\label{lem:transport}
For $0\le a\le d<1$,
\begin{equation}\label{eq:transport}
\boxed{\quad C(d)\le\frac{1-d}{1-a}C(a).\quad}
\end{equation}
Thus $C(d)/(1-d)$ is nonincreasing on $[0,1)$.
\end{lemma}
\begin{proof}
Set $\lambda=(1-d)/(1-a)$. Since $a\le d<1$, we have $0<\lambda\le1$. Let $X_1^n$ be the transmitted input. Let $M_1,\ldots,M_n$ be independent indicators of first-stage selection, each equal to one with probability $\lambda$ and independent of $X_1^n$. Pass the selected input subsequence through an independent deletion channel with parameter $a$. An input bit appears in the final trace if it is selected and is not deleted in the second stage. This event has probability $\lambda(1-a)=1-d$. The two random indicators belonging to different input positions are independent, so the resulting conditional distribution of the trace is exactly that of the $d$-deletion channel.

Let $M=(M_1,\ldots,M_n)$ record the selected positions and let $Y$ be the final trace. The receiver is given $M$, not the values of the selected input bits. By \cref{lem:side-information}, adjoining $M$ cannot decrease mutual information. Independence gives $I(X_1^n;M)=0$, so the information chain rule gives
\[
I(X_1^n;Y)\le I(X_1^n;Y,M)=I(X_1^n;Y\mid M).
\]
Fix a mask selecting $N$ positions and call the corresponding input subsequence $S$. The conditional channel from $X_1^n$ to $Y$ depends only on $S$. Indeed, the unselected input positions cannot affect any received symbol. The information chain rule therefore gives $I(X_1^n;Y\mid M=m)=I(S;Y\mid M=m)$. The marginal law of $S$ may have arbitrary dependence, but the maximum in the definition of $C_N(a)$ covers every length-$N$ law. Thus this conditional information is at most $NC_N(a)$. Averaging over the masks gives
\[
I(X_1^n;Y)\le\E[NC_N(a)],\qquad N\sim\operatorname{Bin}(n,\lambda).
\]
When $N=0$, the trace is empty and its information is zero; we set the product $NC_N(a)$ to zero without defining $C_0$.

The bound still contains $C_N(a)$ at the random selected length $N$, whereas the desired result contains the limiting capacity $C(a)$. We control all small values of $N$ by a bounded remainder and all large values by convergence to capacity. Fix $\eta>0$. Since $C_N(a)\to C(a)$, there is an integer $N_0$ such that $C_N(a)\le C(a)+\eta$ whenever $N\ge N_0$. A binary input carries at most $N$ bits, so $0\le C_N(a)\le1$ for all positive $N$. Split the expectation according to whether $N<N_0$. The short part is at most $N_0$; the long part is at most $(C(a)+\eta)\E N$. As $\E N=n\lambda$, this gives
\[
\frac1n I(X_1^n;Y)\le\frac{N_0}{n}+\lambda\bigl(C(a)+\eta\bigr).
\]
Neither $N_0/n$ nor $\lambda(C(a)+\eta)$ depends on the chosen length-$n$ input distribution. Maximizing the left side therefore gives the same upper bound on $C_n(d)$. With $\eta$ fixed, $N_0$ is fixed, so $N_0/n$ tends to zero as $n\to\infty$. The capacity limit then gives $C(d)\le\lambda(C(a)+\eta)$. This holds for every positive $\eta$; letting $\eta\downarrow0$ proves $C(d)\le\lambda C(a)$.
\end{proof}
An upper bound at $a$ therefore extends to every $d\ge a$: substitute that upper bound for $C(a)$ in \eqref{eq:transport}. For a lower bound, fix $d\le b<1$ and apply the lemma with $d$ as its smaller parameter. It gives $C(b)\le(1-b)C(d)/(1-d)$. Multiplying by the positive factor $(1-d)/(1-b)$ shows that $C(d)\ge(1-d)C(b)/(1-b)$. Replacing $C(b)$ by a lower bound preserves this direction. Thus upper bounds extend right and lower bounds extend left.

\begin{examplebox}[title={Two point bounds control a whole interval}]
For a small numerical example, weaken two of the certified point bounds to $C(1/5)\le2/5$ and $C(1/2)\ge1/10$. The upper point lies to the left of $[3/10,2/5]$, and the lower point lies to its right. Dividing each bound by its survival probability gives the constant ratios
\[
\frac{2/5}{1-1/5}=\frac12,\qquad
\frac{1/10}{1-1/2}=\frac15.
\]
The two transport directions just proved therefore give
\[
\frac{1-d}{5}\le C(d)\le\frac{1-d}{2}
\qquad\text{for every }d\in[3/10,2/5].
\]
At the left endpoint the interval is $[7/50,7/20]$; at the right endpoint it is $[3/25,3/10]$. Subtracting lower from upper gives width $3(1-d)/10$, which decreases as $d$ increases. Its largest value is thus $21/100$, attained on the left. Half that width bounds the midpoint error everywhere in the interval. These deliberately loose numbers illustrate the comparison; the production anchors give much narrower intervals. No interpolation assumption about capacity is used.
\end{examplebox}

\subsection{Why the plotted lower bound has small downward steps}\label{sec:certificate-steps}
The lower bound is assembled from pointwise certificates with one-sided validity. A bound proved at $b$ can be transported to $d\le b$. Immediately to the right of $b$, that particular transport rule no longer applies. The strongest remaining bound may be smaller. This creates steps in the certificate even when all its components were evaluated accurately.
\begin{lemma}[Steps caused by expiring point certificates]\label{lem:certificate-steps}
Suppose $\ell_j\le C(b_j)$ are finitely many lower bounds with $b_j<1$, and $c(1-d)\le C(d)$ holds throughout the parameter range. For $0<d<1$, set
\[
B(d)=\max\left(\{c\}\cup
\left\{\frac{\ell_j}{1-b_j}:b_j\ge d\right\}\right),
\qquad L(d)=(1-d)B(d).
\]
Then $L(d)\le C(d)$, and $B$ is a nonincreasing step function. At an interior parameter $t$,
\[
L(t^+)-L(t^-)=(1-t)\bigl[B(t^+)-B(t^-)\bigr]\le0.
\]
At the point itself, $L(t)=L(t^-)$.
\end{lemma}
\begin{proof}
Every term in the maximum is a valid normalized lower bound by \cref{lem:transport}. Taking their maximum preserves the bound. As $d$ increases, the eligible set loses points and gains none. Its maximum cannot increase, and is constant between successive point parameters. Multiplication of the two one-sided limits by $1-t$ gives the stated jump. Equality $b_j=t$ remains eligible at $t$, which gives the left-hand value there.
\end{proof}
For upper points, the eligible set grows as $d$ increases, and one takes its minimum. The upper curve therefore takes its right-hand value at a step. Intersecting adjacent cell enclosures chooses exactly these stronger endpoint values.

The exact audit finds 4,419 downward lower steps and no upward lower steps. Every selected cell coefficient equals the strongest eligible coefficient among the saved bounds. The largest lower step occurs at $t=1393/4096$: the value at $t$ is $0.20488639494$, while the right-hand limit is $182099642271/896000000000$. Their difference is approximately $0.001650187048$ bit per input bit. Both neighboring certificates are stationary-source bounds. Merely changing source names is therefore not the explanation: the expiring domain of a point certificate is essential. A maximum of finitely many globally continuous rate curves would itself be continuous.

\Cref{fig:certificate-steps} contrasts the one-sided transport mechanism with its largest observed step.

\begin{figure}[!htbp]\centering
\includegraphics[width=\linewidth]{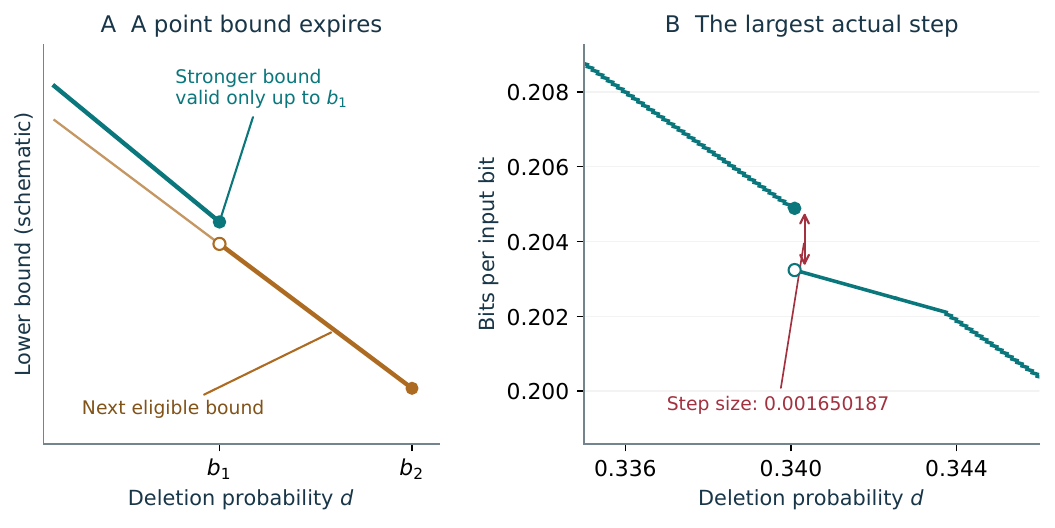}
\caption{The ruggedness of the reported lower bound comes from the finite parameter cover. Panel A is a schematic: a stronger lower point at $b_1$ extends only leftward; just to its right, the weaker certificate at $b_2$ remains available. Panel B plots exact linear cell segments from the published data near its largest downward step. A filled point is the bound at the parameter; an open point is the limit immediately to its right. Gaps are not interpolated. The $0.001650187$-bit step is small on the full vertical scale of \cref{fig:intro-capacity} but visible under magnification. These features describe the proved enclosure, not the shape of the unknown capacity function.}\label{fig:certificate-steps}
\end{figure}
Additional certified parameter points, or direct interval evaluation of a fixed source expression, can reduce these steps. Smoothing a plotted curve would not prove the smoothed values. The exact boundary audit and plotting data are supplied with the reproduction records.

\subsection{A single rational inequality verifies an entire cell}
A point bound will be called an \emph{anchor}. It consists of a parameter and a proved upper or lower value at that parameter. Consider a parameter interval $[x,y]$ and an upper anchor $C(a)\le u_a$ with $a\le x$ and $a<1$. The preceding comparison gives $C(d)\le(1-d)u_a/(1-a)$ throughout the interval. Similarly, a lower anchor $\ell_b\le C(b)$ with $b\ge y$ and $b<1$ gives $C(d)\ge(1-d)\ell_b/(1-b)$. Define the normalized coefficients to be these two constant ratios,
\[
A=\frac{u_a}{1-a},\qquad B=\frac{\ell_b}{1-b}.
\]
They are constant on the chosen interval; the only remaining dependence on $d$ is multiplication by $1-d$. The erasure bound in \cref{lem:erasure}, $C(d)\le1-d$, is the special upper coefficient $A=1$. A separately proved bound $C(d)\ge c(1-d)$ supplies the lower coefficient $B=c$ directly, including intervals ending at one. These analytic bounds avoid a division by $1-b$ at $b=1$.

\begin{certbox}[title={The exact continuum row}]
A cell record contains $(x,y,a,b,A,B)$ and certificate identifiers. Check the applicable anchor directions, exact coefficient identities, $A\ge B$, and
\begin{equation}\label{eq:continuum-row}
0\le x<y\le1,\qquad (1-x)(A-B)\le\frac{19}{1000}.
\end{equation}
Then throughout that closed cell,
\[
(1-d)B\le C(d)\le(1-d)A\le(1-d)B+\frac{19}{1000}.
\]
\end{certbox}

\begin{proof}[Proof of the continuum row]
The upper anchor and lower anchor have the directions established in \cref{lem:transport}; an analytic bound such as $C(d)\le1-d$ can be used directly. Thus the first two inequalities in the cell enclosure hold for every $d\in[x,y]$. Their difference is $(1-d)(A-B)$, which is nonincreasing in $d$ because $A\ge B$. Its maximum on $[x,y]$ is the left endpoint value in \eqref{eq:continuum-row}. For a cell ending at one, extend continuously and use $C(1)=0$. At zero, the noiseless channel has $C(0)=1$.
\end{proof}
The acceptance target $19/1000$ is fixed independently of the resulting cell values. After the enclosures are established, their exact largest width is computed and reported as $\Delta_\star$; it is smaller than that target. Every comparison is rational. Checking the cell's left endpoint is therefore an exact analytic maximization, not a mesh approximation.

\subsection{The complete finite verification algorithm}
Each numerical leaf is a record of a required numerical inequality and the objects used to check it. The assembly below takes already checked leaves as premises. It verifies their identities and how their capacity bounds are combined; it does not establish an unchecked numerical inequality merely by reading a recorded value.
\begin{certbox}[breakable=false,title={Verify the capacity approximation}]
\begin{enumerate}
\item Parse the main transport certificate with duplicate keys and approximate mathematical numbers rejected; component records use their separately specified schemas. Check each leaf identifier, parameter, rule, and content hash against the sealed inventory and completed replay.
\item For every point anchor, require its matching numerical bound and mathematical family. Require a separate transfer theorem for every universal lower coefficient.
\item Start at $x=0$. Read cells in their listed order. Require each left endpoint to equal $x$, require strict positive cell length, and check the continuum row with exact fractions. Set $x$ to its right endpoint.
\item Require $x=1$, the two endpoint identities, and equality between the recorded maximum and the maximum of all exact cell errors. Reject any unrecognized node type or missing dependency.
\end{enumerate}
\end{certbox}
This assembly terminates because its input lists 6,997 numerical nodes and 4,931 cells, and each fraction comparison uses finitely many integer operations. The checker for an individual numerical node must additionally enumerate every state specified by its analytic result and bound every omitted infinite sum. Those obligations were established before the node entered the assembly. The algorithm is a verification specification, not a claim that the search for such a certificate has a polynomial dependence on requested accuracy.

\subsection{The exact constant in the approximation theorem}
The certificate records the following interval endpoints and normalized coefficients for the cell with the largest computed enclosure width:
\[
[x,y]=\left[\frac{521}{1000},\frac{1043}{2000}\right],\quad
A=\frac{41527369}{160000000},\quad
B=\frac{105217927}{478500000}.
\]
The upper coefficient is supplied by shared-tail witness \texttt{U\_final\_grid\_d0520}. The lower coefficient is supplied by stationary-source point \texttt{I\_checked\_source\_extensions\_v1\_143}. The computed width is $(1-x)(A-B)$, the largest width of the linear enclosure on this cell, as proved above. Replacing $x,A,B$ by the preceding rational values, multiplying their numerators, and reducing the resulting fraction gives
\[
\left(1-\frac{521}{1000}\right)
\left(\frac{41527369}{160000000}-\frac{105217927}{478500000}\right)
=\frac{2908466681147}{153120000000000}.
\]
This rational equality identifies the constant in \cref{thm:main}. \Cref{fig:anchors} shows where the point bounds supplying this assembly were evaluated. It does not identify where actual capacity lies within the interval.

\begin{figure}[tb]
\centering\includegraphics[width=\linewidth]{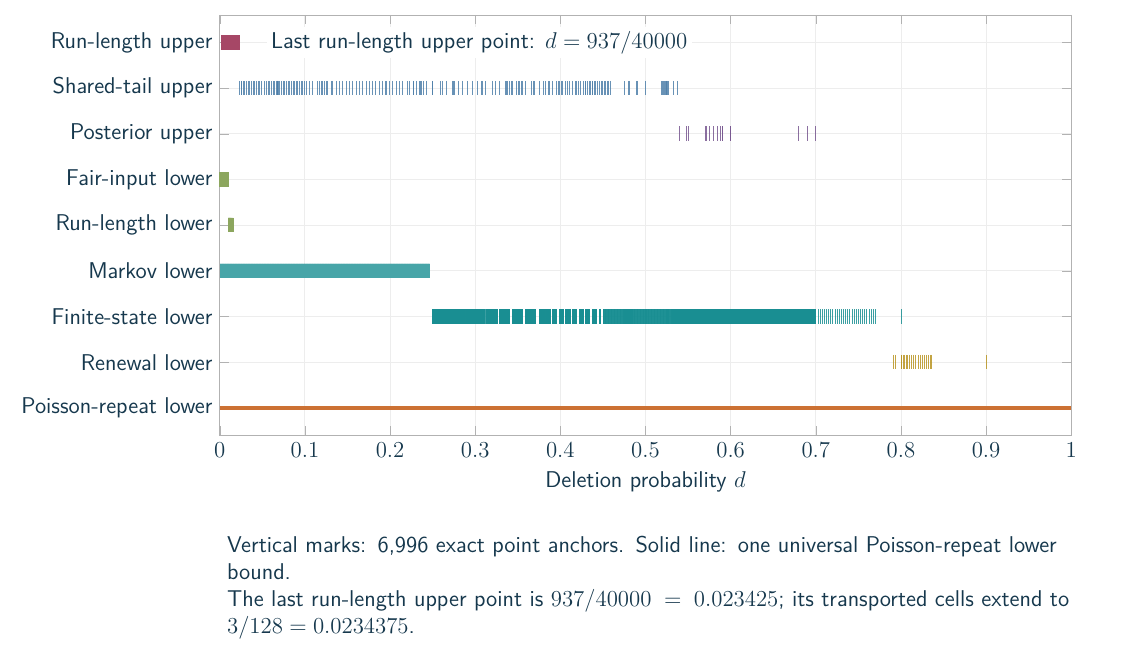}
\caption{Certificate families along the deletion parameter. Marks show the locations of exact point anchors, grouped by mathematical family. The row grouping distinguishes the mathematical roles of the anchors. The universal renewal coefficient is a separate all-parameter node, not a point at $d=1$.}\label{fig:anchors}
\end{figure}

\subsection{A shorter proof route for the uniform tolerance}\label{sec:family-simplification}
The published enclosure selects several lower-bound families near zero. This does not mean that every family is required to reach the worst-case tolerance. We can test that question using only the already admitted anchor values and the transport theorem, without running a new numerical experiment.

Remove the anchors labeled \texttt{lower\_iid} and \texttt{lower\_small\_run}. These are the two fair-input refinement families whose derivations are in \cref{sec:fair-input}. Retain the finite-memory and renewal lower bounds, the Poisson-repeat coefficient, and all upper anchors. On each interval between consecutive retained or original anchor parameters, choose the smallest eligible upper ratio and largest eligible lower ratio. Eligibility has exactly the directions proved above: an upper anchor lies at or left of the interval, and a lower anchor lies at or right of it. The resulting ratios therefore bound capacity throughout the closed interval.

This exact recomposition gives $4{,}813$ cells. Their maximum enclosure width is still
\[
\frac{2908466681147}{153120000000000}=2\eps_\star.
\]
The maximum occurs on $[521/1000,1043/2000]$, using the same two anchor values as the original maximizing cell. Every cell's endpoints, anchor identifiers, ratios, and exact width are recorded in \texttt{data/SIMPLIFIED\_UNIFORM\_COVER.json}. The checker verifies the anchor directions, the full interval cover, and each rational width. By the continuum-row argument, half this maximum bounds the midpoint error for every deletion probability.

The same retained bounds also certify the original published estimate, rather than requiring the reader to adopt a different midpoint. Write $m(d)=(L(d)+U(d))/2$ for that original estimate, and let $L_s(d),U_s(d)$ be the recomposed bounds. Since $C(d)$ lies between $L_s(d)$ and $U_s(d)$, its distance from $m(d)$ is at most the larger distance to an endpoint:
\[
|C(d)-m(d)|\le
\max\{m(d)-L_s(d),\,U_s(d)-m(d)\}.
\]
To verify the right side, subdivide at every boundary of either enclosure. Within each open subinterval, each displayed difference is a constant times $1-d$, so its supremum is found at a one-sided endpoint. At shared boundaries, use the prescribed intersections of adjacent intervals; at $d=0,1$, use the exact capacities. The rational check records a maximum of precisely $\eps_\star$, using only retained capacity bounds. The values defining $m$ enter this check as exact numbers, without assuming the validity of the omitted fair-input lower bounds. The supplied program \certpath{data/audit_original_midpoint.py} checks every common subinterval and boundary intersection and writes \certpath{data/ORIGINAL_MIDPOINT_RETAINED_CHECK.json}.

The two enclosures have different purposes. The shorter route establishes the uniform tolerance without the detailed fair-input refinements. The published enclosure retains those refinements to improve individual values and its mean certified radius. To make the distinction quantitative, the original mean radius is at most $0.006521572012$, while that of the recomposed enclosure is at most $0.006523144395$ bits per input bit, under uniform weighting of $d$. These are bounds on estimation error, not measured errors relative to an unknown capacity. The manuscript continues to report the original $L(d)$, $U(d)$, and mean; the simplified cover is a separate proof route for the worst-case accuracy.

\paragraph{What the omission test does and does not establish.}
The same check can remove one available family at a time. If a recomposed cover exceeds the requested width, this calculation has failed to prove that tolerance using the remaining stored anchors and these comparison rules. It does not show that the family is mathematically necessary, or that a better construction could not replace it. The successful simultaneous omission above is stronger: it supplies an explicit, checkable cover that attains the original tolerance. This is why the corresponding long derivations can be placed among supplementary refinements while the main lower-bound part focuses on finite-memory and renewal inputs.

%% file: sections/arithmetic.tex
\section{Numerical inequalities and rigorous arithmetic}
\label{sec:arithmetic}

The analytic results specify numerical inequalities that suffice for the capacity bounds. On the upper side, the inputs to a check are probability assignments used to bound output entropy and a bounded auxiliary function whose differences cancel when summed. On the lower side, the input is an exact probability law for a stationary source, together with finite entropy expressions and explicit remainders. To use a numerical evaluation in a proof, we must know on which side of the exact answer it lies. An upper bound must remain an upper bound after every rounding step, including any subtraction. A lower bound requires the opposite direction. This section explains the arithmetic rules that preserve those directions.

The certificate checkers use three arithmetic methods: exact integers and fractions; intervals from MPFR or Arb that contain the exact answer; and IEEE-754 operations with a specified rounding direction. A downward operation returns a representable value no larger than the exact result; an upward operation returns one no smaller. An interval calculation propagates both endpoints. The arithmetic contract, the permitted input ranges, and the rejection checks are part of the numerical premise of the theorem. For example, rounding $1/3$ to three decimal places gives the enclosure $[0.333,0.334]$: the lower endpoint is below the exact fraction and the upper endpoint is above it. Such \emph{outward rounding} preserves a statement of the form ``the exact value lies in this interval.'' If an upper bound of $0.334$ is required, this interval proves it. It cannot prove an upper bound of $0.3336$, even though the interval midpoint is smaller. The actual checkers use the precision specified for their families. Numerical search selects their inputs; directed arithmetic verifies the required inequalities for those fixed inputs.

\begin{examplebox}[title={Why subtraction reverses the rounding choice}]
Suppose an exact contribution is $E=A-B-C$, with
\[
A\in[9,10],\qquad B\in[3,4],\qquad C\in[5,6].
\]
To bound $E$ from below, make its positive term as small as permitted and its subtracted terms as large as permitted. This gives $E\ge9-4-6=-1$. If an analytic lemma has also proved $E\ge0$, the valid lower bound improves to $\max\{0,-1\}=0$.

Subtracting all lower endpoints would instead give $9-3-5=1$, which need not be a lower bound: the admissible values $A=9,B=4,C=5$ give $E=0$. The alignment-increment checker evaluates a difference with exactly these signs: a longer-prefix entropy term minus two shorter-prefix terms. It therefore rounds the first term downward and the subtracted terms upward. The additional maximum with zero is valid only because the analytic lemma proves that the exact difference is nonnegative. The intervals here are illustrative, not the precision used in the calculation.
\end{examplebox}

\subsection{Enclosing logarithms by a positive series}

For a probability assignment $q$, the entropy bound uses $-\log_2q$; entropy formulas likewise contain logarithms of masses. Thus enclosing these logarithms is necessary before enclosing either an upper information bound or a lower one. We first reduce a positive argument to a bounded interval. Write $x=2^e m$, where $e$ is an integer and $1\le m<2$. Then $\log_2x=e+\ln m/\ln2$, so it suffices to enclose the natural logarithm of $m$. Set $t=(m-1)/(m+1)$. This gives $m=(1+t)/(1-t)$ and $0\le t<1/3$.

The derivative of $\ln((1+t)/(1-t))$ is $2/(1-t^2)$, and its value at zero is zero. To integrate this derivative, use the finite geometric identity
\[
\frac1{1-z^2}=\sum_{j=0}^{N-1}z^{2j}+\frac{z^{2N}}{1-z^2}.
\]
Multiplying the right side by $1-z^2$ verifies the identity by cancellation of consecutive powers. Integrating from $z=0$ to $z=t$ therefore gives
\begin{equation}
 \ln m=2\sum_{j=0}^{N-1}\frac{t^{2j+1}}{2j+1}+R_N(t),\qquad
 0\le R_N(t)\le\frac{2t^{2N+1}}{(2N+1)(1-t^2)}.
 \label{eq:atanh-enclosure}
\end{equation}
Indeed, the remainder is the nonnegative integral $2\int_0^t z^{2N}/(1-z^2)\,dz$. Throughout that integral, $1/(1-z^2)\le1/(1-t^2)$. Replacing the denominator by this upper bound and integrating $z^{2N}$ proves the displayed remainder estimate. This argument also holds at $t=1/3$, which will be used to enclose $\ln2$.

The shared-tail and posterior checkers use $N=20$ to construct logarithm enclosures at the table points described below. Each checker encloses $t$ from below and above, propagates its positive powers and the partial sum with directed rounding, and adds the upper remainder only to the upper endpoint. Suppose the resulting interval for $\ln m$ is $[a_-,a_+]$ and an interval for $\ln2$ is $[\ell_-,\ell_+]$, with $\ell_->0$. Because division by a larger positive denominator gives a smaller value, the desired binary logarithm lies in
\[
 e+\frac{a_-}{\ell_+}\ \le\ \log_2x\ \le\ e+\frac{a_+}{\ell_-}.
\]
The endpoint divisions and additions are themselves directed. The exponent $e$ is exact; no decimal approximation to it is introduced.

To reduce repeated work, these checkers first enclose $\log_2c$ at the 4,096 points $c=1+k/4096$, $0\le k<4096$, using that 20-term calculation. For each mantissa $m$, choose the largest such $c\le m$. Then
\[
\log_2m=\log_2c+\frac{\ln(m/c)}{\ln2},\qquad
s=\frac{m-c}{m+c},\quad 0\le s<\frac1{8192}.
\]
The bound on $s$ follows from $m-c<1/4096$ and $m+c\ge2$. Apply \eqref{eq:atanh-enclosure} to $m/c=(1+s)/(1-s)$ with only $N=3$: its positive remainder is at most $2s^7/[7(1-s^2)]$. Directed addition of the table enclosure and this residual enclosure bounds each requested logarithm. Thus 20 terms are used once per table point, while three terms and an explicit tail are used per subsequent argument.

The adjacent binary64 endpoints used for $\ln2$ are
\[
 \ell_- = \texttt{0x1.62e42fefa39efp-1},\qquad
 \ell_+ = \texttt{0x1.62e42fefa39f0p-1}.
\]
These hexadecimal expressions specify exact dyadic rationals, meaning integers divided by powers of two. Their validity follows from a rational calculation: take $t=1/3$ and $N=20$ in \eqref{eq:atanh-enclosure}, and let $S_{20}$ be its partial sum and $R_{20}^{+}$ its upper remainder. Exact fractions give
\[
\ell_-<S_{20}<S_{20}+R_{20}^{+}<\ell_+.
\]
The series already places $\ln2$ between the middle two quantities, so these strict comparisons establish the claimed endpoints. The auxiliary arithmetic receipt records all four rational values and both strict margins. This check establishes the constants independently of a logarithm library.

Exponent and mantissa extraction are exact for finite binary64 inputs. A positive subnormal input is first multiplied by $2^{52}$, an exact scaling, and its exponent is subsequently adjusted back. Invalid or nonpositive logarithm inputs cause rejection. Structural zero entropy terms are handled separately by the convention $0\log0=0$.

Some lower bounds need only one side of this calculation. For the nonnegative alignment-entropy contributions in the repeat-channel lower bound, the checker keeps the first \emph{four} terms of the same series and discards the positive remainder. Its arguments are $1+v$ with $v\ge0$, so the binary exponent is nonnegative. Each positive operation is rounded down, and division uses an upper endpoint for $\ln2$. These choices produce a lower logarithm. The discarded remainder can only increase the exact logarithm. The shorter sum may weaken the resulting information-rate lower bound, but it cannot invalidate it.

\subsection{Nonnegative entropy and summation error}

Entropy expressed as a difference of large terms can lose numerical accuracy. Here it is useful to rewrite it as a sum of nonnegative terms before evaluating it. For nonnegative masses $w_1,\ldots,w_r$, let $W=\sum_iw_i$. Their homogeneous entropy is $W\log_2W-\sum_iw_i\log_2w_i$. For positive $w_i$, distributing the first term over the sum gives
\begin{align}
\Phi(w_1,\ldots,w_r)
 &=\sum_iw_i\bigl(\log_2W-\log_2w_i\bigr)\notag\\
 &=\sum_iw_i\log_2\left(1+\frac{\sum_{j\ne i}w_j}{w_i}\right).
 \label{eq:positive-Phi}
\end{align}
The last equality uses $W=w_i+\sum_{j\ne i}w_j$. Every logarithm has an argument at least one, and every summand is therefore nonnegative. A zero mass contributes zero by continuity.

Lower bounds on the masses also give a lower bound on $\Phi$. To see this directly, differentiate its original expression with respect to a positive coordinate $w_i$. The constant derivative terms cancel, leaving $\partial\Phi/\partial w_i=\log_2(W/w_i)\ge0$. Continuity extends this monotonicity to zero coordinates. Thus replacing any mass by a smaller nonnegative value can only reduce the entropy.

The checker rounds positive products and additions down. When it needs the other-mass sum in \eqref{eq:positive-Phi}, it subtracts the exact stored $w_i$ from a downward total and rounds down once more. The result is at most the exact sum of the other stored masses. Clamping a negative result to zero remains valid because that exact sum is nonnegative. Ratios are rounded down. If a ratio overflows, it may be clipped to $2^{1023}$, which is still below the exact ratio. The subsequent logarithm argument is checked to be finite and at least one. Applying the lower logarithm just proved then preserves the lower direction term by term.

One positive summation uses a round-to-nearest NumPy reduction and is corrected by an aggregate error bound. Let $x_1,\ldots,x_n$ be nonnegative binary64 values, let $u=2^{-53}$, and require $nu<1$. A normal, finite rounded addition has the form $(a+b)(1+\delta)$ with $|\delta|\le u$. Expanding a sequence of additions expresses the final result as a sum in which each $x_i$ is multiplied by a product of these factors. Each input traverses at most $n-1$ additions, whether the reduction is sequential or arranged as a tree.

Every such product is at most $(1+u)^n$. The binomial theorem and the inequality $\binom nk\le n^k$ give
\[
(1+u)^n
 =\sum_{k=0}^n\binom nk u^k
 \le\sum_{k=0}^{\infty}(nu)^k
 =\frac1{1-nu}.
\]
The final equality is the geometric series, whose convergence is exactly the reason for the condition $nu<1$. Since all $x_i$ are nonnegative, the multiplier bound can be applied to every term of the expanded reduction. Defining $\gamma_n=nu/(1-nu)$ therefore yields
\[
 \operatorname{fl}\!\left(\sum_{i=1}^nx_i\right)
 \le (1+\gamma_n)\sum_{i=1}^nx_i,
 \qquad 1+\gamma_n=\frac1{1-nu}.
\]
A positive addition with a subnormal exact result is exact: both inputs are integer multiples of the smallest positive subnormal, and their sum is representable while it remains in that range. Such an addition has $\delta=0$ and needs no separate underflow correction.

To obtain a lower bound on the exact sum from this upper bound on the rounding error, the checker divides the floating result by an upward binary64 enclosure of the exact rational $1+\gamma_n$ and rounds the quotient down. It rejects nonfinite values, negative inputs, and invalid reduction sizes. This is the specified Higham-type summation bound \citep{Higham02}; it controls this positive reduction, while the other operations retain their own directed rules.

\paragraph{Combining histogram masses and entropies.}
The histogram calculation groups outcomes that have the same recorded counts. For each collection of outcomes, it stores a pair $(p,E)$: $p$ is the sum of their probability masses, and $E$ is their homogeneous entropy. The two operations used by the checker correspond to taking products of collections and to joining disjoint collections. Their formulas follow directly from \eqref{eq:positive-Phi}.

For collections $(a_i)$ and $(b_j)$ with total masses $p$ and $q$, the product collection has masses $a_ib_j$ and total mass $pq$. Its entropy is
\[
\sum_{i,j}a_ib_j\log_2\frac{pq}{a_ib_j}
 =\sum_{i,j}a_ib_j\left(\log_2\frac p{a_i}+\log_2\frac q{b_j}\right)
 =qE+pF.
\]
In the last step, summing $b_j$ gives $q$ in the first term, and summing $a_i$ gives $p$ in the second. For the disjoint union, split $\log_2((p+q)/a_i)$ into $\log_2(p/a_i)+\log_2((p+q)/p)$, and do the same for the $b_j$ terms. This gives entropy $E+F+p\log_2((p+q)/p)+q\log_2((p+q)/q)$; its last two terms are $\Phi(p,q)$. Hence the two pair operations are
\[
 (p,E)\otimes(q,F)=(pq,qE+pF),\qquad
 (p,E)\oplus(q,F)=(p+q,E+F+\Phi(p,q)).
\]
The formulas extend to zero masses by continuity. Every resulting coordinate increases with the nonnegative input coordinates, so lower mass and entropy endpoints remain lower endpoints after either operation.

The initial histogram mass and entropy are enclosed by Arb and then rounded down into these positive operations. Arb may use cancellation internally; its returned interval already encloses the exact expression. A negative lower endpoint for a quantity known to be nonnegative can be clamped to zero. Subsequent positive arithmetic preserves that lower direction. The stated rules assume IEEE-754 binary64 operations and gradual underflow, and the reduction implementation is recorded.

\subsection{Dyadic potentials and rational probabilities}

A potential is a bounded auxiliary function $V$ of the state in an upper-bound construction. Its expected change can have either sign. Although these changes cancel in the information-rate argument, each finite inequality still contains them and must enclose them with the correct sign. An upper witness stores its potential values as dyadic rationals with bounded integer numerators. The allowed numerator and exponent ranges ensure that each potential difference is exactly representable in binary64. These representability conditions are checked before the potential is used.

For deletion probability $d=p/q$, the expected one-step potential change is
\[
 \frac pq\,\Delta V_{\rm delete}+\frac{q-p}{q}\,\Delta V_{\rm survive}.
\]
This is the conditional expectation over the two channel outcomes: deletion has probability $p/q$, and survival has probability $(q-p)/q$. The potential changes are the next value minus the current value in the corresponding outcome. The checker evaluates each signed numerator product upward, divides by the positive denominator upward, and then adds upward. Multiplication and division by a positive number preserve order even for negative potential differences, which proves that the result is an upper endpoint.

The required bound $q\le2^{53}$ ensures exact representation of the integers $p$, $q$, and $q-p$. Thus this calculation uses the specified rational probability, rather than treating its nearest binary64 or decimal approximation as exact. The shared and posterior kernels likewise use explicit directed operations for probability propagation and code lengths. When the upper inequality subtracts a conditional entropy, it subtracts a lower endpoint: if $h_-\le h$, then $-h\le-h_-$. This simple reversal under subtraction is necessary for the final upper direction.

The checker enumerates every state and action required by the upper-bound result. In the shared-tail construction, an additional analytic error bound covers the probability of deletion patterns that the numerical propagation omits. After checking all these inequalities, it rounds the final upper value upward to the exact rational endpoint used by the parameter-interval calculation. A selected set of favorable rows or a maximal cycle would not establish the inequality for all rows required by the theorem.

\subsection{Exact input laws and integer accumulation}

The lower-bound calculation evaluates the exact stationary rational Markov law specified in the source file. Integer propagation may replace a prefix probability by a smaller number, but that replacement is used only as a lower weight for the same source. To make this distinction precise, suppose $q_i(x)$ is a lower bound on the true probability of prefix $x$ and hidden state $i$. If $T_{ij}$ is a nonnegative transition probability, then $q_i(x)T_{ij}$ is no larger than the corresponding true extended-prefix contribution. Rounding this product down and summing the incoming contributions over $i$ preserves the inequality. Induction from the initial lower masses proves pointwise domination at every depth.

These propagated weights form a submeasure of the exact prefix distribution: they are nonnegative and may have smaller total mass. These weights multiply nonnegative contributions to the entropy of deletion positions given the input and trace. Replacing a probability by a smaller weight can only decrease that sum. The resulting value therefore remains a lower bound for the specified stationary source. The submeasure itself need not be stationary and is not used as a replacement input process.

The entropy-count cache encloses $n\log_2n$ over its entire finite integer domain. Its regeneration recipe uses directed MPFR evaluations of the scaled value. For non-powers of two, the lower and upper directed evaluations have the same integer floor. The common floor is therefore also the floor of the exact scaled value. The stored lower and upper integers, that floor minus two and plus three, lie on the required sides of the exact value. At zero and powers of two, the entropy term is exact. The complete regeneration comparison is the premise used by the current package; the historical empirical generator supplies no replacement for this comparison.

Integer operations require a separate size bound because exact formulas cease to be exact if their machine integers overflow. In the reconstructed base weighting, prefix numerators at each depth sum to at most $2^{61}$, while each nonnegative cached coefficient is less than $2^{62}$. If the prefix numerators are $a_x$, coefficients are $c_x$, and $c_{\max}=\max_xc_x<2^{62}$, the bound is
\[
\sum_x a_xc_x
 \le c_{\max}\sum_xa_x
 \le c_{\max}2^{61}<2^{62}2^{61}=2^{123}.
\]
The first inequality uses nonnegativity and the largest coefficient; the second uses the total prefix mass. The strict final inequality uses $c_{\max}<2^{62}$. Every partial worker sum is at most the complete nonnegative sum and satisfies the same bound. An unsigned 128-bit accumulator therefore has sufficient range. The reviewed transition-product bound is less than $2^{109}$ before its 48-bit shift, coming from the corresponding prefix and transition numerator bounds. This intermediate product also fits the accumulator before truncation. Counting the number of words and multiplying by a per-word maximum would lose the total-mass normalization on which the useful bound depends.

\subsection{Alignment cells and source weights}\label{sec:alignment-arithmetic}

An embedding of a trace in an input word is an increasing list of input positions whose bits equal the trace. Its count depends only on those two words, not on the probability distribution used to generate the input. They can therefore be computed once and later weighted for a specified source. For two nonnegative integer counts $a,b$, homogeneous entropy is $\Phi(a,b)=g(a+b)-g(a)-g(b)$, where $g(n)=n\log_2n$. Let $g_-$ and $g_+$ be scaled integer lower and upper endpoints for this function, all at the same scale. Subtracting the upper endpoints for the two negative terms gives the scaled lower bound
\[
\max\{0,g_-(a+b)-g_+(a)-g_+(b)\}.
\]
The expression inside the maximum is no larger than the exact scaled entropy. The maximum with zero is still no larger because entropy is nonnegative. Summing these lower integers over every required trace gives a lower endpoint for $F_j(x)$, the embedding-count entropy sum defined in \cref{sec:source}, for every input word $x$ at the specified depth.

The embedding counts themselves follow a simple exhaustive recurrence. There is one embedding of the empty word. When a new input bit is appended, every embedding either omits the new position or uses it as its last position; the latter choice is possible only when that bit matches the last trace bit. These two alternatives are disjoint and cover every embedding. Induction therefore gives the correct integer counts. Applying the recurrence to both possible appended bits, zero and one, covers every input word at the specified depth. This is an enumeration of all those words, not a sample from the chosen source. The source probabilities enter only when the resulting counts are weighted.

An additional alignment contribution $E_j(x_1^h)$ measures the increase in the embedding-count entropy sum when the input prefix grows from length $h-1$ to length $h$. Its defining expression contains one positive $F$ term and two negative ones, so lower endpoints for all three would not give a lower bound on the difference. At the common scale $2^{24}$, the directed row is
\[
E_{j,-}(x_1^h)=2^{-24}\max\{0,\,F^{[24]}_{j,-}(x_1^h)
-F^{[24]}_{j,+}(x_1^{h-1})-F^{[24]}_{j-1,+}(x_1^{h-1})\}.
\]
Here $F^{[24]}_{j,-}$ and $F^{[24]}_{j,+}$ are integer endpoints enclosing $2^{24}F_j$. The longer-prefix term has a positive sign and uses its lower endpoint. The two shorter-prefix terms are subtracted and use their upper endpoints. Their difference is therefore a lower bound on the exact cell. \Cref{lem:positive-cell} proves that the exact cell is nonnegative, which justifies the maximum with zero; dividing by $2^{24}$ returns to the original scale.

The implemented row is \texttt{nowlo[j]-previous[j]-previous[j-1]}, with \texttt{previous=nowhi}, in \certpath{kernels/lower/positive_cell_cache_batched_azure.py}. Exact integer accumulation and the checked count bounds prevent subtraction overflow. Source weights are then formed by downward path products, quantized separately at scale $2^{62}$ before exact state summation. The prefix induction in the preceding subsection proves that these weights remain dominated by the exact source probabilities. Multiplication by the cached nonnegative entropy contributions preserves the lower direction. Any probability lost by rounding is assigned contribution zero; this can weaken the lower bound but cannot raise it. Exact complement symmetry of the specified source justifies doubling the contribution of words starting in zero: complementing every bit pairs each such word with a word starting in one, with the same source probability and alignment contribution.

The current complete weighting has the same reviewed aggregate bounds: propagated prefix numerators total at most $2^{61}$ at each depth, and cached nonnegative coefficients are less than $2^{62}$. The preceding weighted-sum argument therefore bounds the complete result and every partial sum by $2^{123}$, within unsigned 128-bit range.

\subsection{Completion of the numerical verification}

Each family ends with a specified inequality. The small-deletion modules check every positive-matrix row and the analytic bound on all run lengths beyond the cutoff, using Arb enclosures of logarithms, exponentials and, in the first module, a square root. Ordinary renewal modules use exact rational source probabilities and directed entropy terms, together with the explicit bound on all omitted cases in which deleted runs join neighboring runs. The repeat-channel module combines a directed evaluation of the jigsaw information expression with nonnegative entropy contributions from disjoint blocks specified by its analytic construction. The shared and posterior upper modules check every local inequality specified in \cref{eq:shared-row,eq:posterior-row}. Finally, the transport assembly uses exact rational endpoints, verifies adjacent closed cells, and handles $d=0,1$ separately.

The verifier rejects malformed mathematical input, including duplicate JSON keys, approximate numerical tokens where exact values are required, nonfinite arrays, inconsistent source identities, missing objects, dimension mismatches, incomplete row coverage, illegal transport directions, and uncovered cells. These checks ensure that the numerical calculation applies to the objects and complete index sets named in the analytic theorem. Hashes identify the input objects and checker versions before and after a replay. Establishing an object's identity is separate from proving its inequalities, so a matching hash or a saved success label does not replace a completed replay.

These computations establish the numerical premises of a computer-assisted proof under the stated mathematical, integer-machine, MPFR/Arb, compiler, and directed-CUDA semantics. On the upper side, the analytic implication turns the checked local inequalities into bounds for every admissible input law. On the lower side, the finite entropy expressions give a rate for the one specified stationary source. In both arguments, the bounded end terms disappear after division by transmitted length. The exact parameter comparison then extends the resulting capacity bounds to every $d$. The finite size of a witness does not restrict the allowed coding-block length. It does not assert formal verification of the compiler or hardware. Independent literal-mask oracles, precision changes, and parser-corruption tests support implementation review; the numerical premises of the theorem remain the complete outward inequalities and their exact composition.

%% file: sections/reproduction.tex
\section{Computation, experiments, and reproduction}\label{sec:reproduction}
The analytic theorems reduce capacity bounds to explicit finite inequalities. This section identifies their numerical premises and separates three operations: evaluating the reported functions, checking their exact composition, and replaying the inequalities that established them.

\subsection{The boundary between a theorem and its numerical premises}
An analytic theorem first specifies a finite test. Numerical search chooses the probabilities and other parameters appearing in that test. A checker then evaluates the fixed choice with controlled rounding. Finally, exact arithmetic combines completed capacity bounds over the deletion parameter. Figure~\ref{fig:proofmap} separates these operations; \Cref{tab:compute-contracts} identifies the inputs and tests for each method.
\begin{table}[!htbp]\centering\small
\begin{tabular}{@{}>{\raggedright\arraybackslash}p{2.65cm}>{\raggedright\arraybackslash}p{4.8cm}>{\raggedright\arraybackslash}p{5.5cm}@{}}\toprule
Method and derivation & Chosen numerical inputs & What is actually checked\\\midrule
Shared outside sequence, \cref{sec:shared} & Output probabilities, a potential on raw-input windows and outside survivor words, inspected depth and retained masks & Every binary row of \eqref{eq:shared-row}, including the explicit omitted-pattern remainder. The maximum certified row gives an upper bound.\\
Posterior cancellation, \cref{sec:posterior} & Output and selector probabilities, finite potential and inspected depth & Every row of \eqref{eq:posterior-row}, after the entropy cancellation makes outside probabilities linear.\\
Small-deletion run converse, \cref{sec:small-deletion} & Run comparison probabilities, multipliers and truncation thresholds & The finite run inequalities and analytic conditions controlling all longer runs.\\
Finite-state input, \cref{sec:source-computation} & Transition probabilities, stationary law, output-history length and selected mask counts & Source normalization; lower output entropy; lower residual mask entropy; the arithmetic combination in \eqref{eq:lower-row}.\\
Renewal input, \cref{sec:renewal} & Run-length probabilities and retained finite run/output events & A lower run-group rate plus a lower conditional run-start entropy, with every subtracted tail bounded above.\\\bottomrule
\end{tabular}
\caption{The numerical work required by each analytic route. A probability table found by optimization is an input to a checker, not its conclusion. Every upper test covers all input choices specified by its theorem; every lower test uses one source throughout.}\label{tab:compute-contracts}
\end{table}
A finite maximum and a finite expectation require different checks. For a maximum, every row must be bounded; a few small sampled values cannot establish a converse. For an expectation with nonnegative summands, selected terms may be retained if their exact probability weights are used and overlap is excluded. This asymmetry is the reason selective counting is available to the lower bound while the upper bound needs complete coverage or a proved remainder.

\input{sections/computation_experiments}

\subsection{Why a finite computation proves a capacity bound}\label{sec:verification-meaning}
The analytic theorems specify finite tests whose successful completion implies a capacity bound. The checker establishes those hypotheses for fixed numerical inputs: all required rows and remainder bounds for an upper bound, or the weighted entropy and alignment terms of one fixed source for a lower bound. Optimization finds useful parameters; its convergence and optimality are not premises of the proof.

Controlled rounding makes the result of each numerical evaluation an interval known to contain the exact value. To prove that an expression is at most $u$, the checker requires the interval's upper endpoint to be at most $u$. For example, an enclosure $[0.1297,0.1298]$ proves an upper bound of $0.13$, whereas $[0.1297,0.1301]$ is inconclusive. To prove a lower bound, the checker instead uses a lower endpoint. These illustrative numbers describe the acceptance rule, not a capacity calculation. The arithmetic justifying the enclosures is given in \cref{sec:arithmetic}.

The finite computations supply pointwise capacity bounds at rational deletion probabilities. The transport theorem then extends those bounds to intervals of deletion probabilities. Exact rational arithmetic checks that the 4,931 closed intervals cover $[0,1]$, that each selected point bound is transported in its proved direction, and that the resulting enclosure has width below $0.019$. Its midpoint is consequently within $0.0095$ bits per transmitted bit of capacity. Thus the passage from finitely many calculations to every real deletion probability is an analytic argument followed by a finite coverage check.

\subsection{The complete release and its three checks}\label{sec:release-checks}
The public repository \url{https://github.com/anadim/binary-deletion-channel-capacity} contains the manuscript source, the exact certificate tables, the portable checkers, and a link to the complete release folder with the large numerical inputs. The complete release is organized as one folder containing the paper, its source and verification programs in \texttt{reproducibility.zip}, and 23 external proof files, including compressed numerical arrays and supporting records, under \texttt{proof-inputs/}. These files total 98,883,080,540 bytes, approximately 98.9 GB. They include the compressed parent calculation and the additional objects needed by its refinements. The small ZIP can be read and checked independently of the large downloads. The release also includes a file manifest and a single entry point, \texttt{verify\_release.py}. \Cref{tab:scopes} explains the conclusions of its checks.

\begin{table}[!htbp]\centering\small
\begin{tabular}{@{}>{\raggedright\arraybackslash}p{2.6cm}>{\raggedright\arraybackslash}p{6.0cm}>{\raggedright\arraybackslash}p{4.4cm}@{}}\toprule
Check & Conclusion & Resources\\\midrule
File integrity & Each inventoried file has the published size and SHA-256 digest. This identifies the numerical inputs; it does not prove an inequality. & Reads the complete release; no GPU.\\
Exact assembly & The included point-bound records imply the reported enclosure, coverage, midpoint, and error statistics. The recorded numerical premises are reused. & Small package and standard Python.\\
Numerical replay & Every required finite numerical inequality is recomputed, and the verified bounds are then combined over the full parameter range. & Complete inputs and the Linux/GPU environment below.\\\bottomrule
\end{tabular}
\caption{Three checks of the same release. File integrity identifies the supplied objects. Exact assembly checks their stated consequences. Numerical replay establishes the finite numerical premises again. The analytic proofs in this article explain why those premises bound capacity.}\label{tab:scopes}
\end{table}

A SHA-256 digest is a compact identifier computed from a file's bytes. Comparing a downloaded file with its published digest detects changes in transfer or storage. An unchanged digest does not imply that the file's assertions are true: that is the purpose of the numerical checker and the analytic proof. The manifest binds the manuscript, programs, and archived inputs to one release. It is kept separate from the original proof inventories, which remain unchanged.

\subsection{Running the checks}
Run the following commands from the complete release folder. The first checks the small files and the sizes of the archived inputs. The second checks the exact assembly of the capacity enclosure:
\begin{verbatim}
python3 verify_release.py quick
python3 verify_release.py light
\end{verbatim}
The wrapper uses a separate temporary directory for its programs, leaving the release inputs unchanged. To read and hash every supplied file, run
\begin{verbatim}
python3 verify_release.py files
\end{verbatim}
This reads approximately 99 GB. Its running time depends on the storage device; a successful result establishes file integrity only. A reader interested in a particular deletion probability can also query the exact bounds and midpoint:
\begin{verbatim}
python3 verify_release.py query 1/2
\end{verbatim}
These operations require Python 3.9 or later and no GPU. The query intersects the applicable enclosures at a shared endpoint and accounts for rounding in its displayed error bound.

For a full numerical replay, first prepare the archived inputs in a new directory on the compute machine. The example paths below are user-chosen locations outside the release folder:
\begin{verbatim}
python3 verify_release.py prepare --data /scratch/proof-data
python3 verify_release.py full --data /scratch/proof-data \
  --work /scratch/fresh-run --lock /scratch/bdc-gpu.lock
\end{verbatim}
\Needspace{11\baselineskip}
The preparation command checks the input identities and reconstructs the required directory layout. The full command runs the original numerical checkers and then the exact continuum calculation. Completion requires the status
\begin{quote}\small
\texttt{PASS\_FULL\_FRESH\_FINAL\_NUMERICAL\_AND\_CONTINUUM\_REPLAY}
\end{quote}
and a result record with \texttt{fresh\_numerics\_this\_invocation} set to \texttt{true}. A command that stops early, a stored earlier result, or the lightweight assembly check does not establish this completion status. Missing inputs and failed checks are reported as failures.

On 15 September 2026, all eight numerical components were replayed on the recorded Linux/A100 environment and combined by the original exact checker. Their saved outputs reproduce the same 4,931-cell enclosure and maximum width. The run began at 01:48 UTC and its final composition was recorded at 04:38 UTC, about two hours and fifty minutes later. It required one restart after six completed components: the 63 GB scratch disk had too little free space for the next reconstruction. After regenerable temporary arrays were removed, the remaining two components were run with the original arguments and all eight newly generated results were combined. No earlier component result was substituted. This establishes a completed numerical replay, but is not a successful uninterrupted execution on that disk.

The release preserves the outer component completion records, selected execution evidence, exact identities, timings, and a check of their consistency in \texttt{replay-20260915/}. The saved archive omits some intermediate outputs referenced by these records; a full replay regenerates them. The saved composition record uses its freshness flag for the resumed campaign; the final composition command itself reads the newly generated results. Rechecking the saved records establishes their binding and exact composition; only rerunning the numerical programs establishes another fresh replay.

\subsection{Execution environment and interpretation}
The recorded numerical environment is Linux with CUDA, an NVIDIA A100 with 80 GB device memory, and approximately 216 GB host RAM. The parent archive alone expands to approximately 188.3 GB. In addition to archived and extracted inputs, the original replay instructions require at least 64 GiB of scratch space. The complete dependency versions and commands are in \texttt{VERIFYING\_THE\_RESULT.md} and the reproducibility package. These describe a recorded working environment, not proved minimum resource requirements. The component times in \cref{sec:computation-experiments} measure the operations explicitly named there; they do not measure the search campaign or a newly timed installation from scratch.

The programs use integer arithmetic and outward numerical enclosures as specified in \cref{sec:arithmetic}. Their execution relies on Python and its numerical packages, the C++ and CUDA compilers, the interval libraries, and the processor instructions implementing the stated rounding operations. A numerical replay checks the inequalities through those implementations. It does not formally verify the analytic lemmas or the software and hardware in a proof assistant. This is the usual distinction between a computer-assisted proof and a machine-checked formalization of its entire derivation.

\subsection{Exact tables and algorithmic descriptions}
The large numerical inputs include auxiliary functions and conditional probability assignments for the upper bounds, together with finite source and alignment data for the lower bounds. They specify finite proof computations rather than samples from a channel experiment. A potential table may equally be specified by an algorithm that returns its exact entries. The finite inequalities and their analytic consequences are unchanged if those entries are reconstructed exactly.

The release already uses this principle. Shared objects are stored once, alignment counts for the lower bounds are regenerated by finite counting programs, and one posterior refinement is encoded as integer differences from parent tables. The checker reconstructs the integers and verifies their identities before evaluating the required rows. A shorter formula yielding different coefficients would also be admissible if the resulting inequalities passed at the claimed endpoints. Supplying only an optimization procedure would leave a further obligation: to show that it actually produces a passing finite certificate. Compact descriptions can reduce storage, but reconstruction time and complete numerical verification must still be accounted for.

%% file: sections/computation_experiments.tex
\subsection{Experiments that distinguish source design from bound evaluation}\label{sec:computation-experiments}
There are three ways to improve a computed lower bound: change the input process, improve its output-entropy estimate, or count more residual deletion uncertainty. The recorded experiments below separate these choices. Their timings describe particular completed jobs, not a reconstruction of the total research cost. The reproduction package supplies the exact source laws, finite counts, and full records behind these comparisons.

\paragraph{Changing the source while fixing the calculation.}
The source construction in \cref{sec:source-law} controls whether to end a run, conditional on the current run age and two completed-run classes. \Cref{tab:source-search} reports successive refinements at $d=3/5$. All use $m=13$, the depth-28 base rectangle with at most 12 survivors, and the same selected additional count coefficients. Each enlargement starts with the preceding input law represented exactly in the larger state space. Thus the comparison holds the entropy and count calculations fixed while permitting more source parameters.
\begin{center}\begin{minipage}{\linewidth}\centering\small
\begin{tabular}{@{}rrrrrl@{}}\toprule
States & Free probabilities & Certified lower bound & Evaluations & Search time & Termination\\\midrule
60 & 30 & $0.079296612362$ & 95 & $873.0$ s & Gradient tolerance\\
96 & 48 & $0.079394294360$ & 144 & $476.1$ s & Gradient tolerance\\
120 & 60 & $0.079466829830$ & 151 & $681.1$ s & Evaluation limit\\\bottomrule
\end{tabular}
\captionof{table}{Source refinement at deletion probability $3/5$. Values are rounded downward from separately certified rational results; search objectives themselves are floating-point proposals. Times include the recorded search job, not construction of the reusable count tables or later certification. A gradient stopping test does not prove a global optimum; the 120-state run did not even reach that stopping test.}\label{tab:source-search}
\end{minipage}\end{center}

The selected process uses its memory strongly. When a new run begins after two singleton runs, its probability of ending after its first bit is about $0.9053$. When the immediately preceding run has length two and the earlier run is a singleton, that probability is about $0.003919$. Both numbers are read from the same 120-state transition table, with complementary bits treated equally. They demonstrate dependence between neighboring run lengths. They do not prove that these particular probabilities are optimal or isolate the information gain caused by either transition.

\paragraph{Holding the source fixed and increasing the output history.}
For that same 120-state source at $d=3/5$, \cref{cor:birch-monotone} guarantees that a larger $m$ cannot weaken the exact Birch estimate. The computation enumerates $2^{m-1}$ binary histories for each initial state. \Cref{tab:birch-cost} reports the two interval checks.
\begin{center}\begin{minipage}{\linewidth}\centering\small
\begin{tabular}{@{}rrrrr@{}}\toprule
$m$ & Histories & State/history terms & $B_m$ (approximately) & Check time\\\midrule
13 & $4{,}096$ & $491{,}520$ & $0.725031271517$ & $10.7$ s\\
16 & $32{,}768$ & $3{,}932{,}160$ & $0.725039287627$ & $84.4$ s\\\bottomrule
\end{tabular}
\captionof{table}{Two full output-entropy checks for one fixed input law, both at 224-bit Arb precision. The number of terms increases eightfold. The lower bound on capacity increases by approximately $(1-d)(B_{16}-B_{13})=3.20644\times10^{-6}$ bit per input bit. These are output-entropy check times only; no deletion-pattern enumeration or source search is included.}\label{tab:birch-cost}
\end{minipage}\end{center}

This small gain for a substantially larger finite sum is a measured limitation of this refinement at this source and parameter. It is not a bound on the improvement available from a different source or converse.

\paragraph{Holding the source fixed and counting more ambiguity.}
After the output-history increase and intervening count refinements, the previously computed lower bound was approximately $0.079533471664$ at $d=3/5$. A subsequent extension using disjoint selected input words counted $8{,}060{,}928$ entries in 89 groups at $d=3/5$. Its recorded check took $131.6$ seconds and added approximately $0.000066669239$ bit per input bit to the inherited lower bound, giving the certified value $0.079600140902$. The additional terms retain their original probabilities and are disjoint from previously counted terms; \cref{lem:positive-cell} explains why omitted cases do not invalidate the bound. This timing is for the added component. It excludes the inherited source and count proofs.

\paragraph{A complete upper-bound calculation at the same parameter.}
The posterior certificate at $d=3/5$ uses $m=8$ surviving bits and $R=21$ inspected input bits. Its output table $Q$ has 256 entries, normalized in 128 pairs. Its potential has $2^{21}\cdot2^8=2^{29}$ integer entries; its deletion-decision table has the same number. Only one selector probability is free for each compatible output word: survival forces its first bit to equal the prepended bit, and normalization supplies the complementary deletion probability. Thus the selector needs $2^{R+m}$ stored entries, despite the two actions in the row test. Each of the latter arrays uses 4 GiB of 64-bit integer payload. These are comparison parameters for the all-input inequality, not a simulated input process.

For each of the $2^{21}$ input windows, the checker computes the grouped mask masses in \cref{sec:posterior-computation}. It then covers $2^8$ outside words and two prepended bits. The recorded coverage is therefore $2^{29}$ states and $2^{30}$ action rows. An upper endpoint for the largest row gives
\[
 C(3/5)\le\frac{3348235675749483}{36028797018963968}
 \approx0.0929322084772.
\]
The fraction is the certified upper value; the decimal is only explanatory. No assumption that a sampled input resembles a worst-case input enters this maximum.

The recorded reconstruction-and-replay job took $141.3$ seconds. It checked a streamed reconstruction of the compressed array changes and replayed the directed row computation, reusing already present arrays only after verifying that their bytes were identical. This is not a clean end-to-end download, decompression, and fresh-file-write benchmark. A separate earlier row-check receipt records $38.5$ seconds for the row checker alone, excluding reconstruction. It includes loading and identity checks as well as the row calculation. The reproduction package supplies both calculation records and the array dimensions. A separate calculation replayed eight upper points in $378.5$ seconds; the full inherited numerical replay took $5{,}014.9$ seconds. These jobs have different scopes and overlapping inherited work, so their times are not added into a claimed total.

\paragraph{The three contributions to the universal renewal bound.}
\Cref{cor:renewal-universal-assembly} combines three finite calculations for the same run-length law and Poisson intensity $\lambda=19/100$. \Cref{tab:poisson-components} shows their sizes and recorded evaluation times. The baseline uses merge-count cutoff $q=5$ and cutoff $K_0=192$ for the positive output entropy $H(K)$; the subtracted conditional count entropies use adaptive cutoffs and the explicit remainder terms of \cref{sec:renewal-baseline-computation}. The two corrections describe disjoint run groups of lengths $3$--$8$ and $9$--$16$, respectively. Their coefficients are already normalized by intensity; they must not be divided by $\lambda$ again.
\begin{center}\begin{minipage}{\linewidth}\centering\small
\begin{tabular}{@{}lrr@{}}\toprule
Contribution & Normalized coefficient (approximately) & Recorded time\\\midrule
Run-group baseline & $0.1217237467$ & $13.1$ s\\
Short-group correction, $3\le k\le8$ & $0.0030420710$ & $23.2$ s\\
Long-group correction, $9\le k\le16$ & $0.0003847524$ & $80.0$ s\\\midrule
Sum & $0.1251505701$ & ---\\\bottomrule
\end{tabular}
\captionof{table}{Components of the universal renewal coefficient. The exact lower endpoints in a recorded September 15 component replay sum to more than $0.12515$. The times concern these individual component evaluations; they come from a separate run from the GPU calculations above and do not measure an end-to-end replay. The first correction brings the coefficient close to the final value, while the longer groups add about $0.000385$.}\label{tab:poisson-components}
\end{minipage}\end{center}
The Poisson-transfer lemma therefore gives $C(d)\ge0.12515(1-d)$ for every deletion probability. This calculation explains what the additional conditional-entropy terms contribute beyond the run-group baseline; the smaller long-group contribution remains necessary for the particular coefficient stated here.

\paragraph{What is learned from these comparisons.}
Refining a source and evaluating it more accurately are distinct operations. The first changes the input distribution in both positive entropy terms of \eqref{eq:source-search-objective}. The second changes a lower estimate for a fixed distribution. The experiments provide evidence that both were useful, with different costs. They do not determine whether the remaining capacity uncertainty is primarily a loose converse or a suboptimal achievable rate. The decomposition in \cref{rem:source-score} states the unresolved quantities explicitly.

%% file: sections/discussion.tex
\section{Consequences and limitations}\label{sec:discussion}
\subsection{Accuracy beyond the worst case}
\Cref{tab:coverage} summarizes how much of the parameter interval meets each error tolerance.
\input{data/average_error.tex}
\Cref{fig:error-profile} gives the fraction of the deletion-probability interval covered at each certified error tolerance.
\begin{figure}[tb]\centering
\includegraphics[width=\linewidth]{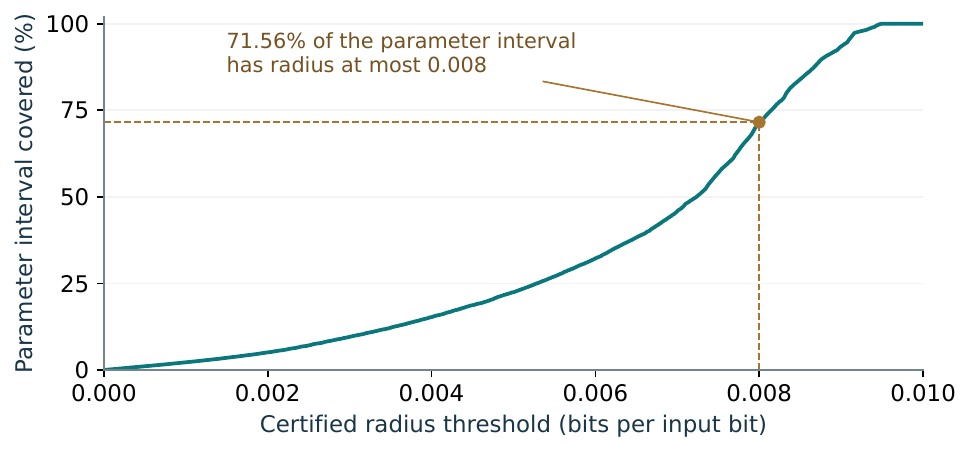}
\caption{Distribution of the certified approximation error over deletion probability. At horizontal value $r$, the curve gives the fraction of $d\in[0,1]$ for which $(U(d)-L(d))/2\le r$, with equal weight on all deletion probabilities. At least $71.56\%$ of the interval has radius at most $0.008$ bit per input bit; the whole interval has radius below $0.0095$. The fraction is computed from the linear radius on every certificate cell. Boundary intersections affect only finitely many points and therefore do not change these fractions.}\label{fig:error-profile}
\end{figure}
\FloatBarrier

\subsection{Comparison with published bounds}
\Cref{tab:literature} compares the reported bounds at the same deletion probabilities.
The comparison records printed values from primary sources at their stated precision. The present values are rounded outward from the exact certificate. Published computations supply historical comparisons; the new theorem follows from its own finite inequalities and completed numerical checks. In particular, the recent run-based lower-bound machinery overlaps with \citet{KD10,Chen26}, as identified in \cref{sec:intro-related,sec:renewal}; the claim here is the explicit uniform enclosure obtained with the specified constructions and parameters.
\input{data/literature_comparison.tex}
\FloatBarrier

\subsection{Remaining error and computational bottlenecks}\label{sec:limits}
The calculation stopped after certifying an interval width below $0.019$ for every deletion probability. It did not establish optimality of the auxiliary distributions, input sources, or truncation parameters. The largest remaining bound on interval width occurs on $[0.521,0.5215]$: a shared-tail upper bound and a finite-state lower bound determine the enclosure there. Thus the stopping point reflects an attained accuracy target, not a proved barrier to further improvement.

The finite-block comparison in \cref{lem:block} gives arbitrary accuracy in principle: increase $n$ until $\log_2(n+1)/n$ is small and evaluate the length-$n$ channel exactly. Its input alphabet already has $2^n$ elements. Our finite tests avoid that full capacity optimization, but their own sizes still limit straightforward enlargement. Table~\ref{tab:bottlenecks} separates the principal sources of conservatism and computational cost. These are limitations of the present constructions, not lower bounds on the complexity of approximating capacity.

\begin{table}[tb]\centering\small
\begin{tabular}{@{}>{\raggedright\arraybackslash}p{2.7cm}>{\raggedright\arraybackslash}p{5.0cm}>{\raggedright\arraybackslash}p{5.4cm}@{}}\toprule
Component & What can leave the bound loose & Cost or requirement of improving it\\\midrule
Shared-tail converse & Finite output history, worst-case outside values, and the explicitly bounded omitted probability & Increase input depth $R$, output order $m$, or retained-pattern coverage; all $2^{R+m-b+1}$ rows must still be covered.\\
Posterior converse & The chosen output and selector probabilities and the restricted finite potential & Optimize these parameters or enlarge the representation; the stated test has $2^{R+m+1}$ rows.\\
Finite-state lower bound & The selected source may be suboptimal; finite output conditioning and omitted nonnegative mask terms lose information & Improve the source law, deepen the output calculation, or count more disjoint events. All terms must use the same source.\\
Renewal lower bound & Independent run lengths restrict the source; finite run-window calculations leave a bounded remainder & Enlarge the run distribution or window, with valid normalization and complete tail control.\\
Parameter extension & A bound is reused away from its evaluation point & Add checked point bounds where the exact cell widths are largest; a denser plot alone changes nothing.\\\bottomrule
\end{tabular}
\caption{Where further improvement can enter the current proof. Here $R$ is the number of inspected input bits, $m$ the number of modeled survivors, and $b$ the retained survivor threshold in the shared-tail test. The row counts follow from enumerating the binary indices in \eqref{eq:shared-row} and \eqref{eq:posterior-row}.}\label{tab:bottlenecks}
\end{table}

One quantitative obstruction is visible before running a checker. An inspected window of $R$ input bits contains fewer than $m$ survivors with probability
\[
 \Pr\{\operatorname{Bin}(R,1-d)<m\}.
\]
The count is binomial because each position independently survives with probability $1-d$. Couple two possible outside sequences using the same deletions inside the window. Their first $m$ output bits agree whenever the window already contains $m$ survivors. Their output distributions can therefore differ only on the displayed event. Increasing $d$ makes a fixed window more likely to be exhausted; increasing $R$ reduces this probability but increases the state count exponentially. This probability measures sensitivity to the unseen input, not capacity error itself.

The lower bound has a different tradeoff. A selected length-$h$ input event with $j$ retained positions carries a factor $(1-d)^j d^{h-j}$ from its deletion mask. Counting such an event helps most where its probability is appreciable. Uncounted events remain valid omissions because their entropy contributions are nonnegative. Enlarging this calculation improves evaluation of a fixed source, but does not prove that the source is close to optimal. Similarly, an upper bound on one source's information rate is not an upper bound on capacity. The available certificate does not separate the remaining uncertainty into exact losses attributable to the converse, source choice, and source evaluation.

Two fair-input lower refinements are optional for the headline tolerance. The exact omission check in \cref{sec:family-simplification} removes both and retains the same uniform midpoint guarantee. They improve some local values and the mean radius, so their complete derivations remain in the appendices. This provides a shorter proof route without discarding the sharper published estimate.

\subsection{Mathematical and computational takeaways}
The enclosure establishes a numerical approximation of unrestricted capacity. It also distinguishes several mechanisms that a single final number would conceal.

\paragraph{Input dependence can be represented and tested concretely.}
The source refinements in \cref{sec:source-computation} allow switching probabilities to depend on more run ages while retaining the preceding source as an exact special case. Their certified improvements show that this extra flexibility was useful for the finite objective. They do not establish optimality of correlated runs or identify the exact capacity-achieving input process.

\paragraph{Finite expectations and universal maxima need different computation.}
The lower calculation can discard complete nonnegative contributions and reuse source-independent counts with new source weights. The upper calculation must cover every row or prove a bound on omitted rows. In both cases, the analytic sign or coverage argument is what licenses the finite computation. More GPU evaluations without that argument would not prove either claim.

\paragraph{Removing unknown probabilities can be more useful than enlarging a window.}
The posterior construction in \cref{sec:posterior} cancels a nonlinear entropy term before maximizing over outside input behavior. The shared-sequence construction in \cref{sec:shared} instead preserves common outside values while controlling omitted masks. These are two ways of reducing an infinite dependence to a finite all-input test. Their usefulness is demonstrated by completed upper certificates; the present data do not establish which is intrinsically stronger at equal computational cost.

\paragraph{The remaining uncertainty is not assigned to one side.}
The midpoint error bound gives no location of capacity inside the interval. It does not prove that the lower curve is nearly tight and the upper curve loose, or the converse. \Cref{rem:source-score} separates source evaluation from source quality; an additional bound on a fixed source's information rate would help diagnose that distinction, but would still not upper-bound unrestricted capacity. The small steps in the lower curve have a separate, completely identified cause: point certificates cease to be eligible under leftward transport, as proved in \cref{lem:certificate-steps}.

The transferable ideas are the finite sufficient tests, common-outside-sequence constraints, entropy cancellation, and reusable nonnegative counts. Their use for another synchronization channel requires that channel's own identities and remainder bounds. The present theorem concerns independent binary deletions; it does not by itself cover insertions, substitutions, correlated errors, or multiple reads \citep{Mit09,MBT10,CR21}.

%% file: data/average_error.tex
For each deletion probability $d$, the lower and upper functions enclose capacity. Define their half-distance by $r(d)=(U(d)-L(d))/2$. The midpoint inequality in \cref{thm:main} says that the actual absolute error is at most $r(d)$, and that $r(d)\le\eps_\star$ everywhere. To average this statement, specify a probability measure $\mu$ on $[0,1]$. Integrating inequalities preserves their direction, and the integral of the constant $\eps_\star$ is $\eps_\star$ because $\mu$ has total mass one. Therefore
\[
\int |\widehat C(d)-C(d)|\,\mathrm d\mu(d)\le\int r(d)\,\mathrm d\mu(d)\le\eps_\star.
\]
For the numerical mean reported in this paper, $\mu$ is the uniform distribution on deletion probability, or Lebesgue measure on $[0,1]$. On the interior of a cell $[x,y]$, the enclosure has lower value $(1-d)B$ and upper value $(1-d)A$. Subtracting and dividing by two gives $r(d)=(A-B)(1-d)/2$. The coefficients are constant on the cell, so they can be taken outside the integral. An antiderivative of $1-d$ is $d-d^2/2$; evaluating it at $y$ and $x$ gives
\[
\int_x^y r(d)\,\mathrm dd=\frac{A-B}{2}
\left[(y-x)-\frac{y^2-x^2}{2}\right].
\]
Summing these rational terms over all 4,931 cells gives the following bound, in bits per input bit:
\begin{equation}\label{eq:mean-error}
\boxed{\quad\int_0^1|\widehat C(d)-C(d)|\,\mathrm dd
\le\int_0^1r(d)\,\mathrm dd
\le0.006521572012\quad}
\end{equation}
The exact sum is retained in \texttt{data/AVERAGE\_ERROR.json}; its decimal expansion starts $0.006521572011751958\ldots$. Intersecting the bounds at shared endpoints changes a measure-zero set and does not change this integral. The computed value is the average \emph{certified radius}, hence an upper bound on average absolute error. It is not a measurement of the actual error.

The median certified radius belongs to $(0.007231523565,\,0.007231523566]$. This interval is verified by exact comparisons of the parameter measure with $1/2$. To compute a coverage probability at threshold $t$, solve the inequality $r(d)\le t$ separately on each cell. When $A>B$, substituting its formula gives $(1-d)(A-B)/2\le t$. Multiplying by $2/(A-B)$ and rearranging gives $d\ge1-2t/(A-B)$. Thus the part of the cell meeting the threshold is
\[
[x,y]\cap\left[1-\frac{2t}{A-B},1\right]
\]
and its length is its uniform probability. When $A=B$, the radius is zero throughout the cell, so the entire cell meets any nonnegative threshold. Adding the lengths of these disjoint cell interiors gives the coverage percentages in \cref{tab:coverage}; finitely many shared endpoints have zero uniform probability.
\begin{table}[tb]\centering
\begin{tabular}{@{}rr@{}}\toprule
Certified radius at most & Fraction of deletion parameters (at least)\\\midrule
$0.0050$ & $22.4604\%$\\
$0.0075$ & $56.8839\%$\\
$0.0080$ & $71.5630\%$\\
$0.0090$ & $93.4404\%$\\
$0.0100$ & $100\%$\\\bottomrule
\end{tabular}
\caption{Uniform-parameter coverage of tighter approximation tolerances. Percentages are rounded downward; all exact measures are supplied in the data. These are measures of the certificate's local error guarantee.}\label{tab:coverage}
\end{table}

%% file: data/literature_comparison.tex
\begin{table}[!htbp]\centering\small
\begin{tabular}{@{}rrrrrr@{}}\toprule
&\multicolumn{3}{c}{Reported literature comparison}&\multicolumn{2}{c}{Current certificate}\\
$d$ & Lower & Upper & Midpoint tolerance & Midpoint & Error at most\\\midrule
0.05 & 0.729100 & 0.8039 & 0.037400 & 0.733147060 & 0.003311547\\
0.1 & 0.563800 & 0.6577 & 0.046950 & 0.569524720 & 0.002918167\\
0.2 & 0.348200 & 0.4574 & 0.054600 & 0.362013501 & 0.005265031\\
0.3 & 0.222500 & 0.3314 & 0.054450 & 0.244937527 & 0.005902107\\
0.4 & 0.149810 & 0.2480 & 0.049095 & 0.170863553 & 0.007353678\\
0.5 & 0.110324 & 0.1896 & 0.039638 & 0.122351098 & 0.007810285\\
0.6 & 0.071838 & 0.1438 & 0.035981 & 0.086266174 & 0.006666035\\
0.64 & 0.060698 & 0.1288 & 0.034051 & 0.075564854 & 0.008074134\\
\bottomrule\end{tabular}
\caption{Primary-table comparison. Lower values use Venkataramanan--Tatikonda--Ramchandran and Rubinstein--Con, with Chen\textquotesingle s reported bound at $d=1/2$~\cite{Chen26}; upper values use Pinto--Ribeiro. The lower at $0.64$ transports Rubinstein--Con\textquotesingle s $0.65$ value. Literature bounds are displayed to the shown precision; the lower values at $0.5$ and $0.64$ are rounded down to six decimal places. These are not independently replayed enclosures. Their midpoint tolerances are arithmetic on those reported values. Current errors include rounding of the displayed estimate.}\label{tab:literature}
\end{table}

%% file: sections/standard.tex
\section{Elementary information inequalities}\label{sec:standard}
We collect the information inequalities used in the proofs, including their directions under conditioning and probability comparison. The entropy results are classical; see \citet[Chapters~2 and~4]{CT06}.

Entropy identities involving subtraction are stated for finite-valued variables, or more generally when the displayed entropies are finite. Conditioning on an infinite input sequence means averaging the conditional probabilities of the variable being described; it does not require the input sequence itself to have finite entropy.

\subsection{Entropy, conditional entropy, and mutual information}
Let $A$ have probabilities $p(a)$. Its entropy is the expected negative log probability:
\[
H(A)=\sum_{a:p(a)>0}p(a)\log_2\frac1{p(a)}.
\]
Zero-probability terms contribute zero. Conditional entropy averages the entropies of $p(a\mid b)$ with weights $p(b)$. Mutual information is
\[
I(A;B)=H(A)-H(A\mid B).
\]
Conditionally, $I(A;B\mid G)=H(A\mid G)-H(A\mid B,G)$. Logarithms are base two and entropies are in bits. Write $A_1^k=(A_1,\ldots,A_k)$, empty when $k=0$.

\begin{lemma}[Chain rule]\label{lem:entropy-chain}
The uncertainty of a pair can be evaluated by describing its first entry and then its second:
\[
H(A,B\mid G)=H(A\mid G)+H(B\mid A,G).
\]
Consequently $I(A;B)=I(B;A)$ and
\[
I(A;B,G)=I(A;B)+I(A;G\mid B).
\]
\end{lemma}
\begin{proof}
Factor $p(a,b\mid g)=p(a\mid g)p(b\mid a,g)$, take negative logarithms, and average. This proves the entropy identity. Expanding $H(A,B)$ in both orders and rearranging proves symmetry. Adding and subtracting $H(A\mid B)$ in $H(A)-H(A\mid B,G)$ proves the information identity.
\end{proof}
Iteration gives $H(A_1^k\mid G)=\sum_{j=1}^kH(A_j\mid A_1^{j-1},G)$.

\subsection{Nonnegativity and the effect of side information}
For distributions $p,q$ with $q>0$ on the support of $p$, define relative entropy
\[
D(p\Vert q)=\sum_{a:p(a)>0}p(a)\log_2\frac{p(a)}{q(a)}.
\]
The function $t-1-\ln t$ has derivative $1-1/t$ and minimum zero at $t=1$, so $-\ln t\ge1-t$. Substitute $t=q(a)/p(a)$, multiply by $p(a)$, and sum:
\[
(\ln2)D(p\Vert q)\ge\sum_{a:p(a)>0}(p(a)-q(a))\ge0.
\]
The last inequality holds because $p$ sums to one and $q$ over a subset sums to at most one. This is Gibbs' inequality.

\begin{lemma}[Conditioning and side information]\label{lem:side-information}
Extra information cannot increase conditional entropy. Moreover, revealing a discrete variable $G$ increases mutual information by at most its entropy:
\[
H(A\mid B,G)\le H(A\mid G),\qquad
I(X;Y,G)-H(G)\le I(X;Y)\le I(X;Y,G).
\]
If $G$ is independent of $X$, then $I(X;Y,G)=I(X;Y\mid G)$.
\end{lemma}
\begin{proof}
For each $g$, expanding $D(p(a,b\mid g)\Vert p(a\mid g)p(b\mid g))$ gives $I(A;B\mid G=g)$. Gibbs' inequality makes it nonnegative. Average over $g$ and use the definition of conditional mutual information.

The chain rule gives $I(X;Y,G)-I(X;Y)=I(X;G\mid Y)\ge0$. By symmetry and nonnegative conditional entropy, this difference is $H(G\mid Y)-H(G\mid X,Y)\le H(G\mid Y)\le H(G)$. Finally, expansion in the other order gives $I(X;Y,G)=I(X;G)+I(X;Y\mid G)$, and independence makes its first term zero.
\end{proof}
Supplying extra information is often called a \emph{genie}. If it has at most $M$ possible values, its entropy is at most $\log_2M$: compare with uniform masses $1/M$ in Gibbs' inequality. Thus revealing it increases mutual information by at most $\log_2M$.

\begin{lemma}[Data processing]\label{lem:data-processing}
If, conditional on $B$, the variable $C$ is independent of $A$, then $I(A;C)\le I(A;B)$.
\end{lemma}
\begin{proof}
The two chain-rule expansions give $I(A;B,C)=I(A;B)$, because $I(A;C\mid B)=0$, and $I(A;B,C)=I(A;C)+I(A;B\mid C)\ge I(A;C)$. Equate them.
\end{proof}
The same proof holds with an additional variable conditioned on throughout.

\begin{lemma}[Erasure upper bound]\label{lem:erasure}
For every $0\le d\le1$, $C(d)\le1-d$.
\end{lemma}
\begin{proof}
Let $D_1^n$ be the independent deletion mask and $Y_n$ the trace of an arbitrary input $X_1^n$. Revealing the mask can only increase mutual information. Its independence of the input, followed by determinism of the trace given input and mask, gives
\[
I(X_1^n;Y_n)\le I(X_1^n;Y_n,D_1^n)
=I(X_1^n;Y_n\mid D_1^n)=H(Y_n\mid D_1^n).
\]
A fixed mask retaining $k$ bits leaves at most $2^k$ possible traces, so its conditional entropy is at most $k$. Averaging over masks yields $H(Y_n\mid D_1^n)\le\E N_n=n(1-d)$. Maximize over input laws, divide by $n$, and use \eqref{eq:capacity-limit}.
\end{proof}

\subsection{Probability codes and mixtures}
\begin{lemma}[Entropy is bounded by probability-code length]\label{lem:cross-entropy}
Let $A$ have law $p$, and let $q$ be any normalized probability law positive on the support of $p$. Then
\[
H(A)\le \E[-\log_2q(A)].
\]
This also holds conditionally, with a normalized law $q(\cdot\mid B=b)$ for each known value $b$.
\end{lemma}
\begin{proof}
Subtract $H(A)$ and combine logarithms: the difference is $\sum_a p(a)\log_2[p(a)/q(a)]=D(p\Vert q)\ge0$. Conditionally, apply this argument at each $b$ and average.
\end{proof}
The true distribution $p$ supplies the expectation; the chosen comparison $q$ supplies the logarithmic loss $-\log_2q(a)$, also called ideal code length. This is an entropy bound, not a decoding algorithm.

For example, a fair bit has entropy one. With $q(0)=3/4,q(1)=1/4$, its expected loss is $\tfrac12\log_2(4/3)+\tfrac12\log_24>1$. The weights remain the true probabilities $1/2,1/2$.

For a string, multiply normalized conditional table entries along its symbols. Summing these products over the last symbol gives one for that row, then summing successively backwards proves total normalization. This justifies sequential comparison distributions.

\begin{lemma}[Entropy of a selected distribution]\label{lem:mixture}
Suppose a random selector $J$ is drawn first, and then $Z$ is drawn from a distribution depending on $J$. For any normalized conditional code $K(j\mid z)$ positive on the possible pairs,
\[
H(Z)=H(Z\mid J)+H(J)-H(J\mid Z),\qquad
I(J;Z)\ge H(J)-\E[-\log_2K(J\mid Z)].
\]
\end{lemma}
\begin{proof}
Expand $H(J,Z)$ in both orders and rearrange. Conditional cross entropy bounds $H(J\mid Z)$ by the expected logarithmic loss of $K$; subtracting from $H(J)$ gives the inequality.
\end{proof}
For each observed $z$, the \emph{posterior code} $K(j\mid z)$ assigns normalized probabilities to the possible selector values. It need not be the true posterior or recover the selector. For deletion, $J$ can indicate whether a specified input bit survives.

\subsection{Combining nonnegative masses}
\begin{lemma}[Log-sum inequality and homogeneous entropy]\label{lem:logsum}
For nonnegative lists $(a_i),(b_i)$, the log-sum inequality is
\[
\sum_i a_i\log_2\frac{a_i}{b_i}\ge
\left(\sum_i a_i\right)\log_2\frac{\sum_i a_i}{\sum_i b_i}.
\]
A positive $a_i$ with $b_i=0$ gives an infinite left side; zero terms use continuity. Define the binary homogeneous entropy
\[
\Phi(a,b)=(a+b)\log_2(a+b)-a\log_2a-b\log_2b.
\]
It is nonnegative, at most $a+b$, and satisfies $\Phi(v+w)\ge\Phi(v)+\Phi(w)$ for nonnegative two-dimensional vectors $v,w$.
\end{lemma}
\begin{proof}
For totals $A=\sum_i a_i>0$ and $B=\sum_i b_i>0$, the difference between the log-sum sides is $A D((a_i/A)_i\Vert(b_i/B)_i)$, by expanding logarithms. Gibbs' inequality proves nonnegativity; zero cases use continuity.

Writing $r=a/(a+b)$ and $\hb(r)=-r\log_2r-(1-r)\log_2(1-r)$ gives $\Phi(a,b)=(a+b)\hb(r)$. This proves the bounds since binary entropy lies in $[0,1]$. Moreover, $\hb''(r)=-1/[(\ln2)r(1-r)]<0$, so binary entropy is concave. Combining mass pairs averages their proportions with weights equal to their total masses. Apply concavity and multiply by the combined mass to obtain $\Phi(v+w)\ge\Phi(v)+\Phi(w)$. Continuity handles zero masses.
\end{proof}
Scaling the masses by $c\ge0$ scales $\Phi$ by $c$: the $\log_2c$ terms cancel. For any finite mass list, homogeneous entropy is total mass times the entropy of its normalized proportions. Thus two retained pattern classes of masses $a,b$ contribute $\Phi(a,b)$, including their observation mass $a+b$; omitted patterns do not justify renormalizing it.

\subsection{Stationarity and entropy rates}
A process is stationary when every finite block's distribution is independent of its starting position. Symbols may remain dependent at arbitrarily large separations.

\begin{lemma}[Entropy rate from longer observed pasts]\label{lem:entropy-rate}
For a stationary finite-alphabet process $Z$, the numbers $H(Z_0\mid Z_{-k}^{-1})$ decrease as $k$ increases, and
\[
h(Z):=\lim_{n\to\infty}\frac1nH(Z_1^n)
=\lim_{k\to\infty}H(Z_0\mid Z_{-k}^{-1}).
\]
\end{lemma}
\begin{proof}
Conditioning on another past symbol decreases entropy; nonnegativity gives a limit. By the chain rule and stationarity, $H(Z_1^n)$ sums the first $n$ conditional entropies $H(Z_0\mid Z_{-(j-1)}^{-1})$. Their averages have the same limit: finitely many initial terms have vanishing weight, while all subsequent terms approach that limit.
\end{proof}
Finite-past conditional entropy therefore upper-bounds the entropy rate. The Birch lower bound instead reveals a source's internal state (\cref{lem:birch}). Conditioning on an infinite input means taking the decreasing limit over finite input prefixes.

\subsection{Potentials and random-length prefixes}
We seek to bound the average of a real function $r(w,a)$ over transitions. A potential is a bounded real function of the state. Adding its expected one-step change leaves stationary averages unchanged. In the channel application, a transition index $a$ is a raw input bit.

\begin{lemma}[Cancellation of a bounded potential]\label{lem:potential}
Fix transition kernels $P_a(w,\mathrm dw')$, each specifying a probability law for the next state, and a bounded measurable potential $V$, and write $(P_aV)(w)=\int V(w')P_a(w,\mathrm dw')$. Let $(W_i,A_i)$ be any process satisfying
\[
\E[V(W_{i+1})\mid W_i,A_i]=(P_{A_i}V)(W_i)
\quad\text{almost surely for every }i.
\]
If an integrable reward $r$ satisfies the row inequality
\[
r(w,a)+(P_aV)(w)-V(w)\le u
\]
for every permitted pair $(w,a)$, then stationary state marginals give $\E r(W_i,A_i)\le u$. The same inequality bounds nonstationary expected averages up to a boundary term $\osc(V)/n$, where $\osc(V)=\sup V-\inf V$.
\end{lemma}
\begin{proof}
Average each row inequality. Conditional expectation replaces $(P_{A_i}V)(W_i)$ by $V(W_{i+1})$ in expectation. Sum over $i$; intermediate potentials cancel:
\[
\sum_{i=1}^n\E r(W_i,A_i)\le nu+\E V(W_1)-\E V(W_{n+1}).
\]
Stationarity makes the endpoint expectations equal. Otherwise their difference is at most $\sup V-\inf V$; divide by $n$.
\end{proof}
The conditional-expectation hypothesis links the fixed numerical operator to the actual process. It does not require a Markov input or history-independent actions. For alternating states with rewards $2,0$, choose potentials $1,0$. The adjusted rewards are $2+0-1=1$ and $0+1-0=1$. The corrections cancel over each cycle, reducing the row bound from two to the true average one.

\begin{lemma}[The cost of a short observed prefix]\label{lem:random-prefix}
Suppose $Z_1,Z_2,\ldots$ are binary, and $Y=(Z_1,\ldots,Z_N)$ is a random prefix whose length $N$ is readable from $Y$. For any additional information $F$ and fixed integer $k\ge1$,
\[
H(Z_1^k\mid F)\le H(Y\mid F)+k\Pr(N<k).
\]
No independence between $N,Z$, and $F$ is required.
\end{lemma}
\begin{proof}
Add $Y$ to the described variables, then apply the chain rule:
\[
H(Z_1^k\mid F)\le H(Z_1^k,Y\mid F)
=H(Y\mid F)+H(Z_1^k\mid Y,F).
\]
Fix an observed word $Y=y$. If $|y|\ge k$, its prefix determines $Z_1^k$ and the remaining entropy is zero. Otherwise there are at most $2^k$ binary words, giving entropy at most $k$. Averaging yields $H(Z_1^k\mid Y,F)\le k\Pr(N<k)$. No extra event indicator is needed: $y$ reveals its length.
\end{proof}
For $N\sim\operatorname{Bin}(n,1-d)$, independent survival indicators give mean $n(1-d)$ and variance $nd(1-d)$. Markov's inequality applied to the squared deviation gives
\[
\Pr\{|N-n(1-d)|\ge t\}\le\frac{nd(1-d)}{t^2}.
\]
Taking $t=n^{2/3}$ makes this Chebyshev bound $d(1-d)n^{-1/3}\to0$. A prefix this far below the mean length is therefore missing with vanishing probability.

\subsection{From information bounds to reliable communication}
Let $W$ be uniform on $M$ messages, and let a decoder observing $Y$ have error probability $P_e$. Reveal the error indicator, of entropy at most one. The message is then determined on success and has at most $M-1$ possibilities on error. Hence $H(W\mid Y)\le1+P_e\log_2M$. Subtract from $H(W)=\log_2M$ and rearrange to obtain Fano's bound:
\[
(1-P_e)\log_2M\le I(W;Y)+1.
\]
Dividing by input length bounds every code sequence whose error tends to zero. Identifying the limiting maximum mutual information with operational capacity additionally uses the synchronization-channel coding theorem in \cref{sec:setting}.

\subsection{The finite-block capacity benchmark}
Here $C_n(d)$ is the maximum mutual information per input bit over all distributions on $n$-bit words; $C(d)$ is asymptotic capacity. The output is the variable-length trace. The difference is bounded by the information in artificial block boundaries.
\begin{lemma}[Finite-block capacity comparison]\label{lem:block}
For every $d\in[0,1]$ and integer $n\ge1$,
\[
C_n(d)-\frac{\log_2(n+1)}n\le C(d)\le C_n(d).
\]
\end{lemma}
\begin{proof}
Split a length-$kn$ input into $k$ blocks of length $n$, with separate traces $V_j$. Revealing these traces gives at least the information in their concatenation. Although the input blocks may be dependent, $H(V_1,\ldots,V_k)\le\sum_jH(V_j)$. Given the input, independent block deletions make conditional output entropies add. Subtracting them yields
\[
I(X_1^{kn};Y_{kn})\le I(X_1^{kn};V_1,\ldots,V_k)
\le\sum_{j=1}^k I(X_{(j-1)n+1}^{jn};V_j)
\le knC_n(d).
\]
Each block marginal is bounded by $nC_n(d)$. Maximizing over the length-$kn$ input and dividing by $kn$ gives $C_{kn}(d)\le C_n(d)$. Take $k\to\infty$ using \eqref{eq:capacity-limit}.

For the lower bound, concatenate $k$ independent blocks drawn from a maximizing length-$n$ distribution. Their separate mutual informations add to $knC_n(d)$. The vector of trace lengths identifies the block boundaries and has at most $(n+1)^k$ values, hence entropy at most $k\log_2(n+1)$. Removing it loses at most this entropy by \cref{lem:side-information}:
\[
I(X_1^{kn};Y_{kn})\ge knC_n(d)-k\log_2(n+1).
\]
The maximum defining $C_{kn}(d)$ is at least this source's value. Divide by $kn$ and take $k\to\infty$.
\end{proof}

For these independent blocks, the length vector actually has entropy $kH(N_n)$, where $N_n\sim\operatorname{Bin}(n,1-d)$. Using this value instead gives $C_n(d)-H(N_n)/n\le C(d)$, with the same upper bound \citep[Section~VIII]{FD10}.

%% file: sections/fair_input.tex
\section{Fair-input refinements near zero}\label{sec:fair-input}
This section improves the lower bound for independent fair input bits when the deletion probability $d$ is small. The refinements enter the reported pointwise enclosure and mean certified radius, although \cref{sec:family-simplification} proves the worst-case tolerance without them. We first lower-bound $H(D_1^n\mid X_1^n,Y_n)$ by conditioning on the deletion count in each constant input run. We then retain one-deletion terms to obtain an explicit estimate near $d=0$.

\subsection{An achievable rate from independent fair input bits}
Conditional on its length, the trace of independent fair input is a uniform binary word. Its entropy rate is therefore $1-d$ per input bit. Subtracting the full mask entropy $h_2(d)$ gives the baseline $(1-d)-h_2(d)$. The exact mutual-information identity adds $n^{-1}H(D_1^n\mid X_1^n,Y_n)$ to this calculation. We improve the baseline by lower-bounding that conditional entropy.

A run is a maximal constant input segment. If a run of length $\ell$ has $j$ deleted positions, each of the $\binom\ell j$ location sets produces the same surviving symbols. Conditional on $j$, these sets are equally likely and have entropy $\log\binom\ell j$. The input and all per-run deletion counts determine the full trace, even when complete runs vanish. Conditioning on those counts therefore leaves a conditional mask entropy that can be evaluated within each run and is no larger than $H(D_1^n\mid X_1^n,Y_n)$.

\begin{lemma}[Fair-input run-count lower bound]
\label{lem:iid-run-lower}
For every $0\le d\le1$,
\begin{equation}
C(d)\ge (1-d)-h_2(d)+\frac12\sum_{\ell\ge1}2^{-\ell}
\sum_{j=1}^{\ell-1}\binom\ell j d^j(1-d)^{\ell-j}\log\binom\ell j.
\label{eq:iid-run-full}
\end{equation}
Every finite subset of the summands is also a valid lower bound.
\end{lemma}
\begin{proof}
\textbf{Step 1: evaluate the output entropy for this source.}
Take independent fair input bits $X_1^n$ and let $Y_n$ be their trace. Conditional on its length, the trace is a uniform binary word. The entropy chain rule therefore gives $H(Y_n)=n(1-d)+H(|Y_n|)$.

\textbf{Step 2: retain only within-run deletion ambiguity.}
The mask $D_1^n$ records which positions are deleted and is independent of the input, with entropy $nh_2(d)$. The conditional output entropy is smaller by $H(D_1^n\mid X_1^n,Y_n)$, the uncertainty in the mask after input and trace are known. Expanding $H(D_1^n,Y_n\mid X_1^n)$ in its two orders gives
\[
I(X_1^n;Y_n)=H(Y_n)-nh_2(d)+H(D_1^n\mid X_1^n,Y_n).
\]
Let $S$ contain the number of survivors in every input run. By \cref{lem:side-information}, $H(D_1^n\mid X_1^n,Y_n)\ge H(D_1^n\mid X_1^n,Y_n,S)$. Since $X_1^n$ and $S$ determine $Y_n$, the latter entropy equals $H(D_1^n\mid X_1^n,S)$. Conditional on the counts, the survivor locations are independent between runs and uniform within each run. Their entropies therefore add, with the per-run value calculated above.

\textbf{Step 3: count this contribution per transmitted bit.}
The conditional entropy is a sum over runs. Dividing by the number of input bits requires the expected number of length-$\ell$ runs per input bit, which we calculate next.

A specified run of length $\ell$ has $j$ deletions with probability $\binom\ell j d^j(1-d)^{\ell-j}$. The factor counts the possible masks, each with the stated product probability. The cases $j=0$ and $j=\ell$ contribute zero entropy, explaining the inner summation range. A length-$\ell$ run begins at a specified interior input position when the preceding bit changes, the next $\ell-1$ comparisons agree, and the following comparison changes. For fair independent bits this event has probability $2^{-\ell-1}$. That is the number of such runs per input bit in expectation, and it explains the outer coefficient.

First retain only finitely many run lengths. Their expected counts divided by $n$ converge to these coefficients. The two boundary runs have geometric length tails and bounded expected lengths, so their descriptions contribute zero after division by $n$. Also $H(|Y_n|)\le\log_2(n+1)$, which vanishes per input bit. The capacity limit in \eqref{eq:capacity-limit} now proves the inequality with the finite sum. Increasing the retained set preserves the lower bound because all its terms are nonnegative, and yields the displayed series.
\end{proof}

The numerical certificate needs a lower bound on a whole parameter interval, so its left endpoint cannot simply be substituted into the preceding expression. Fix $[a,b]\subset[0,1/2]$. It is sufficient to replace the baseline term and every retained positive summand by a value no larger than that term anywhere in the interval. The derivative of $(1-d)-h_2(d)$ is $-1-\log_2((1-d)/d)$, which is negative on the interior of this range. For $1\le j<\ell$, the logarithmic derivative of $d^j(1-d)^{\ell-j}$ is $(j-\ell d)/(d(1-d))$. Thus this weight first increases and then decreases, and its minimum on any closed interval is attained at an endpoint. Truncating at run length $N$ gives
\begin{align}
L_{a,b}^{(N)}={}&(1-b)-h_2(b)+\frac12\sum_{\ell=2}^{N}2^{-\ell}
\sum_{j=1}^{\ell-1}\binom\ell j\log\binom\ell j\notag\\
&\hspace{20mm}\cdot\min\{a^j(1-a)^{\ell-j},b^j(1-b)^{\ell-j}\}.
\label{eq:iid-cell-lower}
\end{align}
Each retained term is at most its value at every $d\in[a,b]$, so $L_{a,b}^{(N)}\le C(d)$ throughout the interval.

\subsection{An explicit polynomial remainder near zero}
Near $d=0$, we need an explicit lower expression to compare with the erasure upper bound $1-d$. Retaining only runs with one deleted position gives a linear term whose decrease with $d$ can be bounded by a quadratic term. The following calculation gives numerical bounds on both coefficients; it leaves no unspecified asymptotic remainder.

\textbf{Step 1: bound the one-deletion terms by a quadratic expression.}
Keeping only the terms with one deletion in \eqref{eq:iid-run-full} gives $d\ell(1-d)^{\ell-1}\log\ell$ for a run of length $\ell$. The elementary inequality $(1-d)^k\ge1-kd$ follows by induction: multiplying the bound for $k$ by $1-d$ gives $1-(k+1)d+kd^2\ge1-(k+1)d$. Applying it with $k=\ell-1$ yields
\[
C(d)\ge(1-d)-h_2(d)+dA-d^2B,
\]
where the two constants collect the linear and quadratic terms:
\[
A=\frac12\sum_{\ell\ge1}\ell2^{-\ell}\log\ell,
\qquad B=\frac12\sum_{\ell\ge1}\ell(\ell-1)2^{-\ell}\log\ell.
\]
\textbf{Step 2: bound the two coefficients in the required directions.}
Because $A$ is added and $B$ is subtracted, it suffices to prove a lower estimate for $A$ and an upper estimate for $B$. Both estimates below use finite calculations or explicit convergent sums.

All weights are nonnegative, so a finite partial sum bounds $A$ from below. Let $m_\ell$ be the integer determined by $2^{m_\ell}\le\ell^{256}<2^{m_\ell+1}$. Then $\log_2\ell\ge m_\ell/256$. Using the first twenty terms and a common denominator gives
\[
A\ge\sum_{\ell=1}^{20}\frac{\ell m_\ell}{512\,2^\ell}
=\frac{345626757}{268435456}>\frac{32}{25}.
\]
This is a finite integer calculation: compute the twenty $m_\ell$ by comparing integer powers, place the summands over denominator $512\,2^{20}$, and add their numerators. The numerator before cancellation is $691253514$.

For the other constant, $\log_2\ell\le(\ell+1)/2$. Equivalently $\ell^2\le2^{\ell+1}$; check $\ell=1,2,3$, then note that $((\ell+1)/\ell)^2<2$ for $\ell\ge3$ to continue by induction. Hence
\[
B\le\frac14\sum_{\ell\ge1}(\ell^3-\ell)2^{-\ell}
=\frac{26-2}{4}=6.
\]
The numbers $2$ and $26$ are obtained by applying the operation $x\frac{d}{dx}$ once and three times to $\sum_{\ell\ge0}x^\ell=(1-x)^{-1}$ and then setting $x=1/2$. The successive expressions are $x/(1-x)^2$, then $x(1+x)/(1-x)^3$, then $x(1+4x+x^2)/(1-x)^4$; each follows by differentiating the preceding expression and multiplying by $x$. Their first and third values at $x=1/2$ are $2$ and $26$.

\textbf{Step 3: put the binary-entropy term in the same form.}
The remaining term is $(1-d)-h_2(d)$. Its logarithmic singularity at zero will be retained explicitly as $d\log_2d$; only its linear remainder needs a numerical bound.

Two elementary logarithm bounds finish the estimate. Substitution in the integral for $\ln2$ gives $\ln2=2\int_0^{1/3}(1-t^2)^{-1}\,dt\ge2(1/3+1/81)>20/29$, because $(1-t^2)^{-1}\ge1+t^2$. Also $(1-d)\ln(1-d)\ge-d$: the difference plus $d$ vanishes at zero and has derivative $-\ln(1-d)\ge0$. Expanding $h_2(d)$ therefore gives
$(1-d)-h_2(d)\ge1+d\log_2d-(1+29/20)d$.
Combine this with $A\ge32/25$ and $B\le6$. The linear coefficient is $1+29/20-32/25=117/100$, so
\begin{equation}
C(d)\ge1+d\log_2d-\frac{117}{100}d-6d^2.
\label{eq:iid-quadratic}
\end{equation}
The endpoint convention is $d\log d=0$ at zero. Every remainder and constant in this bound is explicit.

%% file: sections/small_assembly.tex
\section{Assembly of the refined small-deletion intervals}\label{sec:small-assembly}
This section combines the run-based converse in \cref{sec:small-deletion} with the fair-input refinements in \cref{sec:fair-input}. Both ingredients have already been proved. The calculation preserves the original detailed enclosure near zero; the simplified worst-case proof uses the alternative cover of \cref{sec:family-simplification}.

\subsection{An upper anchor and a whole-interval lower bound}
The run-based checker proves an upper bound on capacity at one specified deletion probability. To cover an interval, we need a lower bound already valid throughout that interval and an upper bound that remains valid when deletions increase. We obtain the latter from channel degradation. This final step concerns where the proved bounds apply; it does not require another run-length approximation. Fix $[a,b]\subset[0,1/2]$. Let $L^-_{a,b}$ be a numerical lower endpoint for the exact fair-input expression in \eqref{eq:iid-cell-lower}: directed arithmetic places it at or below that expression. It already satisfies $L^-_{a,b}\le C(d)$ for every $d$ in that interval. At the left endpoint $a$, choose the length coefficient $\beta$ in the positive-row test so that $(1-\omega)B_0+\beta=L^-_{a,b}+\varepsilon$. The number $L^-_{a,b}+\varepsilon$ is the proposed upper target. Choosing $\beta$ to express that target does not prove it. The converse follows only after every finite row and the bound on all omitted run lengths pass their required inequalities, giving $C(a)\le L^-_{a,b}+\varepsilon$.

For $d\ge a$, the $d$-deletion channel can be produced by applying additional independent deletions to the $a$-deletion output. The extra deletion probability is $(d-a)/(1-a)$, since the two survival probabilities multiply to $1-d$. Data processing, \cref{lem:data-processing}, implies $C(d)\le C(a)$. The two bounds therefore give
\[
\boxed{L^-_{a,b}\le C(d)\le L^-_{a,b}+\varepsilon
\quad\text{for every }d\in[a,b].}
\]
The upper bound uses monotonicity from a single point, whereas the lower bound comes from an expression already proved valid throughout $[a,b]$. These are separate reasons for the two inequality directions in the enclosure. This is a valid whole-interval construction, not the definition of the published table. The published $L(d)$ uses right-endpoint lower anchors transported left by \cref{lem:transport}, and $U(d)$ uses left-endpoint upper anchors transported right, as specified in \eqref{eq:cell-functions}.

The first scalar construction records these checks through the function $F(d)=1+d\log_2d-(13/25)d+2d^2$. At $d_i=i/40000$, for $68\le i\le399$, it verifies $C(d_i)\le F(d_{i+1})$. The function decreases on the relevant range: $F'(d)=\log_2d+1/\ln2-13/25+4d$, and $d\le1/100<2^{-6}$ together with $1/\ln2<29/20$ makes this derivative negative. For $d_i\le d\le d_{i+1}$, degradation gives $C(d)\le C(d_i)\le F(d_{i+1})\le F(d)$. This explains its whole-interval conclusion.

The mixture parameter in that construction is $\omega_i=(19/100)\sqrt{d_i}$ and its finite run cutoff is $L_i=\lceil20/d_i\rceil$. The square root is evaluated by an enclosing interval; it is not silently treated as a rational witness parameter. The numerical interval contains the exact square root, and every row and remainder inequality is evaluated so that its endpoint bound holds for all values in that interval. In particular it holds at the intended value of $\omega_i$.

\textbf{The interval that includes the noiseless channel.}
The scalar checks start at a positive deletion probability. To include zero, we use the explicit lower bound \eqref{eq:iid-quadratic} and compare it directly with the erasure upper bound. It is sufficient to show that their separation increases up to the stated endpoint and is small there.

This gives a hand-checkable alternative enclosure near zero; it does not replace the published cells. For this comparison, subtract the lower bound \eqref{eq:iid-quadratic} from the erasure upper bound $C(d)\le1-d$ in \cref{lem:erasure}. Their difference is
$W(d)=-d\log_2d+(17/100)d+6d^2$, with $W(0)=0$.
Put $t=17/10000$. For $0<d\le t<2^{-9}$, differentiating gives
\[
W'(d)=-\log_2d-\frac1{\ln2}+\frac{17}{100}+12d
>9-\frac{29}{20}+\frac{17}{100}>0.
\]
Thus the width is largest at $t$. The exact integer comparison $10000^{32}\le17^{32}2^{295}$ implies $-\log_2t\le295/32$ after taking logarithms and dividing by $32$. Substitution yields the entirely rational enclosure-width bound
\[
W(d)\le\frac{17}{10000}\left(\frac{17}{100}+\frac{295}{32}\right)
+6\left(\frac{17}{10000}\right)^2
=\frac{3195643}{200000000}.
\]
Its midpoint error is at most half this width, namely $3195643/400000000<0.008$. The other small-deletion intervals use the same separation of obligations: the row and remainder inequalities prove the upper point values, the fair-input calculation proves the interval lower values, and the exact endpoints specify which real deletion probabilities are covered.

%% file: sections/equivalences.tex
\section{Duality identities and notation}\label{app:equivalent}
The following finite-channel identities place the converse constructions in their standard duality framework.

\subsection{Output codes and information radius}

Consider a finite channel with conditional probabilities $W(y\mid x)$, an input distribution $P_X(x)$, and its output law $P_Y(y)=\sum_xP_X(x)W(y\mid x)$. Let $Q(y)>0$ be a candidate output distribution. Relative entropy measures the expected extra description length incurred by using $Q$ in place of the true distribution. We now separate that description length into the information carried by the channel and an error due to the choice of $Q$.

For every summand of positive probability, split the logarithm by inserting the true output probability:
\[
\log_2\frac{W(y\mid x)}{Q(y)}
 =\log_2\frac{W(y\mid x)}{P_Y(y)}
  +\log_2\frac{P_Y(y)}{Q(y)}.
\]
Multiplying by $P_X(x)W(y\mid x)$ and summing gives
\begin{align*}
\E_X D\bigl(W(\cdot\mid X)\Vert Q\bigr)
 &=\sum_{x,y}P_X(x)W(y\mid x)
       \log_2\frac{W(y\mid x)}{P_Y(y)}\\
 &\quad+\sum_y\left(\sum_xP_X(x)W(y\mid x)\right)
       \log_2\frac{P_Y(y)}{Q(y)}\\
 &=I(X;Y)+D(P_Y\Vert Q).
\end{align*}
The first sum is the definition of mutual information. In the second, the inner sum is $P_Y(y)$, leaving the definition of output relative entropy. Zero-probability summands contribute zero and cause no division issue. Rearranging, using nonnegativity of relative entropy, and bounding an average by its largest term now yield
\[
I(X;Y)
 =\E_X D\bigl(W(\cdot\mid X)\Vert Q\bigr)-D(P_Y\Vert Q)
 \le\max_xD\bigl(W(\cdot\mid x)\Vert Q\bigr).
\]
Thus any fixed candidate $Q$ provides an upper bound that holds for every input distribution. This is why an output code can support a converse even when it was chosen without knowing the optimizing input law.

Minimizing the right side over $Q$ gives the finite-channel information-radius dual. For completeness, its reverse inequality follows from an optimizing input law $P_X^\star$, which exists by continuity on the finite probability simplex. Write $Q^\star$ for its output law and $C_W$ for its mutual information. Restrict the output alphabet to reachable symbols. Then $Q^\star$ is positive on this alphabet: adding a small input mass that reaches a missing output symbol increases output entropy by a term proportional to $t\log(1/t)$, while the remaining entropy terms change by at most order $t$. A missing reachable symbol would therefore contradict optimality.

For any input letter $x$, mix $P_X^\star$ with a mass $t$ on $x$. The output law becomes $Q_t=(1-t)Q^\star+tW(\cdot\mid x)$. Mutual information is output entropy minus average row entropy. Differentiating this expression at $t=0$ gives
\begin{align*}
\left.\frac{d}{dt}I(P_X^\star(1-t)+t\delta_x;W)\right|_{t=0}
 &=-\sum_y\bigl(W(y\mid x)-Q^\star(y)\bigr)\log_2Q^\star(y)\\
 &\quad-H\bigl(W(\cdot\mid x)\bigr)
       +\sum_zP_X^\star(z)H\bigl(W(\cdot\mid z)\bigr)\\
 &=D\bigl(W(\cdot\mid x)\Vert Q^\star\bigr)-C_W.
\end{align*}
The constant derivative of $-p\log_2p$ disappears because $\sum_y(W(y\mid x)-Q^\star(y))=0$. The last line follows by collecting the terms with $W(\cdot\mid x)$ into a divergence and the remaining terms into $-C_W$. Optimality makes this one-sided derivative nonpositive, so every row divergence to $Q^\star$ is at most $C_W$. Combining this with the preceding universal upper bound proves the dual characterization.

The predictor used in this paper assigns probabilities sequentially to a received word. It serves the same admissible-output-code role, while the survivor argument handles the conditional entropy of the deletion channel's variable-length output. Universal-coding and redundancy-capacity results provide the broader interpretation \citep{MF98,Ryabko79}; the elementary entropy identities are standard \citep[Sections~2.3--2.6]{CT06}.

\subsection{Finite-state potentials and average-reward duality}

A potential transfers contributions between neighboring states without changing their long-run average. To see the exact cancellation, let $\eta(w,a)$ be nonnegative state-action frequencies with total mass one. Let $P_a(w,w')$ be the probability of moving from state $w$ to state $w'$ under action $a$, and let $\rho(w)=\sum_a\eta(w,a)$ be the state frequency. An exact stationary occupation measure satisfies flow balance:
\[
\rho(w')=\sum_{w,a}\eta(w,a)P_a(w,w')
\quad\hbox{for every state }w'.
\]
The kernels are fixed independently of the unknown source. Their flow identity must be established for that source, as in \cref{lem:potential}; it is not automatic for an arbitrary proposed transition table. The right side counts arrivals to $w'$ and the left side counts visits to that state. Stationarity requires the two frequencies to agree.

For a potential $V$, its expected next value is $(P_aV)(w)=\sum_{w'}P_a(w,w')V(w')$. Expanding this expectation and applying flow balance gives
\begin{align*}
\sum_{w,a}\eta(w,a)(P_aV)(w)
 &=\sum_{w'}V(w')\sum_{w,a}\eta(w,a)P_a(w,w')\\
 &=\sum_{w'}\rho(w')V(w')
 =\sum_{w,a}\eta(w,a)V(w).
\end{align*}
All sums are finite, so rearranging their order is valid. Subtracting the last expression from the first proves
\[
\sum_{w,a}\eta(w,a)\bigl[(P_aV)(w)-V(w)\bigr]=0.
\]
If each row satisfies $r(w,a)+(P_aV)(w)-V(w)\le u$, multiply by $\eta(w,a)$ and sum. The potential terms cancel, and $\sum_{w,a}\eta(w,a)=1$, leaving the average-reward bound $\sum_{w,a}\eta(w,a)r(w,a)\le u$.

This is weak duality between stationary flow constraints and potential inequalities. The capacity proof needs the row inequalities: they control every occupation measure induced by an admissible physical source. Approximate flow balance would leave a residual potential term, and a separate relaxation's objective would not by itself be an achievable channel rate. \citet{HSPKS21} develop the direct connection between output test distributions and dynamic programming for finite-state-channel capacity. General average-reward optimization is treated in \citet{Puterman94}.

\subsection{Notation register}
\begin{longtable}{@{}p{3cm}p{9.3cm}@{}}\toprule Symbol & Meaning\\\midrule\endfirsthead
\toprule Symbol & Meaning\\\midrule\endhead
$d,1-d$ & Deletion and survival probabilities\\
$X_1^n,Y_n,N_n$ & Raw input, trace, and trace length\\
$C_n,C$ & Normalized finite-block information capacity and operational limit\\
$L,U,\widehat C$ & Lower, upper, and midpoint functions from the exact cell certificate\\
$\eps_\star,\Delta_\star$ & Certified uniform midpoint-error bound and twice that bound\\
$F,Z,G_i$ & Infinite raw future, survivor process, and independent geometric spacings\\
$Q,\ell_Q$ & Global normalized output predictor and its code loss\\
$\Psi_Q$ & Conditional relative-entropy perspective on nonnegative measures\\
$R,m,b$ & Raw-window depth, number of desired survivors, and retained-mask cutoff in the converse; the tail prefix has length $m-b$\\
$w,u,a,c$ & Raw window, survivor tail prefix, prepended bit, and outgoing raw bit\\
$V,t,K$ & Finite potential, entropy weight, and normalized action or selector code; its role is explicit in each row\\
$B_a,T_a,P_w$ & Prepend/drop operator, deletion mixture, and composed window operator\\
$T,\pi,M$ & Exact source transition, stationary law, and survivor transition\\
$\Phi,F_j,E_j$ & Homogeneous entropy, embedding entropy, and nonnegative incremental alignment cell\\
$H,h,j,k$ & Maximum alignment raw depth, raw-depth index, survivor count, and rectangle survivor cutoff\\
$B_m$ & Birch conditional-entropy lower bound\\
$A_i,B_i$ & Normalized upper and lower coefficients on continuum cell $i$\\
$P_\ell,\overline L,D_0$ & Renewal run-length law, its mean, and probability a complete run vanishes\\
$\lambda,\nu$ & Poisson-repeat intensity and output-run density in the renewal section; transport uses its own explicitly defined retention $\lambda$\\\bottomrule
\end{longtable}
Local symbols are introduced within their certificate families and their roles are recalled when needed. The index $h$ in alignment cells denotes a raw input depth.